%% file: main.tex
\PassOptionsToPackage{dvipsnames}{xcolor}
\documentclass{lmcs}
\pdfoutput=1

\usepackage{lastpage}
\lmcsdoi{21}{3}{6}
\lmcsheading{}{\pageref{LastPage}}{}{}%
{Jan.~23,~2024}{Jul.~23,~2025}{}

\newif\ifshowpageframe

\newif\ifInlineTypeCheck

\input{macros/lengths_rules.tex}
\input{packages.tex}

\input{macros/lmcs_qol.tex}

\newif\ifdraftmode
\newif\ifshownotes
\newif\ifshowAppendixTOC

\ifdraftmode    
\ifshownotes
\else\fi\else\fi

\newif\ifshowdepends

\newif\ifshowUnusedLemmas
\long\def\IsUnused#1{\ifshowUnusedLemmas#1\else\fi}

\usepackage{layout}

\NewDocumentCommand{\highlight}{O{red} m}{{\color{#1}#2}}

\input{macros/notes.tex}

\NewDocumentCommand{\todo}{s m O{blue}}{\IfBooleanTF{#1}{\note*[#3]{#2}[ToDo]}{\note[#3]{#2}}\NoteContentsLine[#3]{#2}{ToDo}}

\NewDocumentCommand{\TODO}{s m O{red}}{\IfBooleanTF{#1}{\note*[#3]{#2}[ToDo]}{\note[#3]{#2}}\NoteContentsLine[#3]{#2}{Urgent}}

\NewDocumentCommand{\checkme}{s m O{magenta}}{\IfBooleanTF{#1}{\note*[#3]{#2}[Check me]}{\note[#3]{#2}}\NoteContentsLine[#3]{#2}{Check me}}

\NewDocumentCommand{\anno}{m O{cyan}}{\note[#2]{#1}}

\NewDocumentCommand{\updateme}{m O{green}}{\note[#2]{update: #1}\NoteContentsLine[#2]{#1}{Update}}

\NewDocumentCommand{\jnote}{s m O{orange}}{\IfBooleanTF{#1}{\note*[#3]{#2}[Jonah]\NoteContentsLine[#3]{#2}{Jonah note}}{\note[#3]{#2}[\highlight[blue]{J.P}]}\NoteContentsLine[#3]{#2}{Jonah note}}

\NewDocumentCommand{\lnote}{s m O{orange}}{\IfBooleanTF{#1}{\note*[#3]{#2}[Laura]\NoteContentsLine[#3]{#2}{Laura note}}{\note[#3]{#2}[\highlight[blue]{L.B}]}\NoteContentsLine[#3]{#2}{Laura note}}

\NewDocumentCommand{\anote}{s m O{orange}}{\IfBooleanTF{#1}{\note*[#3]{#2}[Andy]\NoteContentsLine[#3]{#2}{Andy note}}{\note[#3]{#2}[\highlight[blue]{A.K}]}\NoteContentsLine[#3]{#2}{Andy note}}

\NewDocumentCommand{\mnote}{s m O{orange}}{\IfBooleanTF{#1}{\note*[#3]{#2}[Maurizio]\NoteContentsLine[#3]{#2}{Maurizio note}}{\note[#3]{#2}[\highlight[blue]{M.M}]}\NoteContentsLine[#3]{#2}{Maurizio note}}

\ifshowAppendixTOC\ifshownotes
    \AtEndDocument{%
        \clearpage\par\noindent \highlight{(this table of contents will be removed automatically when i turn my notes macro off - Jonah.)}%
        \hypersetup{linkcolor=black}\addtocontents{toc}{\protect\hypertarget{toc}{}}%
        \tableofcontents%
        }
\else\fi\else\fi

\definecolor{lightblue}{RGB}{179, 246, 255}
\definecolor{lightpink}{RGB}{255, 181, 230}

\definecolor{codepurple}{rgb}{0.58,0,0.82}

\newcommand{\REVISEDII}[1]{{\color{codepurple}{#1}}}

\ifshownotes
\fi

\input{macros/global.tex}

\input{macros/tools.tex}

\input{macros/symbols.tex}

\input{macros/plaintext.tex}

\input{macros/environments.tex}

\usepackage{environ}
\input{macros/environment_tweaks/proof.tex}

\input{macros/environment_tweaks/exa.tex}

\input{macros/environment_tweaks/restatable_equation.tex}

\input{macros/cref_config.tex}

\input{macros/citations.tex}

\input{macros/theory/general.tex}

\input{macros/theory/sessiontypes.tex}

\input{macros/theory/processcalculus.tex}

\input{macros/theory/actions.tex}

\input{macros/theory/type_checking.tex}

\def\middot{\textperiodcentered~}
\keywords{Session types \middot Mixed-choice \middot Timeouts \middot \texorpdfstring{$\pi$}{pi}{-}calculus}

\makeatletter
\renewcommand{\subsubsection}{%
  \@startsection{subsubsection}{3}%
  \z@{0pt}{-.5em}%
  {\normalfont\itshape}%
}
\makeatother

\makeatletter
\renewcommand{\subsection}{%
  \@startsection{subsection}{2}%
  \z@{0pt}{-.5em}%
  {\normalfont\bfseries}%
}
\makeatother

\makeatletter
\let\ltxxlabel\ltx@label
\makeatother

\begin{document}

\title[Timeout Asynchronous Session Types]{%
  Timeout Asynchronous Session Types:\texorpdfstring{\\}{\ }
  Safe Asynchronous Mixed-Choice for Timed Interactions}

\titlecomment{This work is an extended version of~\cite{Pears2023} that was accepted at {COORDINATION} 2023.}

\thanks{This work has been partially supported by EPSRC project EP/T014512/1 (STARDUST) and the BehAPI project funded by the EU H2020 RISE under the Marie Sklodowska-Curie action (No: 778233).
  We thank Simon Thompson
  for his insightful comments on an early version of this work.}

\author[J.~Pears]{Jonah Pears\lmcsorcid{0000-0003-4492-4072}}[a]
\author[L.~Bocchi]{Laura Bocchi\lmcsorcid{0000-0002-7177-9395}}[a]
\author[M.~Murgia]{Maurizio Murgia\lmcsorcid{0000-0001-7613-621X}}[b]
\author[A.~King]{Andy King\lmcsorcid{0000-0001-5806-4822}}[a]

\address{University of Kent, Canterbury, UK}
\email{%
  j.pears@kent.ac.uk,~%
  l.bocchi@kent.ac.uk,~%
  a.m.king@kent.ac.uk%
}

\address{Gran Sasso Science Institute, L'Aquila, Italy}
\email{%
  maurizio.murgia@gssi.it%
}

\def\middot{\textperiodcentered~}
\begin{abstract}
  Mixed-choice has long been barred from models of asynchronous communication since it compromises the decidability of key properties of communicating finite-state machines.
  Session types inherit this restriction, which precludes them from fully modelling timeouts -- a core property of web and cloud services.
  To address this deficiency, we present (binary) Timeout Asynchronous Session Types ({TOAST}) as an extension to (binary) asynchronous timed session types, that permits mixed-choice.
  {TOAST} deploys timing constraints to regulate the use of mixed-choice so as to preserve communication safety.
  %
  %
  We provide a new behavioural semantics for {TOAST} which guarantees progress in the presence of mixed-choice.
  Building upon {TOAST}, we provide a calculus featuring process timers which is capable of modelling timeouts using a $\mathtt{receive\text{-}after}$ pattern, much like Erlang,
  and capture the correspondence with TOAST specifications via a type system for which we prove subject reduction.

\end{abstract}

\maketitle

\section{Introduction}\label{sec:introduction}

\LoadTypeChecking

Mixed-choice is an inherent feature of models of communications such as communicating finite-state machines (CFSM)~\cite{Brand1983:CFSM} where actions are classified as either send or receive. In this setting, a state of a machine is said to be mixed if there exist both a sending action and a receiving action from that state.
When considering an asynchronous model of communication, deadlock freedom is undecidable in general~\cite{Gouda1984} but can be guaranteed in presence of three sufficient and decidable conditions: determinism, compatibility, and \emph{absence} of mixed-states~\cite{Gouda1984,Denielou2013}. Intuitively, determinism means that it is not possible, from a state, to reach two different states with the same kind of action, and compatibility requires that for each send action of one machine, the rest of the system can eventually perform a complementary receive action.

In the pursuit of deadlock freedom mixed-choice has been cast out,
even though this curtails
the
descriptive capabilities
of CFSM and its derivatives.
Despite the rapid evolution of session types,
even to the point of deployment in
Java~\cite{Hu2008},
Python~\cite{Neykova2013:spy,Neykova2013b:go_dynamic},
Rust~\cite{Lagaillardie2020},
{F\#}~\cite{Neykova2018:fsharp}
and {Go}~\cite{Castro2019},
thus far mixed-choice has only been introduced into the
synchronous binary setting~\cite{Vasconcelos2020}.
In fact, the exclusion of mixed-choice pervades work on asynchronous communication which guarantee of deadlock freedom,
both for
communicating timed automata~\cite{Bocchi2015,Krcal2006:CTA}
and session types~\cite{Bettini2008,Carbone2008,Honda2008,Yoshida2007}.
Determinism and the absence of mixed-choice is baked into the very syntax of session types
(the correspondence between session types and
CFSM is explained in~\cite{Denielou2013}).

Timed session types~\cite{Bartoletti2014,Bocchi2019,Bocchi2014}, which extend session types with time constraints, inherit the same syntactic restrictions of session types, and hence also rule out mixed-choice.
This is unfortunate since in the timed setting, mixed-choice are a useful abstraction for timeouts.
To illustrate, \Cref{fig:timeout_snippets} shows how
an Erlang~\cite{ErlangOTP} \recvafter\ statement (\ref{fig:timeout_snippets_erlang})
can be rendered as a mixed-state CFSM (\ref{fig:timeout_snippets_automata}):
if neither a \erl{data} nor a \erl{done} message are received within 3 time units from co-party \erl{Q},
then a \erl{timeout} message is sent.

Timeouts are important for handling failure and unexpected delays, for instance, the {{SMTP}} protocol stipulates:
\emph{``An {{SMTP}} client \emph{must} provide a timeout mechanism''}~\cite[Section~4.5.3.2]{RFC5321}.
Mixed-states would allow, for example, states where a server is waiting to \emph{receive} a message from the client and, if nothing is received after a certain amount of time, \emph{send} a notification that ends the session.
Current variants of timed session types allow deadlines to be expressed but cannot, because of the absence of mixed-states, characterise (and verify) the behaviour that should follow a missed deadline, e.g., a restart or retry strategy.
In this paper, we argue that time makes mixed-states more powerful (allowing timeouts to be expressed), while just adding sufficient synchronisation to ensure that mixed-states are safe in an asynchronous semantics (cannot produce deadlocks).

\begin{figure}[t]
  \centering
  \hfill
  \begin{subfigure}{0.4\textwidth}
    \begin{lstlisting}[style=Erlang]
receive 
    {Q, data} -> foo();
    {Q, done} -> bar();
after 3000  -> 
    Q ! {self(), timeout},
    handle_timeout()
end.
\end{lstlisting}
    \caption{Erlang snippet}
    \label{fig:timeout_snippets_erlang}
  \end{subfigure}
  \hfill
  \begin{subfigure}{0.4\textwidth}
    \begin{tikzpicture}[shorten >=1pt,node distance=1.5cm,auto]
      \tikzstyle{every state}=[fill={rgb:black,1;white,10}]

      \node[state] (0) at (0,0) {$ra$};
      \node[state] (1) at (4,1)  {$f$};
      \node[state] (2) at (4,0)  {$b$};
      \node[state] (3) at (4,-1) {$t$};

      \path[->] (0) edge[bend left=25]  node[above, midway] {?data $(x<3)$} (1);

      \path[->] (0) edge  node[above] {?done $(x<3)$} (2);

      \path[->] (0) edge[bend right=25] node[below, midway] {!timeout $(x \geq 3)$} (3);
    \end{tikzpicture}
    \caption{Mixed-state CFSM}
    \label{fig:timeout_snippets_automata}
  \end{subfigure}
  \hfill
  \caption[]{An Erlang snippet and its mixed-state machine representation}
  \label{fig:timeout_snippets}
\end{figure}

\paragraph{Contributions}
This work makes three orthogonal contributions to the theory of (binary) session types, with a focus on improving their descriptive capabilities:
\begin{itemize}
  \item We introduce Timeout Asynchronous Session Types (TOAST) to support timeouts.
        Inspired by asynchronous timed (binary) session types~\cite{Bocchi2019}, TOAST shows how timing constraints provide an elegant solution for guaranteeing the safety of mixed-choice. Technically, we provide a semantics for TOAST and a well-formedness condition. We show that well-formedness is sufficient to guarantee progress for TOAST (which may, instead, get stuck in general).
  \item We provide a new process calculus whose functionality extends to support programming motifs such as the widely used \recvafter\ pattern of Erlang for expressing timeouts.
  \item We introduce timers in our process calculus to structure the counterpart of a timeout, as well as time-sensitive conditional statements, where the selection of a branch may be determined by timers. Time-sensitive conditional statements provide processes with knowledge as to which branch should be followed e.g., in the case of timeout.
  \item We provide an informal discussion on the correspondence between TOAST and the aforementioned primitives of our new process calculus.
  \item We show a formal correspondence between TOAST and process calculus via a typing system, for which we establish Subject Reduction.
\end{itemize}

\noindent For simplicity, we focus on binary sessions.

\paragraph{Outline}
In~\Cref{sec:toast} we begin by introducing Timeout Asynchronous Session Types and present their syntax, semantics,
and discuss mixed-choice and other motivations behind our well-formedness conditions.
We end~\Cref{sec:toast} by discussing our proof of progress for TOAST.

In~\Cref{sec:toast_processes} we present a calculus for processes with timeouts, which are designed to correspond to TOAST, and discuss their syntax, structure and reduction rules.
\Cref{sec:expressiveness} discusses the expressiveness of our types and processes, and provides some examples of the kinds of programs that TOAST is capable of modelling, and vice versa, how to express certain behaviour in TOAST and our process calculus.

We present our typing system in~\Cref{sec:type_checking}.
We first introduce the typing judgements, then discuss each rule in detail and conclude with an example evaluation
using our typing system.
In~\Cref{sec:subject_reduction} we present subject reduction for our typing system.
We conclude our work in~\Cref{sec:related_work,sec:conclusion}, where we discuss related work and summarise our contribution.

\UnloadTypeChecking

\section{Timeout Asynchronous Session Types (TOAST)}\label{sec:toast}
This section presents the syntax, semantics and formation rules for (binary) Timeout Asynchronous Session Types (TOAST),
which are an extension of
(binary) Asynchronous Timed Session Types~\cite{Bocchi2019} with a well-disciplined
(and hence safe)
form of mixed-choice.
%
%

\LoadSType*

\paragraph*{Clocks \& Constraints}
We start with a few preliminary definitions borrowed from timed automata~\cite{Alur1994:TA}. Let \SetOfClocks\ be a finite set of clocks denoted $x$, $y$ and $z$.
A (clock) valuation \Val\ is a map $\Val: \SetOfClocks\to\SetRealZ$.
The initial valuation is \Val_0, where
\(
\Val_0=\{
x\mapsto0 \mid x\in\SetOfClocks
\}
\).
Given a time offset
$\t\in\SetRealZ$
and a valuation \Val,
\(
\Val+\t=\{
x\mapsto\Val(x)+\t \mid x\in\SetOfClocks
\}
\).
Given \Val\ and $\Resets\subseteq\SetOfClocks$,
\(
\Val\Reset0=
\{
0\ \text{if\ }(x\in\Resets)%
\text{\ else\ }\Val(x)\mid x\in\SetOfClocks%
\}
\).
Observe $\Val\Reset{\emptyset}=\Val$.
%
%

A clock constraint $\delta$ takes the following form:
%
%
%
%
\begin{requation}{eqDeltaConstraints}\label{eq:delta_constraints}
  \Constraints ::= \True
  \mid {x>n}
  \mid {x=n}
  \mid {x-y>n}
  \mid {x-y=n}
  \mid {\neg\Constraints}
  \mid {\Constraints_1\land\Constraints_2}
  \qquad (\text{where\ }n\in
  \mathbb{Q}_{\geq0}
  )
\end{requation}

\noindent 
  We write $\Val\models\delta$ to denote that the clock valuations in $\nu$ satisfy the constraints $\delta$.
  Defined formally, 
  \(
  (\Val\models\delta)=
    (\forall
    x\in
    \CN{\delta}
    .\delta\Subst{\nu(x)}{x}
    )
  \)
  %
  to say that, for all clocks in the
  \emph{clock names}\footnote{%
  Given in~\Cref{def:clock_names}.
  }%
  of $\delta$
  we substitute the clock name with the corresponding valuation in $\nu$.
  E.g., 
  $\nu[x\mapsto 3]\models(x>2)$ holds since $(x>2)\Subst{\nu(x)}{x}=(3>2)$.
%
%
%
  We write $\Past$ to denote the \emph{weak past} of 
  $\delta$.

\begin{defi}[Weak Past]\label{def:past_delta}
  We write $\Past$ (the past of $\Constraints$)
  for a constraint $\Constraints'$
  such that $\Val\models\Constraints'$,
  if and only if $\exists \t$
  such that $\Val+\t\models\Constraints$.
\end{defi}

\noindent 
  By~\Cref{def:past_delta}, $\Past$ effectively removes any lower-bounds in $\delta$, preserving the upper-bounds.
  For example,
  $\downarrow(3<x<5)=(x<5)$
  and $\downarrow(x>2)=(\True)$.
  %
  In practice
  $\Past$ allows us to reason on constraints being satisfiable at some point in the future,
  since we only require that the clocks valuation does not exceed the upper-bound of the constraint.

\UnloadSType*

\subsection{Syntax of TOAST}\label{ssec:type_syntax}
\LoadSType*
The syntax of TOAST (or just \emph{types}) is given in~\Cref{eq:type_syntax}.
A type \S\ is either a choice \C_i,
recursive definition ${\mu\alpha.\S}$,
call $\alpha$,
or termination type \End.
\begin{requation}{eqTypeSyntax}\label{eq:type_syntax}
  \begin{array}[c]{lcl @{\qquad} lcl}
    \S & ::=   & \C_i \mid {\mu\alpha.\S} \mid \alpha \mid \End
       & \C*   & ::=                                            & \Choice*{\l}[\T]
    \\[-1ex]\\%
    \T & ::=   & \Del
    \mid \mathtt{Unit} \mid \mathtt{Nat} \mid \mathtt{Bool} \mid \mathtt{String} \mid \dots
       & \Comm & ::=                                            & \SendSym \mid \RecvSym
  \end{array}
\end{requation}

\noindent Type \C_i\ models a choice among options $i$ ranging over a non-empty set $I$. Each option $i$ is a selection/send action if $\Comm=!$, or alternatively a branching/receive action if $\Comm=?$.
An option sends (resp. receives) a label \l\ and a message of a specified data type \T\ is delineated by $\langle\cdot\rangle$.
The sending or receiving action of an option is guarded by a time constraint \Constraints. After the action, the clocks within \Resets\ are reset to 0.
Data types, ranged over by $\T, \T_i, \ldots$
can be higher-order types \Del\ to model session delegation,
$\mathtt{Unit}$ types, or \emph{base types} (e.g., $\mathtt{Nat}$, $\mathtt{Bool}$, $\mathtt{String}$). We omit a $\mathtt{Unit}$ type when the payload is nothing.
%
%
We discuss the structure of delegation types further in~\Cref{sssec:wf}.
%
Labels of the options in a choice are pairwise distinct.
Recursion and termination types are as standard~\citeST.

\paragraph{Remarks on the notation}
One convention is to model the exchange of payloads as a separated action with respect to the communication of branching labels.
In this paper we follow~\cite{Bocchi2015,Yoshida2021}, and model them as unique actions.
When irrelevant we omit the payload, yielding a notation closer to that of timed automata.
%
%

\paragraph{Feasibility \& Junk Types}
Unfortunately, when annotating session types with time constraints one may obtain protocols that are infeasible, as shown in~\Cref{exa:junk_types}.
This is a known problem, which has been addressed by providing additional conditions or constraints on timed session types, for example compliance~\cite{Bartoletti2017}, feasibility~\cite{Bocchi2014}, interaction enabling~\cite{Bocchi2015}, and well-formedness~\cite{Bocchi2019}.
(We address feasibility in~\Cref{sssec:wf}.)

\begin{exa}[Junk Types]\label{exa:junk_types}
  Consider type \S\ (defined below) that describes first waiting 3 seconds, then receiving $a$, and finally a mixed-choice between sending $b$ and receiving $c$.
  %
  \[
    \begin{array}[c]{@{}l@{\ }c@{\ }l@{}}
      \S & = & %
      \Option?{a}{x>3}.{\Choice|{
      \Option!{b}{y=2}.{\End}{,},
      \Option?{c}{2<x<5}.{\End}
      }}
    \end{array}
  \]

  %
  \noindent Initially all clocks are 0.
  After $a$ is received all clocks hold values greater than 3, 
  since time passes evenly over all clocks.
  %
    Therefore, $b$ is never able to be sent since $(y>3)$
    and $c$ is only \emph{sometimes} able to be sent
    (i.e., if $a$ is received when $(x<5)$ which constraint $(x>3)$ does nothing to guarantee).
    %
    Types with unsatisfiable constraints are called \emph{junk types}~\cite{Bocchi2019}.
    %
    %
    %
    %
  \[
    \begin{array}[c]{@{}l@{\ }c@{\ }l@{}}
      \S' & = & %
      \Option?{a}{x>3}[x].{\Choice|{
      \Option!{b}{y=2}.{\End}{,},
      \Option?{c}{2<x<5}.{\End}
      }}
    \end{array}
  \]

  \noindent 
    Above we show an amended $\S$ which now resets clock $x$ after receiving $a$ in order to guarantee that receiving $c$ is always feasible.
    By further amending $\S'$ such that clock $y$ is also reset after receiving $a$ then sending $b$ would also become feasible.
  %
\end{exa}

\subsection{Semantics of TOAST}\label{ssec:semantics}
We define the semantics
using three layers:
configurations 
{(\Cfg*)},
configurations with queues $\Que$ that model asynchronous interactions 
{(\Cfg*;)},
and systems that model the parallel composition of configurations with queues.
The semantics, for all layers, are defined over the transition labels $\ell$,
formally defined along with $\Que$, \Cfg*\ and \Cfg*;\  in~\Cref{eq:lts_labels}. 
\begin{requation}{eqTypeSemanticsLabels}\label{eq:lts_labels}
  %
  \ell  ::=  {\Comm\Msg}\mid\t\mid\tau%
  %
  \qquad%
  %
  \Msg  ::=  {\l\left\langle{\T}\right\rangle}%
  %
  \qquad%
  \Que  ::=  {\emptyset \mid \Msg;\,\Que}%
  \qquad
  \Cfg*  ::=  \Cfg
  \qquad
  \Cfg*;{\Que}  ::=  \Cfg;{\Que}
\end{requation}

\noindent 
  where
  transition label
  $\ell$
  is
  either
  an interaction $(\Comm\Msg)$,
  time-step ($t\in\SetRealZ$)
  or silent action $(\tau)$,
  and $\Msg$ is a message and $\Que$ is a FIFO message queue.
  Communication directions
  $\Comm$
  are the same as in~\Cref{eq:type_syntax},
  where $!$ denotes a send/output action
  and $?$ a receive/input action.
%


\UnloadSType*
\input{figs/type_semantics.tex}

\LoadSType*

\subsubsection{Configurations}\label{sssec:cfg_tuple}
A configuration \Cfg*\ is a pair \Cfg.
The semantics for configurations are defined by a Labelled Transition System (LTS) over configurations, the labels in~\Cref{eq:lts_labels}
and the rules given
in~\Cref{eq:type_semantics_tuple}.
A transition $\Cfg\Trans:{\t,\Comm\Msg}\Cfg'$ indicates $\Cfg\Trans:{\t}\Cfg*'{2}\Trans:{\Comm\Msg}\Cfg'$, where \Cfg*'{2}\ is some intermediate configuration.
%
%
Changes from~\cite{Bocchi2019} are \tcbox[on line, size=title, boxsep=0mm, left=2pt, right=2pt, colback=\tcolorChanged, colframe=black!0]{highlighted}.
%
%
\begin{requation}{eqTypeSemanticsTuple}\label{eq:type_semantics_tuple}
  %
  \begin{array}[c]{@{}c@{}}
    \tcboxmath[size=title,boxsep=0mm,colback=\tcolorChanged,colframe=black!0]{\begin{array}[c]{@{}c@{}}\TypeSemanticsAct\end{array}}
    \\[-1ex]\\%
    \begin{array}[c]{@{}c @{\qquad} c}
      \begin{array}[c]{@{}c@{}}\TypeSemanticsUnfold\end{array}
       & %
      \begin{array}[c]{@{}c@{}}\TypeSemanticsTick\end{array}
    \end{array}
  \end{array}
\end{requation}

\noindent 
  By rule $\Rule{act}$ a configuration may perform
  one action $j\in I$, provided the corresponding $\delta_j$ is satisfied by the current valuation of clocks $\Val$.
  The resulting configuration has all clocks in $\Resets_j$ reset to 0.
%
%
Rule \Rule{tick} describes time passing.
Rule \Rule{unfold} unfolds recursive types.

\begin{rdefi}{defFutureEnabledConfigurations}[Future-enabled Configurations]\label{def:cfg_fe}
  \Cfg\ is \fe,
  written either as $\Cfg\fe=$ or $\Cfg\text{\ is\ }\mathtt{fe}$,
  iff $\exists \t$
  such that $\Cfg \Trans:{\t,\Comm\Msg}$ via the rules
  in~\Cref{eq:type_semantics_tuple}.
\end{rdefi}

\subsubsection{Configurations with queues}\label{sssec:cfg_triple}
A configuration with queues \Cfg*;\ is a triple \Cfg;{\Que}, where \Que\ is a FIFO queue of messages which have been received but not yet processed.
A queue takes the form $\Que ::= \emptyset \mid \Msg;\Que$ and thus is either empty, or has a message at its head.
The semantics of configurations with queues is defined by an LTS over the labels in~\Cref{eq:lts_labels} and the rules
in~\Cref{eq:type_semantics_triple}.
The transition $\Cfg*;{\Que}\Trans:{\t,\Comm\Msg}\Cfg*;{\Que}'$ is defined analogously to $\Cfg\Trans:{\t,\Comm\Msg}\Cfg'$.
\begin{requation}{eqTypeSemanticsTriple}\label{eq:type_semantics_triple}
  %
  \begin{array}[c]{@{}c@{}}
    \begin{array}[c]{c @{\qquad} c}
      \begin{array}[c]{@{}c@{}}\TypeSemanticsSend\end{array}
       & %
      \begin{array}[c]{@{}c@{}}\TypeSemanticsRecv\end{array}
    \end{array}
    \\[-1ex]\\%
    \begin{array}[c]{@{}c@{}}\TypeSemanticsQue\end{array}
    \\[-1ex]\\%
    \tcboxmath[size=title,boxsep=0mm,colback=\tcolorChanged,colframe=black!0]{\begin{array}[c]{@{}c@{}}\TypeSemanticsTime\end{array}}
  \end{array}
\end{requation}

\noindent
Rule \Rule{send} emits a message \Msg.
Message reception is handled by two rules: rule \Rule{que} inserts a message at the back of \Que, and rule \Rule{recv} removes a message from the front of \Que.
%


Rule \Rule{time} is for time passing which is formulated in terms of a future-enabled configuration, given in~\Cref{def:cfg_fe}.
The second condition in the premise of rule \Rule{time} $\text{(persistency)}$ ensures the ``latest-enabled'' action is never missed by advancing the clocks.
Notably, this differs from the corresponding behaviour in~\cite{Bocchi2019} which, in the presence of mixed-choice, is too restrictive.
By not allowing a participant to pass its last sending action, the latest-enabled receiving actions can never be reached (i.e., miss sending to receive a timeout). 
The third condition $\text{(urgency)}$ models an urgent semantics, ensuring messages are processed as they arrive.
Urgency is critical for reasoning about progress.
%
%
Later, in~\Cref{exa:weak_persistency} we further discuss the changes made to rule \Rule{time} and illustrate why they are necessary.

\subsubsection{Systems}\label{sssec:cfg_systems} Systems are the parallel composition of two (as we consider binary types) configurations with queues,
written as $\Cfg;{\Que}_1 \mid \Cfg;{\Que}_2$
or $\Cfg*;{\Que}_2 \mid \Cfg*;{\Que}_2$.
The semantics of systems is defined by an {{LTS}} over the labels in~\Cref{eq:lts_labels} and the transition rules
in~\Cref{eq:type_semantics_system}.
%
\begin{requation}{eqTypeSemanticsSystem}\label{eq:type_semantics_system}
  %
  \begin{array}[c]{@{}c @{\quad} c@{\quad} c@{}}
    \begin{array}[c]{@{}c@{}}
      \begin{array}[c]{@{}c@{}}\TypeSemanticsComL\end{array}
      \\[-1ex]\\%
      \begin{array}[c]{@{}c@{}}\TypeSemanticsComR\end{array}
    \end{array}
     &
    \begin{array}[c]{@{}c@{}}
      \begin{array}[c]{@{}c@{}}\TypeSemanticsParL\end{array}
      \\[-1ex]\\%
      \begin{array}[c]{@{}c@{}}\TypeSemanticsParR\end{array}
    \end{array}
     &
    \begin{array}[c]{@{}c@{}}\TypeSemanticsWait\end{array}
  \end{array}
\end{requation}

\noindent 
  Rule \Rule{com}[l]* handles
  $\Cfg*;_1$ asynchronously sending $\Msg$ to the queue of $\Cfg*;_2$;
  rule $\Rule{com}[r]*$ is symmetric.
%
%
Rule \Rule{par}[l]* allows \Cfg*;_1\ to process the message at the head of \Que_1\ via \Rule{recv};
  rule $\Rule{par}[r]*$ is symmetric.
%
By rule \Rule{wait} time passes consistently across systems.

\begin{exa}[Unsafe Mixed-choice]\label{exa:mixed_bad}
  The use of mixed-choice in asynchronous communications may result in infeasible protocols or, more concretely, systems (or types) that get stuck.
  A mixed-choice is considered \emph{unsafe} if actions of different directions compete to be performed (i.e., they are both viable at the same point in time).
  Consider the system $\Cfg{\Val_0},{\S_1}; \mid \Cfg{\Val_0},{\S_2};$, where \S_1\ and \S_2\ are \dual\ and defined as follows:
  \begin{equation}\label{eq:mixed_bad:cfgs}
    \begin{array}[c]{@{}c p{4ex} c@{}}
      \begin{array}[c]{@{}l@{\ }c@{\ }l@{}}
        \S_1 & = &
        \Choice|{
        \Option?{a}{x<5}.{\End},
        \Option!{b}{x=0}.{\S_1'}
        }
      \end{array}
       &  & %
      \begin{array}[c]{@{}l@{\ }c@{\ }l@{}}
        \S_2 & = & %
        \Choice|{
        \Option!{a}{y<5}.{\End},
        \Option?{b}{y=0}.{\S_2'}
        }
      \end{array}
    \end{array}
  \end{equation}

  \noindent In the system $\Cfg*;_1 \mid \Cfg*;_2$,
  it is possible for both $\Cfg{\Val_0},{\S_1}; \Trans:{!b} \Cfg*;_1'$
  and $\Cfg{\Val_0},{\S_2}; \Trans:{!a} \Cfg*;_2'$ to occur at the same time.
  The resulting system $\Cfg{\Val_0},{\S_1'};{a} \mid \Cfg{\Val_0},{\End};{b}$
  is unable to receive either of message $a$ or $b$,
  and \Cfg*;_1'\ may be stuck waiting for interactions
  from \Cfg*;_2'\ indefinitely.
  %
  Even though a deadlock can be avoided if message $a$ is a delegation and \S_1'\ prescribes receiving $a$, since message $b$ is never received by \Cfg*;_2, then~\Cref{eq:mixed_bad:cfgs} still prescribes unsafe behaviour.
  %
\end{exa}

\subsection{Duality, Well-formedness, and Progress (of TOAST)}\label{ssec:dual:wf:progress}

In the untimed scenario, the composition of a well-formed binary type with its dual characterises a \emph{protocol}, which specifies the ``correct'' set of interactions between a party and its co-party.
The dual of a type,
formally defined below,
is obtained by swapping the directions ($!$ or $?$) of each interaction:
\begin{rdefi}{defTypeDuality}[Type Duality]\label{def:type_duality}
  Given a type \S, we define its \emph{dual} type \D\ as follows:
  \[
    \begin{array}[c]{@{}c@{}}
      \begin{array}[c]{@{}c @{\qquad} c @{\qquad} c @{\qquad} c @{\qquad} c@{}}
        \begin{array}[c]{@{}c@{}}
          \Dual{\End}=\End
        \end{array}
         &
        \begin{array}[c]{@{}c@{}}
          \Dual{\alpha}=\alpha
        \end{array}
         &
        \begin{array}[c]{@{}c@{}}
          \Dual{\mu\alpha.\S}=\mu\alpha.\D
        \end{array}
         &
        \begin{array}[c]{@{}c@{}}
          \Dual{?}=!
        \end{array}
         &
        \begin{array}[c]{@{}c@{}}
          \Dual{!}=?
        \end{array}
      \end{array}
      \\[-1ex]\\
      \begin{array}[c]{@{}c@{}}
        \Dual{\Choice{\l}[\T].\S_i}=\DualChoice_i[i]
      \end{array}
    \end{array}
  \]
\end{rdefi}

%
%
\noindent 
Notice that the time constraints on actions are identical for a type and its dual. 
For example, $\Cfg{\Val_0},{\Option!{a}{x=3}.{\S'}} \mid \Cfg{\Val_0},{\Option?{a}{x=3}.{\D'}}$ is clearly a system with dual types and identical clock valuations.
It is crucial to remember that both the clock constraints and valuations of each configuration are on their own local set of clocks, and that `$x$' is not a clock shared by each. For this reason, in our examples we often use different clock names for dual types. 
E.g., $\Cfg{\Val_0},{\Option!{a}{x=3}.{\S'}} \mid \Cfg{\Val_0},{\Option?{a}{y=3}.{\D'}}$.
This is purely a matter of presentation. 
%
%
%
%

Later, in~\Cref{rem:toast:delegation} we shall discuss 
the reason for our
`simple' definition of duality
in regard to recursive types, as opposed to~\cite{Bono2012,Lindley2016}.

\input{figs/type_wellformedness_rules.tex}

\subsubsection{Well-formedness}\label{sssec:wf}
%
We require a notion of \emph{well-formedness} so that we may later provide guarantees of progress for types with mixed-choice. 
%
%
The 
formation rules for types are given in~\Cref{fig:type_wf_rules}.
Again, since we extend~\cite{Bocchi2019} any rules that differ
are
clearly \tcbox[on line, size=title, boxsep=0mm, left=2pt, right=2pt, colback=\tcolorChanged, colframe=black!0]{highlighted}.
%
Types are evaluated against judgements of the form: $\RecTypeEnv;\Constraints\Entails\S$ where
\RecTypeEnv\ is an environment containing recursive variables 
and \Constraints\ is a constraint over all clocks characterising the times in which state \S\ can be  reached.
%
Formally $\RecTypeEnv::=\emptyset\mid\alpha:\Constraints,\RecTypeEnv$.
%
%
%
%
%
%
\begin{requation}{eqTypeWfRules}\label{fig:type_wf_rules}
  \begin{array}[c]{@{}c@{}}
    \begin{array}[c]{@{}c@{}}
      \tcboxmath[size=title,boxsep=0mm,colback=\tcolorChanged,colframe=black!0]{\begin{array}[c]{@{}c@{}}\TypeWFChoice\end{array}}
    \end{array}
    \\\\[-1ex]
    \begin{array}[c]{@{}c@{}}
      \begin{array}[c]{@{}c@{}}
        \TypeWFEnd
      \end{array}
      \qquad
      \begin{array}[c]{@{}c@{}}
        \TypeWFRec
      \end{array}
      \qquad
      \begin{array}[c]{@{}c@{}}
        \TypeWFVar
      \end{array}
    \end{array}
  \end{array}
\end{requation}

\noindent Rule \Rule{choice} checks well-formedness of choices with three conditions:
the first and third conditions are from the branching and delegation rules in~\cite{Bocchi2019}, respectively,
and 
the second condition is new and critical to ensure progress of mixed-choice,
%
By the first condition $\text{(feasibility)}$ a choice is well-formed with respect to the weakest past among all options ($\Past{\bigvee_{i\in I}\Constraints}_i$) given that each continuation \S_i\ is well-formed
  with respect to an environment $\gamma_i$, which models the corresponding guard updated with resets $\lambda_i$.
  By abuse of notation we write
  $\delta[\lambda\mapsto0]\models\gamma$
  to say $\forall \Val:(\Val\models\delta[\lambda\mapsto0]) \implies (\Val\models\gamma)$, where ($\implies$) is an implication.
This ensures that in every choice there is at least one viable action and, for example, would rule out~\Cref{exa:junk_types}. 
The second condition $\text{(mixed-choice)}$ requires all actions that can occur at the same time to have the same (send/receive) direction.
This condition allows for types modelling timeouts
(discussed later in~\Cref{exa:weak_persistency})
and rules out scenarios as the one in~\Cref{exa:mixed_bad}.
The third condition $\text{(delegation)}$ checks for well-formedness of each delegated session with respect to their corresponding initialisation constraint \Constraints'.
  Similar to the $\text{(feasibility)}$ premise, by abuse of notation we write $\delta'\models\gamma'$ to say $\forall\Val:(\Val\models\delta')\implies(\Val\models\gamma')$.
%
%
%
Rule \Rule{end} ensures termination types are \emph{always} well-formed.
Rule $\Rule{rec}$ associates variable $\alpha$ with
invariant $\delta$ in $A$.
Rule \Rule{var} ensures recursive calls are defined.

\begin{rdefi}{defWellformedConfigurations}[Well-formed Configurations]\label{def:cfg_wf}
  Given {\Cfg}, \S\ is \wf\ against \Val\ if,
  $\exists \Constraints$ such that $\emptyset;\Constraints\Entails\S$
  and $\Val\models\Constraints$.
  A type \S\ is \wf\ if it is \wf\ against \Val_0.
\end{rdefi}

\noindent The rules in~\Cref{fig:type_wf_rules} check that in every reachable state (which includes every possible clock valuation) it is possible to perform the next action immediately or at some point in the future, unless the state is final.
(This is formalised as the progress property in~\Cref{def:type_progress}.)
By these rules,
the type in~\Cref{exa:junk_types,exa:mixed_bad}
would not be well-formed.

\begin{rem}[Delegation of Higher-Order TOAST]\label{rem:toast:delegation}
    Recall~\Cref{eq:type_syntax} in~\Cref{ssec:type_syntax},
    where higher-order types for session delegation are defined as $(\delta, S)$.
    By condition $(\text{delegation})$ in the premise of rule $\Rule{choice}$
    in~\Cref{fig:type_wf_rules},
    any interaction with a delegation payload $T=(\delta',S')$,
    delegation type $S'$ must be \emph{well-formed}
    by $\emptyset;\gamma'\Entails\S'$ where environment $\gamma'$ corresponds to $\delta'$ and the environment for recursive variables
    is empty.
    In short, a
    \emph{well-formed}
    TOAST cannot
    delegate a recursive variable ($\alpha$).
    Furthermore,
    higher-order types
    must be
    \emph{flattened}
    (i.e., finitely represented) and
    TOAST is unable to express
    self-referential recursive higher-order types
    (e.g.,
    $S'=\mu\alpha.!l\left\langle(\delta',S')\right\rangle(\delta,\emptyset).\alpha$).
    This contrasts
    with~\cite{Bernardi2017,Bono2012},
    which allow
    recursive variables to be delegated
    directly,
    and~\cite{Lindley2016}
    which
    discusses self-referential recursive higher-order types.
    We feel this
    is
    orthogonal to the
    primary
    focus of this work.
  \end{rem}

\subsubsection{Progress of Types}\label{sssec:types:progress}
Well-formedness, together with the
persistency and receive urgency featured in rule \Rule{time} of the semantics in~\Cref{eq:type_semantics_triple}
ensures that the composition of a well-formed type \S\ with its dual \D\ enjoys progress.
A system enjoys progress if its configurations with queues can continue communicating until reaching the end of the protocol; formally:
\begin{rdefi}{defTypeProgress}[Type Progress]\label{def:type_progress}
  A configuration with queues \Cfg;{\Que}\ is \emph{final} if $\S\utu\End$ and $\Que=\emptyset$.
  A system $\Cfg*;{\Que}_1 \mid \Cfg*;{\Que}_2$ \emph{satisfies progress} if,
  for all $\Cfg*;{\Que}_1' \mid \Cfg*;{\Que}_2'$ reachable from $\Cfg*;{\Que}_1 \mid \Cfg*;{\Que}_2$:
  \begin{enumerate}
    \item\label[cond]{cond:type_progress_final} either $\Cfg;{\Que}_1'$ and $\Cfg;{\Que}_2'$ are \emph{final};
    \item\label[cond]{cond:type_progress_live} or, there exists a $\t\in\SetRealZ$ such that $\Cfg;{\Que}_1' \mid \Cfg;{\Que}_2' \Trans:{\t,\tau}$.
  \end{enumerate}
\end{rdefi}

\noindent We give the definition of ($\utu$) in~\Cref{def:up_to_unfolding}.
We write $\S\utu\End$ if \S\ is equivalent to \End, up-to the unfolding of recursive types. (e.g., $\mu\alpha_1\dots\mu\alpha_i.\End\utu\End$.)

\begin{restatable}[Progress of Systems]{thm}{thmTypeProgress}\label{thm:type_progress}
  If \S\ is \wf\ against \Val_0,
  then:
  \newline
  \hspace*{\fill} $\Cfg{\Val_0}; \mid \Cfg{\Val_0},{\D};$ \emph{satisfies progress}. \hspace*{\fill}
\end{restatable}

\noindent The main result of this section is that for a system composed of a \emph{well-formed} \S\ and its dual
$\Cfg,{\S};{\Que}_1 \mid \Cfg,{\D};{\Que}_2$, any state reached is either \emph{final}, or allows for further
communication;
i.e., the system \emph{satisfies progress}.
%
  Progress is critical to ensure, via type-checking, that a protocol implementation does not reach deadlock and is free from communication mismatches.
  (We introduce such a system for type-checking using TOAST
  in~\Cref{sec:type_checking}.)

The main differences from~\cite{Bocchi2019} are not in the formulation of the theory
(e.g., \Cref{def:type_progress} and the statement of~\Cref{thm:type_progress} are basically unchanged)
but in proving
rule \Rule{choice}
is sufficient to ensure progress of asynchronous mixed-choice.
%
%
%
  Additionally, the proof of progress
  in~\cite{Bocchi2019}
  relies on a notion of persistency to a participant does not reach a point where there are no viable actions, and does so by requiring time-steps do not pass beyond the participants latest-enabled receiving action.
  Due to mixed-choice, it is necessary to reformulate (and relax) this condition in the semantic rule
  \Rule{time} in~\Cref{eq:type_semantics_triple}
  to require only the latest-enabled action is not missed.
(See~\Cref{exa:weak_persistency} for a discussion).

\subsubsection{Compatibility}\label{sssec:compatibility}
The proof of~\Cref{thm:type_progress} is in~\Cref{sec:proof_type_progress},
and proceeds by showing that
a property of systems
called \emph{compatibility}
is preserved by transitions;
given formally in~\Cref{def:cfg_compat}.
%
%
%
%
  As typical for binary systems, our proof of progress builds upon a notion of \emph{duality} between participants.
  Since communication is asynchronous, each party may break duality with their co-party.
  Compatibility allows the duality of participants to be broken, only if, doing so does not violate \emph{communication safety}.
  More specifically, compatibility requires that any message that arrives in a message queue
  is expected and therefore, able to be received
  and that the resulting configurations are still \emph{compatible}.
  In practice,
  compatibility allows each party to behave
  independently, regardless the state of the other party, while retaining the essence of \emph{duality}, and guaranteeing that all messages will eventually be received.%

\begin{rdefi}{defCompatibleConfigurations}[Compatibility]\label{def:cfg_compat}
  We write $\Cfg;{\Que}_1 \Compat \Cfg;{\Que}_2$
  iff \Cfg;{\Que}_1\ and \Cfg;{\Que}_2\ are \emph{compatible}.
  Given \Cfg;{\Que}_1\ and \Cfg;{\Que}_2,
  we define \emph{compatibility} as the largest relation satisfying each of the following conditions:
  \begin{enumerate}
    \item\label[cond]{cond:cfg_compat_neq} \(
          (\Que_1 = \emptyset)
          \text{\ or\ }
          (\Que_2 = \emptyset)
          \)
    \item\label[cond]{cond:cfg_compat_recv} \(
          (\Que_1=\Msg;\Que_1')
          \implies
          (\Cfg_1 \Trans:{?\Msg} \Cfg_1'
          \text{\ and\ }
          \Cfg;{\Que}_1' \Compat \Cfg;{\Que}_2)
          \)
    \item\label[cond]{cond:cfg_compat_recv_other} \(
          (\Que_2=\Msg;\Que_2')
          \implies
          (\Cfg_2 \Trans:{?\Msg} \Cfg_2'
          \text{\ and\ }
          \Cfg;{\Que}_1 \Compat \Cfg;{\Que}_2')
          \)
    \item\label[cond]{cond:cfg_compat_dual} \(
          (\Que_1=\emptyset=\Que_2)
          \implies
          (\S_1=\D_2
          \text{\ and\ }
          \Val_1=\Val_2)
          \)
  \end{enumerate}
\end{rdefi}

\noindent Intuitively, $\Cfg*;{\Que}_1$ and $\Cfg*;{\Que}_2$ are compatible, written $\Cfg*;{\Que}_1 \Compat \Cfg*;{\Que}_2$,
if:
(\ref{cond:cfg_compat_neq}) at most one of their queues is non-empty (equivalent to a half duplex automaton),
(\ref{cond:cfg_compat_recv} -- \ref{cond:cfg_compat_recv_other}) a type is always able to process any message that arrives in its queue, and
(\ref{cond:cfg_compat_dual}) if both queues are empty then $\Cfg*;{\Que}_1$ and $\Cfg*;{\Que}_2$ have dual types and same clock valuations.

\begin{exa}[Weak Persistency]\label{exa:weak_persistency}
  In language-based approaches to timed semantics~\cite{Krcal2006:CTA} time actions are always possible, even if they bring the model into a stuck state by preventing available actions. Execution traces are then filtered
  as necessary,
  removing all `bad' traces (defined on the basis of final states).
  In contrast, and to facilitate the reasoning on process behaviour, we adopt a process-based approach, that only allows for actions that characterise \emph{intended} executions of the model 
  (as in~\cite{Bartoletti2018,Bocchi2019}).
  %
  Precisely, we build on the semantics of~\cite{Bartoletti2018} for asynchronous timed automata with mixed-choice,
  where time actions are possible only if they do not disable:
  \begin{enumerate*}[label={(\alph*)},itemjoin={,},itemjoin*={,\ and\ }]
    \item\label[cond]{cond:weak_persistency:le:send}the latest-enabled sending action
    \item\label[cond]{cond:weak_persistency:le:recv}the (latest-enabled) receiving action if the queue is not empty.
  \end{enumerate*}
  This ensures that time actions preserve the viability of at least one action (\emph{weak-persistency}).
  In our scenario, constraint~\ref{cond:weak_persistency:le:send} is too strict. Consider type $\S$ and its dual $\D$ below:
  \begin{equation*}
    \begin{array}[c]{@{}c p{0.5ex} c@{}}
      \begin{array}[c]{@{}l@{\ }c@{\ }l@{}}
        \S & = &
        \Choice|{
        \Option!{data}[\String]{x<3}.{\S'},
        \Option?{timeout}{x\geq4}.{\End}
        }
      \end{array}
       &  & %
      \begin{array}[c]{@{}l@{\ }c@{\ }l@{}}
        \D & = &
        \Choice|{
        \Option?{data}[\String]{y<3}.{\D'},
        \Option!{timeout}{y\geq4}.{\End}
        }
      \end{array}
    \end{array}
  \end{equation*}

  \noindent According to~\ref{cond:weak_persistency:le:send}, it would never be possible for \S\ to take the $timeout$ branch since a time action of $\t\geq 3$ would disable the latest-enabled send ($data$). 
  This is reasonable in a setting where the behaviour of such a mixed-choice has no guarantee of a timeout being received later, such as in~\cite{Bartoletti2018}.
  However, this is not the case in our setting as we can rely on duality 
  to guarantee that \D\ will send $timeout$ in all executions where $y\geq4$.
  %
  Our new \Rule{time} rule -- condition $\text{(persistency)}$ -- implements a more general constraint 
  than~\ref{cond:weak_persistency:le:send}, 
  requiring that one latest-enabled (send or receive) action is preserved.
  Constraint~\ref{cond:weak_persistency:le:recv} remains as condition $\text{(urgency)}$ and, for instance, prevents \D\ from sending $timeout$ if 
  $data$
  is waiting in the queue when $y<3$.
  %
  For example,
  our semantics permit the following sequence of transitions:
  \[
    \Cfg{\Val_0},{\S};{\emptyset}\Trans:{\t=2}
    \Cfg{\Val [x\mapsto 2]},{\S};{\emptyset}\Trans:{!data\langle\String\rangle}
    \Cfg{\Val [x\mapsto 2]},{\S'};{\emptyset}
  \]

  \noindent where $data$ is sent after 2 time units.
  Alternatively,
  by our semantics, the following holds:
  \[
    \Cfg{\Val_0},{\S};{\emptyset}\Trans:{\t=4}
    \Cfg{\Val [x\mapsto 4]},{\S};{\emptyset}\Trans:{?timeout}
    \Cfg{\Val [x\mapsto 4]},{\S};{timeout;\emptyset}\Trans|:{\t'}
    \qquad
    (\t'\in\SetRealZ)
  \]

  \noindent where condition $\text{(urgency)}$ ensures no time can pass
  since a message currently in the queue that can be received via rule \Rule{recv}:
  $\Cfg{\Val [x\mapsto 4]},{\S};{timeout;\emptyset}\Trans:{\tau}\Cfg{\Val [x\mapsto 4]},{\End};{\emptyset}$.%
  %
\end{exa}

\UnloadSType*

\clearpage

\LoadPrcCalc*
\section{A Calculus for Processes with Timeouts}\label{sec:toast_processes}
We present a new calculus for timed processes which extends existing
timed session calculi~\cite{Bocchi2015,Bocchi2019} with
timeouts and time-sensitive conditional statements.
Timeouts are defined on receive actions and may be immediately followed by sending actions, hence providing an instance of mixed-choice -- which is normally not supported.
Time-sensitive conditional statements
(i.e., \IFthenelse\ with conditions on process timers)
provide a natural counter-part to the timeout construct and enhance the expressiveness of the typing system in~\cite{Bocchi2019}.
To better align processes with TOAST, send and select actions have been streamlined by each message consisting of both a label \l\ and some message
variable $\v$,
which is either data or a delegated session; the same holds for receive/branch actions.

Processes are defined by the grammar given in~\Cref{eq:prc_syntax}.
Participants are the \emph{endpoints} of session and are denoted by $p$ and $q$.
Within a (binary) session, endpoints $p$ and $q$ communicate over channels $pq$ and $qp$, where $p$ sends to $q$ over $pq$, and $q$ sends to $p$ over $qp$.
%
%
We
  also
  use
  $a$ and $b$
  as ad-hoc roles,
  with
channels
$ab$ and $ba$,
  when discussing multiple binary sessions.
\begin{requation}{eqProcessSyntax}\label{eq:prc_syntax}
  \begin{array}[c]{@{}l@{\ }c@{\ }l @{\qquad} l@{\ }c@{\ }l @{\quad} l @{}}
    %
    \P,\Q & ::=
          &
    \Set{x}.\P
          &
          & \mid
          & \Scope{\p\q}{\P}
    \\[-2.5ex]\\%
          & \mid
          & \On{\p}\Send{\l}[\v].{\P}
          &
          & \mid
          & \P\mid\Q
    \\[-2.0ex]\\%
          & \mid
          & \On{\p}\Recv^{\e}{\l}[\v]:{\P}_i
          &
          & \mid
          & \Term
    \\[-2.5ex]\\%
          & \mid
          & \On{\p}\Recv^{\Diam\n}{\l}[\v]:{\P}_i~\After:{\Q}
          &
          & \mid
          & \p\q:\h
    \\[-2.5ex]\\%
          & \mid
          &
    \If*{\Cond}~\Then{\P}~\Else{\Q}
          & \e
          & ::=
          & \Diam\n \mid \infty
    & (\n\in \SetRealZ)
    \\[-2.5ex]\\%
          & \mid
          & \Delay{\Duration}.\P
          & \Diam
          & ::=
          & {\hphantom{\Diam\n}\mathllap{<}} \mid ~\leq
    \\[-2.5ex]\\%
          & \mid
          & \Delay{\t}.\P
          & \Cond
          & ::=
          & {\circled{x}\Diam n} \mid {n \Diam \circled{x}} \mid {\circled{x} = n}
    \\[-2.5ex]\\%
          & \mid
          & \Def{\Rec*{\vec{v};\vec{r}}=\P}:{\Q}
          & \Duration
          & ::=
          & {\hphantom{\circled{x}}\mathllap{t}\Diam n} \mid {n \Diam \mathrlap{t}\hphantom{\circled{x}}} \mid {\mathtt{true}}
          & (\t\in \SetRealZ)
    \\[-2.5ex]\\%
          & \mid
          & \Rec
          & \h
          & ::=
          & {\hphantom{\Diam\n}\mathllap{\emptyset}} \mid \h\cdot \l\v
  \end{array}
\end{requation}

\noindent Processes are equipped with  timers that can be set and reset during the process execution.
Let $\SetOfTimers$
be the set of timers
denoted by $\circled{x}$, $\circled{y}$ and $\circled{z}$.
While clocks are used in types to express the time constraints of a protocol (i.e., type), timers are used in processes to acquire awareness of relative time-passing during execution (as timer
concurrency
primitives in languages like Go and Erlang).
Process $\Set{x}.{\P}$ creates a
process timer
$\circled{x}$,
initialises it to $0$ and continues as $\P$.
If a process timer
named
$\circled{x}$ already exists it is reset to $0$.
Note, there is no correspondence between the clocks used in the constraints of types and process timers, regardless if they have matching names/labels.
%

Process \On{\p}\Send{\l}[\v].{\P}\ is the select/send process: it selects label \l\ and sends payload \v\ to endpoint \pq, and continues as \P.
Its counterpart is the branch/receive process \On{\p}\Recv^{\e}{\l}[\v]:{\P}_{i}.
It receives one of the labels \l_i, instantiates \v_i\ with the received payload, and continues as \P_i.
Parameter \e\ is a deadline for the receive action.
%
It can be either $\Diam\n$ or $\infty$.
If \e\ is $\infty$, then the receive action will block until a message is received, waiting potentially forever.
If \e\ is $\Diam\n$, then \n\ is the upper-bound of the receive action and $\Diam$ specifies whether the wait duration is exclusive $(<)$ or inclusive $(\leq)$.
%
  For example, setting $\e$ to $(\leq 3)$ specifies waiting \emph{up-to and including} 3 time units,
  while $(<3)$ specifies waiting \emph{strictly less than} 3 time units.
  Setting $\e$ to $(\leq0)$ models a non-blocking receive action on branches that expect messages to be immediately available to receive.
%
%
%
%
%
  %
    Process $\On{\p}\Recv^{\Diam\n}{\l}[\v]:{\P}_i~\After:{\Q}$
    is the timeout process, an extended version of the branch/receive process, where $\Diam\n$ specifies the receive deadline, after which the timeout process $\Q$ gets triggered.
    We specify timeouts to be $(\Diam\n)$ to tactically ensure timeouts cannot be $\infty$, which would make $\Q$ dead code and would be equivalent to a branch/receive action without a timeout.
%
For simplicity, 
in branch/receive and timeout processes:
(i) we omit the brackets in the case of a single option (i.e., $\left\lvert I \right\rvert = 1$)
and (ii) for options with no payloads we omit $(\v)$.

    Process $\If*{\Cond}~\Then{\P}~\Else{\Q}$ is a conditional statement.
    The condition $\Cond$ is a constraint on process timers, and is notably simpler than those constraints found
    in~\Cref{eq:delta_constraints} for types.
    Process $\Delay{\Duration}.{\P}$ models time passing
    where $\Duration$ models a non-deterministic delay.
    At runtime
    process $\Delay{\Duration}.{\P}$
    is reduced to process $\Delay{\t'}.{\P}$,
    for some
    $\t'\models\Duration\Subst{\t'}{\t}$,
    where the non-deterministic duration $\Duration$ is resolved into a single duration $\t'$.
%

Recursive processes are defined by a process variable \Var\ and
parameters $\vec{v}$ and $\vec{r}$, containing \emph{base type} values and session channels, respectively.
As standard~\cite{Bettini2008,Bocchi2014,Honda2008,Vasconcelos2012,Yoshida2007} the process calculus allows parallel processes $\P\mid\Q$ and scoped processes \Scope{\pq}{\P}\ between endpoints \p\ and \q.
The end process is \Term.
Endpoints communicate over pairs of channels \pq\ and \qp, each with their own unbounded FIFO buffers \h.

\smallskip

\subsection{Well-formed Processes}\label{ssec:wf_processes}
We assume that sessions within a process have already been instantiated, meaning that rather than relying on reduction rules to produce correct session instantiation, we rely on a syntactic well-formedness assumption, as in~\cite{Bocchi2019}.
%
%
%
%
A \emph{well-formed process} $P$
consists of sessions of the form \Scope{\pq}{\bigl(\P'\mid\Q\mid{\qp:\h}\mid{\pq:\h'}\bigr)}\ and can be checked syntactically
by function $\WF$
given in~\Cref{def:func_wf}.
%
%
%
%
%
%
%
\begin{rdefi}{defWfProcesses}[Well-formed Process]\label{def:func_wf}
  The function \WF{\P}\ is defined inductively as:
  \small\[
    \begin{array}[t]{@{}l@{}}
      \WF{\P}
      = \begin{cases}
          \True                                                               &
          \text{if\ }\P\in\begin{array}[c]\{{@{}l@{}}\}
                          \Term,
                          ~\Rec,
                          ~{\qp:\h}
                        \end{array}
          \\[\ArrayExtraTallLineSpacing]
          \WF{\P'} \land \left(\FQ{\P'}=\{\pq,\qp\}\right)                    &
          \text{if\ }\P=\Scope{\pq}{\P'}
          \\[\ArrayTallLineSpacing]
          \WF{\P'} \land \left(\FQ{\P'}=\emptyset\right)                      &
          \text{if\ }\P\in\begin{array}[c]\{{@{}l@{}}\}
                          \On{\p}\Send{\l}[\v].{\P'},
                          ~\Set{x}.\P',
                          \\
                          \Delay{\Duration}.\P',
                          ~\Delay{\t}.\P'
                        \end{array}
          \\[\ArrayExtraTallLineSpacing]
          \WF{\P'} \land \WF{\Q}                                              &
          \begin{array}[t]{@{}l@{}}
          \text{if\ }\P\in\begin{array}[c]\{{@{}l@{}}\}
                            ~\If*{\Cond}~\Then{\P'}~\Else{\Q},
                            \\
                            ~\Def{\Rec*{\vec{v};\vec{r}}=\P'}:{\Q}
                          \end{array}
        \end{array}
          \\[\ArrayExtraTallLineSpacing]
          \WF{\P'} \land \WF{\Q}
          \land \left(\FT{\P'} \cap \FT{\Q}=\emptyset\right)                  &
          \begin{array}[t]{@{}l@{}}
          \text{if\ }\P={\P'\mid\Q}
        \end{array}
          \\[\ArrayExtraTallLineSpacing]
          \bigwedge_{i\in I} \WF{\P_i} \land \left(\FQ{\P_i}=\emptyset\right) &
          \text{if\ }\P =
          \On{\p}\Recv^{\e}{\l}[\v]:{\P}_i
          \\[\ArrayExtraTallLineSpacing]
          \WF{\On{\p}\Recv^{\Diam\n}{\l}[\v]:{\P}_i} \land \WF{\Q}            &
          \text{if\ }\P=\On{\p}\Recv^{\Diam\n}{\l}[\v]:{\P}_i~\After:{\Q}
        \end{cases}
    \end{array}
  \]
\end{rdefi}

\noindent
  where
  function $\FQ$ returns the set of queues present in process $P$ that are \emph{not} contained within the respective scope of their corresponding session,
  and function $\FT$ returns the set of process timers used in process $P$.
  (We relegate the definitions of \FQ\
  and $\FT$
  to the appendix,
  see~\Cref{def:func_fq,def:func_ft}.)
  By~\Cref{def:func_wf},
  a \emph{well-formed} process must:
  (a) not contain any \emph{free queues} other than those belonging to ongoing (binary) sessions contained within \emph{scoped processes} $\Scope{\pq}{\P}$, which should each contain the two queues necessary for participants to exchange messages,
  and (b) not contain any parallel processes that make use of the \emph{same} process timer, i.e., process timers must only be used by a single process.

Hereafter, we assume that a process is \wf, unless stated otherwise.

\smallskip

\subsection{Process Reduction}\label{ssec:process_reduction}

The semantics of processes are given in~\Cref{fig:prc_reduction}, as a reduction relation on pairs of the form \PCfg\ i.e., a processes $P$ with an timer environment $\Timers$ that maps timers to their values in the current state (formally given
below
in~\Cref{def:process_timer_environment}). 

\clearpage

\begin{rdefi}{defTimerEnvironment}[Timer Environment]\label{def:process_timer_environment}
  Recall $\SetOfTimers$ is the
  set of timers,
  ranged over by $\circled{x}$, $\circled{y}$ and $\circled{z}$.
  A \emph{timer environment} $\Timers$ is a linear map $\Timers : \SetOfTimers\mapsto\SetRealZ$,
  from timers to
  valuations, defined:
  $
    \Timers ::= \emptyset \mid {\Timers,~{\circled{x}:n}}%
  $
  where $n\in\SetRealZ$.
  We define
  $\Timers+\t = \{
    \circled{x}\mapsto \Timers(\circled{x}) + \t
    \mid
    \circled{x}\in\SetOfTimers
    \}$
  so that time may pass over $\Timers$,
  and $\Timers[\circled{x}]$ to
  be the map $\Timers[\circled{x}](\circled{y})=\text{\ if\ } (\circled{x}=\circled{y})\text{\ }0\text{\ else\ } \Timers(\circled{y})$.
  We write $\Dom{\Timers}$ for the domain of $\Timers$,
  and $\Timers_1,\Timers_2$ for $\Timers_1\cup\Timers_2$ when $\Dom{\Timers_1}\cap\Dom{\Timers_2}=\emptyset$.
  We write
  $\Timers\models\Cond$
  to denote that
  $\Timers$ satisfies $\Cond$.
  Defined formally,
  $(\Timers\models\Cond)=(\forall\circled{x}\in\FN{\Cond}.\Cond\Subst{\Timers(\circled{x})}{\circled{x}})$.
  %
  %
\end{rdefi}

\input{figs/process_reduction.tex}

\subsubsection{Reduction Rules}\label{sssec:process_reduction_rules}

The reduction relation is defined on two kinds of reduction: instantaneous communication actions (\Reduce-), and time-consuming actions (\Reduce~). We write \Reduce\ to denote a reduction that is either by (\Reduce-) or (\Reduce~).

Rules \Rule{Str}, \Rule{Scope}, \Rule{Par}[L]* and \Rule{Par}[R]* are as standard~\citeST.
The rule for structural congruence \Rule{Str} is the only rule that applies to both instantaneous and time-consuming actions.
Structural equivalence follows as standard~\citeStructCongruence,
with additions from~\cite{Bocchi2014,Bocchi2019}.
  Formally given below:
  \begin{rdefi}{defStructuralCongruence}[Structural Congruence]\label{def:structural_congruence}%
    {\small
      We define
      $\P\equiv\Q$
      as
      the smallest relation on processes:
    }
    {\footnotesize\[
        \begin{array}[c]{@{}l@{}}
          {
          {\P\mid\Term}
          \equiv
          {\P}
          }
          \qquad
          {
            {(P\mid Q)}
            \equiv
            {(Q \mid P)}
          }
          \qquad
          {
            {(P_1\mid P_2)\mid Q}
            \equiv
            {(P_1 \mid Q)\mid P_2}
          }
          \qquad
          {
            {P\mid Q}
            \equiv
            {Q \mid P}
          }
          \qquad
          {
            {\Scope{\pq}{\Term}}
            \equiv
            {\Term}
          }
          \\[-1.5ex]\\
          {
          {\Delay{0}.\P}
          \equiv
          {\P}
          }
          \qquad
          {
            {\Scope{\pq}{\Scope{ab}{\P}}}
            \equiv
            {\Scope{ab}{\Scope{\pq}{\P}}}
          }
          \qquad
          {
            {\Scope{\pq}{\P\mid\Q}}
            \equiv
            {\Scope{\pq}{(\P\mid\Q)}}
            \quad
            \text{if\ }\pq\not\in\FN{\Q}
          }
          \\[-1.5ex]\\
          {
          {\Scope{\pq}{\P}}
          \equiv
          {\Scope{\qp}{\P}}
          }
          \qquad
          {
          {(\Def{\Rec*{\vec{v};\vec{r}}=\P}:{\P'})\mid\Q}
          \equiv
          {\Def{\Rec*{\vec{v};\vec{r}}=\P}:{(\P'\mid\Q)}}
          \quad
          \text{if\ }
          \Rec{\vec{v'};\vec{r'}}\not\in\FPV{\Q}
          }
        \end{array}
      \]}

    \noindent where
    $\FN{\Q}$ and $\FPV{\Q}$ return the set of free names and process variables in $\Q$, respectively.
  \end{rdefi}

\subsubsection{\texorpdfstring{Communication Rules, \Cref{eq:prc_reduction_communication}}{Communication Rules}}
Rules \Rule{Send} and \Rule{Recv} are standard~\cite{Milner1999}.
  By rule $\Rule{Send}$ a process inserts a message $\lv$ into the queue of the other party,
  while by rule $\Rule{Recv}$ a process may remove an expected message from their queue and continue on the corresponding branch.
  Rule $\Rule{Recv}[T]*$ is similar to rule $\Rule{Recv}$ except, as we will discuss shortly
  in~\Cref{ssec:process_time_passing},
  it allows the timeout branch $Q$ to be taken if the process does not receive a message within duration $\e$.
%
(We include rule \Rule{Recv}[T]* for consistency, but hereafter only refer to rule \Rule{Recv}.)

\subsubsection{\texorpdfstring{Recursion Rules, \Cref{eq:prc_reduction_recursion}}{Recursion Rules}}
Both rules \Rule{Def} and \Rule{Call} are unchanged from~\cite{Bocchi2019}.
%
  Rule $\Rule{Def}$ defines a recursive process $P$
  and allows the recursive body $Q$ to reduce.
  Rule $\Rule{Call}$ handles recursive calls of predefined recursive processes and yields a process configuration consisting of the next iteration of the recursive process using the parameters provided by the process variable $X$; 
  the parallel process $Q$ is unaffected.
%
%
%

\subsubsection{\texorpdfstring{Time-sensitive Rules, \Cref{eq:prc_reduction_time}}{Time-sensitive Rules}}
  Rule $\Rule{Reset}$ resets a process timer $\circled{x}$ in the timer environment $\Timers$ to 0.
%
%
%
%
Rule \Rule{If}[T]* selects branch \P\ if time-sensitive condition $\Cond$ holds for \Timers.
Rule \Rule{If}[F]* is symmetric, selecting \Q\ if $\Cond$ does not hold for \Timers.
%
%
%
  Rule $\Rule{Det}$
  determines the exact duration $\t'$ of a non-deterministic runtime delay modelled by $\Duration$.
%
%
Rule \Rule{Delay} handles delay processes, so long as they are defined in \Time\ which, 
in short, returns process \P\ after \t\ units of time have elapsed. 
(We formally define \Time\ below in~\Cref{def:time_passing_func_ef}.)
Additionally, by rule \Rule{Delay} any delay on a process also causes the timers within \Timers\ to elapse accordingly.

%

\UnloadPrcCalc*

\LoadTypeChecking

\subsection{Time Passing}\label{ssec:process_time_passing}
The definition of \Time\ is given in~\Cref{def:time_passing_func_ef}.
The first two cases in~\Cref{def:time_passing_func_ef} model the effect of time passing on branching and timeout processes.
The third case is for time-consuming processes.
The fourth case distributes time passing in parallel compositions and ensures that time passes for all parts of the system equally.
The remaining cases define the processes where time is allowed to pass.
\begin{rdefi}{defTimePassingFunction}[Time Passing Function]\label{def:time_passing_func_ef}
  The time-passing function \Time\ is a partial function only defined for the cases below:
  \footnotesize\[
    \begin{array}[t]{@{}l@{}}
      \Time{\On{\p}\Recv^{\Diam\n}{\l}[\v]:{\P}_i~\After:{\Q}} = %
      \begin{cases}
        %
        \On{\p}\Recv^{\Diam\n-\t}{\l}[\v]:{\P}_i~\After:{\Q}%
         &
        \text{if\ }
        ({\t\Diam\n})
        \\[0.7ex]%
        \Time{\Q}_{\t-\n}%
         & %
        \text{otherwise}
      \end{cases}
      \\[0.55cm]%
      \Time{\On{\p}\Recv^{\e}{\l}[\v]:{\P}_i}
      =
      \begin{cases}
        \On{\p}\Recv^{\e}{\l}[\v]:{\P}_i
         & %
        \text{if\ }
        { \e = \infty}%
        \\[1.2ex]%
        \On{\p}\Recv^{\n-\t}{\l}[\v]:{\P}_i
         &
        \text{if\ }
        { \e = \Diam \n}%
        \text{\ and\ }
        (\t\Diam\n)
      \end{cases}
      \\[0.55cm]%
      \Time{\Delay{\t'}.\P} = %
      \begin{cases}
        \Delay{\t'-\t}.\P%
         & %
        \text{if\ }\t'\geq\t%
        \\[0.ex]%
        \Time{\P}_{\t-\t'}%
         & %
        \text{otherwise}
      \end{cases}
      \\[0.55cm]%
      \Time{\p\q:\h} = \p\q:\h%
      \qquad%
      \Time{\P_1\mid\P_2} = %
      \Time{\P_1} \mid \Time{\P_2}%
      \quad%
      \text{if\ }(\Wait{\P_i}\cap\NEQ{P_j}=\emptyset)%
      \text{\ and\ }
      i\neq j\in\left\{1,2\right\}
      \\[0.3cm]%
      \Time{\Term} = \Term%
      \qquad%
      \Time{\Scope{\p\q}{\P}} = \Scope{\p\q}{\Time{\P}}%
      \qquad
      \Time{\Def{\Rec*{\vec{v};\vec{r}}=\P}:{\Q}} = %
      \Def{\Rec*{\vec{v};\vec{r}}=\P}:{\Time{\Q}}%
    \end{array}
  \]
\end{rdefi}

\noindent The auxiliary functions \Wait\ and \NEQ\ are given in~\Cref{fig:func_wait_neq}.
Together, they indicate which role (if any) is able to receive,
and are used in~\Cref{def:time_passing_func_ef} to ensure
\emph{receive urgency},
  by requiring that there are no processes that are both waiting to receive and have a non-empty queue.
Informally, \Wait\ returns the set of channels on which \P\ is waiting to receive a message,
and \NEQ\ returns the set of endpoints with a non-empty inbound queue.

\begin{exa}[Time Passing Function]\label{exa:time_passing_function}
In~\Cref{def:time_passing_func_ef}
we define $\Time$ as a \emph{partial function}
in order to:
(a) ensure that certain processes are always non-blocking and can never be arbitrarily delayed (e.g., send processes)
and (b) make it easier to differentiate between a process that can not be delayed, and a process that can be delayed for 0.
For example, consider the process:
$P = \mathtt{delay}(5).P' \mid \On{\p}\Send{data}.{P''}$.
Under the current definition of $\Time$, it is clear that the $data$ must be sent prior to any delay occurring, and in this instance, the definition of $\Time$
in~\Cref{def:time_passing_func_ef}
enforces that certain processes, such as send processes, should never be delayed under any circumstances.
(The only exception being a formally defined delay process prefixed to a send process.)
However,
if $\Time$ were to be defined for all cases of $P$
including send processes,
the reduction of our example process $P$ becomes unclear and may be indefinitely reduced via rule $\Rule{Delay}$ for a delay of 0.
Furthermore, by the current definition of
structural congruence in~\Cref{def:structural_congruence},
sending $data$ may be skipped altogether.
Therefore, we tactically define $\Time$ to be only a partial function in order to distinctly group delayable and non-delayable processes.
%
%
%
  %
  Consider the process below:
  \[
    \P = \Scope{\pq}{\bigl(%
      \On{\p}\Recv^{\Diam\n}{\l}[\v]:{\P}_i~\After:{\Q}%
      \mid \qp:\emptyset%
      \mid \Q'\bigr)%
    }%
  \]

  \noindent For a time-consuming action of \t\ to occur on \P\ it is required that \Time\ is defined for all parallel components in \P.
  Suppose $\t = \n+1$ so that we can observe the expiring of the timeout. The evaluation of \Time{\On{\p}\Recv^{\Diam\n}{\l}[\v]:{\P}_{i}~\After:{\Q}}\ results in the evaluation of \Time{\Q}_{1}.
  If \Q\ were e.g., a sending process, then \Time{\Q}_{\t'}\ is not defined for any $\t'>0$.
  Otherwise, if \Time{\Q}_{1} and \Time{\Q'}\ are defined,
  then \t\ may pass and
  \(
  \Time{\P} = \Scope{\pq}{%
    \bigl(\Time{\Q}_{1}
    \mid \qp:\emptyset
    \mid \Time{\Q'}\bigr)
  }
  \).
\end{exa}

\UnloadTypeChecking

\input{figs/process_func_wait_neq.tex}

\LoadTypeChecking

\section{Expressiveness}\label{sec:expressiveness}

In this section we reflect on the expressiveness of our mixed-choice extension, particularly in regard to~\cite{Bocchi2019}, using examples to illustrate differences.
Given the increase in expressiveness, we discuss how type-checking becomes more interesting with the inclusion of
\recvafter\
and give an intuition on the relationship between the TOAST and processes.

\subsection{Missing deadlines}\label{ssec:missing_deadlines}
The process corresponding to $\Option?{a}.\S$
is merely $\On{\p}\Recv*^{\infty}{a}:{\P}$, which waits to receive $a$ forever.
By way of contrast, $\Option?{b}{x<3}.\S'$, cannot receive when $x\geq 3$, requiring the process to take the form:
\(
\On{\p}\Recv*^{<3}{b}:{\P'}
\).
More generally,  
if we consider:
\(
\{?b~(x<3,~\emptyset).S',~?c~(3< x<5,~\emptyset).{S}''\}
\)
when $(x\geq3)$
then, amending the previous process, \(
\On{\p}\Recv*^{<3}{b}:{\P'}
~\After:{\On{\p}\Recv*^{<2}{c}:{\P'{2}}}
\).

\subsection{Ping-pong protocol}
The example in this section illustrates the usefulness of time-sensitive conditional statements. The $ping$-$pong$ protocol consists of two participants exchanging messages between themselves on receipt of a message from the other~\cite{Lagaillardie2022}.
One interpretation of the protocol is the following:
\[
  \mu\alpha.%
  \begin{array}[c]\{{@{}l@{}}\}
    {!ping(x\leq3,~\{x\})}.
    \begin{array}[c]\{{@{}l@{}}\}
      ?ping(x\leq3,~\{x\}).\alpha \\
      ?pong(x>3,~\{x\}).\alpha
    \end{array} \\[-1ex]\\
    {!pong(x>3,~\{x\})}.
    \begin{array}[c]\{{@{}l@{}}\}
      ?ping(x\leq3,~\{x\}).\alpha \\
      ?pong(x>3,~\{x\}).\alpha
    \end{array}
  \end{array}
\]

\noindent where each participant exchanges the role of sender, either sending $ping$ early, or $pong$ late.
Without time-sensitive conditional statements, the setting in~\cite{Bocchi2019} only allows implementations where the choice between the `$ping$' and the `$pong$' branch are hard-coded.
In presence of non-deterministic delays (e.g., $\Delay{t<6}$), the hard-coded choice can only be for the latest branch to `expire', and the highlighted fragment of the $ping$-$pong$ protocol above could be naively implemented as follows (omitting \Q\ for simplicity):
\[
  \Def{\Rec*{\p}=
  \REVISEDII{%
  \left(
  \right.
  }%
  \Delay{t<6}.\On{\p}\Send{pong}.{\Q}
  \REVISEDII{%
  \left.
  \right)
  }%
  }:{\Rec{r}}
\]

\noindent The choice of sending $ping$ is \emph{always} discarded as it may be unsatisfied in \emph{some} executions.
The calculus in this paper, thanks to the time-awareness deriving from a process timer \circled{y}, allows us to \emph{potentially} honour each branch, as follows:
\[
  \Def{\Rec*{\p}=
  \REVISEDII{%
  \left(
  \right.
  }%
  \Set{y}
  .\Delay{t<6}
  .\If{\circled{y}\leq 3}\Then{\On{\p}\Send{ping}.{\Q}}
  ~\Else{\On{\p}\Send{pong}.{\Q'}}
  \REVISEDII{%
  \left.
  \right)
  }%
  }:{\Rec{r}}
\]

\subsection{Mixed-choice Ping-pong protocol}\label{ssec:mc_ping_pong}
An alternative interpretation of the $ping$-$pong$ protocol can result in an implementation with mixed-choice, as shown below:
\[
  \mu\alpha.%
  \begin{array}[c]\{{@{}l@{}}\}
    ?ping(x\leq3,~\{x\}).
    \begin{array}[c]\{{@{}l@{}}\}
      !pong(x\leq3,~\{x\}).\alpha
      \\ ?timeout(x>3,~\emptyset).end
    \end{array} \\[-1ex]\\
    !pong(x>3,~\{x\}).
    \begin{array}[c]\{{@{}l@{}}\}
      ?ping(x\leq3,~\{x\}).\alpha
      \\ !timeout(x>3,~\emptyset).end
    \end{array} \\
  \end{array}
\]
where `$pings$' are responded by `$pongs$' and vice versa.
Notice that if a timely $ping$ is not received, a $pong$ is sent instead, which if not responded to by a $ping$, triggers a timeout.
Similarly, once a $ping$ has been received, a $pong$ must be sent on time to avoid a timeout.
Such a convoluted protocol can be fully implemented:
\(\Def{\Rec*{\p}=\P}:{\Rec{r}}\) where {\P}:
\[
  \begin{array}[t]{@{}l@{}}%
    \P = 
    \begin{array}[t]{@{}l@{}}%
      \On{\p}\Recv*^{\leq 3}{ping}:{\begin{array}[t]{@{}l@{}}%
                                      \Set{x}%
                                      .\Delay{t\leq 4}
                                      .\If{\circled{x}\leq 3}\Then*{%
                                      \On{\p}\Send{pong}.{\Rec{\p}}%
                                      }\Else*{%
                                      \On{\p}\Recv*^{\infty}{timeout}:{\Term}%
                                      }%
                                    \end{array}}%
      \\[\ArrayExtraTallLineSpacing]
      ~\mathtt{after\ }{%
        \On{\p}\Send{pong}{%
          \begin{array}[t]{@{}l@{}}%
            .\On{\p}\Recv*^{\leq 3}{ping}:{\Rec{\p}}%
            \\~\mathtt{after\ }{%
            \On{\p}\Send{timeout}.{\Term}%
            }%
          \end{array}}%
      }%
    \end{array}%
  \end{array}%
\]

\subsection{Message throttling}\label{ssec:message_throttling}
A real-world application of the previous example is \emph{message throttling}.
The rationale behind message throttling is to cull unresponsive processes, which do not keep up with the message processing tempo set by the system.
This avoids a server from becoming overwhelmed by a flood of incoming messages.
In such a protocol, upon receiving a message, a participant is permitted a grace period to respond before receiving another.
The grace period is specified as a number of unacknowledged messages, rather than a period of time.
Below we present a fully parametric implementation of this behaviour, where $m$ is the maximum number of
unacknowledged messages before a timeout is issued.
\begin{equation}\label{eq:message_throttling_type}
  \begin{array}[c]{@{}l@{\ }c@{\ }l@{\quad}l@{}}
    S_0 & = & \mu\alpha^{0}.!msg(x \geq 3, \{x\}). S_1
    \\[-1ex]\\ %
    S_i & = & \mu\alpha^{i}.
    \begin{array}[c]\{{@{}l@{}}\}
      ?ack(x < 3, \{ x \}). \alpha^{i-1},
      ~!msg(x \geq 3, \{ x \}) . S_{i+1}
    \end{array} 
    %
    \\[-1ex]\\ %
    S_m & = &
    \begin{array}[c]\{{@{}l@{}}\}
      ?ack(x < 3, \{ x \}). \alpha^{m-1},
      ~!tout(x \geq 3,~\emptyset) . \End
    \end{array}
  \end{array}
\end{equation}

\noindent 
where $0<i<m$. Below is the corresponding processes (again, $0<i<m$):
%
%
\begin{equation}\label{eq:message_throttling_process}
  \begin{array}[c]{@{}l@{\ }l@{\ }l@{\ }l@{\ }l@{}}
    \P_0
     & = &
     \mathtt{def\ }\Rec*{\p}_0 &=& \Delay{3}.{\On{\p}\Send{msg}.{\P_1}} \mathtt{\ in\ } \Rec{\p}_0
    \\[-1.5ex]\\ %
    \P_i
     & = &
     \mathtt{def\ }\Rec*{\p}_i &=& \On{\p^{<3}}\Recv*{ack}:{\Rec{\p}_{i-1}}%
     ~\mathtt{after\ }{\On{\p}\Send{msg}.{\P_{i+1}}} \mathtt{\ in\ } \Rec{\p}_i
    \\[-1.5ex]\\ %
    \P_m
     & = &
     \mathrlap{
     \On{\p^{<3}}\Recv*{ack}:{\Rec{\p}_{m-1}}
     ~\mathtt{after\ }{\On{\p}\Send{tout}.{\Term}} 
     }
     %
     %
  \end{array}%
\end{equation}

\noindent 
  We
  now
  concretize this example with the $m=2$ instance
  for type $S$
  and process $P$:
\[
  \begin{array}[c]{@{}l@{\ }c@{\ }l@{}}
    S & = &
    \mu\alpha^0
    .!msg(x\geq3,\{x\})
    .\mu\alpha^1
    .\begin{array}[c]\{{@{}l@{}}\}
       ?ack(x<3,\{x\}).\alpha^0,
       \\
       !msg(x\geq3,\{x\})
       .\begin{array}[c]\{{@{}l@{}}\}
         ?ack(x<3,\{x\}).\alpha^1,
         \\
         !tout(x\geq3,\emptyset).\mathtt{end}
       \end{array}
     \end{array}
    \\[-1.5ex]\\ %
    P
    &
    =
    &
    \begin{array}[t]{@{}l@{}}
            \mathtt{def\ }
        \begin{array}[t]{@{}l@{}}
                \Rec*{\p}_0 = \Delay{3}.\On{\p}\Send{msg}.
            \\[0.5ex]
                \mathtt{def\ }
            \begin{array}[t]{@{}l@{}}
                    \Rec*{\p}_1 =\ %
            \begin{array}[t]{@{}l@{}}
                    \On{\p^{<3}}\Recv*{ack}:{\Rec{\p}_0}
             \\[0.5ex]
                    \mathtt{after\ } 
                        \On{\p}\Send{msg}. 
                        \On{\p}^{<3}\Recv*{ack}:{\Rec{\p}_1}
                        ~\mathtt{after\ }
                        \On{\p}\Send{tout}.{\Term}
            \end{array}
            \end{array}
             \\[0.1ex]
                \mathtt{in\ } \Rec{\p}_1
        \end{array}
         \\[0.75ex]
            \mathtt{in\ } \Rec{\p}_0
    \end{array}
  \end{array}
\]

\noindent The system shown in~\Cref{fig:message_throttling}
also illustrates the instance of $m=2$.
\ManualExaQED

\input{figs/message_throttling.tex}

\UnloadTypeChecking

\LoadTypeChecking

\section{Type Checking TOAST Processes}\label{sec:type_checking}
Processes are validated against types using judgements of the form:
$\Gam;\Timers\Entails\P\TypedBy\Session$.
\begin{requation}{eqTypingJudgements}\label{eq:type_check_environments}
  \begin{array}[c]{@{}l@{\ }c@{\ }l@{\quad}l@{}}
    \Gam     & ::= & \emptyset
    \mid {\Gam,~\v:\T}
    \mid {\Gam,~\Rec|}
    & 
    (\text{where \T\ is a \bt})
    \\[\ArrayExtraTallLineSpacing]
    \Timers  & ::= & \emptyset \mid {\Timers,~{\circled{x}:n}}
    & 
    (\text{recall  Definition~\ref{def:process_timer_environment},\ } \n\in\SetRealZ)
    \\[\ArrayExtraTallLineSpacing]%
    \Session & ::= & \emptyset
    \mid {\Session,~{\p:\Cfg}}
    \mid {\Session,~{\qp:\Que}}
  \end{array}
\end{requation}

\noindent \emph{Variable Environments} \Gam\ map values \v\ to data types \T\ and process variables \Var\ to triples \Rec|, where
$\vec{T}$ is a vector of messages (consisting of labels and data types),
and $\FormatRecCheckStyle{\theta}$ and $\FormatRecCheckStyle{\Delta}$ are sets of timer environments and session environments, respectively and both are possibly infinite.
With respect to usual formulations of session types,
we use the mapping $\Timers$ of timers (denoted $\circled{x}$, $\circled{y}$ and $\circled{z}$)
as a 
timer environment (given in~\Cref{def:process_timer_environment})
to correctly type check
time-sensitive conditional statements.
As usual, \emph{Session Environments} \Session\ map roles \p\ and \q\ to configurations and session endpoints \qp\ and \pq\ to queues \Que.
%
%
  Similarly to the syntax of processes
  in~\Cref{sec:toast_processes},
  we may
  use
  $a$ and $b$
  as ad-hoc roles,
  along with endpoints $ab$ and $ba$,
  during discussions of more than one session.

\subsection{Auxiliary definitions \& notation}

We write \Dom{\Session}\ for the domain of \Session, and $\Session_1,\Session_2$ for $\Session_1 \cup \Session_2$ when $\Dom{\Session_1}\cap\Dom{\Session_2}=\emptyset$.
We define $\Session(\p)=\Cfg$ iff $\p\in\Session$ and $\exists \Session'$ such that $\Session=\Session',{\p:\Cfg}$, 
and define 
$\mathtt{val}(\Val,\S) = \Val$
and 
$\mathtt{type}(\Val,\S) = \S$.
%
Given a time offset $\t\in\SetRealZ$, we define:
\begin{itemize}
  \item 
  $\Session+\t= \{ {(\mathtt{val}(\Session(\p))+\t,~\mathtt{type}(\Session(\p)))} \mid {\p\in\Dom{\Session}} \}$
  \item $\Timers+\t=\{\circled{x}\mapsto \Timers(\circled{x}) + \t \mid \circled{x}\in\SetOfTimers\}$ \hfill (recall \Cref{def:process_timer_environment})
\end{itemize}

\noindent We say that \Session\ is \t-reading if for the duration \t, there exist roles within \Session\ that are able to perform a receiving action.
We say that \Session\ is \emph{not} \t-reading if there are no roles within \Session\ able to receive for the duration of \t. Formally:
%
%
%
\begin{rdefi}{defTReading}[\t-reading \Session]\label{def:t_reading_session}
  \Session\ is \t-reading if,
  $\exists \t'<\t, \p \in \Dom{\Session}$
  such that $\Session(\p) = \Cfg
    \implies \Cfg{\Val+\t'}\Trans:{?\Msg}$.
  We write \Session\ is $\Diam\t$-reading if,
  $\exists \t'\Diam\t, \p \in \Dom{\Session}$
  such that $\Session(\p) = \Cfg
    \implies \Cfg{\Val+\t'}\Trans:{?\Msg}$.
\end{rdefi}

\noindent \Cref{def:t_reading_session} is useful to check against violations of \emph{receive urgency}.\footnote{See~\Cref{exa:t_reading_bad_interleavings}.}
Finally, we extend the definition of well-formed configurations to session environments in the obvious way:
\begin{rdefi}{defWfSession}[Well-formed \Session]\label{def:wf_session}
  \Session\ is \wf\ if,
  for all $\p \in \Dom{\Session}$
  such that $\Session = \Session',{\p:\Cfg} \implies $ \Cfg\ is \wf\ by~\Cref{def:cfg_wf}.
\end{rdefi}

\UnloadTypeChecking
\input{figs/process_type_checking.tex} 
%
%
%
\input{figs/process_type_checking_fig_standard.tex}

\LoadTypeChecking

\subsection{Type Checking Rules}
The \tc\ rules are given in~\RefTypeCheck.
\Cref{fig:type_checking} shows standard typing rules in~\Cref{eq:type_checking_standard} and time-sensitive typing rules in~\Cref{eq:type_checking_time_sensitive}.
\Cref{fig:type_checking_communication} shows rules pertaining to communication, with:
sending processes in~\Cref{eq:type_checking_send},
branching in~\Cref{eq:type_checking_recv},
receptions in~\Cref{eq:type_checking_recv_single},
and queues in~\Cref{eq:type_checking_queues}.
We will now discuss each rule.

\subsubsection{\texorpdfstring{Standard Rules, \Cref{eq:type_checking_standard}}{Standard Rules}}
Terminated processes and empty channel buffers are typed using axioms \Rule{End} and \Rule{Empty}, respectively.
%
%
  Rule $\Rule{Weak}$ handles the type-checking
  of a session where an individual process has successfully reached termination, and where the rest of the session may be ongoing.
  Such an instance is to be expected in the asynchronous setting, where one party may perform a sequence of non-blocking sending actions and finish, leaving the other party to process and receive.
  Rule $\Rule{Par}$ divides
  the session environment
  into two disjoint environments,
  $\Session_1$ and $\Session_2$,
  so that each session is carried by only one of $\P$ and $\Q$.
  The timer environment is also split into two disjoint environments, $\theta_1$ and $\theta_2$,
  as processes would otherwise be capable of indirectly communicating or influencing the behaviour of other processes
  outside the interactions prescribed by the types.
%
Rule \Rule{Res}\ ensures scoped processes are well-typed against binary sessions composed of \compat\ configurations (by \Cref{def:cfg_compat}).
Notably, the only change from~\cite{Bocchi2019} in this rule is the well-formedness requirement of $\S_i$ against $\nu_i$ (see \Cref{def:cfg_wf}). In~\cite{Bocchi2019}, well-formedness was naturally entailed by the premises of their
  rules type-checking sending/receiving processes.
%
%
%
In order to tackle the complexity of checking mixed-choice (branching and selections), we have used multiple simpler rules that make it difficult to encapsulate such a requirement. Therefore, without loss of expressiveness, we have made
the
requirement
explicit here.

Rules \Rule{Rec} and \Rule{Var} are essentially unchanged from~\cite{Bocchi2019},
only now extended with process timers.
Rule \Rule{Rec} handles recursion definitions by
evaluating \Q\ and collecting within the triple $({\vec{T};\ScSVListToRecCheck{\theta;\Delta}})$ all messages, timer and session environments present at any corresponding call.
This triple
is mapped to a process variable assigned to \P, and added to $\Gamma$ as the evaluation continues with \Q.
Rule \Rule{Var} determines if a call is well-typed by checking that the provided messages and channels correspond to those in the process variable.

\subsubsection{\texorpdfstring{Time-sensitive Rules, \Cref{eq:type_checking_time_sensitive}}{Time-sensitive Rules}}
Rule \Rule{Timer} handles setting or resetting of timer $\circled{x}$.
We require $\circled{x}$ to be already in \Timers\ in the set case (also handled by this rule) to avoid race conditions on timers with concurrent sessions, which would compromise subject reduction.
%
  Inspection of the syntactic structure of a process is enough to know beforehand which timers will be used.
%
%
%
Rule \Rule{Del}[\Duration] checks that \P\ is well-typed
against all solutions of $\Duration$.
Rule \Rule{Del}[\t] reflects time passing on \Timers\ and the session \Session.
%
%
%
  The right-hand side premise of `$\Delta \text{\ not\ }t\text{-reading}$' enforces \emph{receive urgency} in processes by requiring that no participants in $\Delta$ are able to receive a message from their queue for the duration of $t$.
  (Inverse to~\Cref{def:t_reading_session}.)
%
Rules \Rule{IfTrue} and \Rule{IfFalse} are for typing time-sensitive conditional statements.
If the timers in \Timers\ satisfy the
condition $\Cond$,
then by rule \Rule{IfTrue} the evaluation continues with \P.
Else, by the symmetric rule \Rule{IfFalse} the evaluation continues with \Q.

\UnloadTypeChecking
%
%
%
%
%
\input{figs/process_type_checking_fig_communication.tex}

\LoadTypeChecking

\subsubsection{\texorpdfstring{Sending Rules, \Cref{eq:type_checking_send}}{Sending Rules}}
Rules \Rule{VSend} and \Rule{DSend} are for typing processes that send messages against a choice type, by finding a matching pair of labels in the process and type.
Rule \Rule{VSend} applies for sending base-types, requiring that \v\ is well-typed against \T, and the continuation process \P\ is well-typed against the continuation \S, with any clock resets applied.
Rule \Rule{DSend} is similar,
but for delegating an ongoing session with role \b\ to role \p,
and then
removes role \b\ from the session environment.

\subsubsection{\texorpdfstring{Branching Rules, \Cref{eq:type_checking_recv}}{Branching Rules}}
Rule $\Rule{Branch}$ checks a branching process against a choice type
and is comprised of a single premise.
The first
condition
of the premise
requires that
the type or the processes has more than one element, to ensure that type-checking is deterministic (if both are singleton sets, then either rule \Rule{DRecv} or \Rule{VRecv} in~\Cref{eq:type_checking_recv_single} should be used).
The second
condition
of the premise
requires that each currently enabled receiving action in the type has a corresponding branch in the process (i.e., since all labels are unique, they have matching labels) which when paired together, are \emph{well-typed}.

Rule \Rule{Timeout} checks that a \recvafter\ process is well-typed against a choice type; decomposing the checking into two parts:
the branch process, and the process \Q; requiring that both are \wt.
The first line in the premise forms a judgement to be checked by rule \Rule{Branch}, by discarding the `timeout' part of the original process.
The second line
of the premise
checks
the process is still \wt\ in the case where the
  timeout `triggers' (i.e., time $\n$ passes).
Put simply,
each enabled action prescribed by the type must have a corresponding branch in the process that, when paired together, form a well-typed judgement.

\subsubsection{\texorpdfstring{Reception Rules, \Cref{eq:type_checking_recv_single}}{Reception Rules}}
Rule \Rule{VRecv} is applied when evaluating individual base-type receptions.
The first line of the premise checks that waiting the duration of \n\ adheres to the timings specified by the type's constraints \Constraints; also, this line requires that there are no other roles in \Session\ that have enabled receiving actions.
The second line of the premise evaluates the proceeding process \P\ against the continuation type \S, given any possible delay that could occur ($\t\Diam\n$).
Rule \Rule{DRecv} is similar, but for evaluating individual \emph{delegation} receptions, which are then added to the session environment.
\subsubsection{\texorpdfstring{Queue Rules, \Cref{eq:type_checking_queues}}{Queue Rules}}
Non-empty buffers and queues are typed by \Rule{VQue} or rule \Rule{DQue}.
The former is for \bt\ payloads and the latter for delegated sessions.

\subsection{Type Checking Examples}\label{ssec:type_checking_examples}

The handling of scope restriction, parallel composition, and recursion
is similar to existing formulations.
We give an example of the less obvious rules to handle timers and mixed-choice.
Precisely,
we show
in~\Cref{exa:type_check_ping_pong}
how to type check an implementation of a simplified (non-recursive version) the mixed choice $ping$-$pong$ protocol featured
in~\Cref{ssec:mc_ping_pong}.
Later, in~\Cref{exa:t_reading_bad_interleavings}
  we discuss type-checking multiple binary sessions, and illustrate the issues that can arise from \emph{incompatible interleavings}.

\begin{exa}[Type Checking Ping-pong]\label{exa:type_check_ping_pong}
  Consider the protocol $S$ below:
  \begin{equation}\label{eq:ex_tc_type}
    \S =
    \begin{array}[c]\{{@{}l@{}}\}
      ?ping(x\leq3,~\{x\}).
      \begin{array}[c]\{{@{}l@{}}\}
        !pong(x\leq3,~\{x\}).\End,
        \\ ?timeout(x>3,~\emptyset).\End
      \end{array}, \\[-1ex]\\
      !pong(x>3,~\{x\}).
      \begin{array}[c]\{{@{}l@{}}\}
        ?ping(x\leq3,~\{x\}).\End,
        \\ !timeout(x>3,~\emptyset).\End
      \end{array} \\
    \end{array}
  \end{equation}

  \noindent and its candidate implementations $P$:
  \begin{equation}\label{eq:ex_tc_prc}
    \P = \begin{array}[t]{@{}l@{}}%
      \On{\p}\Recv*^{\leq 3}{ping}:{\begin{array}[t]{@{}l@{}}%
                                      \Set{z}.\Delay{t\leq 4}.
                                      \If{\circled{z}\leq 3}\Then*{%
                                      \On{\p}\Send{pong}.{\Term}%
                                      }\Else*{%
                                      \On{\p}\Recv*^{\infty}{timeout}:{\Term}%
                                      }%
                                    \end{array}}%
      \\[\ArrayExtraTallLineSpacing]
      \mathtt{after\ }\Delay{0.1}.{%
      \On{\p}\Send{pong}.{%
      \begin{array}[t]{@{}l@{}}%
        \On{\p}\Recv*^{\leq 3}{ping}:{\Term}%
        \\~\mathtt{after\ }\Delay{0.1}.{%
        \On{\p}\Send{timeout}.{\Term}%
        }%
      \end{array}}%
      }%
    \end{array}%
  \end{equation}

  \noindent We discuss how one can use our type system to check the following:
  $\emptyset; \circled{z} : 0 \Entails P \TypedBy ~\p : (\nu_0, S)$.
  From here,
  one can apply rule \Rule{Timeout},
  which decomposes the judgement into two:
  one for branching process
  in~\Cref{eq:ex_tc_branch},
  and one for the timeout continuation \Q\
  in~\Cref{eq:ex_tc_timeout}.
  \begin{equation}\label{eq:ex_tc_branch}
    \emptyset; \circled{z} : 0 \Entails
    \On{\p}\Recv*^{\leq 3}{ping}:\Set{z}.\Delay{t\leq 4}%
    .\mathtt{if}\,(\circled{z}\leq 3) \ldots
    \TypedBy ~\p : (\nu_0, S)
  \end{equation}
  \begin{equation}\label{eq:ex_tc_timeout}
    \emptyset; \circled{z}:3 \Entails
    \Delay{0.1}.{}\On{\p}\Send{pong}.R
    \TypedBy ~\p : (\nu_3[x\mapsto 0], S)
  \end{equation}

  \noindent where
  \(
  R = \On{\p}\Recv*^{\leq 3}{ping}:{\Term}%
  ~\mathtt{after\ }\Delay{0.1}.{\On{\p}\Send{timeout}.{\Term}}
  \).
  In~\Cref{eq:ex_tc_timeout}, note that the values of
  process timer $\circled{z}$
  and virtual clock $x$ have incremented by the duration of the timeout, i.e., $3$.
  Following up~\Cref{eq:ex_tc_branch}, one can apply rule \Rule{Branch}, which checks that all viable actions in \S\ are receiving actions, and that each is correctly implemented in the process being typed.
  Clearly this holds, since the only viable receive action in \S\ is \Option?{ping}{x\leq 3}[x], which matches the only
  branch $({ping}:\ldots)$ in the process.
  We next apply rule \Rule{VRecv}:
  \begin{equation}\label{eq:ex_tc_branch_recv}
    \emptyset; \circled{z} : \t' \Entails
    \Set{z}.\Delay{t\leq 4}%
    .\mathtt{If}\,(\circled{z}\leq 3) \ldots
    \TypedBy ~\p : (\nu_{\t'}[x \mapsto 0],  ~ S')
    \qquad (\t' \leq 3)
  \end{equation}

  \noindent with $n\leq 3$ and
  $S'=\{ !pong(x\leq3,~\{x\}).end, ?timeout(x>3,~\emptyset).end\}$.
  Since clock $x$ has been reset to 0,
  and is $x$ the only clock used by our type $S$,
  the set of clock valuations is back to $\nu_0$.
  Application of $\Rule{Set}$ then resets
  timer $\circled{z}$ to $0$.
  Next, $\Rule{Del}[\Duration]$
  decomposes~\Cref{eq:ex_tc_branch_recv}
  into a set of judgements, one for each solution $\t'{2}$
  of $\Duration = (\t\leq 4)$;
  which by $\Rule{Del}[t]$ yields:
  \begin{equation*}\label{eq:ex_tc_delay}
    \emptyset; \circled{z}: \t'+\t'{2} \Entails
    \mathtt{if\ }(\circled{z}\leq 3)\Then{%
    \On{\p}\Send{pong}.{\Term}%
    }~\Else{%
    \On{\p}\Recv*^{\infty}{timeout}:{\Term}%
    }%
    \TypedBy ~\p : (\nu_{\t'+\t'{2}}[x\mapsto\t'{2}], ~  S')
  \end{equation*}

  \noindent Notably, if $\t'{2}$ is smaller than or equal to $3$ ($\t'{2}\leq 3$) we can apply rule \Rule{IfTrue}; e.g., $\t'{2}=3$:
  \begin{equation}\label{eq:ex_tc_if_true}
    \emptyset; \circled{z}:  3 \Entails
    \On{\p}\Send{pong}.{\Term}%
    \TypedBy ~\p : (\nu_{\t'+3}[x\mapsto 3], ~  S')
  \end{equation}

  \noindent and if $\t'{2}$ is greater than $3$ (e.g., $\t'{2}=4$) we can apply rule \Rule{IfFalse}:
  \begin{equation}\label{eq:ex_tc_if_false}
    \emptyset; \circled{z} :  4 \Entails
    \On{\p}\Recv*^{\infty}{timeout}:{\Term}%
    \TypedBy ~\p : (\nu_{\t'+4}[x\mapsto 4],   ~S')
  \end{equation}

  \noindent To proceed the evaluation of~\Cref{eq:ex_tc_if_true}, observe that the process is selecting a label $pong$.
  By rule \Rule{VSend}, we only need to check that $S'$ has a sending branch with label $pong$ viable when $(x = 3)$, which is indeed the case (i.e., $!pong(x\leq3,~\{x\}).end$).
  Similarly, looking at~\Cref{eq:ex_tc_if_false}, we need to use the rule branching to isolate all (one in this case) the receive actions viable at the current time.
  The remaining of both derivations is straightforward.

  The case for the `$\mathtt{after}$' process in~\Cref{eq:ex_tc_timeout} is similar, except it starts with a shifted time in the timers and virtual clocks.
  %
  Note, we insert `$\Delay{0.1}.{\dots}$' immediately upon entering the `$\mathtt{after}$' to compensate for the inclusive upper-bound of ($\leq3$).
  %
  %

  We show
  the full
  derivation tree in~\Cref{sec:appendix:derivation:ping_pong}.
\end{exa}

\begin{exa}[Incompatible Interleavings]\label{exa:t_reading_bad_interleavings}
  Consider the following session environments:
  \[
    \begin{array}[c]{@{}lcl@{}}
      \Session_{ab} & = & {a:\Cfg_1},~{b:\Cfg_2},~{ba:\Que_1},~{ab:\Que_2}
      \\[0.5ex]%
      \Session_{pq} & = & {p:\Cfg_3},~{q:\Cfg_4},~{qp:\Que_3},~{pq:\Que_4}
    \end{array}
  \]

  \noindent where both \Session_{ab}\ and \Session_{pq}\ are \wf\ and composed together, under the session environment $\Session=(\Session_{ab},\Session_{pq})$.
  Both $p$ and $b$ are roles performed by the same participant in two distinct  sessions.
  All clocks within \Session\ are 0.
  Below are the definitions of \S_2\ and \S_3:
  \[
    \begin{array}[c]{@{}l@{\qquad}l@{}}
      \begin{array}[c]{@{}l@{}}
        \S_2  =  \begin{array}[c]{@{}l@{}}\Choice|{
                   \Option?{request}[\String]{x<5}[x].{\S_2'},
                   \Option!{timeout}{x\geq 5}[x].{\S_2'{2}}
                   }\end{array}
      \end{array}
       &
      \begin{array}[c]{@{}l@{}}
        \S_3  =  \Choice*|{\Option?{message}{y>3}.{\S_3'}}
      \end{array}
    \end{array}
  \]

  \noindent Consider the following implementation $\P$, where role $b$ first waits 5 time units for a $request$ before receiving $message$ as role $p$.
  \[
    \P = \begin{array}[t]{@{}l@{}}
      \On{ab}\Recv*^{<5}{request}[\mathtt{req\_str}]:{
      \On{qp}\Recv*^{\infty}{message}:{\P'}
      }
      \\[0.5ex]~\After<{5}:{
      \On{ba}\Send{timeout}.{
      \On{qp}\Recv*^{\infty}{message}:{\P'{2}}
      }
      }
    \end{array}
  \]

  \noindent Assume $message$ arrives in the queue $qp$ at time $3.5$ and $request$ arrives at time $4$. Process $P$ will receive $request$ at time $4$ while $message$ is still pending in the queue, breaking \emph{receive urgency}.
  In fact, $P$ is not $4$-reading since there exists a role $p$ in $\Delta$ that may be able to receive before $4$ time units.
  \Cref{def:t_reading_session} is utilised by the \tc\ rules in~\RefTypeCheck\ to ensure such implementations are \emph{not} \wt.
  In this case, one may observe that there exists no well-typed implementation of $S_2$ and $S_3$ that satisfies receive urgency, hence the roles are
  fundamentally incompatible and unable to be interleaved.
\end{exa}

\section{Subject Reduction}\label{sec:subject_reduction}

In this section we establish subject reduction for our typing system, namely we establish
a relation between well-typed processes and their types that is preserved by reduction.

As usual, subject reduction is given for closed systems (\Cref{thm:subject_reduction}) with $\Gamma$ and $\Delta$ empty.
The proof relies on two lemmas that establish subject reduction for open systems, for time-consuming steps (\Cref{lem:sr_time_step})
and instantaneous steps (\Cref{lem:sr_action_step}).
We give an overview in this section and relegate the full proofs to~\Cref{sec:proof_subject_reduction}.
Crucially, by rule \Rule{Res} in~\Cref{fig:type_checking}~\Cref{eq:type_checking_standard}, in any typing judgement of an open system the sessions in $\Delta$ are both \wf\ and comprised of \compat\ configurations, enabling us to utilise progress of TOAST~\Cref{thm:type_progress}.
Such sessions we say are \fbal.
The notion of balanced session, inspired  by the one in~\cite{Chen2017} is given in~\Cref{def:fbal_session}.

\begin{rdefi}{defBalancedSession}[Balanced \Session]\label{def:bal_session}
  Let \Bal\ be the set containing all session environments \Session, such that if $\Session\in\Bal$, then \Session\ adheres to the following:
  \begin{enumerate}
    \item\label[case]{case:bal_session_recv} $\Session = \Session',{\p:\Cfg},{\qp:{\Msg;\Que}}
            \implies
            \Cfg\Trans:{?\Msg}\Cfg'$
          and $\Session',{\p:\Cfg'},{\qp:\Que}
            \in\Bal
          $.
    \item\label[case]{case:bal_session_missing_queue} $\Session = \Session',{\p:\Cfg_1},{\qp:{\Que_1}},{\q:\Cfg_2}
            \implies \exists \Que_2:
            \Cfg;{\Que}_1\Compat\Cfg;{\Que}_2
          $.
    \item\label[case]{case:bal_session_compat} $\Session = \Session',{\p:\Cfg_1},{\qp:\Que_1},{\q:\Cfg_2},{\pq:\Que_2}
            \implies
            \Cfg;{\Que}_1\Compat\Cfg;{\Que}_2$.
  \end{enumerate}
\end{rdefi}

\begin{rdefi}{defFullyBalancedSession}[Fully\text{-}Balanced \Session]\label{def:fbal_session}
  A \bal\ \Session\ is said to be \fbal\ if:
  \begin{enumerate}
    \item\label[case]{case:fbal_session_cfg} $\Session = \Session',{\p:\Cfg_1}
            \implies
            \exists \Session'{2},\q,\Val_2,\S_2,\Que_1,\Que_2:
            \Session' = \Session'{2},{\qp:\Que_1},{\q:\Cfg_2},{\pq:\Que_2}$.
    \item\label[case]{case:fbal_session_que} $\Session = \Session',{\qp:\Que_1}
            \implies
            \exists \Session'{2},\Val_1,\S_1,\Val_2,\S_2,\Que_2:
            \Session' = \Session'{2},{\p:\Cfg_1},{\q:\Cfg_2},{\pq:\Que_2}$.
  \end{enumerate}
\end{rdefi}

\begin{restatable}[Time Step]%
  {thm}
  {lemTimeStep}\label{lem:sr_time_step}
  Let \Session\ be \fbal\ and \wf.
  If $\Gam;\Timers\Entails\P\TypedBy\Session$ and \Time\ is defined,
  then $\Gam;\Timers+\t\Entails\Time\TypedBy\Session+\t$
  and $\Session+\t$ is \fbal\ and \wf.
\end{restatable}

\noindent By~\Cref{lem:sr_time_step} we prove that, given a \wt\ process \P, for any value of \t\ such that \Time\ is defined by~\Cref{def:time_passing_func_ef}, \t\ passing evenly over the session environment preserves both \wfness\ and \balness.
(Similar to the preservation of \wfness\ and \emph{compatibility} for system configurations in~\Cref{lem:tp_cfg_trans_preserve_wf,lem:tp_cfg_trans_preserve_compat} of~\Cref{thm:type_progress}.)
Clearly, by inspection of the \tc\ rules in~\RefTypeCheck, such processes would require several divisions of the initial session environment (and map of process timers) before each individual process could be evaluated.
Therefore, we rely on the notion of \balness\ (\Cref{def:bal_session}).
Naturally, a \fbal\ session must also be \bal, and this allows us to reason on specific \bal\ components of a larger \fbal\ session environment.

\begin{restatable}[Action Step]%
  {thm}
  {lemActionStep}\label{lem:sr_action_step}
  Let \Session\ be \bal\ and \wf.
  If $\Gam;\Timers\Entails\P\TypedBy\Session$
  and $\PCfg\!\!\Reduce-\!\!\PCfg'$
  then $\exists \Session':
    \Session\!\!\Reduce*\!\!\Session'$
  and
  $\Gam;\Timers'\!\Entails\!\P'\!\TypedBy\!\Session'$
  and
  \Session'\ is \bal\ and \wf.
\end{restatable}

\noindent By~\Cref{lem:sr_action_step} we prove, among other things, that any sending or receiving action performed by a \wt\ process corresponds to an action prescribed by the type of the corresponding configuration, at the correct time.
The best illustration of this is perhaps~\Cref{case:sr_action_step_recv} in~\Cref{lem:sr_action_step},
which shows the case of reduction via rule \Rule{Recv}
in~\Cref{fig:prc_reduction}~\Cref{eq:prc_reduction_recursion}, and that the reduction preserves \wtness.
We relegate the rules for session environment reduction
to~\Cref{fig:session_reduction} in~\Cref{sec:proof_subject_reduction}, along with the full proofs of~\Cref{lem:sr_time_step,lem:sr_action_step}.

\begin{restatable}[Subject Reduction]%
  {cor}
  {thmSubjectReduction}\label{thm:subject_reduction}
  If $\emptyset;\Timers\Entails\!\P\TypedBy\emptyset$
  and $\PCfg\!\rightarrow\!\PCfg'$
  then $\emptyset;\Timers'\!\Entails\!\P'\!\TypedBy\emptyset$.%
\end{restatable}%
\begin{rproof}{proofSubjectReduction}
  Since $\Session=\emptyset$, then by~\Cref{def:fbal_session}, \Session\ is \fbal.
  We proceed by induction on the nature of the reduction (\Reduce):
  \begin{caseanalysis}
    %
    \item\label{case:sr_time_step}
    If
    \(
    \Reduce=\Reduce~
    \)
    then the thesis follows~\Cref{lem:sr_time_step}.
    %
    \item\label{case:sr_action_step}
    If
    \(
    \Reduce=\Reduce-
    \)
    then the thesis follows~\Cref{lem:sr_action_step}.
    \qedhere
  \end{caseanalysis}
\end{rproof}

\noindent 
As a safety property of our typing system, subject reduction guarantees
  that any action performed by a well-typed process is as prescribed by the types; i.e., is timely and correct.

\section{Related Work}\label{sec:related_work}
This paper is an extended version of~\cite{Pears2023} featuring numerous improvements and three substantial additions to the existing work:
(1) the proofs of progress for TOAST given in~\Cref{sec:proof_type_progress};
(2) the type-checking rules in~\Cref{sec:type_checking}; and,
(3) the subject reduction result in~\Cref{sec:subject_reduction} together with its proof in~\Cref{sec:proof_subject_reduction}.
We have also made several changes to the syntax of processes in~\Cref{sec:toast_processes} and the reduction rules in~\Cref{ssec:process_reduction} to better accommodate the new type-checking rules and our focus on subject reduction result. These changes are largely syntactic, and are equivalent to those in the original. For convenience, we report below in detail about the changes to the syntax of processes with respect to the short version of this paper. Most notably, we have simplified (without loss of generality) the definition of $\e$ from being a linear expression over process timers and numeric constants,
to either $\Diam\n$, $\infty$ or $\leq 0$, as defined in~\Cref{eq:prc_syntax}.
For readability, we have streamlined the syntax of branching processes with and without timeout.
We have removed $\Err$ processes, as our focus is not that of showing error-freedom.

\paragraph{Asynchronous Timed Session Types}
Since our work builds on {{Asynchronous Timed Session Types}}~\cite{Bocchi2019}, we begin by highlighting the differences between the two. Our work adds expressiveness to the types in~\cite{Bocchi2019} by allowing mixed-choice. Our processes are also extended.
The processes in~\cite{Bocchi2019} allow a simplified variant of a timeout where a receive action is annotated with a deadline, akin to our branching primitive $\On{\p}\Recv^{\e}{\l}[\v]:{\P}_i$.
The added expressiveness of mixed-choice in TOAST allow us to introduce a process to execute upon timeout to express, for instance, retry strategies, which are not expressible in~\cite{Bocchi2019}.
In addition to timeouts we have allowed our processes to model timers. By using timing constraints over timers as conditional statements it becomes straightforward to derive the corresponding ``\dual'' of a given timeout processes.

  Similarly
  to~\cite{Bocchi2019},
  type-checking using our typing system is decidable under the assumption that
  channel and recursion variables
  are annotated with typing information.
  The rules in our typing system with infinitely many premises, such as `$\forall \t\in \Duration$' in the premise of rule $\Rule{Del}[d]$
  in~\Cref{eq:type_checking_time_sensitive},
  can be accounted for symbolically, such as by exploiting zones and Difference Bound Matricies
  (DBMs)~\cite{Bengtsson2003}
  or Satisfiability Modulo Theory (SMT).
%

In this version of our work we do not feature the $\Err$ process, which is featured in~\cite{Bocchi2019} and~\cite{Pears2023}.
More similarly to the earlier work of {{Timed MPST}}~\cite{Bocchi2014}, instead of letting a process reaching an error state ($\Err$) when a time constraint is violated, we do not allow time to pass.
(See~\Cref{def:time_passing_func_ef}.)
This choice allows us to manage the complexity added by mixed-choice and separate the concerns of subject reduction, safety and progress.
Normally, safety and progress come as consequences of subject reduction.
Conversely, in a timed scenario with error states, subject reduction depends on progress (i.e., if a process gets stuck some time constraint may be violated, yielding an error state reached, hence compromising subject reduction).
In~\cite{Bocchi2019} this circular dependency was overcome by requiring a progress property called receive-liveness, which is defined on the untimed counterpart of a timed process, which can be checked by well-established techniques~\cite{DezaniCiancaglini2007,Bettini2008}.
Unfortunately, not only do techniques such as~\cite{DezaniCiancaglini2007} not apply to mixed-choice, but they are hardly attainable in this setting:
removing time from a mixed-choice is likely to introduce progress violations that would otherwise not be there in the presence of time.
Therefore, removing \Err\ processes became critical in order to break this circular dependency and establish subject reduction.
As a future work, we would like to establish a progress property on the basis of subject reduction.

\subsection*{Runtime Monitoring}
Recent work by~\cite{Pears2024b} presented a work-in-progress toolchain for generating Erlang stub programs from a protocol specification notation derived from the theory of TOAST presented in this paper, with the aim of bridging the gap between protocol specifications and program implementations.
%
%
%
%
%
%
The toolchain generates correct-by-construction Erlang stub programs which provide a bare-bones scaffolding structure for a user to extend with functionality (i.e., what to send and how to handle receptions).

In addition to generating Erlang stub programs, \cite{Pears2024b} provide a generic runtime monitoring program for Erlang, which is provided as a precautionary measure to ensure that user of the toolchain does not unintentionally introduce unspecified behaviour when extending the Erlang stubs with functionality. The monitoring program can be setup to monitor the communication of any Erlang process against a given protocol specification, and can be configured to either: (a) \emph{verify} the process communicates as specified, or (b) try to \emph{enforce} the specified behaviour (e.g., by delaying the sending/receiving of messages).

\subsection*{Other work on timeouts} Outside of session types, a similar construct to timeout processes appears in~\cite{Bocchi2022a}, which builds upon the {{Temporal Process Language}}.
Similarly to our work, \cite{Bocchi2022a} extends the receiving branch processes in~\cite{Bocchi2019} in a manner based on the \recvafter\ pattern in Erlang.
However~\cite{Bocchi2022a} does not provide a typing discipline and rather focuses on characterising reliability properties in a mailbox-based communication setting.
The ``\emph{periodic recursion}'' in {{Rate-Based Session Types}}~\cite{Iraci2023},
features a similar construct where a timeout-recursive loop  runs indefinitely, at a steady and fixed rate.
Instead of timing constraints on individual actions, periods are imposed on recursive loops, specifying the rate at which each iteration must adhere to, and optional deadlines which specify a fixed upper-bound.

\subsection*{Session types}
The (untimed) {{Multiparty, Asynchronous and Generalised $\pi$-calculus}}~\cite{LeBrun2023} ({{MAG$\pi$}})
features the option for a unique ``timeout branch'' in the syntax of both branching types and processes.
Notably, ``timeout branches'' in~\cite{LeBrun2023} are non-deterministic, and may trigger regardless of another branch being viable.
Additionally, the semantics go as far as to facilitate errors, such as by dropping unreliable messages from a queue before they can be received.
In our work, we use timeouts to model safe mixed-choice, structured in such a way that communication mismatches cannot occur.
Our work focuses on the timeliness of messages being sent and received,
and to this end our timeouts ensure only one participant can send at any time.
By comparison, the timeouts in~\cite{LeBrun2023} can also model mixed-choice,
but since they are unstructured and non-deterministic, communication mismatches may occur and as such, the type/process must be designed to handle these failures.
Our processes feature time-sensitive conditional statements in our \IFthenelse\ processes to model the counterpart of a timeout, while~\cite{LeBrun2023} instead staggers timeout processes in an alternating pattern.
While this is effective in their context, it does limit the descriptive capabilities to mixed-choice with only two distinct regions of sending and receiving actions.
In comparison, our work allows mixed-choice composed of any number of sending and receiving actions interleaved together in the same state.

  {{Fault-Tolerant MPST}}~\cite{Peters2022,Peters2023} allows a branching type/process to specify a default value, to be used in the case when nothing is received, abstracting from the decision procedure on when to trigger this case.
(This approach relies on an external failure detector to inform the process to take the default value.)
While different from a timeout, this does allow a process that would otherwise be stuck waiting forever, to instead make progress.
Our timeout process and timers closer align with programming primitives like timeouts and timers in e.g.,
Erlang (with its \recvafter\ pattern and timers) and Go.

  {{Mixed Sessions}}~\cite{Vasconcelos2020}
introduces mixed-choice in untimed \emph{synchronous} session types. They allow for mixed-choice with a single communication primitive, similar to our own types.
Unlike our work, choices in~\cite{Vasconcelos2020} are non-deterministic and labels in a choice do not have to be unique.
Additionally, options in a choice are either linear or \emph{unrestricted}.
This approach leads to patterns such as the producer/consumer pattern being elegantly represented by a single type, whereas in our system several would be needed.
Effectively, the unrestricted action embeds into one of the options action within a choice, a recursive call that returns back to the choice.
This is once again similar to the timeout-recursion behaviour found in~\cite{Bocchi2022a,Iraci2023}, except that~\cite{Vasconcelos2020} is \emph{untimed}.
%
Recent work by~\cite{Peters2024b} 
present a typing system for
mixed-choice in \emph{synchronous} multiparty sessions
and 
explores the expressiveness 
offered by mixed-choice
by evaluating and comparing various other works (with and without mixed-choice).
%

Further afield, coordination structures have been proposed that overlap with mixed-choice, for example, fork and join~\cite{Denielou2012a:MPST:CFSM}, which permit messages within a fork (and its corresponding join) to be sent or received in any order; reminiscent of mixed-choice.
Affine sessions~\cite{Lagaillardie2022,Mostrous2018} support exception handling by enabling an end-point to perform a subset of the interactions specified by their type, but there is no consideration of time, hence timeouts.
Before session types gained traction, timed processes~\cite{Berger2007} were proposed for realising timeouts, but lack any notion of counterpart.

\subsection*{Timed semantics} Our work is based on the timed semantics of {{Communicating Timed Automata}}~\cite{Krcal2006:CTA} ({{CTA}}).
{{CTA}} is the combined product of {{Timed Automata}}~\cite{Alur1994:TA} and {{CFSM}}~\cite{Brand1983:CFSM}, yielding a timed asynchronous semantics.
Our time-passing semantics of \Rule{time}
in~\Cref{eq:type_semantics_triple} is based upon the \emph{urgent semantics}~\cite{Bartoletti2018},
to enforce \emph{receive urgency} and ensure that the \emph{latest enabled} action is not missed~\cite{Bartoletti2018}.
Our semantics for time-passing over configurations differs from~\cite{Bocchi2019}, which requires that sending actions are never missed.
%
%
%
In our setting a configuration featuring a mixed-choice,
such as $\Cfg{\Val_0},{\left\{ \Option!{a}{x<3}.{S}, \Option?{b}{x\geq3}.{S'} \right\}}$
would be unable to reach the latest-enabled receiving action $b$, since the preceding sending action $a$ would \emph{always} happen instead,
effectively rendering $b$ redundant.
However, in this work we only require that the latest-enabled action is never missed, allowing instances such as $a$ to be passed.
This provides implementations of
TOAST protocols
the affordance of not featuring every sending action, and allows them to draw upon the full range of interactions described by the given TOAST protocol.
%
%

The work of~\cite{Bartoletti2017} presents (binary) {{Timed Session Types}} in the synchronous setting.
Instead of relying on
a
notion of \emph{duality}, \cite{Bartoletti2017} presents a means of determining if two parties are \emph{compliant}, given their sets of constraints.
This allows the timing constraints on actions to be constructed without considering the dual,
and therefore, avoiding such issues as those discussed in~\Cref{exa:junk_types}.
%
  E.g., if one party expects to receive $a$ some time after 3 time units,
  but has some proceeding action $c$ that will not be enabled after a total of 5 time units,
  then the party can still be guaranteed to be free from deadlocks if there is some \emph{compliant co-party} that always sends before 4 time units;
  as in~\Cref{exa:junk_types}.
%

\section{Conclusion}\label{sec:conclusion}

We have shown how timing constraints provide an intuitive way of integrating mixed-choice into asynchronous (binary) session types.
There are many conceivable ways to realise mixed-choice using programming primitives.
However, our integration with time, embodied in TOAST, offers new capabilities for modelling timeouts which sit at the heart of protocols and are a widely-used idiom in programming practice.
To provide a bridge to programming languages, we provide a timed session calculus enriched with a \recvafter\ pattern and process timers, providing the means to implement a timeout process and its dual.

Taken altogether, we have lifted a long-standing restriction on asynchronous session types by allowing for safe mixed-choice, through the judicious application of timing constraints.
In this extended version of~\cite{Pears2023}, we have established a means of type-checking our processes with TOAST.
Additionally, we have shown that processes that are well-typed by our type system adhere to subject reduction (\Cref{thm:subject_reduction}) which means that any step the system
takes
maintains
the shape of a well-typed system.
Subject reduction guarantees that when an action happens, its sender, receiver, label, payload, and timing all conform to the types.

\subsection*{Future Work}

  There are two primary focuses for future work:
  (1) an extension to the multiparty setting;
  and (2) a progress property for our typing system.
  Progress was included in previous work on timed session types~\cite{Bocchi2019}. This result is not easily transferable to our work since we deal with mixed-choice.
  In fact, the result in~\cite{Bocchi2019} builds upon existing techniques to ensure global progress~\cite{DezaniCiancaglini2007} on an untimed version of their timed processes.
  In essence, \cite{Bocchi2019} uses an \emph{untimed} version of a process to check for deadlocks in the interaction structures of the process, using the checks described in~\cite{DezaniCiancaglini2007}.
  We are unable to take this approach since: (a) the work in~\cite{DezaniCiancaglini2007} does not account for mixed-choice, and (b) the untimed counterpart of a \emph{safe} timed mixed-choice may not enjoy progress since the timing constraints that make it \emph{safe} would be removed. Proving progress in this setting requires explicit reasoning on the interactions structures.
  Therefore, we will have to take a different approach to in order to provide a stronger property for our typing system.
%

Additionally, in future work we would like to continue the development of the theory, such as to the multiparty setting, and in turn, enable work such as the toolchain presented by~\cite{Pears2024b} to have a broader potential for application to real world systems.
  As part of this growing toolchain, we aim to include an implementation of our typing system.

\bibliographystyle{alphaurl}
\bibliography{main.bib}

\appendix

\section{Proof of Type Progress}\label{sec:proof_type_progress}

The result of~\Cref{thm:type_progress} follows similar reasoning to~\cite{Bocchi2019}.
While~\Cref{thm:type_progress} states progress for \emph{initial system configurations}, progress is also guaranteed for systems composed of \wf\ and \compat\ configurations by~\Cref{lem:tp_progress}.
Our proof shows that our \wfness\ rules in~\Cref{fig:type_wf_rules} ensure that the timing-constraints on actions are \emph{feasible}, and structured such that any mixed-choice are \emph{safe}.
We show that our {{LTS}} in~\Cref{eq:type_semantics_tuple,eq:type_semantics_triple,eq:type_semantics_system} preserve both \wfness\ and compatibility of our configurations.

\subsection{Auxiliary Definitions \& Assumptions}

\begin{defi}[Latest-enabled Configurations]\label{def:cfg_le}
  \Cfg\ is \le\ if, $\exists \Comm, \Comm'$
  such that $\Comm\neq\Comm'$
  and $\Cfg \fe=$
  and $\Cfg \fe=|'$.
\end{defi}

\begin{defi}[Live Configurations]\label{def:cfg_live}
  \Cfg\ is \live\ if, $\S=\End$ or \Cfg\ is \fe.
\end{defi}


\begin{defi}[Clock Names]\label{def:clock_names}
  We define \CN{\delta}, the \emph{clock names} of a constraint $\delta$, as:\
  \[
    \CN{\delta} =
    \begin{cases}
      \{ x \}     
      & \text{if\ }\delta\in
          \begin{array}\{{@{}l@{}}\}
                x>n, x=n
          \end{array}
      \\
      \{ x,y \}     
      & \text{if\ }\delta\in
          \begin{array}\{{@{}l@{}}\}
                x-y>n, x-y=n
          \end{array}
      \\
      \CN{\delta'}
      & \text{if\ }\delta=\neg\delta'
      \\
      \CN{\delta_1}\cup\CN{\delta_2}
      & \text{if\ }\delta=\delta_1\land\delta_2
      \\
      \emptyset             
      & \text{if\ }\delta=\True
    \end{cases}
  \]
\end{defi}

\begin{defi}[Free Names]\label{def:free_names}
  We define \FN{S}, the \emph{free names} of a type \S, as:\footnote{As in~\cite{Pierce2002}.}
  \[
    \FN{\S} =
    \begin{cases}
      \{\alpha\}                  & \text{if\ }\S=\alpha
      \\
      \FN{\S'}\setminus\{\alpha\} & \text{if\ }\S=\mu\alpha.\S'
      \\
      \bigcup_{i\in I} \FN{\S_i}  & \text{if\ }\S=\Choice{\l}[\T]_i
      \\
      \emptyset                   & \text{if\ }\S=\End
    \end{cases}
  \]
\end{defi}

\begin{defi}[Capture Avoiding Substitution]\label{def:capture_avoid_subst}
  Substitution $\Subst{\S'}{\alpha}$ is \emph{capture avoiding} for $\S$ iff:
  \[
    \begin{array}{@{}l@{\quad}c@{\quad}l@{}}
      \S = \Choice{\l}[\T]_i
       & \Longrightarrow &
      \Subst{\S'}{\alpha} \text{\ is capture avoiding for \S_j,\ } \forall j\in I
      \\
      \S = \mu\beta.\S'{2} \land \alpha\neq\beta
       & \Longrightarrow &
      \beta\notin\FN{\S'} \text{\ and\ }
      \Subst{\S'}{\alpha} \text{\ is capture avoiding for}\S'{2}
    \end{array}
  \]

  \noindent With an abuse of notation, we say that $\S\Subst{\S'}{\alpha}$ is capture avoiding when $\Subst{\S'}{\alpha}$ is capture avoiding for \S.
\end{defi}

\begin{defi}[Unfold Equivalence]\label{def:up_to_unfolding}
  We define $(\utu)$ as the smallest relation such that:
  \small\[
    \begin{array}[t]{@{}l@{}}
      \begin{array}[t]{@{}l@{\qquad}l@{\qquad}l@{}}
        \S \utu \S
         &
        \S \utu \S' \implies \S' \utu \S
         &
        \S \utu \S' \land \S' \utu \S'{2} \implies \S \utu \S'{2}
        \\[1ex]
      \end{array}
      \\
      \begin{array}[t]{@{}l@{\qquad}l@{\qquad}l@{}}
        \S\Subst{\mu\alpha.\S}{\alpha}\text{\ is capture avoiding}\implies
        \mu\alpha.\S \utu \S\Subst{\mu\alpha.\S}{\alpha}
         &
        \S \utu \S' \implies
        \mu\alpha.\S \utu \mu\alpha.\S'
        \\[1ex]
      \end{array}
      \\
      \begin{array}[t]{@{}l@{}}
        \S \utu \S' \implies
        \begin{array}[c]\{{@{}l@{\ }l@{}}\}
          {\Comm}{\lT<>}({\Constraints},{\Resets}).{\S},
           &
          {\Comm''}{\lT<>'{2}}({\Constraints'{2}},{\Resets'{2}}).{\S'{2}}
        \end{array}
        \utu
        \begin{array}[c]\{{@{}l@{\ }l@{}}\}
          {\Comm}{\lT<>}({\Constraints},{\Resets}).{\S'},
           &
          {\Comm''}{\lT<>'{2}}({\Constraints'{2}},{\Resets'{2}}).{\S'{2}}
        \end{array}
      \end{array}
    \end{array}
  \]
\end{defi}

\begin{asm}\label{asm:rec_name_def}
  If $\RecTypeEnv,\alpha:\Constraints';\Constraints\Entails\S$,
  then $\alpha\notin\Dom{\RecTypeEnv}$.
\end{asm}

\subsection{Well-formed Types}

\begin{lem}\label{lem:tp_free_name}
  If $\alpha\notin\FN{\S}$
  and $\alpha\notin\Dom{\RecTypeEnv}$,
  then
  \(
  \RecTypeEnv,\alpha:\Constraints';\Constraints\Entails\S
  \iff
  \RecTypeEnv;\Constraints\Entails\S
  \).
\end{lem}
\begin{proof}
  Since $\alpha\notin\Dom{\RecTypeEnv}$ it follows that recursive call $\alpha$ has not yet been defined.
  Additionally, since $\alpha\notin\FN{\S}$, then it follows~\Cref{def:free_names} that
  if $\mu\alpha.\S'$ appears at any point within \S, then $\alpha$ \emph{may} appear at some point within \S'.
  For each case of ($\iff$):
  \begin{caseanalysis}
    %
    \item\label{case:tp_free_name_ltor}
    \(
    \RecTypeEnv,\alpha:\Constraints';\Constraints\Entails\S
    \Longrightarrow
    \RecTypeEnv;\Constraints\Entails\S
    \).
    If $\S=\mu\beta.\S'$, then by rule \Rule{rec}:
    \[
      \infer[\Rule{rec}]{%
        \RecTypeEnv,\alpha:\Constraints';\Constraints\Entails\mu\beta.\S'
      }{%
        \RecTypeEnv,\alpha:\Constraints',\beta:\Constraints;\Constraints\Entails\S'
      }
    \]

    \noindent Notice the environment
    \(
    \RecTypeEnv,\alpha:\Constraints',\beta:\Constraints
    \)
    in the premise of rule \Rule{rec} (above).
    If $\beta=\alpha$ then, the environment would contain duplicate variables.
    Therefore, $\beta\neq\alpha$
    or $\beta$ must be renamed to some $\beta'\notin\Dom{\RecTypeEnv,\alpha:\Constraints'}$
    and $\S=\mu\beta'.\S'\Subst{\beta'}{\beta}$.
    Similarly, if $\S=\beta$ then, since $\alpha\notin\FN{\beta}$, it must be that $\beta\neq\alpha$.
    In both of these cases, the presence of $\alpha:\Constraints$ is immediately inconsequential, and will remain so for the remainder of the evaluation.
    Clearly, if $\S=\End$ then the same holds.
    If $\S=\Choice{\l}[\T]_i$ then by~\Cref{def:free_names}, it must also hold for all $i\in I$.
    In summary, since $\alpha\notin\Dom{\RecTypeEnv}$ and $\alpha\notin\FN{\S}$, it follows that $\alpha$ must correspond to a recursive variable in a different unfolding (either outside \S, or a future unfolding within \S), and that it has been renamed from the corresponding recursive variable in the current unfolding.
    %
    \item\label{case:tp_free_name_rtol}
    \(
    \RecTypeEnv;\Constraints\Entails\S
    \Longrightarrow
    \RecTypeEnv,\alpha:\Constraints';\Constraints\Entails\S
    \).
    Following~\Cref{case:tp_free_name_ltor},
    the presence of $\alpha:\Constraints$ is inconsequential since $\alpha$ corresponds to some already unfolded recursive call outside \S, or a some future unfolding within \S.
    In both cases, the $\alpha$ has been renamed from the corresponding recursive variable in the current unfolding.
    \qedhere
  \end{caseanalysis}
\end{proof}

\begin{lem}[Substitution Lemma]\label{lem:tp_substitution}
  Let $ \RecTypeEnv;\Constraints'\Entails\S'$.
  If \(
  \alpha\notin \Dom{\RecTypeEnv}
  \) and \(
  \S\Subst{\S'}{\alpha}
  \) is \emph{capture avoiding}, then
  \(
  \RecTypeEnv;\Constraints\Entails \S\Subst{\S'}{\alpha}
  \iff
  \RecTypeEnv,\alpha:\Constraints';\Constraints\Entails \S
  \) holds.
\end{lem}
\begin{proof}
  We proceed by induction on the structure of \S:
  \begin{caseanalysis}
    %
    \item\label{case:tp_subst_rec_def}
    If $\S=\mu\beta.\S'{2}$
    then, for each case of ($\iff$):
    \begin{caseanalysis}
      %
      \item\label{case:tp_subst_rec_def_ltor}
      \(
      \RecTypeEnv;\Constraints\Entails
      \mu\beta.\S'{2}
      \Subst{\S'}{\alpha}
      \Longrightarrow
      \RecTypeEnv,\alpha:\Constraints';\Constraints\Entails
      \mu\beta.\S'{2}
      \).
      \begin{caseanalysis}
        %
        \item\label{case:tp_subst_rec_def_ltor_a_eq_b}
        If $\beta=\alpha$, then
        \(
        \mu\beta.\S'{2}\Subst{\S'}{\alpha}
        =
        \mu\beta.\S'{2}\Subst{\S'}{\beta}
        =
        \mu\beta.\S'{2}
        \).
        Therefore:
        \[
          \RecTypeEnv;\Constraints\Entails\mu\beta.\S'{2}
        \]

        \noindent Since $\alpha=\beta$ and $\alpha\notin\Dom{\RecTypeEnv}$, by~\Cref{lem:tp_substitution} we have the thesis:
        \[
          \RecTypeEnv,\alpha:\Constraints';
          \Constraints\Entails
          \mu\beta.\S'{2}
        \]
        %
        \item\label{case:tp_subst_rec_def_ltor_a_neq_b}
        If $\beta\neq\alpha$,
        then by rule \Rule{rec}:
        \[
          \infer[\Rule{rec}]{%
          \RecTypeEnv;\Constraints\Entails
          \mu\beta.\S'{2}
          \Subst{\S'}{\alpha}
          }{%
          \RecTypeEnv,\beta:\Constraints;
          \Constraints\Entails
          \S'{2}
          \Subst{\S'}{\alpha}
          }
        \]

        \noindent By~\Cref{def:capture_avoid_subst}, $\beta\notin\FN{\S'}$ and $\S'{2}\Subst{\S'}{\alpha}$ is \emph{capture avoiding}.
        By~\Cref{lem:tp_free_name}:
        \(
        \RecTypeEnv,\beta:\Constraints;\Constraints';\Entails\S'
        \iff
        \RecTypeEnv;\Constraints'\Entails\S'
        \).
        By the induction hypothesis:
        \[
          \RecTypeEnv,\alpha:\Constraints',\beta:\Constraints;\Constraints\Entails\S'{2}
          \iff
          \RecTypeEnv,\beta:\Constraints;\Constraints;\Entails\S'{2}\Subst{\S'}{\alpha}
        \]

        \noindent We obtain our thesis following the conclusion of rule \Rule{rec} (below), using the former case (above) as our premise:
        \[
          \infer[\Rule{rec}]{%
          \RecTypeEnv,\alpha:\Constraints';\Constraints\Entails\mu\beta.\S'{2}
          }{%
          \RecTypeEnv,\alpha:\Constraints',\beta:\Constraints;\Constraints\Entails\S'{2}
          }
        \]
      \end{caseanalysis}
      %
      \item\label{case:tp_subst_rec_def_rtol}
      \(
      \RecTypeEnv,\alpha:\Constraints';\Constraints\Entails
      \mu\beta.\S'{2}
      \Longrightarrow
      \RecTypeEnv;\Constraints\Entails
      \mu\beta.\S'{2}
      \Subst{\S'}{\alpha}
      \).
      By rule \Rule{rec}:
      \[
        \infer[\Rule{rec}]{%
        \RecTypeEnv,\alpha:\Constraints';\Constraints\Entails
        \mu\beta.\S'{2}
        }{%
        \RecTypeEnv,\alpha:\Constraints',\beta:\Constraints;\Constraints\Entails
        \S'{2}
        }
      \]

      \noindent It cannot be that $\beta=\alpha$, as otherwise the environment
      \(
      \RecTypeEnv,\alpha:\Constraints',\beta:\Constraints
      \)
      would contain a duplicate variable (as in the premise of rule \Rule{rec}, above).
      Therefore, it must be that $\beta\neq\alpha$.
      Since $\mu\beta.\S'{2}\Subst{\S'}{\alpha}$ is \emph{capture avoiding}, by~\Cref{def:capture_avoid_subst}, $\beta\notin\FN{\S'}$ and $\S'{2}\Subst{\S'}{\alpha}$ is also \emph{capture avoiding}.
      By~\Cref{lem:tp_free_name}:
      \(
      \RecTypeEnv;\Constraints'\Entails\S'
      \iff
      \RecTypeEnv,\beta:\Constraints;\Constraints'\Entails\S'
      \).
      By the induction hypothesis:
      \[
        \RecTypeEnv,\alpha:\Constraints',\beta:\Constraints;
        \Constraints\Entails\S'{2}
        \iff
        \RecTypeEnv,\beta:\Constraints;\Constraints\Entails\S'{2}\Subst{\S'}{\alpha}
      \]

      \noindent We obtain our thesis following the conclusion of rule \Rule{rec} (below), using the latter case (above) as the premise:
      \[
        \infer[\Rule{rec}]{%
        \RecTypeEnv;\Constraints\Entails\mu\beta.\S'{2}\Subst{\S'}{\alpha}
        }{%
        \RecTypeEnv,\beta:\Constraints;\Constraints\Entails\S'{2}\Subst{\S'}{\alpha}
        }
      \]
    \end{caseanalysis}
    %
    \item\label{case:tp_subst_rec_call}
    If $\S=\beta$
    then, for each case of ($\iff$):
    \begin{caseanalysis}
      %
      \item\label{case:tp_subst_rec_call_ltor}
      \(
      \RecTypeEnv;\Constraints\Entails
      \beta
      \Subst{\S'}{\alpha}
      \Longrightarrow
      \RecTypeEnv,\alpha:\Constraints';\Constraints\Entails
      \beta
      \).
      \begin{caseanalysis}
        %
        \item\label{case:tp_subst_rec_call_ltor_a_eq_b}
        If $\beta=\alpha$,
        then $\beta\Subst{\S'}{\alpha}=\S'$
        and $\Constraints'=\Constraints$.
        Since $\alpha\notin\Dom{\RecTypeEnv}$, it follows:
        \[
          \infer[\Rule{var}]{%
            \RecTypeEnv',\beta:\Constraints;\Constraints\Entails\beta
          }{}
        \]
        %
        \item\label{case:tp_subst_rec_call_ltor_a_neq_b}
        If $\beta\neq\alpha$,
        then $\beta\Subst{\S'}{\alpha}=\beta$
        and $\RecTypeEnv=\RecTypeEnv',\beta:\Constraints$
        and by rule \Rule{var}:
        \[
          \infer[\Rule{var}]{%
            \RecTypeEnv',\beta:\Constraints;\Constraints\Entails\beta
          }{}
        \]

        \noindent By rule \Rule{var} (and up-to reordering of variables):
        \[
          \infer[\Rule{var}]{%
            \RecTypeEnv',\beta:\Constraints,\alpha:\Constraints';\Constraints\Entails\beta
          }{}
        \]
      \end{caseanalysis}
      %
      \item\label{case:tp_subst_rec_call_rtol}
      \(
      \RecTypeEnv,\alpha:\Constraints';\Constraints\Entails
      \beta
      \Longrightarrow
      \RecTypeEnv;\Constraints\Entails
      \beta
      \Subst{\S'}{\alpha}
      \).
      \begin{caseanalysis}
        %
        \item\label{case:tp_subst_rec_call_rtol_a_eq_b}
        If $\beta=\alpha$,
        then $\Constraints'=\Constraints$
        and $\beta\Subst{\S'}{\alpha}=\S'$.
        The thesis coincides with the hypothesis:
        \[
          \RecTypeEnv;\Constraints'\Entails\S'
        \]
        %
        \item\label{case:tp_subst_rec_call_rtol_a_neq_b}
        If $\beta\neq\alpha$,
        then $\beta\Subst{\S'}{\alpha}=\beta$ and
        \(
        \RecTypeEnv=\RecTypeEnv',\beta:\Constraints
        \)
        and by rule \Rule{var}:
        \[
          \infer[\Rule{var}]{%
            \RecTypeEnv',\alpha:\Constraints',\beta:\Constraints;\Constraints\Entails\beta
          }{}
        \]

        \noindent The thesis follows immediately by rule \Rule{var}:
        \[
          \infer[\Rule{var}]{%
            \RecTypeEnv',\beta:\Constraints;\Constraints\Entails\beta
          }{}
        \]
      \end{caseanalysis}
    \end{caseanalysis}
    %
    \item\label{case:tp_subst_choice}
    If $\S=\Choice{\l}[\T]_i$
    then, for each case of ($\iff$):
    \begin{caseanalysis}
      %
      \item\label{case:tp_subst_choice_ltor}
      \(
      \RecTypeEnv;\Constraints\Entails
      \Choice{\l}[\T]_i
      \Subst{\S'}{\alpha}
      \Longrightarrow
      \RecTypeEnv,\alpha:\Constraints';\Constraints\Entails
      \Choice{\l}[\T]_i
      \).
      \vspace{0.75ex}
      \newline
      By~\Cref{def:capture_avoid_subst}, for all $i\in I$, $\S_i\Subst{\S'}{\alpha}$ is \emph{capture avoiding}.
      By rule \Rule{choice}:
      \begin{equation}\label{eq:tp_subst_choice_ltor_asm}
        \infer[\Rule{choice}]{
        \RecTypeEnv;~\Past_iV_{i\in I}
        \Entails
        \Choice{\l}[\T]{\Constraints}[\Resets].{\S}_i
        \Subst{\S'}{\alpha}
        }{%
        \begin{array}{@{}c @{\hspace{2ex}} l@{}}
          \forall i\in I:
          \RecTypeEnv;~\FutureEnv_i
          \Entails
          \S_i\Subst{\S'}{\alpha}
          ~\land~
          \Constraints_i\Reset_i\models\FutureEnv_i
           & \text{(feasibility)}
          \\[0.5ex]
          \forall i,j\in I:
          i \neq j
          \implies
          \Constraints_i\land\Constraints_j\models\False
          ~\vee~
          \Comm_i=\Comm_j
           & \text{(mixed-choice)}
          \\[0.5ex]
          \forall i\in I:
          \T_i=\Del'{3}
          \implies
          \emptyset;~\FutureEnv'\Entails\S'{3}
          ~\land~
          \Constraints'{3}\models\FutureEnv'
           & \text{(delegation)}
        \end{array}
        }
      \end{equation}

      \noindent where $\Constraints=\Past{\bigvee_{i\in I}\Constraints_i}$.
      By the induction hypothesis:
      \[
        \RecTypeEnv;\FutureEnv_i
        \Entails\S_i\Subst{\S'}{\alpha}
        \iff
        \RecTypeEnv,\alpha:\Constraints';\FutureEnv_i
        \Entails\S_i
        \qquad
        \text{where\ }\Constraints_i\Reset_i\models\FutureEnv_i
      \]

      \noindent Therefore, we obtain our thesis following the conclusion of rule \Rule{choice} (below), using the latter case (above) as the $\text{(feasibility)}$ premise:
      \begin{equation}\label{eq:tp_subst_choice_ltor_conc}
        \infer[\Rule{choice}]{
        \RecTypeEnv,\alpha:\Constraints';~\Past_iV_{i\in I}
        \Entails
        \Choice{\l}[\T]{\Constraints}[\Resets].{\S}_i
        }{%
        \begin{array}{@{}c @{\hspace{2ex}} l@{}}
          \forall i\in I:
          \RecTypeEnv,\alpha:\Constraints';~\FutureEnv_i
          \Entails
          \S_i
          ~\land~
          \Constraints_i\Reset_i\models\FutureEnv_i
           & \text{(feasibility)}
          \\[0.5ex]
          \forall i,j\in I:
          i \neq j
          \implies
          \Constraints_i\land\Constraints_j\models\False
          ~\vee~
          \Comm_i=\Comm_j
           & \text{(mixed-choice)}
          \\[0.5ex]
          \forall i\in I:
          \T_i=\Del'{3}
          \implies
          \emptyset;~\FutureEnv'\Entails\S'{3}
          ~\land~
          \Constraints'{3}\models\FutureEnv'
           & \text{(delegation)}
        \end{array}
        }
      \end{equation}
      %
      \item\label{case:tp_subst_choice_rtol}
      \(
      \RecTypeEnv,\alpha:\Constraints';\Constraints\Entails
      \Choice{\l}[\T]_i
      \Longrightarrow
      \RecTypeEnv;\Constraints\Entails
      \Choice{\l}[\T]_i
      \Subst{\S'}{\alpha}
      \).
      \vspace{0.75ex}
      \newline
      By rule \Rule{choice} as in~\Cref{eq:tp_subst_choice_ltor_conc},
      $\Constraints=\Past{\bigvee_{i\in I}\Constraints_i}$.
      By the $\text{(feasibility)}$ premise of rule \Rule{choice},
      for all $i\in I$,
      \(
      \RecTypeEnv,\alpha:\Constraints';\FutureEnv_i\Entails\S_i
      \)
      and $\Constraints_i\Reset_i\models\FutureEnv_i$.
      Therefore, by the induction hypothesis:
      \[
        \RecTypeEnv,\alpha:\Constraints;\FutureEnv_i\Entails\S_i
        \iff
        \RecTypeEnv;\FutureEnv_i\Entails\S_i\Subst{\S'}{\alpha}
        \qquad
        \text{where\ }\Constraints_i\Reset_i\models\FutureEnv_i
      \]

      \noindent Our thesis follows the conclusion of rule \Rule{choice} in~\Cref{eq:tp_subst_choice_ltor_asm}, where we use the latter case (above) as the $\text{(feasibility)}$ premise.
    \end{caseanalysis}
    %
    \item\label{case:tp_subst_end}
    If $\S=\End$
    then, by rule \Rule{end}:
    \[
      \infer[\Rule{end}]{
        \RecTypeEnv;\True\Entails\End
      }
    \]

    \noindent where $\Constraints=\True$.
    By rule \Rule{end}, \End\ is \wf\ against any \Constraints\ and any \RecTypeEnv.
    It follows that $\End\Subst{\S'}{\alpha}=\End$.
    Therefore, the hypothesis holds for $\S=\End$.
  \end{caseanalysis}

  \noindent We have shown the hypothesis to hold for any \S.
  The interesting cases are when $\S=\mu\beta.\S'{2}$ or $\S=\beta$, which vary depending on $\alpha=\beta$ or $\alpha\neq\beta$.
  Since $\alpha\notin\Dom{\RecTypeEnv}$, we are dealing with either:
  (1) an unfolding of a yet-to-be defined recursion $\alpha$,
  or (2) an future unfolding of a renamed recursion.
  By our thesis, we prove that unfolding and renaming recursive types does not interfere with the judgements of \wfness\ found in the rules in~\Cref{fig:type_wf_rules}.
\end{proof}

\subsection{Well-formed Configurations}

\begin{lem}\label{lem:tp_end_wf}
  \Cfg,{\End}\ is always \wf.
\end{lem}
\begin{proof}
  By~\Cref{def:cfg_wf}, $\exists \Constraints$
  such that $\Val\models\Constraints$
  and $\emptyset;\Constraints\Entails\End$.
  By rule \Rule{end} in~\Cref{fig:type_wf_rules}, $\Constraints=\True$.
  Therefore, our thesis follows $\Val\models\True$.
\end{proof}


\begin{lem}\label{lem:tp_rec_unfold}
  If \Cfg\ is \wf,
  then $\S\neq\alpha$.
\end{lem}
\begin{proof}
  Let us consider by contradiction that $\S=\alpha$.
  By~\Cref{def:cfg_wf}, $\exists \Constraints$ such that $\Val\models\Constraints$ and $\emptyset;\Constraints\Entails\alpha$.
  The only rule applicable to $\alpha$ is \Rule{var}:
  \[
    \infer[\Rule{var}]{
      \RecTypeEnv,\alpha:\Constraints;\Constraints\Entails\alpha
    }{}
  \]

  \noindent Yet,
  \(
  \emptyset\neq\RecTypeEnv,\alpha:\Constraints
  \),
  and so $\S=\alpha$ contradicts with~\Cref{def:cfg_wf} and we obtain our thesis.
\end{proof}

\begin{lem}\label{lem:tp_wf_cfg_utu_aux}
  If $\S\utu\S'$,
  then $\RecTypeEnv;\Constraints\Entails\S\ \implies \RecTypeEnv;\Constraints\Entails\S'$.
\end{lem}
\begin{proof}
  We proceed by induction on the last equation of~\Cref{def:up_to_unfolding} used to establish $\S\utu\S'$.
  \begin{caseanalysis}
    %
    \item\label{case:tp_wf_cfg_utu_aux_rec_def}
    If
    \(
    \S = \mu\alpha.\S_0
    \)
    with $\S_0\Subst{\mu\alpha.\S_0}{\alpha}$ capture avoiding,
    then
    \(
    \S' = \S_0\Subst{\mu\alpha.\S_0}{\alpha}
    \).
    The only rule applicable to
    \(
    \RecTypeEnv;\Constraints\Entails
    \mu\alpha.\S_0
    \)
    is rule \Rule{rec}:
    \[
      \infer[\Rule{rec}]{%
        \RecTypeEnv;\Constraints\Entails\mu\alpha.\S_0
      }{%
        \RecTypeEnv,\alpha:\Constraints;\Constraints\Entails\S_0
      }
    \]

    \noindent Note that $\alpha$ cannot be in $\Dom{\RecTypeEnv}$.
    Since $\S_0\Subst{\mu\alpha.\S_0}{\alpha}$ is capture avoiding,
    by~\Cref{lem:tp_substitution}:
    \[
      \RecTypeEnv,\alpha:\Constraints;\Constraints\Entails\S_0
      \iff
      \RecTypeEnv;\Constraints\Entails\S_0\Subst{\mu\alpha.\S_0}{\alpha}
    \]

    \noindent The thesis follows the conclusion of rule \Rule{rec}, using the latter case (above) as premise:
    \[
      \infer[\Rule{rec}]{%
        \RecTypeEnv;\Constraints\Entails\mu\alpha.\S_0\Subst{\mu\alpha.\S_0}{\alpha}
      }{%
        \RecTypeEnv,\alpha:\Constraints;\Constraints\Entails\S_0\Subst{\mu\alpha.\S_0}{\alpha}
      }
    \]
    %
    %
    \item\label{case:tp_wf_cfg_utu_aux_rec_call}
    If
    \(
    \S = \mu\alpha\S_0
    \),
    then
    \(
    \S' = \mu\alpha\S_0'
    \)
    with $ \S_0\utu\S_0'$.
    The only rule applicable to
    \(
    \RecTypeEnv;\Constraints\Entails
    \mu\alpha.\S_0
    \)
    is rule \Rule{rec}:
    \[
      \infer[\Rule{rec}]{%
        \RecTypeEnv;\Constraints\Entails\mu\alpha.\S_0
      }{%
        \RecTypeEnv,\alpha:\Constraints;\Constraints\Entails\S_0
      }
    \]

    \noindent By the induction hypothesis applied to the premise of the above:
    \(
    \RecTypeEnv,\alpha:\Constraints;\Constraints\Entails\S_0'
    \).
    Therefore, we obtain our thesis following the conclusion of rule \Rule{rec}:
    \[
      \infer[\Rule{rec}]{%
        \RecTypeEnv;\Constraints\Entails\mu\alpha.\S_0'
      }{%
        \RecTypeEnv,\alpha:\Constraints;\Constraints\Entails\S_0'
      }
    \]
    %
    \item\label{case:tp_wf_cfg_utu_aux_choice}
    If
    \(
    \S =
    \begin{array}[c]\{{@{}l@{}}\}
      {\Comm'}{\lT<>'}({\Constraints'},{\Resets'}).{\S_0'}
    \end{array}_{i \in I}
    \)
    then
    \(
    \S' =
    \begin{array}[c]\{{@{}l@{}}\}
      {\Comm'}{\lT<>'}({\Constraints'},{\Resets'}).{\S_0'{2}}
    \end{array}_{i \in I}
    \)
    with $ \S_0'\utu\S_0'{2} $.
    (For brevity and clarity we only show for $\left\lvert I\right\rvert = 1$.)
    The only rule applicable to
    \(
    \RecTypeEnv;\Constraints\Entails
    \begin{array}[c]\{{@{}l@{}}\}
      {\Comm'}{\lT<>'}({\Constraints'},{\Resets'}).{\S_0'}
    \end{array}_{i \in I}
    \)
    is rule \Rule{choice}:
    \[
      \infer[\Rule{choice}]{
      \RecTypeEnv;~\Past_iV_{i\in I}
      \Entails
      \begin{array}[t]\{{@{}l@{}}\}
        {\Comm'}{\lT<>'}({\Constraints'},{\Resets'}).{\S_0'}
      \end{array}_{i \in I}
      }{%
      \begin{array}{@{}c @{\hspace{2ex}} l@{}}
        \RecTypeEnv;~\FutureEnv
        \Entails
        \S_0'
        ~\land~
        \Constraints'\Reset'\models\FutureEnv
         & \text{(feasibility)}
        \\[\ArrayLTSPremiseLineSpacing]
        \T'=\Del'{2}
        \implies
        \emptyset;~\FutureEnv'\Entails\S'{2}
        ~\land~
        \Constraints'{2}\models\FutureEnv'
         & \text{(delegation)}
      \end{array}
      }
    \]

    \noindent where $ \Constraints = \Past{\bigvee_{i\in I}\Constraints_i} = \Past{\Constraints'} $.
    (Since $\left\lvert I\right\rvert = 1$, we omit the $\text{(mixed-choice)}$ premise of rule \Rule{choice}.)
    By the induction hypothesis, applied to the $\text{(feasibility)}$ premise above:
    \[
      \RecTypeEnv;~\FutureEnv
      \Entails
      \S_0'{2}
      ~\land~
      \Constraints'\Reset'\models\FutureEnv
    \]

    \noindent We obtain our thesis following the conclusion of rule \Rule{choice}, when the above is used as the $\text{(feasibility)}$ premise:
    \[
      \infer[\Rule{choice}]{
      \RecTypeEnv;~\Past_iV_{i\in I}
      \Entails
      \begin{array}[t]\{{@{}l@{}}\}
        {\Comm'}{\lT<>'}({\Constraints'},{\Resets'}).{\S_0'{2}}
      \end{array}_{i \in I}
      }{%
      \begin{array}{@{}c @{\hspace{2ex}} l@{}}
        \RecTypeEnv;~\FutureEnv
        \Entails
        \S_0'{2}
        ~\land~
        \Constraints'\Reset'\models\FutureEnv
         & \text{(feasibility)}
        \\[\ArrayLTSPremiseLineSpacing]
        \T'=\Del'{2}
        \implies
        \emptyset;~\FutureEnv'\Entails\S'{2}
        ~\land~
        \Constraints'{2}\models\FutureEnv'
         & \text{(delegation)}
      \end{array}
      }
      \qedhere
    \]
  \end{caseanalysis}
\end{proof}

\begin{lem}\label{lem:tp_wf_cfg_utu}
  If $\S\utu\S'$,
  then \Cfg\ is \wf\ $\implies$ \Cfg,{\S'}\ is \wf.
\end{lem}
\begin{proof}
  By~\Cref{def:cfg_wf}, $\exists \Constraints$ such that $\Val\models\Constraints$ and $\emptyset;\Constraints\Entails\S$.
  The thesis follows immediately from~\Cref{lem:tp_wf_cfg_utu_aux},
  since $\emptyset;\Constraints\Entails\S'$.
\end{proof}

\begin{lem}\label{lem:tp_wf_s}
  If \Cfg\ is \wf,
  then $\S \utu \End$
  or $\S \utu \Choice{\l}[\T]_i$.
\end{lem}
\begin{proof}
  Since \Cfg\ is \wf, by~\Cref{def:cfg_wf} $\exists \Constraints$ such that $\Val\models\Constraints$ and $\emptyset;\Constraints\Entails\S$.
  We proceed by induction on the derivation, analysing the structure of \S.
  By~\Cref{def:up_to_unfolding}, the only interesting case is $\S=\mu\alpha.\S'{2}$ and  $\S\utu\S'\Subst{\mu\alpha.\S'}{\alpha}$.
  By rule \Rule{rec}:
  \[
    \infer[\Rule{rec}]{%
      \emptyset;\Constraints\Entails\mu\alpha.\S'
    }{%
      \alpha:\Constraints;\Constraints\Entails\S'
    }
  \]

  \noindent By~\Cref{lem:tp_substitution},
  \(
  \alpha:\Constraints;\Constraints\Entails\S'
  \iff
  \emptyset;\Constraints\Entails\S'\Subst{\mu\alpha.\S'}{\alpha}
  \).
  Therefore, it holds that \Cfg,{\S'\Subst{\mu\alpha.\S'}{\alpha}}\ is \wf, and so, we obtain our thesis by the induction hypothesis:
  \[
    \Cfg,{\S'\Subst{\mu\alpha.\S'}{\alpha}}
    ~\text{is\ \wf}
    \implies
    \Cfg,{\S'{2}}
    ~\text{is\ \wf}
  \]

  \noindent where $\S'\utu\S'{2}$.
  Since types are \emph{contractive}, we rule out the possibility of $\S'=\mu\alpha_1\dots\mu\alpha_n.\alpha_i$, for some $i\leq n$ where $1<n\in \SetNatural$.
\end{proof}

\subsection{Live Configurations}

\begin{lem}\label{lem:tp_wf_cfg_is_live}
  If \Cfg\ is \wf,
  then \Cfg\ is \live.
\end{lem}
\begin{proof}[\Cref{lem:tp_rec_unfold}]
  By~\Cref{def:cfg_wf}, $\exists\Constraints$
  such that $\Val\models\Constraints$
  and $\emptyset;\Constraints\Entails\S$.
  We proceed by induction on the derivation of $\emptyset;\Constraints\Entails\S$ by the \wfness\ rules in~\Cref{fig:type_wf_rules}:
  \begin{caseanalysis}
    %
    \item\label{case:tp_wf_cfg_is_live_choice}
    If rule \Rule{choice}, then $\S=\Choice{\l}[\T]_i$ and
    $\Constraints=\Past{\bigvee_{i\in I}\Constraints_i}$.
    By~\Cref{def:past_delta}, $\exists t$ such that $\Val+\t\models\bigvee_{i\in I}\Constraints_i$.
    It follows~\Cref{def:cfg_fe} that \Cfg\ is \fe.
    Following~\Cref{def:cfg_live} we obtain our thesis.
    %
    \item\label{case:tp_wf_cfg_is_live_rec_def}
    If rule \Rule{rec}, then $\S=\mu\alpha.\S'$.
    By~\Cref{lem:tp_wf_cfg_utu}, since \Cfg,{\mu\alpha.\S'}\ is \wf, if $\S\utu\S'{2}$ then \Cfg,{\S'{2}}\ is \wf.
    By~\Cref{def:up_to_unfolding}, $\mu\alpha.\S'\utu\S'\Subst{\mu\alpha.\S'}{\alpha}=\S'{2}$.
    Therefore, since \Cfg,{\S'\Subst{\mu\alpha.\S'}{\alpha}}\ is \wf, the thesis follows by induction hypothesis on \Cfg,{\S'\Subst{\mu\alpha.\S'}{\alpha}}\ being \live.
    %
    \item\label{case:tp_wf_cfg_is_live_end}
    If rule \Rule{end}, then $\S=\End$.
    By~\Cref{def:cfg_live}, \Cfg,{\End}\ is \live.
    %
    \item\label{case:tp_wf_cfg_is_live_rec_call}
    Rule \Rule{var} is not applicable following~\Cref{lem:tp_rec_unfold}, as $\S\neq\alpha$.
    \qedhere
  \end{caseanalysis}
\end{proof}

\begin{lem}\label{lem:tp_wf_cfg_past_closed}
  If \Cfg{\Val+\t}\ is \wf, then \Cfg\ is \wf.
\end{lem}
\begin{proof}
  By \Cref{lem:tp_wf_s}, $\S \utu \End$ or $\S \utu \Choice{\l}[\T]_i$.
  By \Cref{lem:tp_wf_cfg_utu}, together with rules \Rule{end} and \Rule{choice} we have that $\emptyset;\Constraints\Entails\S$
  where $\Constraints = \True$
  or $\Constraints = \Past{\bigvee_{i\in I}\Constraints_i}$.
  Obviously, $\Val \models \True$.
  The fact $\Val \models \Past{\bigvee_{i\in I}\Constraints_i}$ instead follows from the definition of $\Past{}$.
\end{proof}

\begin{lem}\label{lem:tp_wf_type_fe_cfg_wf}
  If $\emptyset;\Constraints\Entails\S$
  and \Cfg\ is \fe,
  then \Cfg\ is \wf.
\end{lem}
\begin{proof}[\Cref{def:cfg_fe,lem:tp_wf_type_fe_cfg_wf}]
  By~\Cref{def:cfg_fe}, $\exists \t$ \st\ %
  \(
  \Cfg \Trans:{\t} \Cfg{\Val+\t}
  \Trans:{\Comm\Msg}
  \Cfg'
  \).
  We proceed by induction on the last rule applied for the transition
  \(
  \Cfg{\Val+\t} \Trans:{\Comm\Msg}
  \) of those %
  in~\Cref{eq:type_semantics_tuple}:
  \begin{caseanalysis}
    %
    \item\label{case:tp_wf_type_fe_cfg_wf_act}
    If rule \Rule{act}, then $\S=\Choice{\l}[\T]_i$ and:
    \[
      \infer[\Rule{act}]{
      \Cfg{\Val'{2}},{\Choice{\l}[\T]_i}%
      \Trans:{\Comm_j\Msg}%
      \Cfg{\Val'{2}\Reset_j},{\S_j}%
      }{
      \Val'{2}\models\delta_j%
      \quad%
      \Msg={\lT<>_j}%
      \quad%
      j\in I%
      }%
    \]

    \noindent where $\Comm=\Comm_j$ and $\Val'{2}=\Val+\t$.
    Since $\emptyset;\Constraints\Entails\Choice{\l}[\T]_i$
    by hypothesis, by rule \Rule{choice} we conclude:
    \[
      \RecTypeEnv;~
      \Past_iV_{i\in I} \Entails \Choice{\l}[\T]_i
    \]

    \noindent where $\RecTypeEnv=\emptyset$ and $\Constraints=\Past{\bigvee_{i\in I}\Constraints_i}$.
    By~\Cref{def:cfg_wf}, it remains to show that $\Val\models\Past{\bigvee_{i\in I}\Constraints_i}$.
    First we show that $\Val+\t\models\Past{\bigvee_{i\in I}\Constraints_i}$.
    By the premise of rule \Rule{act}, $\Val'{2}\models\Constraints_j$ (where $\Val'{2}=\Val+\t$).
    Since $\Constraints_j\models\Past{\bigvee_{i\in I}\Constraints_i}$, it follows that
    \(
    \Val+\t
    \models
    \Past{\bigvee_{i\in I}\Constraints_i}
    \)
    and therefore, \Cfg{\Val+\t},{\Choice{\l}[\T]_i}\ is \wf.
    By the conclusion of rule \Rule{choice}, the set of constraints $\Constraints=\Past{\bigvee_{i\in I}\Constraints_i}$ only require a minimum of \emph{one weak past} constraint to be satisfied for the entire set constraint to be satisfied.
    Therefore, given that we know $\Val+\t\models\Past{\bigvee_{i\in I}\Constraints_i}$, it follows that $\Val\models\Past{\bigvee_{i\in I}\Constraints_i}$ must also hold,
    and we obtain our thesis.
    %
    \item\label{case:tp_wf_type_fe_cfg_wf_unfold}
    If rule \Rule{unfold}, then $\S=\mu\alpha.\S'{2}$ and:
    \[
      \infer[\Rule{unfold}]{
      \Cfg{\Val+\t},{\mu\alpha.\S'{2}}%
      \Trans:{\ell}%
      \Cfg'%
      }{
      \Cfg{\Val+\t},{\S'{2}\Subst{\mu\alpha.\S'{2}}{\alpha}}%
      \Trans:{\ell}%
      \Cfg'%
      }%
    \]

    \noindent By the induction hypothesis we have that
    \(
    \Cfg{\Val+\t},{\S'{2}\Subst{\mu\alpha.\S'{2}}{\alpha}}
    \)
    is \wf.
    Then by \Cref{lem:tp_wf_cfg_utu} we have that $\Cfg{\Val+\t},{\mu\alpha.\S'{2}}$ is \wf.
    Finally, by \Cref{lem:tp_wf_cfg_past_closed}, $\Cfg,{\mu\alpha.\S'{2}}$ is \wf, as required.
    \qedhere
  \end{caseanalysis}
\end{proof}

\begin{lem}\label{lem:tp_en_cfg_dual_en}
  If $\Cfg\Trans:{\Comm\Msg}$,
  then $\Cfg,{\D}\Trans:{\Dual{\Comm}\Msg}$.
\end{lem}
\begin{proof}
  We proceed by induction on the last rule applied for the transition $\Cfg\Trans:{\Comm\Msg}$ of those %
  in~\Cref{eq:type_semantics_tuple}:
  \begin{caseanalysis}
    %
    \item\label{case:tp_en_cfg_dual_en_act}
    If rule \Rule{act}, then $\S=\Choice{\l}[\T]_i$ and:
    \[
      \infer[\Rule{act}]{
      \Cfg,{\Choice{\l}[\T]_i}%
      \Trans:{\Comm_j\Msg}%
      \Cfg{\Val\Reset_j},{\S_j}%
      }{
      \Val\models\delta_j%
      \quad%
      \Msg={\lT<>_j}%
      \quad%
      j\in I%
      }%
    \]

    \noindent where $\Comm=\Comm_j$.
    By~\Cref{def:type_duality}, $\D=\DualChoice_i$.
    Therefore, as the preconditions (\Constraints_i) in \S\ and \D\ are identical, it follows that the premise of rule \Rule{act} shown above is equally applicable to \Cfg,{\D}.
    With the only difference being the interactions having opposite directions (sending or receiving), we obtain our thesis.
    %
    \item\label{case:tp_en_cfg_dual_en_unfold}
    If rule \Rule{unfold}, then $\S=\mu\alpha.\S'{2}$ and:
    \[
      \infer[\Rule{unfold}]{
      \Cfg,{\mu\alpha.\S'{2}}%
      \Trans:{\ell}%
      \Cfg'%
      }{
      \Cfg,{\S'{2}\Subst{{\mu\alpha.\S'{2}}}{\alpha}}%
      \Trans:{\ell}%
      \Cfg'%
      }%
    \]

    \noindent where $\ell=\Comm\Msg$.
    By~\Cref{def:type_duality}, $\D=\mu\alpha.\D'{2}$.
    The thesis follows by induction:
    \[
      \Cfg,{\S'{2}\Subst{\mu\alpha.\S'{2}}{\alpha}}
      \Trans:{\Comm\Msg}
      \quad
      \implies
      \quad
      \Cfg,{\D'{2}\Subst{\mu\alpha.\D'{2}}{\alpha}}
      \Trans:{\Dual{\Comm}\Msg}
      \qedhere
    \]
  \end{caseanalysis}
\end{proof}

\subsection{Configuration Transitions}

\begin{lem}\label{lem:tp_cfg_trans}
  Let \Cfg\ be \wf.
  \Cref{claim:tp_cfg_trans_tick,claim:tp_cfg_trans_act} both hold.
\end{lem}
\begin{clm}\label{claim:tp_cfg_trans_tick}
  \(
  \Cfg\Trans:{\t}\Cfg'
  \implies
  {\Val'=\Val+\t}
  \land {\S\utu\S'}
  \)
\end{clm}
\begin{clm}\label{claim:tp_cfg_trans_act}
  \(
  \Cfg\Trans:{\Comm\Msg}\Cfg'
  \implies
  \begin{array}[t]{@{}l@{\ }l@{}}
    \exists \Constraints                    & \ST 
    {\emptyset;\Constraints\Entails\S}
    \land
    {\Val\models\Constraints}
    \\
    \mathllap{\land\ }\exists \Constraints' & \ST %
    {\Constraints'=\Constraints\Reset->0}
    \land
    {\Val'=\Val\Reset->0}
    \\
    \mathllap{\land\ }\exists \FutureEnv    & \ST %
    {\Val'\models\Constraints'\models\FutureEnv}
    \land
    {\emptyset;\FutureEnv\Entails\S'}
  \end{array}
  \)
\end{clm}
\begin{proof}
  We proceed by addressing each claim in turn:
  \begin{description}[font=\normalfont]
    %
    \item[\Cref{claim:tp_cfg_trans_tick}]\label{case:tp_cfg_trans_tick}
          We proceed by induction on the derivation of the transition
          \(
          \Cfg\Trans:{\t}\Cfg'
          \)
          analysing the last rule applied of those given %
          in~\Cref{eq:type_semantics_tuple}:
          \begin{caseanalysis}
            %
            \item\label{case:tp_cfg_trans_tick_tick}
            If rule \Rule{tick},
            then the thesis holds as $\Val'=\Val+\t$ and $\S=\S'$.
            %
            \item\label{case:tp_cfg_trans_tick_unfold}
            If rule \Rule{unfold},
            then $\S=\mu\alpha.\S'{2}$ and:
            \[
              \infer[\Rule{unfold}]{
              \Cfg,{\mu\alpha.\S'{2}}%
              \Trans:{\ell}%
              \Cfg'%
              }{
              \Cfg,{\S'{2}\Subst{{\mu\alpha.\S'{2}}}{\alpha}}%
              \Trans:{\ell}%
              \Cfg'%
              }%
            \]

            \noindent where $\ell=\t$.
            By~\Cref{def:up_to_unfolding}, $\mu\alpha.\S'{2}\utu\S'{2}\Subst{\mu\alpha.\S'{2}}{\alpha}$.
            By~\Cref{lem:tp_wf_cfg_utu}, \Cfg,{\S'{2}\Subst{\mu\alpha.\S'{2}}{\alpha}}\ is \wf.
            By the induction hypothesis,
            \(
            \S'{2}\Subst{\mu\alpha.\S'{2}}{\alpha}
            \utu
            \allowbreak
            \S'
            \).
            By the transitivity equation in~\Cref{def:up_to_unfolding},
            $\mu\alpha.\S'{2} \utu \S'$, as required.
          \end{caseanalysis}
          %
    \item[\Cref{claim:tp_cfg_trans_act}]\label{case:tp_cfg_trans_act}
          We proceed by induction on the derivation of the transition $\Cfg\Trans:{\Comm\Msg}\Cfg'$, analysing the last rule applied of those given %
          in~\Cref{eq:type_semantics_tuple}:
          \begin{caseanalysis}
            %
            \item\label{case:tp_cfg_trans_act_act}
            If rule \Rule{act},
            then $\S=\Choice{\l}[\T]_i$ and:
            \[
              \infer[\Rule{act}]{
              \Cfg,{\Choice{\l}[\T]_i}%
              \Trans:{\Comm_j\Msg}%
              \Cfg{\Val\Reset_j},{\S_j}%
              }{
              \Val\models\delta_j%
              \quad%
              \Msg={\lT<>_j}%
              \quad%
              j\in I%
              }%
            \]

            \noindent where $\Val'=\Val\Reset_j$, $\S'=\S_j$ and $\Constraints=\Constraints_j$.
            Since \Cfg\ is \wf, by~\Cref{def:cfg_wf}, $\exists \Constraints$ such that $\Val\models\Constraints$ and
            \(
            \emptyset;\Constraints\Entails
            \S
            \).
            By rule \Rule{choice}:
            \[
              \infer[\Rule{choice}]{
              \RecTypeEnv;~\Past_iV_{i\in I}
              \Entails
              \Choice{\l}[\T]{\Constraints}[\Resets].{\S}_i
              }{%
              \begin{array}{@{}c @{\hspace{2ex}} l@{}}
                \forall i\in I:
                \RecTypeEnv;~\FutureEnv_i
                \Entails
                \S_i
                ~\land~
                \Constraints_i\Reset_i\models\FutureEnv_i
                 & \text{(feasibility)}
                \\[0.5ex]
                \forall i,j\in I:
                i \neq j
                \implies
                \Constraints_i\land\Constraints_j\models\False
                ~\vee~
                \Comm_i=\Comm_j
                 & \text{(mixed-choice)}
                \\[0.5ex]
                \forall i\in I:
                \T_i=\Del'{2}
                \implies
                \RecTypeEnv;~\FutureEnv'\Entails\S'{2}
                ~\land~
                \Constraints'{2}\models\FutureEnv'
                 & \text{(delegation)}
              \end{array}
              }
            \]

            \noindent where $\RecTypeEnv=\emptyset$.
            As in~\Cref{case:tp_wf_type_fe_cfg_wf_act} of~\Cref{lem:tp_wf_type_fe_cfg_wf}, the set of constraints $\Past{\bigvee_{i\in I}\Constraints_i}$ requires as minimum the \emph{weak past} of only \emph{one} set of constraints (\Constraints_j) to be satisfied for the entire set of constraints to be satisfied; hence $\Constraints_j\models\Past{\bigvee_{i\in I}\Constraints_i}$.
            It remains to show:
            \(
            \Val\Reset_j\models\Constraints_j\Reset_j\models\FutureEnv
            \)
            and $\emptyset;\FutureEnv\Entails\S_j$.
            Clearly, by the $\text{(feasibility)}$ premise of rule \Rule{choice} it holds that $\Constraints_j\Reset_j\models\FutureEnv$ and $\emptyset;\FutureEnv\Entails\S_j$ (as $\RecTypeEnv=\emptyset$).
            By the premise of rule \Rule{act}, it holds that $\Val\models\Constraints_j$, and therefore it follows that $\Val\Reset_j\models\Constraints_j\Reset_j$.
            %
            \item\label{case:tp_cfg_trans_act_unfold}
            If rule \Rule{unfold},
            then $\S=\mu\alpha.\S'{2}$ and:
            \[
              \infer[\Rule{unfold}]{
              \Cfg,{\mu\alpha.\S'{2}}%
              \Trans:{\ell}%
              \Cfg'%
              }{
              \Cfg,{\S'{2}\Subst{{\mu\alpha.\S'{2}}}{\alpha}}%
              \Trans:{\ell}%
              \Cfg'%
              }%
            \]

            \noindent where $\ell=\Comm\Msg$.
            The rest follows~\Cref{case:tp_cfg_trans_tick_unfold} of~\Cref{claim:tp_cfg_trans_tick} (above).
            \qedhere
          \end{caseanalysis}
  \end{description}
\end{proof}

\begin{lem}\label{lem:tp_cfg_que_trans}
  Let \Cfg\ be \wf.
  \Cref{claim:tp_cfg_que_trans_send,claim:tp_cfg_que_trans_recv,claim:tp_cfg_que_trans_que,claim:tp_cfg_que_trans_time} all hold.
\end{lem}
\begin{clm}\label{claim:tp_cfg_que_trans_send}
  $\Cfg;{\Que} \Trans:{!\Msg} \Cfg;{\Que}'
    \implies
    \Que'=\Que
    \land
    \Cfg\Trans:{!\Msg}\Cfg'$
\end{clm}
\begin{clm}\label{claim:tp_cfg_que_trans_recv}
  $\Cfg;{\Que} \Trans:{\tau} \Cfg;{\Que}'
    \implies
    \exists \Msg ~\ST
    \Que=\Msg;\Que'
    \land
    \Cfg\Trans:{?\Msg}\Cfg'$
\end{clm}
\begin{clm}\label{claim:tp_cfg_que_trans_que}
  $\Cfg;{\Que} \Trans:{?\Msg} \Cfg;{\Que}'
    \implies
    \S'=\S
    \land
    \Que'=\Que;\Msg
    \land
    \Val'=\Val$
\end{clm}
\begin{clm}\label{claim:tp_cfg_que_trans_time}
  $\Cfg;{\Que} \Trans:{\t} \Cfg;{\Que}'
    \implies
    \S'\utu\S
    \land
    \Que'=\Que
    \land
    \Val'=\Val+\t$
\end{clm}
\begin{proof}
  We proceed addressing each claim in turn, using the rules in~\Cref{eq:type_semantics_triple}:
  \begin{description}[font=\normalfont]
    %
    \item[\Cref{claim:tp_cfg_que_trans_send}]
          By rule \Rule{send}
          the claim holds.
          By the premise $\Cfg\Trans:{!\Msg}\Cfg'$.
          %
    \item[\Cref{claim:tp_cfg_que_trans_recv}]
          By rule \Rule{recv}
          the claim holds.
          By the premise $\Cfg\Trans:{?\Msg}\Cfg'$.
          %
    \item[\Cref{claim:tp_cfg_que_trans_que}]
          By rule \Rule{que}
          the claim holds.
          %
    \item[\Cref{claim:tp_cfg_que_trans_time}]
          By rule \Rule{time}
          the claim holds.
          By the premise $\Val'=\Val+\t$ via rule \Rule{tick}.
          (See~\Cref{claim:tp_cfg_trans_tick} of~\Cref{lem:tp_cfg_trans}.)
          \qedhere
  \end{description}
\end{proof}

\begin{lem}\label{lem:tp_wf_cfg_safe_mc}
  Let \Cfg\ be \wf.
  If $\Cfg\Trans:{\Comm\Msg}$
  and $\Cfg\Trans:{\Comm'\Msg'}$
  then $\Comm=\Comm'$.
\end{lem}
\begin{proof}
  We proceed by induction on the derivation of the transition $\Cfg\Trans:{\Comm\Msg}$, analysing the last rule applied of those %
  in~\Cref{eq:type_semantics_tuple}.
  We only show the case of rule \Rule{act}, as the only other applicable case
  (rule \Rule{unfold}) follows by induction hypothesis as in~\Cref{claim:tp_cfg_trans_tick,claim:tp_cfg_trans_act} of~\Cref{lem:tp_cfg_trans}.
  Therefore, if rule \Rule{act}, then $\S=\Choice{\l}[\T]_i$ and:
  \begin{equation}\label{eq:tp_wf_cfg_safe_mc}
    \infer[\Rule{act}]{
    \Cfg,{\Choice{\l}[\T]_i}%
    \Trans:{\Comm_j\Msg}%
    \Cfg{\Val\Reset_j},{\S_j}%
    }{
    \Val\models\delta_j%
    \quad%
    \Msg={\lT<>_j}%
    \quad%
    j\in I%
    }%
  \end{equation}

  \noindent where $\Comm=\Comm_j$ and $\Msg=\lT<>_j$.
  Since \Cfg\ is \wf, by~\Cref{def:cfg_wf}, $\exists \Constraints$ such that $\Val\models\Constraints$ and $\emptyset;\Constraints\Entails\Choice{\l}[\T]_i$.
  The only possible rule is \Rule{choice}:
  \[
    \infer[\Rule{choice}]{
    \RecTypeEnv;~\Past_iV_{i\in I}
    \Entails
    \Choice{\l}[\T]{\Constraints}[\Resets].{\S}_i
    }{%
    \begin{array}{@{}c @{\hspace{2ex}} l@{}}
      \forall i\in I:
      \RecTypeEnv;~\FutureEnv_i
      \Entails
      \S_i
      ~\land~
      \Constraints_i\Reset_i\models\FutureEnv_i
       & \text{(feasibility)}
      \\[0.5ex]
      \forall i,j\in I:
      i \neq j
      \implies
      \Constraints_i\land\Constraints_j\models\False
      ~\vee~
      \Comm_i=\Comm_j
       & \text{(mixed-choice)}
      \\[0.5ex]
      \forall i\in I:
      \T_i=\Del'{2}
      \implies
      \emptyset;~\FutureEnv'\Entails\S'{2}
      ~\land~
      \Constraints'{2}\models\FutureEnv'
       & \text{(delegation)}
    \end{array}
    }
  \]

  \noindent As in~\Cref{lem:tp_wf_type_fe_cfg_wf}, it follows that $\Constraints_j\models\Past{\bigvee_{i\in I}\Constraints_i}$.
  Since $\Cfg\Trans:{\Comm'\Msg'}$, we proceed by induction on the derivation of the transition, analysing the last rule applied. Again, we omit the case by rule \Rule{unfold}.
  By rule \Rule{act}, it follows similarly to~\Cref{eq:tp_wf_cfg_safe_mc}: for some $k \in I$, $\Comm'=\Comm_k$, $\Msg=\lT<>_k$ and $\Val\models\Constraints_k$.
  In the case where $j=k$, then the thesis coincides with the hypothesis.
  Otherwise, if $j\neq k$, then by the $\text{(mixed-choice)}$ premise of rule \Rule{choice},
  $\Constraints_j\land\Constraints_k\models\False$ or $\Comm_j=\Comm_k$.
  Since $\Val\models\Constraints_j$ and $\Val\models\Constraints_k$, it follows that $\Constraints_j\land\Constraints_k\not\models\False$.
  Therefore, it must be that $\Comm_j=\Comm_k$ and we obtain our thesis.
  \qedhere
\end{proof}

\subsection{System Transition Preservations}\label{ssec:tp_system_trans_preserve}

\begin{lem}\label{lem:tp_cfg_trans_preserve_wf}
  Let \Cfg_1\ and \Cfg_2\ be \wf,
  and $\Cfg;{\Que}_1 \Compat \Cfg;{\Que}_2$.
  \linebreak%
  If \(
  {\Cfg*;{\Que}_1 \mid \Cfg*;{\Que}_2}
  \ \Trans\ %
  {\Cfg;{\Que}_1' \mid \Cfg;{\Que}_2'}
  \),
  then \Cfg_1'\ and \Cfg_2'\ are \wf.
\end{lem}
\begin{proof}
  We proceed by cases on the rule used in the derivation of the transition:
  \[
    {\Cfg;{\Que}_1 \mid \Cfg;{\Que}_2}
    \ \Trans\ %
    {\Cfg;{\Que}_1' \mid \Cfg;{\Que}_2'}
  \]

  \noindent analysing the last rule applied
  of those given %
  in~\Cref{eq:type_semantics_system}:
  \begin{caseanalysis}
    %
    \item\label{case:tp_cfg_trans_preserve_wf_wait}
    If rule \Rule{wait}, then:
    \[
      \infer[\Rule{wait}]{
      {\Cfg;{\Que}_1 \mid \Cfg;{\Que}_2}
      \Trans:{\t}
      {\Cfg;{\Que}_1' \mid \Cfg;{\Que}_2'}
      }{
      {\Cfg;{\Que}_1}
      \Trans:{\t}
      {\Cfg;{\Que}_1'}
      &
      {\Cfg;{\Que}_2}
      \Trans:{\t}
      {\Cfg;{\Que}_2'}
      }
    \]

    \noindent Hereafter, we only show the case for \Cfg;{\Que}_1, as the transition for \Cfg;{\Que}_2\ is analogous.
    We proceed by inner induction on the derivation of the transition
    \(
    {\Cfg;{\Que}_1}
    \Trans:{\t}
    {\Cfg;{\Que}_1'}
    \)
    analysing the last rule applied.
    By rule \Rule{time}:
    \[
      \infer[\Rule{time}]{
      \Cfg;{\Que}_1%
      \Trans:{\t}%
      \Cfg{\Val_1'},{\S_1'};{\Que_1}%
      }{
      \begin{array}[t]{@{}c @{\quad} l@{}}
        %
        \Cfg_1\Trans:{\t}\Cfg{\Val_1'},{\S_1'}%
         & \text{(configuration)} %
        %
        \\[0.5ex]%
        \Cfg_1~\text{is \fe*}%
        \implies%
        \Cfg{\Val_1'},{\S_1'}~\text{is \fe*}%
         & \text{(persistency)}   %
        %
        \\[0ex]%
        \forall \t'<\t:\Cfg{\Val_1+\t'},{\S_1};{\Que_1}\Trans|:{\tau}%
         & \text{(urgency)}       %
      \end{array}
      }%
    \]

    \noindent By inner induction on the derivation of the transition
    \(
    \Cfg_1\Trans:{\t}\Cfg{\Val_1'},{\S_1'}%
    \)
    analysing the last rule applied of those %
    in~\Cref{eq:type_semantics_tuple}:
    \begin{caseanalysis}
      %
      \item\label{case:tp_cfg_trans_preserve_wf_wait_tick}
      If rule \Rule{tick},
      then by~\Cref{claim:tp_cfg_trans_tick} of~\Cref{lem:tp_cfg_trans}, $\Val_1'=\Val_1+\t$, $\Que_1'=\Que_1$ and $\S_1'\utu\S_1$.
      If $\t=0$ then $\Val_1=\Val_1+\t$ and therefore, the thesis coincides with the hypothesis.
      By the $\text{(urgency)}$ premise of rule \Rule{time}, \t\ must be valued such that no message can be received from the queue \Que_1\ via rule \Rule{recv}.
      If $\t=0$, then $\nexists \t'<\t$, and the $\text{(urgency)}$ premise of rule \Rule{time} always holds, even if $\Cfg;{\Que}_2\Trans:{\tau}$.
      However, since $\Val_1'=\Val_1$, this is the same as no transition occurring via rule \Rule{time}.
      Otherwise, $\t>0$.
      Since \Cfg_1\ is \wf, it follows~\Cref{lem:tp_cfg_fe}, \Cfg_1\ is also \fe.
      Since \Cfg_1\ is \fe, it follows the $\text{(persistency)}$ premise of rule \Rule{time} that \Cfg{\Val_1'},{\S_1'} is also \fe.
      By~\Cref{lem:tp_wf_type_fe_cfg_wf}, \Cfg{\Val_1+\t},{\S_1'}\ is \wf.
      %
      %
      \item\label{case:tp_cfg_trans_preserve_wf_wait_unfold}
      If rule \Rule{unfold},
      then $\S_1=\mu\alpha.\S_1'{2}$ and:
      \[
        \infer[\Rule{unfold}]{%
        \Cfg{\Val_1},{\mu\alpha.\S_1'{2}}\Trans:{\ell}\Cfg{\Val_1'},{\S_1'}%
        }{%
        \Cfg{\Val_1},{\S_1'{2}\Subst{\mu\alpha.\S_1'{2}}{\alpha}}\Trans:{\ell}\Cfg{\Val_1'},{\S_1'}%
        }
      \]

      \noindent where $\ell=\t$.
      Since \Cfg{\Val_1},{\mu\alpha.\S_1'{2}}\ is \wf,
      $\exists \Constraints$ such that $\Val_1\models\Constraints$ and $\emptyset;\Constraints\Entails\mu\alpha.\S_1'{2}$.
      By rule \Rule{rec}:
      \[
        \infer[\Rule{rec}]{%
        \emptyset;\Constraints\Entails\mu\alpha.\S_1'{2}
        }{%
        \alpha:\Constraints;\Constraints\Entails\S_1'{2}\Subst{\mu\alpha.\S_1'{2}}{\alpha}
        }
      \]

      \noindent By~\Cref{lem:tp_substitution}:
      \(
      \alpha:\Constraints;\Constraints\Entails\S_1'{2}\Subst{\mu\alpha.\S_1'{2}}{\alpha}
      \iff
      \emptyset;\Constraints\Entails\S_1'{2}
      \).
      By~\Cref{def:up_to_unfolding},
      it follows
      \(
      \mu\alpha.\S_1'{2}\utu\S_1'{2}\Subst{\mu\alpha.\S_1'{2}}{\alpha}
      \).
      By~\Cref{lem:tp_wf_cfg_utu_aux}, since
      \(
      \emptyset;\Constraints\Entails\S_1'{2}
      \),
      then
      \(
      \emptyset;\Constraints\Entails\mu\alpha.\S_1'{2}
      \).
      By~\Cref{lem:tp_wf_cfg_utu},
      if \Cfg{\Val_1},{\S_1'{2}\Subst{\mu\alpha.\S_1'{2}}{\alpha}}\ is \wf,
      then \Cfg{\Val_1},{\mu\alpha.\S_1'{2}}\ is \wf.
      The thesis holds by the induction hypothesis. (As in~\Cref{case:tp_wf_type_fe_cfg_wf_unfold} in~\Cref{lem:tp_wf_type_fe_cfg_wf}.)
    \end{caseanalysis}
    %
    \item\label{case:tp_cfg_trans_preserve_wf_com}
    If rule \Rule{com}[l]*, then:
    \[
      \infer[\Rule{com}[l]*]{
      {\Cfg;{\Que}_1\mid\Cfg;{\Que}_2}%
      \Trans:{\tau}%
      \Cfg;{\Que}_1'%
      \mid%
      \Cfg;{\Que}_2'%
      }{
      \Cfg;{\Que}_1%
      \Trans:{!\Msg}%
      \Cfg;{\Que}_1'%
      &%
      \Cfg;{\Que}_2%
      \Trans:{?\Msg}%
      \Cfg;{\Que}_2'%
      }%
    \]

    \noindent First, we proceed by cases on the derivation of the transition
    \(
    \Cfg;{\Que}_2%
    \Trans:{?\Msg}%
    \Cfg;{\Que}_2'%
    \)
    analysing the last rule applied.
    By rule \Rule{que}:
    \[
      \Cfg;{\Que}_2%
      \Trans:{?\Msg}%
      \Cfg{\Val_2},{\S_2};{\Que_2;\Msg}%
      \quad%
      \Rule{que}%
    \]

    \noindent where, as in~\Cref{claim:tp_cfg_que_trans_que} of~\Cref{lem:tp_cfg_que_trans}, $\Val_2'=\Val_2$, $\S_2'=\S_2$ and $\Que_2'=\Que_2;\Msg$.
    Therefore, it follows that \Cfg_2'\ is \wf.
    Next, by cases on the derivation of the transition
    \(
    \Cfg;{\Que}_1%
    \Trans:{!\Msg}%
    \Cfg;{\Que}_1'%
    \) analysing the last rule applied.
    By rule \Rule{send}:
    \[
      \infer[\Rule{send}]{
      \Cfg;{\Que}_1%
      \Trans:{!\Msg}%
      \Cfg;{\Que}_1'%
      }{
      \Cfg_1%
      \Trans:{!\Msg}%
      \Cfg_1'%
      }%
    \]

    \noindent where, as in~\Cref{claim:tp_cfg_que_trans_send} of~\Cref{lem:tp_cfg_que_trans}, $\Que_1'=\Que_1$ and
    \(
    \Cfg_1%
    \Trans:{!\Msg}%
    \Cfg_1'%
    \).
    As in~\Cref{claim:tp_cfg_trans_tick,claim:tp_cfg_trans_act} of~\Cref{lem:tp_cfg_trans}, we proceed to only show the case of rule \Rule{act} (as the only other applicable rule is rule \Rule{unfold}, which follows by induction).
    Therefore,
    \(
    \S_1 = \Choice{\l}[\T]_i
    \)
    and for some $j \in I$, $\Val_1\models\Constraints_j$, $\Msg=\lT<>_j$, $!=\Comm_j$, $\Val_1'=\Val_1\Reset_j$ and $\S_1'=\S_j$.
    It remains to show that \Cfg{\Val_1\Reset_j},{\S_j}\ is \wf.
    The thesis follows~\Cref{claim:tp_cfg_trans_act} of~\Cref{lem:tp_cfg_trans}, since
    \(
    \Val_1\Reset_j\models\Constraints_j\Reset_j\models\FutureEnv_j
    \)
    and $\emptyset;\FutureEnv_j\Entails\S_j$, for some $j \in I$.
    By~\Cref{def:cfg_wf}, \Cfg{\Val_1\Reset_j},{\S_j}\ is \wf.
    %
    \item\label{case:tp_cfg_trans_preserve_wf_par}
    If rule \Rule{par}[l]*, then:
    \[
      \infer[\Rule{par}[l]*]{
      {\Cfg;{\Que}_1\mid\Cfg;{\Que}_2}%
      \Trans:{\tau}%
      {\Cfg;{\Que}_1'\mid\Cfg;{\Que}_2}%
      }{
      \Cfg;{\Que}_1
      \Trans:{\tau}%
      \Cfg;{\Que}_1'%
      }%
    \]

    \noindent where, $\Val_2'=\Val_2$, $\S_2'=\S_2$ and $\Que_2'=\Que_2$.
    It holds that \Cfg_2'\ remains \wf.
    We proceed by inner induction on the derivation of the transition
    \(
    \Cfg;{\Que}_1
    \Trans:{\tau}%
    \Cfg;{\Que}_1'%
    \)
    analysing the last rule applied of those %
    in~\Cref{eq:type_semantics_triple}.
    By rule \Rule{recv}:
    \[
      \infer[\Rule{recv}]{
      \Cfg;{\Que}_1
      \Trans:{\tau}%
      \Cfg;{\Que}_1'%
      }{
      \Cfg_1%
      \Trans:{?\Msg}%
      \Cfg_1'%
      }%
    \]

    \noindent As in~\Cref{case:tp_cfg_trans_preserve_wf_com},
    the thesis follows~\Cref{claim:tp_cfg_trans_act} of~\Cref{lem:tp_cfg_trans}
    (except $?=\Comm_j$).
    %
    %
    \item Rules \Rule{com}[r]* and \Rule{par}[r]* are analogous to~\Cref{case:tp_cfg_trans_preserve_wf_com,case:tp_cfg_trans_preserve_wf_par} respectively.
    \qedhere
  \end{caseanalysis}
\end{proof}

\begin{lem}\label{lem:tp_compat_cfg_delay_empty_que}
  Let \Cfg_1\ and \Cfg_2\ be \wf, and $\Cfg;{\Que}_1 \Compat \Cfg;{\Que}_2$.
  \linebreak
  If ${\Cfg;{\Que}_1 \mid \Cfg;{\Que}_2}\ \Trans:{\t}\ {\Cfg;{\Que}_1' \mid \Cfg;{\Que}_2'}$
  and $\t>0$,
  then $\Que_1=\emptyset=\Que_2$.
\end{lem}
\begin{proof}
  We proceed by cases on the derivation of the transition
  \(
  \Cfg;{\Que}_1\mid\Cfg;{\Que}_2
  \Trans:{\t}
  \)
  analysing the last rule applied of those %
  in~\Cref{eq:type_semantics_system}.
  By rule \Rule{wait}:
  \[
    \infer[\Rule{wait}]{
    {\Cfg;{\Que}_1\mid\Cfg;{\Que}_2}%
    \Trans:{\t}%
    {\Cfg;{\Que}_1'\mid\Cfg;{\Que}_2'}%
    }{
    {\Cfg;{\Que}_1}
    \Trans:{\t}
    \Cfg;{\Que}_1'
    &
    {\Cfg;{\Que}_2}
    \Trans:{\t}
    \Cfg;{\Que}_2'
    }
  \]

  \noindent where, as in~\Cref{claim:tp_cfg_que_trans_time} of~\Cref{lem:tp_cfg_que_trans}, $\Que_1'=\Que_1$ and $\Que_2'=\Que_2$.
  Hereafter, we only show \Cfg;{\Que}_1\ as \Cfg;{\Que}_2\ is analogous.
  By cases on the derivation of the transition
  \(
  {\Cfg;{\Que}_1}
  \Trans:{\t}
  \Cfg;{\Que}_1'
  \)
  analysing the last rule applied.
  By rule \Rule{time}:
  \[
    \infer[\Rule{time}]{
    \Cfg;{\Que}_1%
    \Trans:{\t}%
    \Cfg{\Val_1'},{\S_1'};{\Que_1}%
    }{
    \begin{array}[t]{@{}c @{\quad} l@{}}
      %
      \Cfg_1\Trans:{\t}\Cfg{\Val_1'},{\S_1'}%
       & \text{(configuration)} %
      %
      \\[-2ex]\\%
      \Cfg_1~\text{is \fe*}%
      \implies%
      \Cfg{\Val_1'},{\S_1'}~\text{is \fe*}%
       & \text{(persistency)}   %
      %
      \\[-2.25ex]\\%
      \hphantom{\forall \t'<\t:\Cfg{\Val_1+\t'},{\S_1};{\Que_1}\Trans|:{\tau}}%
      \mathllap{\smash{\forall \t'<\t:\Cfg{\Val_1+\t'},{\S_1};{\Que_1}\Trans|:{\tau}}}%
       & \text{(urgency)}       %
    \end{array}
    }%
  \]

  \noindent where, as in~\Cref{claim:tp_cfg_que_trans_time} of~\Cref{lem:tp_cfg_que_trans}, $\Val_1'=\Val_1+\t$ and $\S_1'\utu\S_1$.
  Suppose by contradiction that one of the queues $\Que_1$ were \emph{non-empty}, and $\Que_1=\Msg;\Que_1'{2}$.
  Since $\Cfg;{\Msg;\Que}_1 \Compat \Cfg;{\Que}_2$, by~\Cref{cond:cfg_compat_recv} in~\Cref{def:cfg_compat}:
  \[
    \Que_1=\Msg;\Que_1'
    \implies
    \Cfg_1 \Trans:{?\Msg} \Cfg_1'
    \land \Cfg;{\Que}_1' \Compat \Cfg;{\Que}_2
  \]

  \noindent Such a transition is only viable by rule \Rule{recv}.
  By~\Cref{claim:tp_cfg_que_trans_recv} of~\Cref{lem:tp_cfg_que_trans}, it follows
  that $\Cfg_1 \Trans:{?\Msg} \Cfg_1'$ is made by the premise of rule \Rule{recv}.
  However, by the $\text{(urgency)}$ premise of rule \Rule{time}, $\forall \t'<\t$ no transitions by rule \Rule{recv} are viable
  -- which contradicts with~\Cref{cond:cfg_compat_recv} in~\Cref{def:cfg_compat}.
  Therefore, it cannot be that either queue is \emph{non-empty} for a transition \t, where $\t>0$.
  For such a transition, both queues must be \emph{empty}.
\end{proof}

\begin{lem}\label{lem:tp_cfg_trans_preserve_compat}
  Let \Cfg_1\ and \Cfg_2\ be \wf,
  and $\Cfg;{\Que}_1 \Compat \Cfg;{\Que}_2$.
  \linebreak%
  If ${\Cfg;{\Que}_1 \mid \Cfg;{\Que}_2}\ \Trans\ {\Cfg;{\Que}_1' \mid \Cfg;{\Que}_2'}$,
  then $\Cfg;{\Que}_1' \Compat \Cfg;{\Que}_2'$.
\end{lem}
\begin{proof}
  We proceed by cases on the derivation of the transition:
  \[
    {\Cfg;{\Que}_1 \mid \Cfg;{\Que}_2}\ \Trans\ {\Cfg;{\Que}_1' \mid \Cfg;{\Que}_2'}
  \]

  \noindent analysing the last rule applied
  of those given %
  in~\Cref{eq:type_semantics_system}:
  \begin{caseanalysis}
    %
    \item\label{case:tp_cfg_trans_preserve_compat_wait}
    If rule \Rule{wait}, then:
    \[
      \infer[\Rule{wait}]{
      {\Cfg;{\Que}_1\mid\Cfg;{\Que}_2}%
      \Trans:{\t}%
      {\Cfg;{\Que}_1'\mid\Cfg;{\Que}_2'}%
      }{
      \Cfg;{\Que}_1\Trans:{\t}\Cfg;{\Que}_1'%
      \quad%
      \Cfg;{\Que}_2\Trans:{\t}\Cfg;{\Que}_2'%
      }%
    \]

    \noindent By~\Cref{claim:tp_cfg_que_trans_time} of~\Cref{lem:tp_cfg_que_trans}, it follows that $\Val_1'=\Val_1+\t$, $\S_1'\utu\S_1$ and $\Que_1'=\Que_1$ (and similarly for \Val_2', \S_2'\ and \Que_2').
    If $\t=0$ then $\Val_1=\Val_1+\t$ (and $\Val_2=\Val_2+\t$) and therefore, the thesis coincides with the hypothesis.
    Else $\t>0$, and by~\Cref{lem:tp_compat_cfg_delay_empty_que}, $\Que_1=\emptyset=\Que_2$.
    By~\Cref{cond:cfg_compat_dual} of~\Cref{def:cfg_compat}:
    \(
    \Que_1=\emptyset=\Que_2
    \implies
    \S_1=\D_2 ~\land~ \Val_1=\Val_2
    \).
    Since $\Val_1+\t=\Val_2+\t$, it follows that $\Cfg{\Val_1+\t},{\S_1'};{\Que_1}\Compat\Cfg{\Val_2+\t},{\S_2'};{\Que_2}$.
    %
    \item\label{case:tp_cfg_trans_preserve_compat_com}
    If rule \Rule{com}[l]*, then:
    \[
      \infer[\Rule{com}[l]*]{
      {\Cfg;{\Que}_1\mid\Cfg;{\Que}_2}%
      \Trans:{\tau}%
      {\Cfg;{\Que}_1'\mid\Cfg;{\Que}_2'}%
      }{
      \Cfg;{\Que}_1\Trans:{!\Msg}\Cfg;{\Que}_1'%
      \quad%
      \Cfg;{\Que}_2\Trans:{?\Msg}\Cfg;{\Que}_2'%
      }%
    \]

    \noindent First, by induction on the derivation of the transition
    \(
    \Cfg;{\Que}_2\Trans:{?\Msg}\Cfg;{\Que}_2'
    \)
    analysing the last rule applied of those %
    in~\Cref{eq:type_semantics_triple}.
    By rule \Rule{que}:
    \[
      \Cfg;{\Que}_2\Trans:{?\Msg}\Cfg{\Val_2},{\S_2};{\Que_2;\,\Msg}%
      \quad%
      \Rule{que}%
    \]

    \noindent where $\Val_2'=\Val_2$, $\S_2'=\S_2$ and $\Que_2'=\Que_2;\Msg$.
    Next, by induction on the derivation of the transition $\Cfg;{\Que}_1\Trans:{!\Msg}\Cfg;{\Que}_1'$ analysing the last rule applied.
    By rule \Rule{send}:
    \[
      \infer[\Rule{send}]{
      \Cfg;{\Que}_1%
      \Trans:{!\Msg}%
      \Cfg{\Val_1'},{\S_1'};{\Que_1}%
      }{
      \Cfg_1%
      \Trans:{!\Msg}%
      \Cfg_1'%
      }%
    \]

    \noindent where $\Que_1'=\Que_1$.
    By inner induction on the derivation of the transition
    \(
    \Cfg_1%
    \Trans:{!\Msg}%
    \Cfg_1'%
    \)
    analysing the last rule applied of those %
    in~\Cref{eq:type_semantics_tuple}.
    We omit the case of rule \Rule{unfold}, since following~\Cref{lem:tp_wf_s}, $\S_1\utu\Choice{\l}[\T]_i$, and by \Cref{lem:tp_wf_cfg_utu_aux}, \Cfg{\Val_1},{\Choice{\l}[\T]_i}\ is \wf.
    Therefore, by rule \Rule{act}:
    \[
      \infer[\Rule{act}]{
      \Cfg{\Val_1},{\Choice{\l}[\T]_i}%
      \Trans:{\Comm_j\Msg}%
      \Cfg{\Val_1\Reset_j},{\S_j}%
      }{
      \Val_1\models\delta_j%
      \quad%
      \Msg={\lT<>_j}%
      \quad%
      j\in I%
      }%
    \]

    \noindent where $\Val_1'=\Val_1\Reset_j$, $\S_1'=\S_j$ and $\Msg=\lT<>_j$, for some $j\in I$ such that $\Val_1\models\Constraints_j$.
    If \Que_1\ were \emph{non-empty} (i.e.: $\Que_1=\Msg';\Que_1'{2}$),
    then by~\Cref{cond:cfg_compat_recv} of~\Cref{def:cfg_compat}, it must be that $\Cfg_1\Trans:{?\Msg'}$.
    However, by~\Cref{lem:tp_wf_cfg_safe_mc}, this is plainly not the case since $\Cfg_1\Trans:{!\Msg}$.
    Therefore, $\Que_1=\emptyset$.
    We proceed by case analysis on the contents of \Que_2:
    \begin{caseanalysis}
      %
      \item\label{case:tp_cfg_trans_preserve_compat_com_empty}
      If $\Que_2=\emptyset$,
      then by~\Cref{cond:cfg_compat_dual} of~\Cref{def:cfg_compat}, $\S_1=\D_2$ and $\Val_1=\Val_2$.
      Therefore, by the induction hypothesis:
      \[
        \begin{array}[c]{@{}l@{\ArrWrap}l@{}}
          \Cfg{\Val_1},{\Choice{\l}[\T]_i};{\emptyset}
          \mid
          \Cfg{\Val_2},{\DualChoice_i};{\emptyset}
          \Trans:{\tau}
          \\[\ArrayLTSLineSpacing]
           & \mathllap{
          \Cfg{\Val_1\Reset_j},{\S_j};{\emptyset}
          \mid
          \Cfg{\Val_2},{\DualChoice_i};{\Msg}
          }
        \end{array}
      \]

      \noindent It remains to show that:
      \(
      \Cfg{\Val_1\Reset_j},{\S_j};{\emptyset}
      \Compat
      \Cfg{\Val_2},{\DualChoice_i};{\Msg}
      \).
      Since $\Que_2'=\Msg;\Que_2'{2}$, then by~\Cref{cond:cfg_compat_recv_other} of~\Cref{def:cfg_compat}:
      \begin{equation}\label{eq:tp_cfg_trans_preserve_compat_com_empty}
        \Cfg{\Val_2},{\DualChoice_i}
        \Trans:{?\Msg}
        \Cfg_2'{2}
        \text{\ and\ }
        \Cfg{\Val_1\Reset_j},{\S_j};{\emptyset}
        \Compat
        \Cfg;{\Que}_2'{2}
      \end{equation}

      \noindent The thesis follows~\Cref{lem:tp_en_cfg_dual_en}, since
      \(
      \Cfg_1\Trans:{!\Msg}
      \)
      and $\Val_1=\Val_2$,
      then
      \(
      \Cfg{\Val_2},{\D_2}\Trans:{?\Msg}
      \).
      %
      \item\label{case:tp_cfg_trans_preserve_compat_com_neq}
      If $\Que_2=\Msg';\Que_2'{2}$,
      then $\Que_2'=\Msg';\Que_2'{2};\Msg$ and therefore, by~\Cref{cond:cfg_compat_recv_other} of~\Cref{def:cfg_compat}, it follows~\Cref{eq:tp_cfg_trans_preserve_compat_com_empty}.
      The thesis follows by the induction hypothesis:
      \[
        \Cfg{\Val_1},{\Choice{\l}[\T]_i};{\emptyset}
        \mid
        \Cfg{\Val_2},{\S_2};{\Que_2}
        \Trans:{\tau}
        \Cfg{\Val_1\Reset_j},{\S_j};{\emptyset}
        \mid
        \Cfg{\Val_2},{\S_2};{\Msg';\Que_2'{2};\Msg}
      \]

      \noindent If $\Que_2'{2}=\emptyset$, then it follows~\Cref{case:tp_cfg_trans_preserve_compat_com_empty}.
      Otherwise, it follows this case.
    \end{caseanalysis}
    %
    %
    \item\label{case:tp_cfg_trans_preserve_compat_par}
    If rule \Rule{par}[l]*, then by~\Cref{claim:tp_cfg_que_trans_recv} of~\Cref{lem:tp_cfg_que_trans} it follows that $\Que_1=\Msg;\Que_1'$, and by~\Cref{cond:cfg_compat_recv} in~\Cref{def:cfg_compat}, we obtain our thesis.
    %
    \item If rule \Rule{com}[r]* or rule \Rule{par}[r]*, then it follows~\Cref{case:tp_cfg_trans_preserve_compat_com,case:tp_cfg_trans_preserve_compat_par} respectively.
    \qedhere
  \end{caseanalysis}
\end{proof}

\subsection{System Progress}\label{ssec:tp_system_progress}

\begin{lem}\label{lem:tp_cfg_fe}
  Let \Cfg_1\ and \Cfg_2\ be \wf.
  If $\Cfg;{\Que}_1 \Compat \Cfg;{\Que}_2$,
  then \Cfg;{\Que}_1\ and \Cfg;{\Que}_2\ are \emph{final},
  or \(
  \exists \t:
  {\Cfg;{\Que}_1 \mid \Cfg;{\Que}_2}\ %
  \Trans:{\t,\tau}%
  \).
\end{lem}
\begin{proof}
  Since both \Cfg_1\ and \Cfg_2\ are \wf,
  by~\Cref{lem:tp_wf_cfg_is_live}, both \Cfg_1\ and \Cfg_2\ are \live.
  By~\Cref{def:type_progress}, \Cfg;{\Que}_1\ is \emph{final} if $\S_1\UTU\End$
  (and similarly for \S_2).
  By~\Cref{def:cfg_live} if $\S_1\NUTU\End$, then \Cfg_1\ is \fe\
  (and similarly for \Cfg_2).
  By~\Cref{def:cfg_fe}, $\exists \t'$ such that $\Cfg_1\Trans:{\t',\Comm\Msg}$
  (and similarly for \Cfg_2).
  Therefore, we have obtained our thesis as it holds that if, \Cfg;{\Que}_1\ is \emph{not} final the by~\Cref{def:cfg_live}, it \emph{must} be \fe\ by~\Cref{def:cfg_fe}, which coincides with the hypothesis.
\end{proof}

\begin{lem}\label{lem:tp_cfg_progress}
  If \Cfg_1\ and \Cfg_2\ are \wf,
  and $\Cfg;{\Que}_1 \Compat \Cfg;{\Que}_2$
  and ${\Cfg;{\Que}_1 \mid \Cfg;{\Que}_2}\Reduce*{\Cfg;{\Que}_1' \mid \Cfg;{\Que}_2'}$,
  then ${\Cfg;{\Que}_1' \Compat \Cfg;{\Que}_2'}$ \allowbreak and \Cfg_1'\ and \Cfg_2'\ are \wf.
\end{lem}
\begin{proof}
  We proceed by induction on the length of ($\Reduce*$).
  The base case is trivial, as
  \(
  {\Cfg;{\Que}_1 \mid \Cfg;{\Que}_2}%
  \Reduce^{0}
  {\Cfg;{\Que}_1 \mid \Cfg;{\Que}_2}
  \)
  and hence the thesis coincides with the hypothesis.
  For the induction case, suppose:
  \[
    {\Cfg;{\Que}_1 \mid \Cfg;{\Que}_2}%
    \Reduce^{n}%
    {\Cfg;{\Que}_1'{2} \mid \Cfg;{\Que}_2'{2}}%
    \Reduce%
    {\Cfg;{\Que}_1' \mid \Cfg;{\Que}_2'}
  \]

  \noindent By the induction hypothesis:
  $\Cfg;{\Que}_1'{2}\Compat\Cfg;{\Que}_2'{2}$,
  and both \Cfg_1'{2}\ and \Cfg_2'{2}\ are \wf.
  It remains for us to show that
  both \Cfg_1'\ and \Cfg_2'\ are \wf,
  and $\Cfg;{\Que}_1'\Compat\Cfg;{\Que}_2'$.
  By~\Cref{lem:tp_cfg_trans_preserve_wf,lem:tp_cfg_trans_preserve_compat}, we obtain our thesis.
\end{proof}

\begin{lem}\label{lem:tp_progress}
  The following holds:
  \[
    \forall \Val, \S : %
    \Cfg\text{\ is\ \wf}
    \implies
    {\Cfg;\mid\Cfg,{\D};}
    \text{\ satisfies progress.}
  \]
\end{lem}
\begin{proof}
  Since \Cfg\ is \wf, by~\Cref{def:type_duality} \Cfg,{\D}\ is \wf.
  By~\Cref{def:cfg_compat}, ${\Cfg;\Compat\Cfg,{\D};}$.
  If $\S=\End$ then $\D=\End$ and therefore, by~\Cref{def:type_progress}, the thesis coincides with the hypothesis.
  Otherwise, ${\Cfg;\mid\Cfg,{\D};}$ \emph{satisfies progress} if, for all $\Cfg*;'\mid\Cfg*;'{2}$ reachable from ${\Cfg;\mid\Cfg,{\D};}$
  then,
  either \Cfg*;'\ and \Cfg*;'{2}\ are both \emph{final}, or
  $\exists \t$ such that:
  \(
  \Cfg*;'\mid\Cfg*;'{2}
  \Trans:{\t,\tau}
  \).
  By~\Cref{lem:tp_cfg_progress}, any transition made by ${\Cfg;\mid\Cfg,{\D};}$ will preserve both \wfness\ and \compat.
  By~\Cref{lem:tp_cfg_fe}, we obtain our thesis.
\end{proof}

\thmTypeProgress*
\begin{proof}
  The thesis follows immediately from~\Cref{lem:tp_progress}.
\end{proof}

\noindent As typical for binary systems, our proof of progress builds upon a notion of \emph{duality} between participants.
Since communication is asynchronous, each party may break duality with their co-party.
Compatibility allows the duality of participants to be broken, only if, doing so does not violate \emph{communication safety}.
More specifically, compatibility requires that any message received by a queue is expected (i.e., able to be received), and that the resulting configurations are still \compat.
In practice, using compatibility allows each party to behave and progress independently, regardless the state of the other party, while retaining the essence of \emph{duality}, and guaranteeing that all messages will eventually be received.

\clearpage

\section{Proof of Subject Reduction}\label{sec:proof_subject_reduction}

%
%
%
\NewDocumentCommand{\BranchCardinalityPremise}{}{\WrapMathMode{%
    \neg (\left\lvert J \right\rvert = \left\lvert I \right\rvert = 1)
  }}

\subsection{Auxiliary Definitions, Figures \& Assumptions}

\input{figs/process_func_fq.tex}

\input{figs/process_func_ft.tex}

\input{figs/process_func_fn.tex}

\input{figs/process_func_fpv.tex}

\subsubsection{Session Environments}

\begin{defi}[Live \Session]\label{def:live_session}
  \Session\ is \live\ if,
  for all $\p \in \Dom{\Session}$
  such that $\Session = \Session',{\p:\Cfg} \implies $ \Cfg\ is \live\ by~\Cref{def:cfg_live}.
\end{defi}

\IsUnused{
  \begin{defi}[Role Actions]\label{def:role_actions}
    Given a role $\p:\Cfg$ we write \p\ as shorthand for {\Cfg}.
  \end{defi}
}

\begin{defi}[Delayable \Session]\label{def:delayable_session}
  A session environment \Session\ is \del\ if:
  \[
    \forall \qp\in\Dom{\Session} :
    \Session(\qp)\neq\emptyset \implies \p\notin\Dom{\Session}
  \]
\end{defi}

\subsubsection{Session Environment Reduction}
See~\Cref{fig:session_reduction}.

\input{figs/session_reduction.tex}

\subsection{Delayable Processes}

\begin{lem}\label{lem:sr_def_delay_no_recv}
  If \Time\ is defined,
  then $\Wait\cap\NEQ=\emptyset$.
\end{lem}
\begin{proof}
  The hypothesis holds by the definition of \Time\ in~\Cref{def:time_passing_func_ef}.
\end{proof}

\subsection{Well-formed \texorpdfstring{\Session}{Session Environments}}

\begin{lem}\label{lem:sr_wt_bal_session_reduce_preserve_bal}
  Let \Session\ be \wf.
  If \Session\ is \bal\ and $\Session\Reduce\Session'$,
  then \Session'\ is \bal.
  Furthermore, if \Session\ is \fbal, then \Session'\ is \fbal.
\end{lem}
\begin{proof}
  We proceed by induction on the derivation of the reduction $\Session\Reduce\Session'$, analysing the last rule applied of those given in~\Cref{fig:session_reduction}:
  \begin{caseanalysis}
    %
    \item\label{case:sr_wt_bal_session_reduce_preserve_bal_send}
    If rule \Rule[\Session]{Send}, then $\Session={\p:\Cfg},{\pq:\Que}$ and:
    \[
      \infer[\Rule[\Session]{Send}]{%
        {\p:\Cfg},{\pq:\Que}
        \Reduce
        {\p:\Cfg'},{\pq:\Que;\lT<>}
      }{%
        \Cfg
        \Trans:{!\lT<>}
        \Cfg'
      }
    \]

    \noindent We have that ${\p:\Cfg'},{\pq:\Que;\lT<>}$ is \bal, as every condition of~\Cref{def:bal_session} holds trivially. Regarding the furthermore part, it holds
    trivially as $\Session$ is not \fbal.
    %
    \item\label{case:sr_wt_bal_session_reduce_preserve_bal_recv}
    If rule \Rule[\Session]{Recv},
    then $\Session={\p:\Cfg},{\qp:\lT<>;\Que}$ and:
    \[
      \infer[\Rule[\Session]{Recv}]{%
      {\p:\Cfg},{\qp:\lT<>;\Que}
      \Reduce
      {\p:\Cfg'},{\qp:\Que}
      }{%
      \Cfg;{\lT<>;\Que}
      \Trans:{?\lT<>}
      \Cfg;{\Que}'
      }
    \]

    \noindent By rule \Rule{recv} of~\Cref{eq:type_semantics_triple}:
    \[
      \infer[\Rule{recv}]{
      \Cfg;{\Msg;\,\Que}%
      \Trans:{\tau}%
      \Cfg{\Val'},{\S'};{\Que}%
      }{
      \Cfg%
      \Trans:{?\Msg}%
      \Cfg'%
      }%
    \]

    \noindent By the premise of the above,
    thanks to~\Cref{case:bal_session_recv} of~\Cref{def:bal_session},
    it follows that $\Cfg{\Val'},{\S'};{\Que}$ is \bal, as required.
    The furthermore case holds trivially.
    %
    \item\label{case:sr_wt_bal_session_reduce_preserve_bal_par}
    If rule \Rule[\Session]{L} then
    $\Session = \Session_1,~\Session_2$ and
    $\Session' = \Session'_1,~\Session_2$ and:
    \[
      \infer[\Rule[\Session]{L}]{
        \Session_1,~\Session_2
        \Reduce
        \Session_1',~\Session_2
      }{
        \Session_1
        \Reduce
        \Session_1'
      }
    \]

    \noindent It is easy to see that $\Session_1$ and $\Session_2$ are both balanced. By the induction hypothesis also $\Session'_1$ is balanced. Notice that the
    composition of balanced environments is not necessarily balanced. So we need to show
    that all the items in~\Cref{def:bal_session} hold for $\Session'_1$.
    We show only for~\Cref{case:bal_session_recv} of~\Cref{def:bal_session}, as the other cases are similar.
    So, suppose:
    \(
    \Session' = \Session'{2},{\p:\Cfg},{\qp:{\Msg;\Que}}
    \),
    then we have to show:
    \(
    \Cfg\Trans:{?\Msg}\Cfg'
    \) and \(
    \Session',{\p:\Cfg'},{\qp:\Que} \in\Bal
    \).
    The interesting cases are:
    \begin{itemize}
      \item $\Session'_1(\p) = \Cfg$, $\Session_2(\qp) = \Msg;\Que$ and
            $\Session_1(\p) \neq \Session'_1(\p)$. We argue that this case is
            impossible. So suppose, by contradiction, it is not the case. Then,
            since queue $\qp$ did not change during the transition, it must be
            $\Session_1(\p) = (\nu_0,\S_0)$ for some $\nu_0,\S_0$ such that:
            \(
            (\nu_0,\S_0) \Trans:{!\Msg'} \Cfg
            \)
            and $\Session_1(\pq) = \Que'$ and $\Session'_1(\pq) = \Que';
              \Msg'$.
            By~\Cref{lem:tp_wf_cfg_safe_mc}, $(\nu_0,\S_0) \Trans|:{?\Msg}$
            and hence $\Session$ was not balanced in the first place: contradiction.
      \item $\Session_2(\p) = \Cfg$, $\Session'_1(\qp) = \Msg;\Que$ and
            $\Session_1(\qp) \neq \Session'_1(\qp)$. Then $\Que = \Que';\Msg'$ and $\Session_1(\q) = (\nu_0,\S_0)$, $\Session'_1(\q) = (\nu_0',\S'_0)$ and
            $\Session_1(\qp) = \Msg;\Que'$ with:
            \(
            (\nu_0,\S_0) \Trans:{!\Msg'} (\nu_0',\S'_0)
            \).
            By~\Cref{case:bal_session_missing_queue} of~\Cref{def:bal_session}, for some $\Que_0$:
            \(
            (\nu,\S,\Que)\Compat(\nu_0,\S_0,\Que_0)
            \).
            by rule \Rule{com}[r]* %
            in~\Cref{eq:type_semantics_system}:
            \[
              {(\nu,\S,\Msg;\Que')\mid(\nu_0,\S_0,\Que_0)}%
              \Trans:{\tau}%
              {(\nu,\S,\Msg;\Que)\mid(\nu'_0,\S'_0,\Que_0)}
            \]

            \noindent By~\Cref{lem:tp_cfg_trans_preserve_compat}:
            \[
              (\nu,\S,\Msg;\Que)\Compat(\nu_0,\S_0,\Que_0)
            \]

            \noindent Hence $(\nu,\S) \Trans:{?\Msg} (\nu',\S')$ for some $(\nu',\S')$ such
            that $(\nu',\S',\Que)\Compat(\nu'_0,\S'_0,\Que_0)$.
            It remains to show that $\Session',{\p:\Cfg'},{\qp:\Que}$ is \bal\ by~\Cref{def:bal_session}.
            For~\Cref{case:bal_session_missing_queue,case:bal_session_compat} of~~\Cref{def:bal_session} it surely holds.
            Regarding~\Cref{case:bal_session_recv} of~\Cref{def:bal_session}, it follows by a simple induction on the length of $\Que$.
            \qedhere
    \end{itemize}
  \end{caseanalysis}
\end{proof}

\subsection{Well-typed Processes}

\begin{lem}[Inversion Lemma]\label{lem:tc_inversion}
  Let \Session\ be \wf.
  \Cref{claim:tc_inversion_send,claim:tc_inversion_recv,claim:tc_inversion_branch,claim:tc_inversion_timeout,claim:tc_inversion_timer,claim:tc_inversion_det,claim:tc_inversion_delay,claim:tc_inversion_if,claim:tc_inversion_par,claim:tc_inversion_scope,claim:tc_inversion_def,claim:tc_inversion_def_par,claim:tc_inversion_rec,claim:tc_inversion_empty_queue,claim:tc_inversion_neq} each hold.
\end{lem}
\begin{clm}\label{claim:tc_inversion_par}\anno{par}
  If
  \(
  \Gam;\Timers\Entails
  \P\mid\Q
  \TypedBy\Session
  \), then \(
  \Session = (\Session_1,\Session_2)
  \) and \(
  \Timers = (\Timers_1,\Timers_2)
  \) and $\Gam;\Timers_1\Entails\P\TypedBy\Session_1$
  and $\Gam;\Timers_2\Entails\Q\TypedBy\Session_2$.
\end{clm}
\begin{clm}\label{claim:tc_inversion_scope}\anno{scope}
  If
  \(
  \Gam;\Timers\Entails
  \Scope{\pq}{\P}
  \TypedBy\Session
  \), then
  \(
  \Gam;\Timers\Entails
  \P
  \TypedBy\Session,{\p:\Cfg_1},{\qp:\Que_2},{\q:\Cfg_2},{\pq:\Que_2}
  \) and \(
  \Cfg;{\Que}_1\Compat\Cfg;{\Que}_2
  \) and both \Cfg_1\ and \Cfg_2\ are \wf.
\end{clm}
\begin{clm}\label{claim:tc_inversion_def}\anno{def}
  If
  \(
  \Gam;\Timers\Entails
  \Def{\Rec*{\vec{v};\vec{r}}=\P}:{\Q}
  \TypedBy\Session
  \), then $\Gam,\Rec|;\Timers\Entails\Q\TypedBy\Session$ and:
  \[
    \forall (\vec{\Val},~\vec{\S})
    \in\FormatRecCheckStyle{\Delta},
    {\theta'}\in\FormatRecCheckStyle{\theta}:\quad %
    \Gam,
    {\vec{v}:\vec{T}},%
    {\Rec|}%
    ;\Timers'
    \Entails{\P}
    \TypedBy
    \vec{r}: (\vec{\Val},~\vec{\S})%
  \]
\end{clm}
\begin{clm}\label{claim:tc_inversion_def_par}\anno{def par}
  If
  \(
  \Gam,\Rec|;\Timers\Entails
  \Def{\Rec*{\vec{v};\vec{r}}=\P}:{\P'\mid\Q}
  \TypedBy\Session
  \), then $\Session=(\Session_1,\Session_2)$, $\Timers=(\Timers_1,\Timers_2)$ and:
  \begin{itemize}
    \item $\Gam,\Rec|;\Timers_1\Entails
            \Def{\Rec*{\vec{v};\vec{r}}=\P}:{\P'}
            \TypedBy\Session_1$.
    \item $\Gam,\Rec|;\Timers_2\Entails
            \Def{\Rec*{\vec{v};\vec{r}}=\P}:{\Q}
            \TypedBy\Session_2$.
  \end{itemize}
\end{clm}
\begin{clm}\label{claim:tc_inversion_rec}\anno{rec}
  If
  \(
  \Gam;\Timers\Entails
  \Rec
  \TypedBy\Session
  \), then \(
  \Gam = \Gam',~\Rec|
  \)
  and:
  \[
    \forall i: \Gam'\Entails\vec{v}_i : \vec{T}_i
    \quad\text{and}\quad
    \Timers\in\FormatRecCheckStyle{\theta}
    \quad\text{and}\quad
    \Session=\vec{r}:(\vec{\nu},\vec{S})\in\FormatRecCheckStyle{\Delta}
  \]
\end{clm}
\begin{clm}\label{claim:tc_inversion_send}\anno{send}\TODO{update}
  If
  \(
  \Gam;\Timers\Entails
  \On{\p}\Send{\l}[\v].{\P}
  \TypedBy\Session
  \), then \(
  \Session = \Session',{\p:\Cfg,{\Choice{\l}[\T]_i}}
  \) and \NEQ{\Session}0\ and 
  $\exists j\in I$ such that
  $\Val\models\Constraints_j$ and:
  \begin{itemize}
    \item $\T_j~\text{\bt} \implies \Gam;\Timers\Entails\P\TypedBy\Session',{\p:\Cfg{\Val\Reset_j},{\S_j}}$ and $\Gam\Entails\v:\T_j$.
    \item $\T_j=(\delta',\S') \implies \Session'=\Session'{2},{\v:\Cfg'}$ and $\Val'\models\Constraints'$ and {{\small$\Gam;\Timers\Entails\P\TypedBy\Session'{2},{\p:\Cfg{\Val\Reset_j},{\S_j}}$}}.
  \end{itemize}
\end{clm}
\begin{clm}\label{claim:tc_inversion_recv}\anno{recv}\TODO{update}
  If
  \(
  \Gam;\Timers\Entails
  \On{\p}\Recv*^{\e}{\l}[\v]:{\P}
  \TypedBy\Session
  \), then \(
  \Session = \Session',{\p:\Cfg,{\Option?-{\l}[\T]{\delta}[\lambda].{\S}}}
  \) and
  \Session'\ not $\e$-reading and:
  \begin{itemize}
    \item $\forall \t$ such that $\Val+\t\models\Constraints \Bimplies \t\in\e$
    \item $\forall \t\in\e$:
          \begin{itemize}
            \item $\T=\Del' \implies \Gam;\Timers+\t\Entails\P\TypedBy\Session'+\t,{\p:\Cfg{\Val+\t\Reset},{\S}},{\v:\Cfg'}$ and $\Val'\models\Constraints'$.
            \item $\T~\text{\bt} \implies \Gam,\v:\T;\Timers+\t\Entails\P\TypedBy\Session'+\t,{\p:\Cfg{\Val+\t\Reset},{\S}}$.
          \end{itemize}
  \end{itemize}
\end{clm}
\begin{clm}\label{claim:tc_inversion_branch}\anno{branch}
  If
  \(
  \Gam;\Timers\Entails
  \On{\p}\Recv^{\e}{\l}[\v]:{\P}_i
  \TypedBy\Session
  \), then
    {{\small$\Session = \Session',{\p:\Cfg,{\Choice{\l}[\T]_j}}$}} and
  \BranchCardinalityPremise\ and,
  $\forall j\in J$ such that $\Val\models\Constraints_j$ implies:
  \[
    \Comm_j=?
    \quad\text{and}\quad
    \Bigl(
    \exists i\in I:
    \Gam;\Timers\Entails\On{\p}\Recv*^{\e}{\l}[\v]:{\P}_i\TypedBy\Session',{\p:\Cfg,{\Choice*{\l}[\T]_j}}
    \text{\ and\ }
    \l_i=\l_j
    \Bigr)
  \]
\end{clm}
\begin{clm}\label{claim:tc_inversion_timeout}\anno{timeout}
  If
  \(
  \Gam;\Timers\Entails
  \On{\p}\Recv^{\Diam\n}{\l}[\v]:{\P}_i~\After:{\Q}
  \TypedBy\Session
  \), then
  $\Session = \Session',{\p:\Cfg,{\C_j}}$
  and:
  \begin{itemize}
    \item $\Gam;\Timers\Entails\On{\p}\Recv^{\Diam\n}{\l}[\v]:{\P}_i\TypedBy\Session',~{\p:\Cfg,{\C_j}}$.
    \item $\Gam;\Timers+\n\Entails\Q\TypedBy\Session'+\n,~{\p:\Cfg{\Val+\n},{\C_j}}$.
  \end{itemize}
\end{clm}
\begin{clm}\label{claim:tc_inversion_timer}\anno{timer}
  If
  \(
  \Gam;\Timers\Entails
  \Set{x}.{\P}
  \TypedBy\Session
  \),
  then $\Timers=\Timers',\circled{x}:\t$
  and
  $ \Gam;\Timers',\circled{x}:0\Entails{\P}\TypedBy\Session$.
\end{clm}
\begin{clm}\label{claim:tc_inversion_det}\anno{det}
  If
  \(
  \Gam;\Timers\Entails
  \Delay{\Duration}.{\P}
  \TypedBy\Session
  \), then $\forall \t\in\Duration: \Gam;\Timers\Entails\Delay{\t}.{\P}\TypedBy\Session$.
\end{clm}
\begin{clm}\label{claim:tc_inversion_delay}\anno{delay}
  If
  \(
  \Gam;\Timers\Entails
  \Delay{\t}.{\P}
  \TypedBy\Session
  \), then \Session\ not \t-reading and $\Gam;\Timers+\t\Entails\P\TypedBy\Session+\t$.
\end{clm}
\begin{clm}\label{claim:tc_inversion_if}\anno{if}
  If
  \(
  \Gam;\Timers\Entails
  \If{\Cond}\Then{\P}~\Else{\Q}
  \TypedBy\Session
  \), then:
  \begin{itemize}
    \item $\Timers\models\Cond \implies \Gam;\Timers\Entails{\P}\TypedBy\Session$,
    \item $\Timers\not\models\Cond \implies \Gam;\Timers\Entails{\Q}\TypedBy\Session$.
  \end{itemize}
\end{clm}
\begin{clm}\label{claim:tc_inversion_empty_queue}\anno{empty queue}
  If
  \(
  \Gam;\Timers\Entails
  \qp:\emptyset
  \TypedBy\Session
  \), then $\Session = \Session',{\qp:\emptyset}$.
\end{clm}
\begin{clm}\label{claim:tc_inversion_neq}\anno{neq}
  If
  \(
  \Gam;\Timers\Entails
  \qp:\lv\cdot\h
  \TypedBy\Session
  \), then \(
  \Session = \Session',{\qp:\lT<>;\Que}
  \) and \NEQ{\Session'}0\ and:
  \begin{itemize}
    \item $\T=\Del' \implies \Session'=\Session'{2},{\q:\Cfg'}$ and $\Gam;\Timers\Entails\qp:\h\TypedBy\Session'{2},{\qp:\Que}$ and $\Val'\models\Constraints'$.
    \item $\T~\bt \implies \Gam\Entails\v:\T$ and $\Gam;\Timers\Entails\qp:\h\TypedBy\Session',{\qp:\Que}$.
  \end{itemize}
\end{clm}
\begin{proof}
  As standard, it follows by induction on each derivation, using the \tc\ rules given in~\RefTypeCheck.
\end{proof}

\subsubsection{Well-typed Enabled Actions}

%
\begin{lem}\label{lem:sr_wt_wf_role_recv_prc_wait}
  Let $\Session,{\p:\Cfg}$ be \wf.
  If $\Gam;\Timers\Entails\P\TypedBy\Session,{\p:\Cfg}$
  and \(
  \Cfg\Trans:{?\Msg}
  \)
  then $\p\in\Wait$.
\end{lem}
\begin{proof}
  By induction on the derivation of the transition
  \(
  \Cfg\Trans:{?\Msg}
  \), analysing the last rule applied of those given
  in~\Cref{eq:type_semantics_tuple}.
  The only applicable rules are \Rule{act} and \Rule{unfold}.
  As in~\Cref{lem:tp_wf_cfg_safe_mc}, we only show rule \Rule{act}:
  \[
    \infer[\Rule{act}]{
    \Cfg,{\Choice{\l}[\T]_i}%
    \Trans:{\Comm_j\Msg}%
    \Cfg{\Val\Reset_j},{\S_j}%
    }{
    \Val\models\delta_j%
    \quad%
    \Msg={\lT<>_j}%
    \quad%
    j\in I%
    }%
  \]

  \noindent where $\S=\Choice{\l}[\T]_i$.
  By inspecting~\Cref{fig:func_wait_neq} for when $\p\in\Wait{\P}$, it remains for us to show that \P\ is structured as follows:
  \[
    \P \in
    \begin{array}[t]\{{@{\ }l@{\ \ }l@{\ }}\}
      \On{\p}\Recv^{\e}{\l}[\v]:{\P}_j ,
       &
      \On{\p}\Recv^{\Diam\n}{\l}[\v]:{\P}_j ~\After:{\Q}
    \end{array}
  \]

  \noindent We proceed by induction on the derivation of
  $\Gam;\Timers\Entails\P\TypedBy\Session,{\p:\Cfg}$,
  analysing the last rule applied of those given in~\RefTypeCheck:
  \begin{caseanalysis}
    %
    \item\label{case:sr_wt_wf_role_recv_prc_wait_branch}
    If rule \Rule{Branch},
    then \(
    \P = \On{\p}\Recv^{\e}{\l}[\v]:{\P}_j
    \) and we obtain our thesis.
    %
    \item\label{case:sr_wt_wf_role_recv_prc_wait_}
    If rule \Rule{Timeout},
    then \(
    \P = \On{\p}\Recv^{\Diam\n}{\l}[\v]:{\P}_j ~\After:{\Q}
    \).
    %
    \item\label{case:sr_wt_wf_role_recv_prc_wait_vrecv}
    If rule \Rule{VRecv},
    then \(
    \P = \On{\p}\Recv*^{\e}{\l}[\v]:{\P'}
    \).
    (The same holds for rule \Rule{DRecv}.)
    %
    \item\label{case:sr_wt_wf_role_recv_prc_wait_vsend}
    If rule \Rule{VSend},
    then \(
    \P = \On{\p}\Send{\l}[\v].{\P'}
    \).
    Clearly $\p\notin\Wait$, and it remains for us to prove that this case is not applicable.
    By rule \Rule{VSend}:
    \[
      \infer[\Rule{VSend}]{%
      \Gam;\Timers\Entails
      \On{\p}\Send{\l}[\v].{\P'}
      \TypedBy\Session,{\p:\Cfg,{\Choice{\l}[\T]_i}}
      }{%
      \begin{array}[t]{@{}c@{}}
        \exists j\in I: %
        (\l=\l_j)
        \land
        (\Val\models\Constraints_j)%
        \land
        (\T_j~\text{\bt})
        \land%
        (\Gam\Entails\v:\T_j)
        \land\mbox{}%
        \\[\ArrayTallLineSpacing]%
        (\Comm_j=!)
        \land
        \Gam;\Timers\Entails{\P}\TypedBy\Session,~{\p:\Cfg{\Val\Reset_j},{\S_j}}%
      \end{array}%
      }
    \]

    \noindent By~\Cref{lem:tp_wf_cfg_safe_mc}, since \Cfg\ is \wf, it cannot be that both $\Cfg\Trans:{?\Msg}$ and $\Cfg\Trans:{!\Msg'}$, where $\Msg'=\lT<>$.
    Therefore, this case is not applicable since it cannot be that a \emph{sending} process \P\ is \wt\ against \p, which has an enabled \emph{receiving} action.
    (The same holds for rule \Rule{DSend}.)
    \qedhere
  \end{caseanalysis}
\end{proof}

\subsubsection{Type-checking Processes}

\begin{lem}\label{lem:sr_wt_equiv_prc}
  Let \Session\ be \wf.
  If $\Gam;\Timers\Entails\P\TypedBy\Session$ and $\P\equiv\Q$,
  then $\Gam;\Timers\Entails\Q\TypedBy\Session$.
\end{lem}
\begin{proof}
  As standard~\cite{Bocchi2019,Honda2008,Yoshida2007}, with the additions in~\Cref{sssec:process_reduction_rules}.
\end{proof}

\subsubsection{Well-typed Substitutions}

\begin{lem}[Substitution]\label{lem:sr_wt_substitution}
  Let \Session\ be \wf.
  \Cref{claim:sr_wt_substitution_role,claim:sr_wt_substitution_var} both hold.
\end{lem}
\begin{clm}\label{claim:sr_wt_substitution_role}
  If $\Gam;\Timers\Entails\P\TypedBy\Session,{\p:\Cfg}$
  and $\b\notin\Dom$,
  then $\Gam;\Timers\Entails{\P\Subst{\b}{\p}}\TypedBy\Session,{\b:\Cfg}$.
\end{clm}
\begin{clm}\label{claim:sr_wt_substitution_var}
  If $\Gam_1,{\Msg:\T};\Timers\Entails\P\TypedBy\Session$
  and $\Gam_2\Entails{\v:\T}$
  and $\Dom{\Gam_1}\cap\Dom{\Gam_2}=\emptyset$,
  then $\Gam_1,\Gam_2;\Timers\Entails{\P\Subst{\v}{\Msg}}\TypedBy\Session$.
\end{clm}
\begin{proof}
  \Cref{claim:sr_wt_substitution_role,claim:sr_wt_substitution_var} both hold as standard~\citeST.
\end{proof}

\subsubsection{Type-checking Messages}

\begin{lem}[Well-typed Messages]\label{lem:sr_wt_msg}
  Let
  \(
  \Gam;\Timers\Entails{\qp:\h}\TypedBy\Session,{\qp:\Que}
  \)
  and $\Session,{\qp:\Que}$ be \wf.
  \Cref{claim:sr_wt_msg_bt,claim:sr_wt_msg_del} both hold, for all \l.
\end{lem}
\begin{clm}\label{claim:sr_wt_msg_bt}
  If $\Gam\Entails{\v:\T}$,
  then
  $\Gam;\Timers\Entails{\qp:\h\cdot {\lv}}\TypedBy\Session,{\qp:\Que;{\lT<>}}$.
\end{clm}
\begin{clm}\label{claim:sr_wt_msg_del}
  If $\T=\Del$
  and $\b\neq\Dom{\Session}$
  and $\exists \Val$
  such that $\Val\models\Constraints$,
  then:
  \[
    \Gam;\Timers\Entails{\qp:\h\cdot\l\b}
    \TypedBy
    \Session,{\qp:\Que;\lT<>},{\b:\Cfg}
  \]
\end{clm}
\begin{proof}
  We proceed by addressing each claim in turn:
  \begin{description}[font=\normalfont]
    %
    \item[\Cref{claim:sr_wt_msg_bt}]
          Holds following the conclusion of rule \Rule{VQue} in~\Cref{fig:type_checking_communication}~\Cref{eq:type_checking_queues}, where we use the given assumptions as premise:
          \[
            \infer[\Rule{VQue}]{
            \Gam;\Timers\Entails{\qp:{\l\v}\cdot\h}\TypedBy\Session,%
            ~{\qp:{\lT<>};\Que}%
            }{
            \T~\text{\bt}%
            \quad%
            \Gam\Entails{\v:\T}%
            \quad%
            \Gam;\Timers\Entails{\qp:{\h}}\TypedBy\Session,%
            ~{\qp:{\Que}}%
            }%
          \]
          %
    \item[\Cref{claim:sr_wt_msg_del}]
          Holds following the conclusion of rule \Rule{DQue} in~\Cref{fig:type_checking_communication}~\Cref{eq:type_checking_queues}, similar to the case of~\Cref{claim:sr_wt_msg_bt} (above).
          \qedhere
  \end{description}
\end{proof}

\subsubsection{Type-checking Message Handling}

\begin{lem}\label{lem:sr_neq_prc}
  Let $\Session,{\qp:\lT<>;\Que}$ be \wf.
  If $\Gam;\Timers\Entails{\P}\TypedBy\Session,{\qp:\lT<>;\Que}$,
  then $\p\in\NEQ$.
\end{lem}
\begin{proof}
  By induction on the typing derivation,
  the key rules are \Rule{VQue} are \Rule{DQue}, which are structured such that:
  \(
  \Gam;\Timers\Entails{\qp:{\lv}\cdot\h}\TypedBy\Session,~{\qp:\Msg;\Que}
  \)
  where $\Msg=\lT<>$.
  The thesis follows~\Cref{fig:func_wait_neq}, as $\p\in\NEQ{\qp:{\lv}\cdot\h}$ holds.
\end{proof}

\begin{lem}\label{lem:sr_neq_cfg}
  Let \Session\ be \wf.
  If $\Gam;\Timers\Entails{\P}\TypedBy\Session$
  and $\p\in\NEQ$,
  then $\exists \Session'$ such that $\Session=\Session',~{\qp:\lT<>;\Que}$.
\end{lem}
\begin{proof}
  By the definition of \NEQ\ in~\Cref{fig:func_wait_neq}, it must be that
  $\P\equiv{\qp:\lv\cdot\h}\mid\Q$,
  and by~\Cref{claim:tc_inversion_neq} of~\Cref{lem:tc_inversion}, we obtain our thesis.
\end{proof}

\begin{lem}\label{lem:sr_empty_queue_cfg}
  Let \Session\ be \wf.
  If $\Gam;\Timers\Entails{\P}\TypedBy\Session$
  and \NEQ0,
  then \NEQ{\Session}0.
\end{lem}
\begin{proof}
  Since \NEQ0, it follows that $\forall \qp:\h$ such that $\P\equiv \qp:\h\mid\Q$, then $\h=\emptyset$.
  By~\Cref{claim:tc_inversion_par} in~\Cref{lem:tc_inversion}, $\Session=\Session_1,\Session_2$ and,
  \(
  \Gam;\Timers\Entails\qp:\h\TypedBy\Session_1
  \)
  and
  \(
  \Gam;\Timers\Entails\Q\TypedBy\Session_2
  \).
  By~\Cref{claim:tc_inversion_empty_queue} in~\Cref{lem:tc_inversion}, it follows that $\Session_1=\Session_1',{\qp:\Que}$ and $\Que=\emptyset$.
  We obtain our thesis:
  given that all message queues in \P\ are empty, if \P\ is \wt\ against \Session\ then, there exists corresponding queues in \Session\ that are all also empty.
\end{proof}

\subsubsection{Delayable Session Environments}

\begin{lem}\label{lem:sr_delay_session_preserve_bal}
  Let \Session\ be \wf.
  If \Session\ is \bal\
  and \del\
  and $\Session+\t$ is \wf,
  then $\Session+\t$ is \bal.
\end{lem}
\begin{proof}
  We show that $\Session+\t$ satisfies all conditions of \Cref{def:bal_session}:
  \begin{caseanalysis}
    %
    \item\label{case:sr_delay_session_preserve_bal_recv}
    Since \Session\ is \del, by~\Cref{def:delayable_session}, all queues in the domain of \Session\ are \emph{empty}.
    Therefore, the case where $\Session=\Session',{\p:\Cfg},{\qp:\Msg;\Que}$, conflicts with the hypothesis, and is not applicable.
    %
    \item\label{case:sr_delay_session_preserve_bal_scope}
    If $\Session=\Session',{\p:\Cfg_1},{\qp:\Que_1},{\q:\Cfg_2},{\pq:\Que_2}$,
    then $\Cfg;{\Que}_1\Compat\Cfg;{\Que}_2$.
    Since \Cfg_1\ and \Cfg_2\ are \wf,
    then by~\Cref{lem:tp_cfg_trans_preserve_wf,lem:tp_cfg_trans_preserve_compat},
    \Cfg{\Val_1+\t},{\S_1}\ and \Cfg{\Val_2+\t},{\S_2}\ are \wf,
    and $\Cfg{\Val_1+\t},{\S_1};{\Que}\Compat\Cfg{\Val_2+\t},{\S_2};{\Que_2}$.
    Therefore,
    $\Session+\t$ is \bal.
    \item\label{case:sr_delay_session_preserve_bal_missing_queue}
    If $\Session=\Session',{\p:\Cfg_1},{\qp:\Que_1},{\q:\Cfg_2}$,
    then $\exists \Que_2:\Cfg;{\Que}_1\Compat\Cfg;{\Que}_2$ and the thesis follows~\Cref{case:sr_delay_session_preserve_bal_scope}.
    \qedhere
  \end{caseanalysis}
\end{proof}

\begin{lem}\label{lem:sr_not_t_reading_session_delayable}
  If \Session\ is \bal, \wf\ and not \t-reading,
  then \Session\ is \del.
\end{lem}
\begin{proof}
  If we suppose by contradiction that \Session\ is \emph{not} \del,
  then by~\Cref{def:delayable_session}, $\exists \qp,\Msg,\Que$ such that $\qp\in\Dom{\Session}$ and $\Session(\qp)=\Msg;\Que$ and $\exists \p\in\Dom{\Session}$.
  Since \Session\ is \bal, for some $\Session(\p)=\Cfg'$
  it follows~\Cref{case:bal_session_recv} of~\Cref{def:bal_session}:
  \[
    \Session=\Session',{\p:\Cfg'},{\qp:\Msg;\Que}
    \implies
    \Cfg'\Trans:{?\Msg}\Cfg'{2}
    \quad\text{and}\quad
    \Session',{\p:\Cfg'{2}},{\qp:\Que}\in\Bal
  \]

  \noindent However, this contradicts with \Session\ not \t-reading:
  by~\Cref{def:t_reading_session}, there are no roles $\q\in\Dom{\Session}$ such that $\Session(\q)=\Cfg'{3}$ and $\Cfg{\Val'{3}+\t'},{\S'{3}}\Trans:{?\Msg'}$, for all $\t'<\t$.
  Therefore, \Session\ must be \del.
\end{proof}

\begin{lem}\label{lem:sr_wt_delay}
  Let \Session\ be \wf.
  If $\Gam;\Timers\Entails{\P}\TypedBy\Session$
  and \Time\ is \emph{defined}
  and \Session\ is \bal,
  then \Session\ is \del.
\end{lem}
\begin{proof}
  Since \Time\ is \emph{defined}, we
  assume that $\t>0$.
  We proceed by induction on the typing derivation of \P:
  \begin{caseanalysis}
    %
    \item\label{case:sr_wt_delay_branch_recv}
    If \(
    \P = \On{\p}\Recv^{\e}{\l}[\v]:{\P}_i \mid \qp:\h
    \)
    then by~\Cref{claim:tc_inversion_par} of~\Cref{lem:tc_inversion},
    $\Timers=(\Timers_1,\Timers_2)$,
    $\Session=(\Session_1,\Session_2)$ and:
    \[
      \Gam;\Timers_1\Entails\On{\p}\Recv^{\e}{\l}[\v]:{\P}_i\TypedBy\Session_1
      \qquad
      \Gam;\Timers_2\Entails\qp:\h\TypedBy\Session_2
    \]

    \noindent By~\Cref{claim:tc_inversion_branch} of~\Cref{lem:tc_inversion},
    \(
    \Session_1=\Session_1',{\p:\Cfg_1}
    \)
    where $\S_1=\Choice{\l}[\T]_j$ and
    \BranchCardinalityPremise.
    It also follows that $\forall j\in J$ such that $\Val_1\models\delta_j$, then $\Comm_j=?$
    and:
    \[
      \exists i\in I : \l_i=\l_j
      \quad\text{and}\quad
      \Gam;\Timers_1\Entails\On{\p}\Recv*^{\e}{\l}[\v]:{\P}_i\TypedBy\Session_1',{\p:\Cfg{\Val_1},{\Choice*{\l}[\T]_j}}
    \]

    \noindent By~\Cref{claim:tc_inversion_recv} of~\Cref{lem:tc_inversion},
    \Session_1'\ is not \e-reading, and $\forall\t:\Val_1+\t\models\Constraints_j\iff\t\in\e$.
    Since $\Val_1\models\Constraints_j$ and $\Comm_j=?$,
    by~\Cref{claim:tp_cfg_trans_act} of~\Cref{lem:tp_cfg_trans},
    by rule \Rule{act} %
    in~\Cref{eq:type_semantics_tuple}:
    \[
      \Cfg{\Val_1},{\Choice*{\l}[\T]_j}
      \Trans:{?\lT<>_j}
      \Cfg{\Val_1\Reset_j},{\S_j}
    \]

    \noindent By~\Cref{lem:sr_wt_wf_role_recv_prc_wait}, $\p\in\Wait$.
    By~\Cref{def:time_passing_func_ef}, $\Wait\cap\NEQ=\emptyset$.
    Therefore, $\p\notin\NEQ$ and, since $\p\notin\NEQ$, then $\qp:\h=\emptyset$.
    By~\Cref{claim:tc_inversion_empty_queue} of~\Cref{lem:tc_inversion},
    \(
    \Session_2=\Session_2',{\qp:\Que_1}
    \)
    and since $\h=\emptyset$ then by~\Cref{lem:sr_empty_queue_cfg}, $\Que_1=\emptyset$.
    It follows that \Session_2\ is \del, since there cannot be any other queues present in \Session_2'\ given that $\Gam;\Timers_2\Entails\qp:\h\TypedBy\Session_2',{\qp:\Que_1}$.
    It remains to show that \Session_1\ is \del\ in order to obtain our thesis and prove that \Session\ is \del.
    Since \Session\ is \bal, it follows \Session_1'\ is also \bal,
    Since \Session_1'\ is \bal, \wf\ and not \e-reading, the thesis follows~\Cref{lem:sr_not_t_reading_session_delayable}.
    %
    \item\label{case:sr_wt_delay_timeout_recv}
    If \(
    \P = \On{\p}\Recv^{\Diam\n}{\l}[\v]:{\P}_i~\After:{\Q} \mid \qp:\h
    \)
    then by~\Cref{claim:tc_inversion_par} of~\Cref{lem:tc_inversion},
    $\Timers=(\Timers_1,\Timers_2)$,
    $\Session=(\Session_1,\Session_2)$ and:
    \[
      \Gam;\Timers_1\Entails\On{\p}\Recv^{\Diam\n}{\l}[\v]:{\P}_i~\After:{\Q}\TypedBy\Session_1
      \qquad
      \Gam;\Timers_2\Entails\qp:\h\TypedBy\Session_2
    \]

    \noindent By~\Cref{claim:tc_inversion_timeout} of~\Cref{lem:tc_inversion},
    \(
    \Session_1=\Session_1',{\p:\Cfg_1}
    \)
    where $\S_1=\C_j$ and:
    \begin{equation}\label{eq:sr_wt_delay_timeout_branch}
      \Gam;\Timers_1\Entails\On{\p}\Recv^{\Diam\n}{\l}[\v]:{\P}_i\TypedBy\Session_1',{\p:\Cfg{\Val_1},{\C_j}}
    \end{equation}
    \begin{equation}\label{eq:sr_wt_delay_timeout_timeout}
      \Gam;\Timers_1+\n\Entails\Q\TypedBy\Session_1+\n',{\p:\Cfg{\Val_1+\n},{\C_j}}
    \end{equation}

    \noindent By~\Cref{def:time_passing_func_ef}, if $\t\Diam\n$ then the thesis follows by the induction hypothesis on~\Cref{eq:sr_wt_delay_timeout_branch}.
    (See~\Cref{case:sr_wt_delay_branch_recv}.)
    Otherwise, $\neg(\t\Diam\n)$ and $\Time{\On{\p}\Recv^{\Diam\n}{\l}[\v]:{\P}_i~\After:{\Q}}=\Time{\Q}_{\t-\n}$, and the thesis follows by induction on~\Cref{eq:sr_wt_delay_timeout_timeout}.
    %
    \item\label{case:sr_wt_delay_scope}
    If \(
    \P = \Scope{\pq}{\P'}
    \)
    then by~\Cref{claim:tc_inversion_scope} of~\Cref{lem:tc_inversion}:
    \begin{equation}\label{eq:sr_wt_delay_scope}
      \Gam;\Timers\Entails\P'\TypedBy\Session,
      ~{\p:\Cfg_1},
      ~{\qp:\Que_1},
      ~{\q:\Cfg_2},
      ~{\pq:\Que_2}
    \end{equation}

    \noindent where $\Cfg;{\Que}_1\Compat\Cfg;{\Que}_2$ and both \Cfg_1\ and \Cfg_2\ are \wf.
    By~\Cref{def:fbal_session},
    \(
    \Session,
    {\p:\Cfg_1},
    {\qp:\Que_1},
    {\q:\Cfg_2},
    {\pq:\Que_2}
    \)
    is \fbal, and therefore the thesis follows by the induction hypothesis of~\Cref{eq:sr_wt_delay_scope}.
    %
    \item\label{case:sr_wt_delay_delay}
    If \(
    \P = \Delay{\t'}.{\P'}
    \)
    then by~\Cref{claim:tc_inversion_delay} of~\Cref{lem:tc_inversion},
    \Session\ is not \t'-reading and $\Gam;\Timers+\t'\Entails\P'\Session+\t'$.
    By~\Cref{def:time_passing_func_ef}:
    \begin{caseanalysis}
      %
      \item\label{case:sr_wt_delay_delay_t_leq}
      If $\t\leq\t'$, then $\Time{\Delay{\t'}.{\P'}}=\Delay{\t'-\t}.\P'$.
      Since \Session\ is not \t'-reading, it follows \Session\ is not \t-reading also.
      The thesis follows~\Cref{lem:sr_not_t_reading_session_delayable}.
      %
      \item\label{case:sr_wt_delay_delay_t_gtr}
      If $\t>\t'$, then $\Time{\Delay{\t'}.{\P'}}=\Time{\P'}_{\t-\t'}$ and
      the thesis follows by the induction hypothesis on:
      \(
      \Gam;\Timers+\t'\Entails\Time{\P'}_{\t-\t'}\TypedBy\Session+\t'
      \).
    \end{caseanalysis}
  \end{caseanalysis}

  \noindent We only show the key cases of \P.
  The other cases of \P\ follow~\Cref{def:time_passing_func_ef}.
\end{proof}

\begin{lem}\label{lem:session_env_live}
  If $\Gam;\Timers\Entails\P\TypedBy\Session$ then $\Session$ is \fe.
\end{lem}
\begin{proof}
  A simple induction on the typing derivation.
\end{proof}

\begin{lem}\label{lem:session_env_wf_live}
  If $\Session$ is \wf\ and $\Session + \t$ is \live,
  then $\Session + \t$ is \wf.
\end{lem}
\begin{proof}
  Suppose $\Session + \t(\p) = (\Val,~\S)$ and $\S \neq \End$. It
  must be $\Val = \Val' + \t$ with $\Session(\p) = (\Val',~\S)$ and
  hence $(\Val',~\S)$ is \wf. 
  By~\Cref{lem:tp_wf_cfg_past_closed} we have that $(\Val'+\t,~\S)$ is \wf,
  and by~\Cref{lem:tp_wf_cfg_is_live} it is \live, as required.
\end{proof}

\subsection{Reduction Steps}\label{ssec:proof:sr:steps}

\subsubsection{Time-passing Steps}\label{sssec:proof:sr:steps:time}

\begin{lem}\label{lem:sr_time_step_preserve_wf}
  Let \Session\ be \wf.
  If \(
  \Gam;\Timers\Entails\P\TypedBy\Session
  \),
  \Time\ is defined and \NEQ0,
  then $\Session+\t$ is \wf\
  and \(
  \Gam;\Timers+\t\Entails\Time\TypedBy\Session+\t
  \).
\end{lem}
\begin{proof}
  We proceed by induction on the derivation of $\Gam;\Timers\Entails\P\TypedBy\Session$
  analysing the last rule applied of the \tc\ rules given in~\RefTypeCheck:
  \begin{caseanalysis}
    %
    \item\label{case:sr_time_step_preserve_wf_vrecv}
    If rule \Rule{VRecv},
    then \(
    \P = \On{\p}\Recv*^{\e}{\l}[\v]:{\P'}
    \)
    and following~\Cref{fig:type_checking_communication}~\Cref{eq:type_checking_recv_single}:
    \[
      \infer[\Rule{VRecv}]{
      \Gam;\Timers\Entails{\On{\p}\Recv*^{\e}{\l}[\v]:{\P}}
      \TypedBy\Session',~{\p:\Cfg,{\Option?-{\l}[\T]{\delta}[\lambda].{\S}}}%
      }{
      \begin{array}[t]{@{}c@{}}%
        \T~\text{\bt}%
        \quad%
        \Session'~\text{not\ $\e$-reading}
        \quad
        {%
          \forall\t':\Val+\t'\models\Constraints\Bimplies
          \t'\in\e
        }
        \\[\ArrayTallLineSpacing]
        \forall\t'\in\e: %
        \Gam,{\v:\T};\Timers+\t'\Entails{\P}\TypedBy\Session'+\t',%
        ~{\p:\Cfg{\Val+\t'\Reset},{\S}}%
      \end{array}%
      }%
    \]

    \noindent where by~\Cref{claim:tc_inversion_recv} of~\Cref{lem:tc_inversion}, \(
    \Session=\Session',{\p:\Cfg,{\Option?-{\l}[\T]{\delta}[\lambda].{\S}}}
    \).
    Since \Time\ is defined, by~\Cref{def:time_passing_func_ef}, it must be that $\t\in\e$.
    (If $\e=\infty$ then this always holds.)
    Coinciding with the first premise of rule \Rule{VRecv}, it follows that $\Val+\t\models\Constraints$; additionally, by the second premise it follows that the time-step preserves \wtness:
    \[
      \forall\t'\in\e:
      \Gam,\v:\T;\Timers+\t'\Entails
      \P'
      \TypedBy
      \Session'+\t',~{\p:\Cfg{\Val+\t'\Reset},{\S}}
    \]

    \noindent It remains to show that $\Session+\t$ is \wf.
    By \Cref{lem:session_env_live} $\Session+\t$ is \live.
    By \Cref{lem:session_env_wf_live} it is \wf.
    %
    \item\label{case:sr_time_step_preserve_wf_drecv}
    If rule \Rule{DRecv},
    then we obtain our thesis analogously to~\Cref{case:sr_time_step_preserve_wf_vrecv}.
    %
    \item\label{case:sr_time_step_preserve_wf_branch}
    If rule \Rule{Branch},
    then \(
    \P = \On{\p}\Recv^{\e}{\l}[\v]:{\P}_i
    \)
    and following~\Cref{fig:type_checking_communication}~\Cref{eq:type_checking_recv}:
    \[
      \infer[\Rule{Branch}]{
      \Gam;\Timers\Entails%
      \On{\p}\Recv^{\e}{\l}[\v]:{\P}_{i}%
      \TypedBy\Session',~{\p:\Cfg,{\Choice{\l}[\T]_{j}}}%
      }{
      \begin{array}[t]{@{}c@{}}%
        \neg (\left\lvert J \right\rvert = \left\lvert I \right\rvert = 1)
        \quad
        \begin{array}[t]{@{}l@{}}%
          \forall j\in J: %
          \Val\models\Constraints_j%
          \implies%
          \Comm_j = ?%
          ~\land~%
          \exists i\in I:
          \l_i=\l_j%
          ~\land~\mbox{}%
          \\[\ArrayTallLineSpacing]%
          \mbox{}\qquad
          \Gam;\Timers\Entails{\On{\p}\Recv*^{\e}{\l}[\v]:{\P}_{i}}%
          \TypedBy\Session',~{\p:\Cfg,{\Choice*{\l}[\T]_{j}}}%
        \end{array}%
      \end{array}%
      }%
    \]

    \noindent where by~\Cref{claim:tc_inversion_branch} of~\Cref{lem:tc_inversion},
    \(
    \Session = \Session',{\p:\Cfg,{\Choice{\l}[\T]_{j}}}
    \).
    Since \Time\ is defined, by~\Cref{def:time_passing_func_ef}, it must be that $\t\in\e$.
    (If $\e=\infty$ then this always holds.)
    By the premise of rule \Rule{Branch}, each enabled action $j$ in \S\ is receiving, and has a corresponding branch $i$ in \P\ where $\l_i=\l_j$ and:
    \[
      \Gam;\Timers\Entails
      \On{\p}\Recv*^{\e}{\l}[\v]:{\P}_i
      \TypedBy
      \Session',{\p:\Cfg,{\Choice*{\l}[\T]_j}}
    \]

    \noindent Therefore, the hypothesis holds by induction.
    (See~\Cref{case:sr_time_step_preserve_wf_vrecv,case:sr_time_step_preserve_wf_drecv}.)
    %
    \item\label{case:sr_time_step_preserve_wf_timeout}
    If rule \Rule{Timeout},
    then
    \(
    \On{\p}\Recv^{\Diam\n}{\l}[\v]:{\P}_{i}~\After:{\Q}
    \)
    and following~\Cref{fig:type_checking_communication}~\Cref{eq:type_checking_recv}:
    \[
      \infer[\Rule{Timeout}]{
      \Gam;\Timers\Entails%
      \On{\p}\Recv^{\Diam\n}{\l}[\v]:{\P}_{i}%
      ~\After:{\Q}%
      \TypedBy\Session',~{\p:\Cfg,{\C_j}}%
      }{
      \begin{array}[t]{@{}c@{}}%
        \begin{array}[t]{@{}l@{}}%
          \Gam;\Timers\Entails%
          \On{\p}\Recv^{\Diam\n}{\l}[\v]:{\P}_{i}%
          \TypedBy\Session',~{\p:\Cfg,{\C_{j}}}%
        \end{array}%
        \\[\ArrayTallLineSpacing]
        \begin{array}[t]{@{}l@{}}%
          \Gam;\Timers+\n\Entails{\Q}%
          \TypedBy\Session'+\n,~{\p:\Cfg{\Val+\n},{\C_{j}}}%
        \end{array}%
      \end{array}%
      }%
    \]

    \noindent where by~\Cref{claim:tc_inversion_timeout} of~\Cref{lem:tc_inversion},
    $\Session = \Session',{\p:\Cfg,{\C_j}}$.
    Since \Time\ is defined by~\Cref{def:time_passing_func_ef},
    we show for each case of \t:
    \begin{caseanalysis}
      %
      \item\label{case:sr_time_step_preserve_wf_timeout_t_diam_n}
      If $\t\Diam\n$, then
      \(
      \Time = \On{\p}\Recv^{\Diam\n - \t}{\l}[\v]:{\P}_i~\After:{\Q}
      \).
      Therefore, by the induction hypothesis:
      \[
        \infer[\Rule{Timeout}]{%
        \Gam;\Timers+\t\Entails
        \On{\p}\Recv^{\Diam\n-\t}{\l}[\v]:{\P}_i~\After:{\Q}
        \TypedBy
        \Session'+\t,~{\p:\Cfg{\Val+\t},{\C_j}}
        }{%
        \begin{array}[t]{@{}c@{}}%
          \begin{array}[t]{@{}l@{}}%
            \Gam;\Timers+\t\Entails%
            \On{\p}\Recv^{\Diam\n+\t}{\l}[\v]:{\P}_{i}%
            \TypedBy\Session'+\t,~{\p:\Cfg{\Val+\t},{\C_{j}}}%
          \end{array}%
          \\[\ArrayTallLineSpacing]
          \begin{array}[t]{@{}l@{}}%
            \Gam;\Timers+\n\Entails{\Q}%
            \TypedBy\Session'+\t+\n-\t,~{\p:\Cfg{\Val+\t+\n-\t},{\C_{j}}}%
          \end{array}%
        \end{array}%
        }
      \]

      \noindent The first premise of rule \Rule{Timeout} holds by induction on the derivation.
      (See~\Cref{case:sr_time_step_preserve_wf_branch}.)
      The second premise of rule \Rule{Timeout} coincides with the hypothesis, as $\t+\n-\t=\n$.
      %
      \item\label{case:sr_time_step_preserve_wf_timeout_t_notdiam_n}
      If $\neg(\t\Diam\n)$, then
      \(
      \Time = \Time{\Q}_{\t-\n}
      \).
      The thesis follows the induction hypothesis:
      \[
        \Gam;\Timers+\t\Entails\Time{\Q}_{\t-\n}\TypedBy\Session+\t
      \]
    \end{caseanalysis}
    %
    \item\label{case:sr_time_step_preserve_wf_delay}
    If rule \Rule{Del}[\t],
    then
    \(
    \P = \Delay{\t'}.\P'
    \) and:
    \begin{caseanalysis}
      %
      \item\label{case:sr_time_step_preserve_wf_delay_t_leq}
      If $\t\leq\t'$,
      then
      \(
      \Time=\Delay{\t'-\t}.\P'
      \)
      and following~\Cref{fig:type_checking}~\Cref{eq:type_checking_time_sensitive}:
      \[
        \infer[\Rule{Del}[\t]]{
          \Gam;\Timers\Entails{\Delay{\t'}.\P'}\TypedBy\Session%
        }{
          \Gam;\Timers+\t'\Entails{\P'}\TypedBy\Session+\t'%
          \quad%
          \Session~\text{not \t'-reading}
        }%
      \]

      \noindent Since \Session\ is not \t'-reading, by~\Cref{def:t_reading_session}, it follows that \Session\ is also not \t-reading.
      Therefore, we obtain our thesis:
      \(
      \Gam;\Timers+\t\Entails
      \Delay{\t'-\t}.\P'
      \TypedBy\Session+\t
      \).
      %
      \item\label{case:sr_time_step_preserve_wf_delay_t_gtr}
      If $\t>\t'$,
      then
      \(
      \Time=\Time{\P'}_{\t-\t'}
      \).
      The thesis follows the induction hypothesis:
      \[
        \Gam;\Timers+\t\Entails
        \Time{\P'}_{\t-\t'}
        \TypedBy\Session+\t
      \]
    \end{caseanalysis}
    %
    \item\label{case:sr_time_step_preserve_wf_par}
    If rule \Rule{Par},
    then
    \(
    \P = \P'\mid\Q
    \)
    and following~\Cref{fig:type_checking}~\Cref{eq:type_checking_standard}:
    \[
      \infer[\Rule{Par}]{
        \Gam;\Timers\Entails{\P'\mid\Q}\TypedBy(\Session_1,\Session_2)%
      }{
        \Gam;\Timers_1\Entails{\P'}\TypedBy\Session_1%
        \quad%
        \Gam;\Timers_2\Entails{\Q}\TypedBy\Session_2%
      }%
    \]

    \noindent where $\Timers=(\Timers_1,\Timers_2)$ and $\Session=(\Session_1,\Session_2)$.
    Since \Session\ is \wf, it follows both \Session_1\ and \Session_2\ are also \wf.
    By the induction hypothesis:
    \[
      \infer[\Rule{Par}]{
        \Gam;(\Timers_1+\t,\Timers_2+\t)\Entails\Time{\P'\mid\Q}\TypedBy
        (\Session_1,\Session_2)+\t%
      }{
        \Gam;\Timers_1+\t\Entails\Time{\P'}\TypedBy\Session_1+\t%
        \quad%
        \Gam;\Timers_2+\t\Entails\Time{\Q}\TypedBy\Session_2+\t%
      }%
    \]

    \noindent where $\Session+\t=(\Session_1,\Session_2)+\t=(\Session_1+\t,\Session_2+\t)$.
    (See other cases for \P'\ and \Q.)
    %
    \item\label{case:sr_time_step_preserve_wf_res}
    \mnote{Probably can be simplified. The fact $\Que_1=\emptyset=\Que_2$
      is true but the way is proved doesn't look right: the hypothesis
      \NEQ0 does not apply to $\P'$!}
    If rule \Rule{Res},
    then
    \(
    \P = \Scope{\pq}{\P'}
    \)
    and it follows~\Cref{fig:type_checking}~\Cref{eq:type_checking_standard}:
    \[
      \infer[\Rule{Res}]{
        \Gam;\Timers\Entails{\Scope{\pq}{\P'}}\TypedBy\Session%
      }{
        \begin{array}[c]{c}
          \Gam;\Timers\Entails{\P'}\TypedBy\Session,%
          ~{\p:\Cfg_1},%
          ~{\qp:\Que_1},%
          ~{\q:\Cfg_2},%
          ~{\pq:\Que_2}%
          \\[\ArrayTallLineSpacing]
          \Cfg;{\Que}_1\Compat\Cfg;{\Que}_2%
          \qquad%
          \forall i \in \{1,2\}.\, \S_i \text{ well-formed against }\nu_i
        \end{array}
      }%
    \]

    \noindent By the second premise of rule \Rule{Res}, $\Cfg;{\Que}_1\Compat\Cfg;{\Que}_2$ and both \Cfg_1\ and \Cfg_2\ are \wf.
    Since \NEQ0\ then by~\Cref{lem:sr_empty_queue_cfg},
    $\NEQ{\Session}=\emptyset$
    and therefore,
    it follows that
    $\Que_1=\emptyset=\Que_2$.
    By~\Cref{cond:cfg_compat_dual} of~\Cref{def:cfg_compat}, since $\Que_1=\emptyset=\Que_2$, then $\S_1=\D_2$ and $\Val_1=\Val_2$.
    Hereafter, we only show for \p, as the proof is similar for \q, and .
    By the premise of rule \Rule{Res}, \p\ is \wf.
    By~\Cref{lem:tp_wf_cfg_is_live}, \p\ is \live.
    By~\Cref{def:cfg_live}, either $\S_1=\End$, or \p\ is \fe.
    \begin{caseanalysis}
      %
      \item\label{case:sr_time_step_preserve_wf_res_end}
      If $\S_1=\End$, then by~\Cref{lem:tp_end_wf} we obtain our thesis as \End\ is always \wf.
      %
      \item\label{case:sr_time_step_preserve_wf_res_fe}
      Otherwise, by~\Cref{lem:tp_cfg_fe} \p\ is \fe.
      By~\Cref{def:cfg_fe}, $\exists \t'$ such that $\Cfg_1\Trans:{\t',\Comm\Msg}$.
      It follows that $\forall \t'{2}\leq\t'$ \Cfg{\Val+\t'{2}}\ is also \fe,
      since $\Cfg{\Val+\t'{2}}\Trans:{\t'-\t'{2},\Comm\Msg}$.
      Therefore, \Cfg{\Val_1+\t'},{\S_1}\ is \wf.
      (The same holds for \Cfg_2.)
      It follows that $\Cfg{\Val_1+\t'},{\S_1};{\Que_1}\Compat\Cfg{\Val_2+\t'},{\S_2};{\Que_2}$.

      It remains for us to show that $\t\leq\t'$ for the hypothesis to hold.
      We proceed by case analysis on the structure of \Comm:
      \begin{caseanalysis}
        %
        \item\label{case:sr_time_step_preserve_wf_res_fe_recv}
        If $\Comm=?$,
        then by~\Cref{lem:sr_wt_wf_role_recv_prc_wait}, $\p\in\Wait{\P}$.
        By~\Cref{fig:func_wait_neq}, \P\ must be structured $\P=\P'\mid\Q$ where
        \(
        \P'\equiv\On{\p}\Recv^{\Diam\n}{\l}[\v]:{\P}_i~\After:{\Q}
        \).
        (We only show this case as it can be applied for the other cases, as discussed in~\Cref{sssec:process_reduction_rules}.)
        By~\Cref{claim:tc_inversion_timeout} of~\Cref{lem:tc_inversion}, $\S_1=\Choice{\l}[\T]_i$.
        The thesis follows~\Cref{case:sr_time_step_preserve_wf_timeout}.
        %
        \item\label{case:sr_time_step_preserve_wf_res_fe_send}
        If $\Comm=!$,
        then since $\S_1=\D_2$, by~\Cref{lem:tp_en_cfg_dual_en}, $\Cfg{\Val_2+\t'},{\D_2}\Trans:{?\Msg}$.
        (See~\Cref{case:sr_time_step_preserve_wf_res_fe_recv}.)
        In the case that \P\ is structured
        \(
        \P=\P'\mid\Q
        \)
        where
        \(
        \P'\equiv\On{\p}\Send{\l}[\v].{\P'{2}}
        \),
        then by~\Cref{def:time_passing_func_ef}, \Time\ is not defined, as sending actions cannot be delayed.
      \end{caseanalysis}
    \end{caseanalysis}
    %
    \item\label{case:sr_time_step_preserve_wf_rec}
    If rule \Rule{Rec},
    then
    \(
    \P = \Def{\Rec*{\vec{v};\vec{r}}=\P'}:{\Q}
    \)
    and:
    \[
      \Time = \Def{\Rec*{\vec{v};\vec{r}}=\P'}:{\Time{\Q}}
    \]

    \noindent and, following~\Cref{fig:type_checking}~\Cref{eq:type_checking_standard}:
    \[
      \infer[\Rule{Rec}]{
      \Gam;\Timers\Entails{\Def{\Rec*{\vec{v};\vec{r}}=\P'}:{\Q}}\TypedBy\Session%
      }{
      \begin{array}[t]{@{}c@{}}%
        \forall (\vec{\Val},~\vec{\S})
        \in\FormatRecCheckStyle{\Delta},
        {\theta'}\in\FormatRecCheckStyle{\theta}: %
        \Gam,
        {\vec{v}:\vec{T}},%
        {\Rec|}%
        ;\Timers'
        \Entails{\P'}
        \TypedBy
        \vec{r}: (\vec{\Val},~\vec{\S})%
        \\[\ArrayLineSpacing]
        \Gam,{\Rec|};\Timers\Entails{\Q}\TypedBy\Session%
      \end{array}%
      }%
    \]

    \noindent By the induction hypothesis:
    \(
    \Gam,{\Rec|};\Timers+\t\Entails\Time{\Q}\TypedBy\Session+\t%
    \)
    where $\Session+\t$ is \wf.
    The thesis then follows by rule \Rule{Rec}:
    \[
      \infer[\Rule{Rec}]{
        \Gam;\Timers +\t\Entails{\Time}\TypedBy\Session +\t%
      }{
        \begin{array}[b]{@{}c@{}}%
          \forall (\vec{\Val},~\vec{\S})
          \in\FormatRecCheckStyle{\Delta},
          {\theta'}\in\FormatRecCheckStyle{\theta}: %
          \Gam,
          {\vec{v}:\vec{T}},%
          {\Rec|}%
          ;\Timers'
          \Entails{\P'}
          \TypedBy
          \vec{r}: (\vec{\Val},~\vec{\S})%
          \\[\ArrayLineSpacing]
          \Gam,{\Rec|};\Timers+\t\Entails\Time{\Q}\TypedBy\Session+\t \tag*{\qed}
        \end{array}%
      }%
    \]
  \end{caseanalysis}
  \renewcommand{\qed}{}
\end{proof}

\clearpage

\lemTimeStep*
\begin{proof}
  We proceed by induction on the derivation of $\Gam;\Timers\Entails\P\TypedBy\Session$, analysing the last rule applied of the \tc\ rules given in~\RefTypeCheck.
  Since \Session\ is \fbal, only rules \Rule{Res} and \Rule{Par} are applicable:
  \begin{caseanalysis}
    %
    \item\label{case:sr_time_step_res}
    If rule \Rule{Res},
    then
    \(
    \P = \Scope{\pq}{\P'}
    \)
    and following~\Cref{fig:type_checking}~\Cref{eq:type_checking_standard}:
    \[
      \infer[\Rule{Res}]{
        \Gam;\Timers\Entails{\Scope{\pq}{\P'}}\TypedBy\Session%
      }{
        \begin{array}[c]{c}
          \Gam;\Timers\Entails{\P'}\TypedBy\Session,%
          ~{\p:\Cfg_1},%
          ~{\qp:\Que_1},%
          ~{\q:\Cfg_2},%
          ~{\pq:\Que_2}%
          \\[\ArrayTallLineSpacing]
          \Cfg;{\Que}_1\Compat\Cfg;{\Que}_2%
          \qquad%
          \forall i \in \{1,2\}.\, \S_i \text{\ well-formed\ against\ }\nu_i
        \end{array}
      }%
    \]

    \noindent Let $\Session'=\Session,
      \allowbreak{\p:\Cfg_1},
      \allowbreak{\qp:\Que_1},
      \allowbreak{\q:\Cfg_2},
      \allowbreak{\pq:\Que_2}$.
    By the second premise of rule \Rule{Res}, it follows that \Session'\ is \wf.
    By~\Cref{def:fbal_session}, it follows that \Session'\ is \fbal.
    By~\Cref{def:time_passing_func_ef},
    \(
    \Time=\Scope{\pq}{\Time{\P'}}
    \).
    Therefore, by the induction hypothesis:
    \[
      \infer[\Rule{Res}]{
        \Gam;\Timers+\t\Entails
        \Time{\Scope{\pq}{\P'}}
        \TypedBy\Session+\t%
      }{
        \begin{array}[c]{c}
          \Gam;\Timers+\t\Entails\Time{\P'}\TypedBy\Session+\t,%
          ~{\p:\Cfg{\Val_1+\t},{\S_1}},%
          ~{\qp:\Que_1},%
          ~{\q:\Cfg{\Val_2+\t},{\S_2}},%
          ~{\pq:\Que_2}%
          \\[\ArrayTallLineSpacing]
          \Cfg;{\Que}_1\Compat\Cfg;{\Que}_2%
          \qquad%
          \forall i \in \{1,2\}.\, \S_i \text{\ well-formed\ against\ }\nu_i
        \end{array}
      }%
    \]

    \noindent where $\Session'+\t$ is \wf\ and \fbal.
    The second premise of rule \Rule{Res} holds:
    \begin{itemize}
      \item Since $\Session'+\t$ is \fbal, by~\Cref{def:fbal_session},
            \(
            \Cfg{\Val_1+\t},{\S_1};{\Que_1}
            \Compat
            \Cfg{\Val_2+\t},{\S_2};{\Que_2}
            \).
      \item The latter coincides with~\Cref{def:cfg_wf} and holds since $\Session'+\t$ is \wf.
    \end{itemize}

    \noindent The first premise holds by induction.
    It remains for us to show that $\Session+\t$ is \wf\ and \fbal.
    Since $\Dom{\Session+\t}\subseteq\Dom{\Session'+\t}$, it follows that $\Session+\t$ is \wf\ and \fbal, and we have obtained our thesis.
    %
    \item\label{case:sr_time_step_par}
    If rule \Rule{Par},
    then
    \(
    \P = \P'\mid\Q
    \)
    and it follows~\Cref{fig:type_checking}~\Cref{eq:type_checking_standard}:
    \[
      \infer[\Rule{Par}]{
        \Gam;\Timers\Entails{\P'\mid\Q}\TypedBy(\Session_1,\Session_2)%
      }{
        \Gam;\Timers_1\Entails{\P'}\TypedBy\Session_1%
        \qquad%
        \Gam;\Timers_2\Entails{\Q}\TypedBy\Session_2%
      }%
    \]

    \noindent where $\Timers=(\Timers_1,\Timers_2)$ and $\Session=(\Session_1,\Session_2)$.
    It follows that both \Session_1\ and \Session_2\ are \wf.
    By~\Cref{def:bal_session}, both \Session_1\ and \Session_2\ are \bal.
    (Their \fbalness\ does not matter.)
    By~\Cref{lem:sr_wt_delay}, both \Session_1\ and \Session_2\ are \del.
    By~\Cref{lem:sr_delay_session_preserve_bal}, both \Session_1\ and \Session_2\ are \bal.
    By~\Cref{def:time_passing_func_ef}, $\Time=\Time{\P'}\mid\Time{\Q}$.
    By the induction hypothesis:
    \[
      \infer[\Rule{Par}]{
        \Gam;(\Timers_1+\t,\Timers_2+\t)\Entails\Time{\P'\mid\Q}\TypedBy(\Session_1,\Session_2)+\t%
      }{
        \Gam;\Timers_1+\t\Entails\Time{\P'}\TypedBy\Session_1+\t%
        \qquad%
        \Gam;\Timers_2+\t\Entails\Time{\Q}\TypedBy\Session_2+\t%
      }%
    \]

    \noindent where $\Timers+\t=(\Timers_1+\t,\Timers_2+\t)$
    and $\Session+\t=(\Session_1,\Session_2)+\t=(\Session_1+\t,\Session_2+\t)$.
    By~\Cref{case:sr_time_step_preserve_wf_par} in~\Cref{lem:sr_time_step_preserve_wf}, $\Session+\t$ is \wf.
    It remains for us to prove that $\Session+\t$ is \fbal.
    Since \Session\ is \fbal, by~\Cref{case:fbal_session_que} of~\Cref{def:fbal_session}, it follows that $\forall \p$ such that $\p\in\Dom{\Session}$ and $\Session(\p)=\Cfg_1$, then $\exists \Session',\q,\Val_2,\S_2,\Que_1,\Que_2$ such that
    \(
    \Session=\Session',{\p:\Cfg_1},{\qp:\Que_1},{\q:\Cfg_2},{\pq:\Que_2}
    \)
    is \bal.
    (By~\Cref{case:fbal_session_que} of~\Cref{def:fbal_session}, the same analogously holds $\forall \qp$.)
    By~\Cref{def:bal_session}, $\Cfg;{\Que}_1\Compat\Cfg;{\Que}_2$.
    By~\Cref{cond:cfg_compat_dual} of~\Cref{def:cfg_compat}, if $\Que_1=\emptyset=\Que_2$, then $\S_1=\D_2$ and $\Val_1=\Val_2$.
    In such a case, it follows that
    \(
    \Session+\t=\Session'+\t,{\p:\Cfg{\Val_1+\t},{\S_1}},{\qp:\Que_1},{\q:\Cfg{\Val_2+\t},{\S_2}},{\pq:\Que_2}
    \)
    and $\Session+\t$ is \fbal.
    We now show that it must be that $\Que_1=\emptyset=\Que_2$.
    Suppose by contradiction, that $\Que_1=\Msg;\Que_1'$.
    By~\Cref{def:bal_session}, if
    \(
    \Session,{\p:\Cfg_1},{\qp:\Msg;\Que_1'}
    \)
    then
    \(
    \Cfg_1\Trans:{?\Msg}\Cfg_1'
    \)
    and
    \(
    \Session,{\p:\Cfg_1'},{\qp:\Que_1'}
    \)
    is \bal.
    Since $\Cfg_1\Trans:{?\Msg}$, by~\Cref{lem:sr_wt_wf_role_recv_prc_wait}, $\p\in\Wait$.
    By~\Cref{lem:sr_neq_prc}, $\p\in\NEQ$.
    However, since \Time\ is defined, by~\Cref{lem:sr_def_delay_no_recv}, $\Wait\cap\NEQ=\emptyset$, which is a contradiction.
    Therefore, all queues must be empty and $\Session+\t$ is \fbal.
    \qedhere
  \end{caseanalysis}
\end{proof}

\subsubsection{Action Steps}\label{sssec:proof:sr:steps:action}

\lemActionStep*
\begin{proof}
  We proceed by induction on the derivation of the reduction
  \(
  \PCfg\Reduce-\PCfg'
  \)
  analysing the last rule applied of those given in~\Cref{fig:prc_reduction}:
  \begin{caseanalysis}
    %
    \item\label{case:sr_action_step_send}
    If rule \Rule{Send},
    then
    \(
    \P = {\On{\p}\Send{\l}[\v].{\P'{2}}}
    \mid{\pq:\h}
    \)
    and it follows~\Cref{fig:prc_reduction}~\Cref{eq:prc_reduction_communication}:
    \[
      \PCfg,{
      {\On{\p}\Send{\l}[\v].{\P'{2}}}
      \mid{\pq:\h}
      }
      \Reduce-
      \PCfg,{
      {\P'{2}}
      \mid{\pq:\h\cdot\lv}
      }
      \quad\Rule{Send}
    \]

    \noindent where
    \(
    \P' = \P'{2}\mid{\pq:\h\cdot\lv}
    \)
    and $\Timers'=\Timers$.
    By~\Cref{claim:tc_inversion_par} of~\Cref{lem:tc_inversion},
    $\Timers=(\Timers_1,\Timers_2)$.
    \(
    \Session = (\Session_1,\Session_2)
    \)
    and:
    \begin{equation}\label{eq:sr_action_step_send}
      \Gam;\Timers_1\Entails
      \On{\p}\Send{\l}[\v].{\P'{2}}
      \TypedBy\Session_1
      \qquad
      \Gam;\Timers_2\Entails
      {\pq:\h}
      \TypedBy\Session_2
    \end{equation}

    \noindent By~\Cref{claim:tc_inversion_empty_queue} of~\Cref{lem:tc_inversion},
    \(
    \Session_2=\Session_2',{\pq:\Que}
    \)
    and
    if $\h=\emptyset$ then $\Que=\emptyset$.
    Conversely, if $\h\neq\emptyset$ then, by~\Cref{claim:tc_inversion_neq} of~\Cref{lem:tc_inversion}, $\Que\neq\emptyset$.
    By~\Cref{claim:tc_inversion_send} of~\Cref{lem:tc_inversion},
    \(
    \Session_1=\Session_1',{\p:\Cfg,{\Choice{\l}[\T]_i}}
    \)
    and $\exists j\in I$ \st\ %
    $\l=\l_j$, $!=\Comm_j$, $\Val\models\Constraints_j$ and:
    \begin{caseanalysis}
      %
      \item\label{case:sr_action_step_send_value}
      If \T_j\ is \bt,
      then $\Gam\Entails\v:\T_j$ and
      \(
      \Gam;\Timers_1\Entails\P'{2}\TypedBy
      \Session_1',{\p:\Cfg{\Val\Reset_j},{\S_j}}
      \).
      To obtain our thesis we must prove
      $\exists \Session':\Session\Reduce*\Session'$,
      and \Session'\ is \wf, \bal, and:
      \(
      \Gam;(\Timers_1,\Timers_2)\Entails
      {\P'{2}}
      \mid{\pq:\h\cdot\lv}
      \TypedBy\Session'
      \).
      By~\Cref{claim:tc_inversion_par} of~\Cref{lem:tc_inversion},
      \(
      \Session'= (\Session_1'{2},\Session_2'{2})
      \)
      and:
      \begin{equation}\label{eq:sr_action_step_send_reduction}
        \Gam;\Timers_1\Entails
        {\P'{2}}
        \TypedBy\Session_1'{2}
        \qquad
        \Gam;\Timers_2\Entails
        {\pq:\h\cdot\lv}
        \TypedBy\Session_2'{2}
      \end{equation}

      \noindent We proceed to show the structure of \Session'\ necessary to obtain our thesis.
      By inner induction on the derivation of each in~\Cref{eq:sr_action_step_send}, analysing the last rule applied.

      Firstly:
      \[
        \infer[\Rule{VSend}]{%
        \Gam;\Timers_1\Entails
        \On{\p}\Send{\l}[\v].{\P'{2}}
        \TypedBy\Session_1',~{\p:\Cfg,{\Choice{\l}[\T]_i}}
        }{%
        \begin{array}[t]{@{}c@{}}
          \exists j\in I: %
          (\l=\l_j)
          \land
          (\Val\models\Constraints_j)%
          \land
          (\T_j~\text{\bt})
          \land%
          (\Gam\Entails\v:\T_j)
          \land\mbox{}%
          \\[\ArrayTallLineSpacing]%
          (\Comm_j=!)
          \land
          \Gam;\Timers\Entails{\P'{2}}\TypedBy\Session_1',~{\p:\Cfg{\Val\Reset_j},{\S_j}}%
        \end{array}%
        }
      \]

      \noindent Therefore, it must be that
      \(
      \Session_1'{2} = \Session_1',{\p:\Cfg{\Val\Reset_j},{\S_j}}
      \).
      Since $\Gam\Entails\v:\T_j$
      then by~\Cref{claim:sr_wt_msg_bt} of~\Cref{lem:sr_wt_msg}
      \(
      \Gam;\Timers_2\Entails
      {\pq:\h\cdot\lv}
      \TypedBy
      \Session_2',~
      {\pq:\Que;\lT<>_j}
      \)
      and $\l=\l_j$,
      which coincides with the latter of~\Cref{eq:sr_action_step_send_reduction}.
      Therefore,
      \(
      \Session_2'{2} = \Session_2',{\pq:\Que;\lT<>_j}
      \).

      To summarise:
      \begin{itemize}
        \item \(
              \Session =
              \Session_1',~{\p:\Cfg,{\Choice{\l}[\T]_i}},~
              \Session_2',~{\pq:\Que}
              \)
        \item \(
              \Session' =
              \Session_1',~{\p:\Cfg{\Val\Reset_j},{\S_j}},~
              \Session_2',~{\pq:\Que;\lT<>_j}
              \qquad (j\in I) 
              \)
      \end{itemize}

      \noindent Such a reduction is possible via rule \Rule[\Session]{Send} in~\Cref{fig:session_reduction}:
      \[
        \infer[\Rule[\Session]{Send}]{%
        {
        \Session'{2},
        ~{\p:\Cfg,{\Choice{\l}[\T]_i}}
        ~{\pq:\Que}
        }
        \Reduce
        {
        \Session'{2},
        ~{\p:\Cfg'}
        ~{\pq:\Que;\Msg}
        }
        }{%
        \Cfg,{\Choice{\l}[\T]_i}
        \Trans:{!\Msg}
        \Cfg'
        }
      \]

      \noindent where
      $\Msg=\lT<>_j$ and $\Session'{2}=(\Session_1',\Session_2')$.
      By~\Cref{case:tp_cfg_trans_preserve_wf_com} in~\Cref{lem:tp_cfg_trans_preserve_wf}, it follows that, for some $j\in I$, $\Val\models\Constraints_j$, $\Val'=\Val\Reset_j$ and $\S'=\S_j$.
      Additionally, it follows that \p\ being \wf\ indicates that the constraints on actions are structured to preserve \wfness\ across such transitions.
      Therefore, it follows that
      \(
      \Session'=\Session'{2},~{\p:\Cfg{\Val\Reset_j},{\S_j}}
      ~{\pq:\Que;\Msg}
      \)
      is \wf.
      By~\Cref{lem:sr_wt_bal_session_reduce_preserve_bal}, \Session'\ is \bal.

      Finally, by rule \Rule{Par} (where $\l=\l_j$):
      \[
        \infer[\Rule{Par}]{%
        \Gam;(\Timers_1,\Timers_2)\Entails
        \P'{2}\mid
        {\pq:\h\cdot\lv}
        \TypedBy
        \Session_1',~{\p:\Cfg{\Val\Reset_j},{\S_j}},~%
        \Session_2',~{\pq:\Que;\lT<>_j}
        }{%
        \Gam;\Timers_1\Entails
        {\P'{2}}
        \TypedBy
        \Session_1',~{\p:\Cfg{\Val\Reset_j},{\S_j}}
        \quad
        \Gam;\Timers_2\Entails
        {\pq:\h\cdot\lv}
        \TypedBy
        \Session_2',~
        {\pq:\Que;\lT<>_j}
        }
      \]
      %
      %
      \item\label{case:sr_action_step_send_del}
      If $\T_j=\Del'$,
      then
      \(
      \Session_1'{2} = \Session_1'{3},{\b:\Cfg'}
      \)
      and, $\Val'\models\Constraints'$, $\v=\b$ and:
      \[
        \Gam;\Timers_1\Entails\P'{2}\TypedBy
        \Session_1'{3},~{\p:\Cfg{\Val\Reset_j},{\S_j}}
      \]

      \noindent The rest follows~\Cref{case:sr_action_step_send_value}.
    \end{caseanalysis}

    \noindent In~\Cref{case:sr_action_step_send_value}, we have shown that the process of a configuration yielded by a reduction via rule \Rule{Send} of~\Cref{fig:prc_reduction}~\Cref{eq:prc_reduction_communication}, is \wt\ against a session environment \Session', which itself is reachable via a reduction via the rules in~\Cref{fig:session_reduction}, and that the \wfness\ and \balness\ of \Session'\ is preserved.
    \Cref{case:sr_action_step_send_del} follows similarly, for delegations.
    %
    \item\label{case:sr_action_step_recv}
    If rule \Rule{Recv},
    then
    \(
    \P = \On{\p}\Recv^{\e}{\l}[\v]:{\P}_i
    \mid{\qp:\lv\cdot\h}
    \)
    and by~\Cref{fig:prc_reduction}~\Cref{eq:prc_reduction_communication}:
    \[
      \infer[\Rule{Recv}]{
      \PCfg,{{\On{\p}\Recv^{\e}{\l}[\v]:{\P}_i}\mid{\qp:\lv\cdot\h}}%
      \Reduce-%
      \PCfg,{{\P_j\Subst{\v}{\v_j}}\mid{\qp:\h}}%
      }{
      j\in I%
      \quad%
      \l=\l_j%
      }%
    \]

    \noindent where
    \(
    \P' = \P_j\Subst{\v}{\v_j}\mid{\qp:\h}
    \)
    and $\Timers'=\Timers$.
    By~\Cref{claim:tc_inversion_par} of~\Cref{lem:tc_inversion}, $\Timers=(\Timers_1,\Timers_2)$ and $\Session=(\Session_1,\Session_2)$ and:
    \begin{equation}\label{eq:sr_action_step_recv}
      \Gam;\Timers_1
      \Entails
      {\On{\p}\Recv^{\e}{\l}[\v]:{\P}_i}
      \TypedBy\Session_1
      \qquad
      \Gam;\Timers_2
      \Entails
      {\qp:\lv\cdot\h}
      \TypedBy\Session_2
    \end{equation}

    \noindent By~\Cref{claim:tc_inversion_neq} of~\Cref{lem:tc_inversion},
    \(
    \Session_2=\Session_2',{\pq:\lT<>;\Que}
    \).
    By~\Cref{claim:tc_inversion_recv} of~\Cref{lem:tc_inversion},
    \(
    \Session_1=\Session_1',{\p:\Cfg,{\Choice{\l}[\T]_k}}
    \)
    and $\forall k\in K$ such that $\Val\models\Constraints_k$,
    then $\Comm_k=?$ and $\exists i\in I$ such that:
    \[
      \Gam;\Timers_1\Entails
      \On{\p}\Recv*^{\e}{\l}[\v]:{\P}_i
      \TypedBy\Session_1',~{\p:\Cfg,{\Option?-{\l}[\T]{\delta}[\lambda].{\S}}}
    \]

    \noindent where $\Option?-{\l}[\T]{\delta}[\lambda].{\S}=\Option-{\l}[\T]{\delta}[\lambda].{\S}_k$.
    Hereafter we write $\Option-{\l}[\T]{\delta}[\lambda].{\S}_k$.
    It follows that in~\Cref{eq:sr_action_step_recv}, $\lv=\lv_i$.
    \mnote{I think from here the substitution lemma should be used, simplifying the proof substantially.}
    We proceed by inner induction on the derivation, analysing the last rule applied of those in~\Cref{fig:type_checking_communication}~\Cref{eq:type_checking_recv_single}:
    \begin{caseanalysis}
      %
      \item\label{case:sr_action_step_recv_value}
      If rule \Rule{VRecv},
      then:
      \begin{equation}\label{eq:sr_action_step_recv_value}
        \infer[\Rule{VRecv}]{
        \Gam;\Timers_1\Entails
        {\On{\p}\Recv*^{\e}{\l}[\v]:{\P}_i}
        \TypedBy
        \Session_1',~{\p:\Cfg,{\Option-{\l}[\T]{\delta}[\lambda].{\S}_k}}%
        }{
        \begin{array}[c]{@{}c@{}}%
          \T_k~\text{\bt}%
          \quad%
          \Session_1'~\text{not $\e$-reading}
          \quad
          {\forall\t:\Val+\t\models\Constraints_k\Bimplies\t\in\e}
          \\[\ArrayLTSPremiseTallLineSpacing]
          \forall\t\in\e:
          \Gam,{\v_i:\T_k};\Timers_1+\t\Entails
          {\P_i}
          \TypedBy
          \Session_1'+\t,%
          ~{\p:\Cfg{\Val+\t\Reset_k},{\S_k}}%
        \end{array}%
        }%
      \end{equation}

      \noindent where $\l_k=\l_i$ (implicitly).
      It remains to show that $\exists \Session'$ such that $\Session\Reduce*\Session'$, and \Session'\ is \wf, \bal, and:
      \(
      \Gam;(\Timers_1,\Timers_2)
      \Entails
      {\P_j\Subst{\v}{\v_j}\mid{\qp:\h}}
      \TypedBy
      \Session'
      \).
      By~\Cref{claim:tc_inversion_par} of~\Cref{lem:tc_inversion}, $\Timers'=(\Timers_1',\Timers_2')$ and $\Session'=(\Session_1'{2},\Session_2'{2})$ and:
      \[
        \Gam;\Timers_1'
        \Entails
        {\P_j\Subst{\v}{\v_j}}
        \TypedBy\Session_1'{2}
        \qquad
        \Gam;\Timers_2
        \Entails
        {\qp:\h}
        \TypedBy\Session_2'{2}
      \]

      \noindent Clearly, following~\Cref{eq:sr_action_step_recv_value},
      \(
      \Session_1'{2} = \Session_1',{\p:\Cfg{\Val\Reset_k},{\S_k}}
      \).
      By~\Cref{claim:tc_inversion_empty_queue,claim:tc_inversion_neq} of~\Cref{lem:tc_inversion},
      \(
      \Session_2'{2} = \Session_2',{\qp:\Que}
      \)
      and, if $\h=\emptyset$ then $\Que=\emptyset$.
      (Simiarly, if $\h\neq\emptyset$ then $\Que\neq\emptyset$.)
      Therefore, it must be that
      \(
      \Session'=
      \Session_1',
      ~{\p:\Cfg{\Val\Reset_k},{\S_k}},
      ~\Session_2',
      ~{\qp:\Que}
      \)
      and, it follows $\Session\Reduce*\Session'$ is possible via rule \Rule[\Session]{Recv} in~\Cref{fig:session_reduction}:
      \[
        \infer[\Rule[\Session]{Recv}]{
        \Session'{2},~{\p:\Cfg,{\Option-{\l}[\T]{\delta}[\lambda].{\S}_k}},
        ~{\qp:{\lT<>;\Que}}
        \Reduce
        \Session'{2},~{\p:\Cfg'},
        ~{\qp:\Que}
        }{
        \Cfg,{\Choice*{\l}[\T]_k};{\lT<>;\Que}
        \Trans:{\tau}
        \Cfg{\Val'},{\S'};{\Que}
        }
      \]

      \noindent where $\Session'{2}=(\Session_1',\Session_2')$.
      By~\Cref{case:tp_cfg_trans_preserve_wf_par} in~\Cref{lem:tp_cfg_trans_preserve_wf}, $\lT<>=\lT<>_k$ and, the \wfness\ rules in~\Cref{fig:type_wf_rules} are enough to guarantee that \wfness\ is preserved across transitions via rules
      in~\Cref{eq:type_semantics_tuple,eq:type_semantics_triple,eq:type_semantics_system}.
      Therefore, \Session'\ is \wf.
      By~\Cref{lem:sr_wt_bal_session_reduce_preserve_bal}, \Session'\ is also \bal.
      By induction on the latter derivation in~\Cref{eq:sr_action_step_recv}, analysing the last rule applied.
      Since \T_k\ is \bt, by rule \Rule{VQue}:
      \begin{equation}\label{eq:sr_action_step_recv_que}
        \infer[\Rule{VQue}]{
        \Gam;\Timers_2\Entails
        {\qp:\lv_i\cdot\h}
        \TypedBy
        \Session_2',~{\qp:\lT<>_k;\Que}%
        }{
        \T_k~\text{\bt}%
        \quad%
        \Gam\Entails{\v_i:\T_k}%
        \quad%
        \Gam;\Timers_2\Entails
        {\qp:{\h}}
        \TypedBy\Session_2',%
        ~{\qp:{\Que}}%
        }%
      \end{equation}

      \noindent It remains to show that
      \(
      \Gam;(\Timers_1,\Timers_2)
      \Entails
      {\P_j\Subst{\v}{\v_j}\mid{\qp:\h}}
      \TypedBy
      \Session'
      \).
      Since $\l_k=\l_i$ and $\lv=\lv_i$, then by reformulation of~\Cref{eq:sr_action_step_recv}:
      \[
        \Gam;\Timers_1
        \Entails
        {\On{\p}\Recv*^{\e}{\l}[\v]:{\P}_i}
        \TypedBy
        \Session_1',~{\p:\Cfg,{\Option-{\l}[\T]{\delta}[\lambda].{\S}_k}}
      \]\[
        \Gam;\Timers_2
        \Entails
        {\qp:\lv_i\cdot\h}
        \TypedBy
        \Session_2',~{\qp:\lT<>_k;\Que}
      \]

      \noindent The induction of the reformulated derivations (above) follow the same as~\Cref{eq:sr_action_step_recv_value} and~\Cref{eq:sr_action_step_recv_que} respectively.
      Therefore, we conclude by rule \Rule{Par}:
      \[
        \infer[\Rule{Par}]{%
        \Gam;(\Timers_1,\Timers_2)
        \Entails
        {\P_j\Subst{\v}{\v_j}\mid{\qp:\h}}
        \TypedBy
        \Session_1',
        ~{\p:\Cfg,{\p:\Cfg{\Val\Reset_k},{\S_k}}},
        ~\Session_2',
        ~{\qp:\h}
        }{%
        \Gam;\Timers_1
        \Entails
        \P_j\Subst{\v}{\v_j}
        \TypedBy
        \Session_1',
        ~{\p:\Cfg,{\p:\Cfg{\Val\Reset_k},{\S_k}}},
        \qquad
        \Gam;\Timers_2
        \Entails
        {\qp:\h}
        \TypedBy
        \Session_2',
        ~{\qp:\h}
        }
      \]
      %
      \item\label{case:sr_action_step_recv_del}
      If rule \Rule{DRecv},
      then it follows similarly to~\Cref{case:sr_action_step_recv_value}.
    \end{caseanalysis}

    \noindent By~\Cref{case:sr_action_step_recv_value}, a \wt\ receiving process may reduce and remain \wt\ against a session environment resulting from a corresponding reduction via the rules in~\Cref{fig:session_reduction}.
    %
    \item\label{case:sr_action_step_recv_t}
    If rule \Rule{Recv}[T]*,
    then
    \(
    \P = \On{\p}\Recv^{\Diam\n}{\l}[\v]:{\P}_i~\After:{\Q}
    \mid{\qp:\lv\cdot\h}
    \)
    and it follows~\Cref{fig:prc_reduction}~\Cref{eq:prc_reduction_communication}:
    \[
      \infer[\Rule{Recv}[T]*]{
      \PCfg,{{\On{\p}\Recv^{\Diam\n}{\l}[\v]:{\P}_i~\After:{\Q}}\mid{\qp:\lv\cdot\h}}%
      \Reduce-%
      \PCfg,{{\P_j\Subst{\v}{\v_j}}\mid{\qp:\h}}%
      }{
      j\in I%
      \quad%
      \l=\l_j%
      }%
    \]

    \noindent where
    \(
    \P' = \P_j\Subst{\v}{\v_j}\mid{\qp:\h}
    \).
    The rest follows~\Cref{case:sr_action_step_recv}.
    %
    \item\label{case:sr_action_step_set}
    If rule \Rule{Set},
    then
    \(
    \P = \Set{x}.\P'
    \)
    and following~\Cref{fig:prc_reduction}~\Cref{eq:prc_reduction_time}:
    \[
      \PCfg,{\Set{x}.\P'}%
      \Reduce-%
      \PCfg[\Timers[\circled{x}]],{\P'}%
      \hspace{0.75ex} \Rule{Set}%
    \]

    \noindent where $\Timers'=\Timers[\circled{x}]$.
    By~\Cref{claim:tc_inversion_timer} in~\Cref{lem:tc_inversion},
    it follows that $\Session=\Session'$.
    Therefore, the thesis coincides with the hypothesis.
    %
    \item\label{case:sr_action_step_det}
    If rule \Rule{Det},
    then
    \(
    \P = \Delay{\Duration}.\P'{2}
    \)
    and following~\Cref{fig:prc_reduction}~\Cref{eq:prc_reduction_time}:
    \[
      \infer[\Rule{Det}]{
      \PCfg,{\Delay{\Duration}.\P'{2}}%
      \Reduce-%
      \PCfg,{\Delay{\t}.\P'{2}}%
      }{
      \t\models\Duration\Subst{\t}{\t'}%
      }%
    \]

    \noindent where $\P' = \Delay{\t}.\P'{2}$.
    By~\Cref{claim:tc_inversion_det} in~\Cref{lem:tc_inversion}, it follows that $\Session=\Session'$.
    Therefore, the thesis coincides with the hypothesis.
    %
    \item\label{case:sr_action_step_if_t}
    If rule \Rule{If}[T]*,
    then
    \(
    \P = \If*{\Cond}~\Then{\P'{2}}~\Else{\Q}
    \)
    and following~\Cref{fig:prc_reduction}~\Cref{eq:prc_reduction_time}:
    \[
      \infer[\Rule{If}[T]*]{
      \PCfg,{\If*{\Cond}~\Then{\P'{2}}~\Else{\Q}}%
      \Reduce-%
      \PCfg,{\P'{2}}%
      }{
      \Timers\models\Cond%
      }%
    \]

    \noindent where $\P'=\P'{2}$.
    By~\Cref{claim:tc_inversion_if} in~\Cref{lem:tc_inversion},
    it follows that $\Session=\Session'$.
    Therefore, the thesis coincides with the hypothesis.
    %
    \item\label{case:sr_action_step_if_f}
    If rule \Rule{If}[F]*,
    then it follows similarly to~\Cref{case:sr_action_step_if_t},
    except $\Timers\not\models\Cond$ and $\P'=\Q$.
    %
    \item\label{case:sr_action_step_par}
    If rule \Rule{Par}[L]*,
    then
    \(
    \P = \P'{2}\mid\Q
    \)
    and following~\Cref{fig:prc_reduction}~\Cref{eq:prc_reduction_standard}:
    \[
      \infer[\Rule{Par}[L]*]{
      \PCfg[\Timers_1,\Timers_2],{\P'{2}\mid\Q}%
      \Reduce-%
      \PCfg[\Timers_1',\Timers_2],{\P'{3}\mid\Q}%
      }{
      \PCfg[\Timers_1],{\P'{2}}%
      \Reduce-%
      \PCfg[\Timers_1'],{\P'{3}}%
      }%
    \]

    \noindent where $\Timers'=(\Timers_1',\Timers_2)$ and $ \P' = \P'{3}\mid\Q $.
    By~\Cref{claim:tc_inversion_par} of~\Cref{lem:tc_inversion},
    $\Timers=(\Timers_1,\Timers_2)$ and
    \(
    \Session = (\Session_1,\Session_2)
    \)
    and:
    \[
      \Gam;\Timers_1\Entails\P'{2}\TypedBy\Session_1
      \qquad
      \Gam;\Timers_2\Entails\Q\TypedBy\Session_2
    \]

    \noindent We must show $\exists \Session'$ such that $\Session\Reduce*\Session'$ and
    \(
    \Gam;\Timers'\Entails\P'{3}\mid\Q\TypedBy\Session'
    \).
    Again, by~\Cref{claim:tc_inversion_par} of~\Cref{lem:tc_inversion},
    $\Timers'=(\Timers_1',\Timers_2)$ and
    \(
    \Session' = (\Session_1',\Session_2')
    \)
    and:
    \[
      \Gam;\Timers_1'\Entails\P'{3}\TypedBy\Session_1'
      \qquad
      \Gam;\Timers_2'\Entails\Q\TypedBy\Session_2'
    \]

    \noindent Clearly, $\Timers_2'=\Timers_2$ and $\Session_2'=\Session_2$.
    By rule \Rule[\Session]{L} in~\Cref{fig:session_reduction}:
    \[
      \infer[\Rule[\Session]{L}]{
        \Session_1,~\Session_2
        \Reduce
        \Session_1',~\Session_2
      }{
        \Session_1
        \Reduce
        \Session_1'
      }
    \]

    \noindent The thesis follows by the induction hypothesis:
    since
    \(
    \Gam;\Timers\Entails\P'{2}\TypedBy\Session_1
    \),
    if
    \(
    \PCfg[\Timers_1],{\P'{2}}%
    \Reduce-%
    \PCfg[\Timers_1'],{\P'{3}}%
    \)
    then
    $\exists \Session_1'$ such that
    \(
    \Session_1\Reduce*\Session_1'
    \)
    and
    \(
    \Gam;\Timers_1'\Entails\P'{3}\TypedBy\Session_1'
    \)
    and \Session_1'\ is \bal\ and \wf.
    %
    \item\label{case:sr_action_step_par_r}
    If rule \Rule{Par}[R]*,
    then it follows similarly to~\Cref{case:sr_action_step_par}.
    %
    \item\label{case:sr_action_step_scope}
    If rule \Rule{Scope},
    then
    \(
    \P = \Scope{\pq}{\P'{2}}
    \)
    and following~\Cref{fig:prc_reduction}~\Cref{eq:prc_reduction_standard}:
    \[
      \infer[\Rule{Scope}]{
      \PCfg,{\Scope{\pq}{\P'{2}}}%
      \Reduce-%
      \PCfg',{\Scope{\pq}{\P'{3}}}%
      }{
      \PCfg,{\P'{2}}%
      \Reduce-%
      \PCfg',{\P'{3}}%
      }%
    \]

    \noindent where $ \P' = \Scope{\pq}{\P'{3}} $.
    By~\Cref{claim:tc_inversion_scope} in~\Cref{lem:tc_inversion}, it follows that $\Session=\Session'$, as the scope restriction ensures that the inner session (corresponding to \pq) is independent of the rest of the session.
    Therefore, the thesis coincides with the hypothesis.
    %
    \item\label{case:sr_action_step_def}
    If rule \Rule{Def},
    then
    \(
    \P = \Def{\Rec*{\vec{v};\vec{r}}=\P'{2}}:{\Q}
    \)
    and following~\Cref{fig:prc_reduction}~\Cref{eq:prc_reduction_recursion}:
    \[
      \infer[\Rule{Def}]{
      \PCfg,{\Def{\Rec*{\vec{v};\vec{r}}=\P'{2}}:{\Q}}%
      \Reduce-%
      \PCfg',{\Def{\Rec*{\vec{v};\vec{r}}=\P'{2}}:{\Q'}}%
      }{
      \PCfg,{\Q}%
      \Reduce-%
      \PCfg',{\Q}'%
      }%
    \]

    \noindent where $ \P' = \Def{\Rec*{\vec{v};\vec{r}}=\P'{2}}:{\Q'} $.
    By~\Cref{claim:tc_inversion_def} in~\Cref{lem:tc_inversion}, it follows that:
    \[
      \Gam,{\Rec|};\Timers\Entails{\Q}\TypedBy\Session%
    \]

    \noindent The thesis follows by the induction hypothesis:
    since
    \(
    \Gam,{\Rec|};\Timers\Entails{\Q}\TypedBy\Session%
    \),
    if
    \(
    \PCfg,{\Q}%
    \Reduce-%
    \PCfg',{\Q}'%
    \)
    then $\exists \Session'$ such that $\Session\Reduce*\Session'$ and
    \(
    \Gam,{\Rec|};\Timers'\Entails{\Q'}\TypedBy\Session'%
    \)
    and \Session'\ is \bal\ and \wf.
    %
    \item\label{case:sr_action_step_call}
    If rule \Rule{Call},
    then
    \(
    \P = \Def{\Rec*{\vec{v'};\vec{r'}}=\P'{2}}:{\Rec\mid\Q}
    \)
    and following~\Cref{fig:prc_reduction}~\Cref{eq:prc_reduction_recursion}:
    {{\small\[
      \PCfg,{\Def{\Rec*{\vec{v'};\vec{r'}}=\P'{2}}:{\Rec\mid\Q}}%
      \Reduce-%
      \PCfg,{\Def{\Rec*{\vec{v'};\vec{r'}}=\P'{2}}:{{\P'{2}\Subst{{\vec{v};\vec{r}}}{{\vec{v'};\vec{r'}}}}\mid\Q}}%
      \quad \Rule{Call}
    \]}}

    \noindent where
    \(
    \P' = \Def{\Rec*{\vec{v'};\vec{r'}}=\P'{2}}:{{\P'{2}\Subst{{\vec{v};\vec{r}}}{{\vec{v'};\vec{r'}}}}\mid\Q}
    \)
    and $\Timers'=\Timers$.
    By~\Cref{claim:tc_inversion_def_par,claim:tc_inversion_def} of~\Cref{lem:tc_inversion}:
    $\Timers = (\Timers_1,\Timers_2)$, $\Session=(\Session_1,\Session_2)$ and:
    \begin{equation}\label{eq:sr_action_step_call_1}
      \forall (\vec{\Val},~\vec{\S})
      \in\FormatRecCheckStyle{\Delta},
      {\theta'}\in\FormatRecCheckStyle{\theta}:\quad %
      \Gam,
      {\vec{v}:\vec{T}},%
      {\Rec|}%
      ;\Timers'{2}
      \Entails{\P'{2}}
      \TypedBy
      \vec{r'}: (\vec{\Val},~\vec{\S})%
    \end{equation}
    \begin{equation}\label{eq:sr_action_step_call_2}
      \Gam,
      {\vec{v}:\vec{T}},{\Rec|};\Timers_1\Entails
      \Rec
      \TypedBy\Session_1%
      \qquad
      \Gam,
      {\vec{v}:\vec{T}},{\Rec|};\Timers_2\Entails
      \Q
      \TypedBy\Session_2%
    \end{equation}

    \noindent By~\Cref{claim:tc_inversion_rec} of~\Cref{lem:tc_inversion} on \Cref{eq:sr_action_step_call_2} we have that
    $\Session_1 = \vec {r}:(\vec{\Val},\vec{\S})$ and $(\vec{\Val},\vec{\S})
      \in \FormatRecCheckStyle{\Delta}$.
    We must show $\exists \Session'$ such that $\Session\Reduce*\Session'$ and:
    \[
      \Gam;\Timers'\Entails
      \Def{\Rec*{\vec{v'};\vec{r'}}=\P'{2}}:{{\P'{2}\Subst{{\vec{v};\vec{r}}}{{\vec{v'};\vec{r'}}}}\mid\Q}
      \TypedBy\Session'
    \]

    \noindent By~\Cref{lem:sr_wt_substitution} on~\Cref{eq:sr_action_step_call_1}:
    By~\Cref{lem:sr_wt_substitution} on~\Cref{eq:sr_action_step_call_1}:
    \[
      \forall (\vec{\Val},~\vec{\S})
      \in\FormatRecCheckStyle{\Delta},
      {\theta'}\in\FormatRecCheckStyle{\theta}:\quad %
      \Gam,
      {\vec{v'}:\vec{T}},%
      {\Rec|}%
      ;\Timers'{2}
      \Entails{\P'{2}\Subst{{\vec{v};\vec{r}}}{{\vec{v'};\vec{r'}}}}
      \TypedBy
      \vec{r}: (\vec{\Val},~\vec{\S})%
    \]

    \noindent By rule \Rule{Par}:
    \[
      \infer[\Rule{Par}]{
      \Gam,{\vec{v}:\vec{T}},{\Rec|};\Timers_1,\Timers_2\Entails{{\P'{2}\Subst{{\vec{v};\vec{r}}}{{\vec{v'};\vec{r'}}}\mid\Q}}\TypedBy\Session_1,\Session_2%
      }{
      \begin{array}{c}
        \Gam,{\vec{v}:\vec{T}},{\Rec|};\Timers_1\Entails
        \P'{2}\Subst{{\vec{v};\vec{r}}}{{\vec{v'};\vec{r'}}}
        \TypedBy\Session_1
        \\
        \Gam,{\vec{v}:\vec{T}},{\Rec|};\Timers_2\Entails
        \Q
        \TypedBy\Session_2%
      \end{array}
      }%
    \]

    \noindent By rule \Rule{Rec}:
    \[
      \infer[\Rule{Rec}]{
      \Gam;\Timers\Entails{\Def{\Rec*{\vec{v'};\vec{r'}}=\P'{2}}:{\P'{2}\Subst{{\vec{v};\vec{r}}}{{\vec{v'};\vec{r'}}}\mid\Q}}\TypedBy\Session%
      }{
      \begin{array}[t]{@{}c@{}}%
        \forall (\vec{\Val},~\vec{\S})
        \in\FormatRecCheckStyle{\Delta},
        {\theta'}\in\FormatRecCheckStyle{\theta}: %
        \Gam,
        {\vec{v}:\vec{T}},%
        {\Rec|}%
        ;\Timers'{2}
        \Entails{\P'{2}}
        \TypedBy
        \vec{r}: (\vec{\Val},~\vec{\S})%
        \\[\ArrayLineSpacing]
        \Gam,{\vec{v}:\vec{T}},{\Rec|};\Timers\Entails{{\P'{2}\Subst{{\vec{v};\vec{r}}}{{\vec{v'};\vec{r'}}}\mid\Q}}\TypedBy\Session%
      \end{array}%
      }%
    \]

    \noindent The thesis then holds with $\Session = \Session'$.
  \end{caseanalysis}

  \noindent In the cases above we have shown, given a process \P\ that is \wt\ against a session environment \Session\ that is \bal\ and \wf, any \emph{instantaneous} reduction made by the process yields a process \P'\ that is \wt\ against some session environment \Session'\ reachable from \Session\ using the rules in~\Cref{fig:session_reduction}.
\end{proof}

\subsubsection{Subject Reduction}

\thmSubjectReduction*
\proofSubjectReduction*

\clearpage

\section{Example Derivation: Ping-Pong\texorpdfstring{\ (\Cref{exa:type_check_ping_pong})}{}}\label{sec:appendix:derivation:ping_pong}
%

Recall~\Cref{eq:ex_tc_type,eq:ex_tc_prc} from~\Cref{exa:type_check_ping_pong}:
\begin{equation*}\label{eq:exa:pingpong:type:0}
%
\end{equation*}


\noindent Below is the initial typing judgement:
\begin{equation*}\label{eq:exa:pingpong:0}
  %
\end{equation*}

\noindent The evaluation begins in~\Cref{eq:exa:pingpong:eval:1} on the next page.

\clearpage

%
%
%
%
{{
      \footnotesize
      \begin{equation}\label{eq:exa:pingpong:eval:1}
        \infer[\Rule{Timeout}]
        {%
          %
                                \quad \Rule{End}
                              \end{array}
                            }
                          \end{array}
                        }
                        \\\\
                        {\Val_{6.2}[x\mapsto3.1]} \models (x>3)
                      \end{array}
                    }
                  \end{array}
                }
              \end{array}
            }
            \\\\
            {\Val_{3.1}} \models (x>3)
          \end{array}
        }
      \end{equation}

    }}

\UnloadTypeChecking

\end{document}

%% file: macros/lengths_rules.tex
\newlength{\FigEqHRuleSpacingBefore}
\newlength{\FigEqHRuleSpacingAfter}
\NewDocumentCommand{\FigHRule}{}{\vspace{\FigEqHRuleSpacingBefore}{\color{lightgray}\hrule}\vspace{\FigEqHRuleSpacingAfter}}

\newlength{\ProofReworkGapSpacingBefore}
\newlength{\ProofReworkGapSpacingAfter}
\NewDocumentCommand{\ProofReworkGap}{s}{\vspace{\ProofReworkGapSpacingBefore}{\color{red}{\scriptsize Reworking...}\linebreak\hrule}\vspace{\ProofReworkGapSpacingAfter}\IfBooleanT{#1}{\noindent}}

\newlength{\ArrayLineSpacing}
\newlength{\ArrayTallLineSpacing}
\newlength{\ArrayExtraTallLineSpacing}

\newlength{\ArrayLTSLineSpacing}
\newlength{\ArrayLTSPremiseLineSpacing}
\newlength{\ArrayLTSPremiseTallLineSpacing}
\newlength{\ArrayLTSConclusionLineSpacing}
\newlength{\ArrayLTSConclusionTallLineSpacing}

\newlength{\ArrayRTLWrapping}
\NewDocumentCommand{\ArrWrap}{}{\hspace{\ArrayRTLWrapping}}

\newlength{\ArrayMultiLineSpacing}

%% file: packages.tex
\ifdefined\ifdraft%
    \let\oldifdraft\ifdraft
    \let\ifdraft\relax
    \usepackage{templatetools}
    \let\ifdraft\oldifdraft
\else%
    \usepackage{templatetools}
\fi%

\IfPackageNotLoaded{inputenc}{\usepackage[utf8]{inputenc}}

\IfPackageNotLoaded{fontenc}{\usepackage[T1]{fontenc}}

\IfPackageNotLoaded{microtype}{\usepackage{microtype}}

\IfPackageNotLoaded{babel}{\usepackage[english]{babel}}

\IfPackageNotLoaded{xcolor}{\usepackage[dvipsnames,hyperref]{xcolor}}

\IfPackageNotLoaded{color}{\usepackage{color}}

\IfPackageNotLoaded{graphicx}{\usepackage{graphicx}}

\IfPackageNotLoaded{rotating}{\usepackage{rotating}}

\IfPackageNotLoaded{float}{\usepackage{float}}

\IfPackageNotLoaded{xparse}{\usepackage{xparse}}

\IfPackageNotLoaded{ifthen}{\usepackage{ifthen}}
\IfPackageNotLoaded{ifdraft}{\usepackage{ifdraft}}

\IfPackageNotLoaded{calc}{\usepackage{calc}}

\IfPackageNotLoaded{silence}{\usepackage{silence}}

\IfPackageNotLoaded{etoolbox}{\usepackage{etoolbox}}

\IfPackageNotLoaded{proof}{\usepackage{proof}}

\IfPackageNotLoaded{mathpartir}{\usepackage{mathpartir}}

\IfPackageNotLoaded{mathtools}{\usepackage{mathtools}}
\IfPackageNotLoaded{amsmath}{%
    \let\subseteq\relax%
    \usepackage{amsmath}%
}

\IfPackageNotLoaded{amsthm}{%
    \let\proof\relax%
    \usepackage{amsthm}%
}

\IfPackageNotLoaded{thmtools}{\usepackage{thmtools}}
\IfPackageNotLoaded{thm-restate}{\usepackage{thm-restate}}

\IfPackageNotLoaded{longtable}{\usepackage{longtable}}

\IfPackageNotLoaded{tabularx}{\usepackage{tabularx}}

\IfPackageNotLoaded{enumitem}{\usepackage[inline]{enumitem}}

\IfPackageNotLoaded{algorithm2e}{\usepackage[ruled,vlined,linesnumbered,commentsnumbered,resetcount]{algorithm2e}}

\IfPackageNotLoaded{listings}{\usepackage{listings}}

\IfPackageNotLoaded{wrapfig}{\usepackage{wrapfig}}

\IfPackageNotLoaded{xfrac}{\usepackage{xfrac}}

\IfPackageNotLoaded{delarray}{\usepackage{delarray}}

\IfPackageNotLoaded{tikz}{\usepackage{tikz}}

\IfPackageNotLoaded{multicol}{\usepackage{multicol}}

\IfPackageNotLoaded{multirow}{\usepackage{multirow}}

\IfPackageNotLoaded{caption}{\usepackage{caption}}

\IfPackageNotLoaded{subcaption}{\usepackage{subcaption}}

\IfPackageNotLoaded{amsfonts}{\usepackage{amsfonts}}
\IfPackageNotLoaded{amstext}{\usepackage{amstext}}

\IfPackageNotLoaded{mathabx}{%
    \let\oldnot\not%
    \let\not\relax%
    \let\olddegree\degree%
    \let\degree\relax%
    \usepackage{mathabx}%
    \let\not\oldnot%
    \let\degree\olddegree%
}

\IfPackageNotLoaded{amssymb}{\usepackage{amssymb}}
\let\emptyset\varnothing%

\IfPackageNotLoaded{amsbsy}{\usepackage{amsbsy}}

\IfPackageNotLoaded{stmaryrd}{\usepackage{stmaryrd}}

\IfPackageNotLoaded{AnonymousPro}{\usepackage[ttdefault=true]{AnonymousPro}}

\IfPackageNotLoaded{fontawesome}{%
    \usepackage{fontawesome}
    \let\eyesymbol\faEye
}

\IfPackageNotLoaded{hyperref}{\usepackage[colorlinks=true,final]{hyperref}}

\IfPackageNotLoaded{bookmark}{\usepackage[depth=5,atend]{bookmark}}

\IfPackageNotLoaded{cleveref}{\usepackage[capitalise]{cleveref}}

\IfPackageNotLoaded{tcolorbox}{\usepackage[theorems,xparse]{tcolorbox}}

\IfPackageNotLoaded{soul}{\usepackage{soul}}
\let\st\relax%
\IfPackageNotLoaded{ulem}{\usepackage{ulem}}
\IfPackageNotLoaded{colortbl}{\usepackage{colortbl}}

\IfPackageNotLoaded{doi}{\usepackage{doi}}

%% file: macros/lmcs_qol.tex
\ProvideDocumentEnvironment{thm}{}{}{}
\ProvideDocumentEnvironment{lem}{}{}{}
\ProvideDocumentEnvironment{cor}{}{}{}
\ProvideDocumentEnvironment{prop}{}{}{}
\ProvideDocumentEnvironment{asm}{}{}{}
\ProvideDocumentEnvironment{defi}{}{}{}
\ProvideDocumentEnvironment{conv}{}{}{}
\ProvideDocumentEnvironment{conj}{}{}{}
\ProvideDocumentEnvironment{fact}{}{}{}
\ProvideDocumentEnvironment{algo}{}{}{}
\ProvideDocumentEnvironment{pty}{}{}{}
\ProvideDocumentEnvironment{clm}{}{}{}
\ProvideDocumentEnvironment{nota}{}{}{}
\ProvideDocumentEnvironment{exa}{}{}{}
\ProvideDocumentEnvironment{exas}{}{}{}
\ProvideDocumentEnvironment{rem}{}{}{}
\ProvideDocumentEnvironment{rems}{}{}{}
\ProvideDocumentEnvironment{prob}{}{}{}
\ProvideDocumentEnvironment{probs}{}{}{}
\ProvideDocumentEnvironment{oprob}{}{}{}
\ProvideDocumentEnvironment{oprobs}{}{}{}
\ProvideDocumentEnvironment{obs}{}{}{}
\ProvideDocumentEnvironment{obss}{}{}{}
\ProvideDocumentEnvironment{qu}{}{}{}
\ProvideDocumentEnvironment{qus}{}{}{}

\numberwithin{equation}{section}
\numberwithin{exa}{section}
\numberwithin{defi}{section}
\numberwithin{lem}{section}
\numberwithin{prop}{section}
\numberwithin{thm}{section}
\numberwithin{cor}{section}
\numberwithin{rem}{section}

\Crefname{thm}{Theorem}{Theorems}
\Crefname{lem}{Lemma}{Lemmas}
\Crefname{cor}{Corollary}{Corollaries}
\Crefname{prop}{Proposition}{Propositions}
\Crefname{asm}{Assumption}{Assumptions}
\Crefname{defi}{Definition}{Definitions}
\Crefname{conv}{Convention}{Conventions}
\Crefname{conj}{Conjecture}{Conjectures}
\Crefname{fact}{Fact}{Factos}
\Crefname{algo}{Algorithm}{Algorithms}
\Crefname{pty}{Property}{Properties}
\Crefname{clm}{Claim}{Claims}
\Crefname{nota}{Notation}{Notations}
\Crefname{exa}{Example}{Examples}
\Crefname{exas}{List of Examples}{Lists of Examples}
\Crefname{rem}{Remark}{Remarks}
\Crefname{rems}{List of Remarks}{Lists of Remarks}
\Crefname{prob}{Problem}{Problems}
\Crefname{probs}{List of Problems}{Lists of Problems}
\Crefname{oprob}{Open Problem}{Open Problems}
\Crefname{oprobs}{List of Open Problems}{Lists of Open Problems}
\Crefname{obs}{Observation}{Obserations}
\Crefname{obss}{List of Observations}{Lists of Observations}
\Crefname{qu}{Question}{Questions}
\Crefname{qus}{List of Questions}{Lists of Questions}

%% file: macros/notes.tex
\let\note\relax

\let\oldmarginpar\marginpar
\renewcommand{\marginpar}[1]{%
    \leavevmode%
    \oldmarginpar{#1}%
    \ignorespacesafterend\ignorespaces%
}

\newcounter{notecounter}
\NewDocumentCommand{\NoteContentsLine}{O{black} m m}{%
    \phantomsection%
    \addcontentsline{toc}{paragraph}{%
        \texorpdfstring{%
            \color{#1}\llap{(\thenotecounter)\ $\Rightarrow$\ }%
            \parbox[t]{0.8\linewidth}{#3: #2}%
        }{#3}%
    }%
}

\definecolor{hlbackground}{rgb}{1.0, 0.7, 0.6}
\definecolor{hlbackgroundwhite}{rgb}{1,1,1}

\NewDocumentCommand{\notemark}{s O{blue}}{%
    \IfBooleanF{#1}{\refstepcounter{notecounter}}%
    \ifvmode
        \smash{\clap{\raisebox{4pt}{%
        \tikz[baseline=(noteID.base), font=\tiny, scale=1]{%
        \node[inner sep=0pt,outer sep=0pt,color=#2] (noteID) {(\ref{mn:\thenotecounter})};
        \node[fill=white,fit=(noteID),rounded corners=0.5pt,inner sep=0pt] {};%
        \node[inner sep=0pt,outer sep=0pt,color=#2] at (noteID) {(\ref{mn:\thenotecounter})};
            }%
        }}}%
    \else
        \smash{\clap{\raisebox{7.5pt}{%
        \tikz[baseline=(noteID.base), font=\tiny, scale=1]{%
        \node[inner sep=0pt,outer sep=0pt,color=#2] (noteID) {(\ref{mn:\thenotecounter})};
        \node[fill=white,fit=(noteID),rounded corners=0.5pt,inner sep=0pt] {};%
        \node[inner sep=0pt,outer sep=0pt,color=#2] at (noteID) {(\ref{mn:\thenotecounter})};
            }%
        }}}%
    \fi%
}

\newlength{\jeroenlen}

\NewDocumentCommand{\notelist}{s o o G{\dots} O{black}}{%
    \settowidth{\jeroenlen}{\smash{\llap{\WrapMathMode{^{(\thenotecounter)}}\IfValueT{#2}{\bfseries #2:}}}}%
    \ifshownotes{{%
        \refstepcounter{notecounter}%
        \highlight[#5]{%
            \begin{description}[leftmargin=\jeroenlen,labelwidth=0pt,labelsep=0pt]%
                \item[\smash{\llap{\WrapMathMode{^{(\thenotecounter)}}\IfValueT{#2}{\bfseries #2:\IfValueT{#3}{~}}}}]%
                \IfValueT{#3}{\smash{\rlap{#3}}}%
                \IfBooleanTF{#1}{
                    \begin{itemize}[leftmargin=1.5em,labelsep=.5em]#4\end{itemize}%
                }{
                    \begin{enumerate}[leftmargin=1.5em,labelsep=.5em]#4\end{enumerate}%
                }%
            \end{description}%
        }%
        \ifshowAppendixTOC{{%
            \NoteContentsLine[#5]{#2: #3}{note list}%
        }}\else\fi%
    }}\else\fi%
}

\NewDocumentCommand{\note}{s t| O{black} G{\dots} o t+}{%
    \ifshownotes{{%
        \IfBooleanTF{#1}{
            \IfBooleanTF{#2}{
                \hspace*{\fill}\par\noindent%
                \highlight[#3]{%
                    \IfValueT{#5}{\llap{\bfseries #5:}~}{#4}%
                }%
                \hspace*{\fill}\par%
            }{%
                \hfill%
                \highlight[#3]{%
                    \IfValueT{#5}{{\bfseries #5:}~}{#4}%
                }%
            }%
        }{
            \IfBooleanTF{#6}{%
                \marginpar{%
                    \highlight[#3]{%
                        \scriptsize%
                        \smash{\llap{\WrapMathMode{^{\hyperlink{toc}{(\thenotecounter)}}}}}%
                        \IfValueT{#5}{{\bfseries #5:}~}%
                        {#4}\hspace*{\fill}%
                    }%
                }\label{mn:\thenotecounter}%
            }{%
                \refstepcounter{notecounter}%
                \marginpar{%
                    \highlight[#3]{%
                        \scriptsize%
                        \smash{\llap{\WrapMathMode{^{\hyperlink{toc}{(\thenotecounter)}}}}}%
                        \IfValueT{#5}{{\bfseries #5:}~}%
                        {#4}\hspace*{\fill}%
                    }%
                }\label{mn:\thenotecounter}%
            }%
        }%
   }}\else\fi%
}

\makeatletter
\NewDocumentCommand{\NewAuthorNote}{s O{blue} m m o g}{%
    \IfValueT{#6}{%
        \NewDocumentCommand{#6}{s o G{\dots} O{#2}}{%
            \begingroup%
            \setcounter{tocdepth}{5}%
            \IfBooleanTF{##1}{%
                \notelist*[\IfValueTF{#5}{\highlight[#5]{#4}}{\highlight[#2]{#4}}][##2]{##3}[##4]%
            }{%
                \notelist[\IfValueTF{#5}{\highlight[#5]{#4}}{\highlight[#2]{#4}}][##2]{##3}[##4]%
            }%
            \endgroup%
        }%
    }%
    \NewDocumentCommand{#3}{s t| G{\dots} g O{#2} e| O{hlbackground} t+}{%
        \ifshownotes{{\begingroup%
        \setcounter{tocdepth}{5}%
        \IfBooleanTF{##1}{%
            \IfBooleanF{##8}{\refstepcounter{notecounter}}%
            \IfBooleanTF{##2}{%
                \IfBooleanTF{##8}{%
                    \note*|[##5]{##3\IfValueT{##4}{\ \allowbreak##4}}[%
                        \smash{\llap{\hyperlink{toc}{\WrapMathMode{^{(\thenotecounter)}}}}}%
                        \IfBooleanTF{#1}{
                            \IfValueTF{#5}{\highlight[#5]{#4}}{\highlight[#2]{#4}}%
                        }{{#4}}%
                    ]+%
                }{%
                    \note*|[##5]{##3\IfValueT{##4}{\ \allowbreak##4}}[%
                        \smash{\llap{\hyperlink{toc}{\WrapMathMode{^{(\thenotecounter)}}}}}%
                        \IfBooleanTF{#1}{
                            \IfValueTF{#5}{\highlight[#5]{#4}}{\highlight[#2]{#4}}%
                        }{{#4}}%
                    ]%
                }%
            }{%
                \IfBooleanTF{##8}{%
                    \note*[##5]{##3\IfValueT{##4}{\ \allowbreak##4}}[%
                        \smash{\llap{\hyperlink{toc}{\WrapMathMode{^{(\thenotecounter)}}}}}%
                        \IfBooleanTF{#1}{
                            \IfValueTF{#5}{\highlight[#5]{#4}}{\highlight[#2]{#4}}%
                        }{{#4}}%
                    ]+%
                }{%
                    \note*[##5]{##3\IfValueT{##4}{\ \allowbreak##4}}[%
                        \smash{\llap{\hyperlink{toc}{\WrapMathMode{^{(\thenotecounter)}}}}}%
                        \IfBooleanTF{#1}{
                            \IfValueTF{#5}{\highlight[#5]{#4}}{\highlight[#2]{#4}}%
                        }{{#4}}%
                    ]%
                }%
            }%
        }{%
            \IfBooleanTF{##8}{%
                \note[##5]{##3\IfValueT{##4}{\ \allowbreak##4}}[%
                    \IfBooleanTF{#1}{
                        \IfValueTF{#5}{\highlight[#5]{#4}}{\highlight[#2]{#4}}%
                    }{{#4}}%
                ]+%
            }{%
                \note[##5]{##3\IfValueT{##4}{\ \allowbreak##4}}[%
                    \IfBooleanTF{#1}{
                        \IfValueTF{#5}{\highlight[#5]{#4}}{\highlight[#2]{#4}}%
                    }{{#4}}%
                ]%
            }%
        }%
        \ifshowAppendixTOC{{%
            \NoteContentsLine[##5]{##3\IfValueT{##4}{~[\dots]}}{\texorpdfstring{%
                \IfBooleanTF{#1}{
                    \IfValueTF{#5}{\highlight[#5]{#4}}{\highlight[#2]{#4}}%
                }{{#4}}}{#4}%
            }%
        }}\else\fi%
        \IfValueT{##6}{%
            \sethlcolor{##7}%
            \hl{##6}%
        }%
        \IfBooleanF{##1}{\IfBooleanF{##8}{\notemark*}}%
        \endgroup}}\else\fi%
    }%
}
\makeatother

%% file: macros/global.tex
\let\le\relax%
\let\ne\relax%

\let\ae\relax

\let\xrightarrow\relax%

\newcommand{\tcolorChanged}[0]{green!10}

\let\oldinfer\infer%
\RenewDocumentCommand{\infer}{o m g}{%
    \IfValueTF{#1}{
        \oldinfer[{{#1}}]{#2}{\IfValueTF{#3}{#3}{}}%
    }{
        \oldinfer{#2}{\IfValueTF{#3}{#3}{}}%
    }%
}

\newcolumntype{M}[1]{>{\centering\arraybackslash}m{#1}}

\newcounter{proof}
\AtBeginEnvironment{lemma}{\setcounter{proof}{\thelemma}}

\crefalias{proof}{lemma}

\hypersetup{
    linktocpage=false,
    colorlinks=true,
    bookmarksopen=false,
    linkcolor=blue,
    anchorcolor=red,
    citecolor=blue,
    filecolor=magenta,
    urlcolor=blue,
    citebordercolor={0 0 0},
    filebordercolor={1 0 1},
    linkbordercolor={0 0 0},
    menubordercolor={0 0 0},
    urlbordercolor={0 0 0},
    pdfstartpage=1
}

\usetikzlibrary{shapes,trees,tikzmark,arrows,automata,backgrounds,petri,positioning,fit,calc,decorations.markings,shadows,fadings,patterns,decorations.pathreplacing}

\newcounter{tmkcount}
\tikzset{
    use tikzmark/.style={
            remember picture,
            overlay,
            execute at end picture={
                    \stepcounter{tmkcount}
                },
        },
    tikzmark suffix={-\thetmkcount}
}

\SetLabelAlign{parright}{\parbox[t]{\labelwidth}{\raggedleft#1}}

\setlist[description]{itemsep=0.5ex,labelsep=2ex,parsep=0ex,listparindent=0ex}%

\setlist[itemize]{topsep=0.5pt,itemsep=-2pt,parsep=4pt}%
\setlist[itemize,2]{topsep=5pt}%

\setlist[enumerate]{topsep=1pt,itemsep=1.5pt,parsep=4pt}%

\definecolor{codegreen}{rgb}{0,0.6,0}
\definecolor{codegray}{rgb}{0.5,0.5,0.5}
\definecolor{codepurple}{rgb}{0.58,0,0.82}

\definecolor{codered}{RGB}{236, 68, 109}
\definecolor{codeyellow}{RGB}{242, 167, 64}
\definecolor{codeorange}{RGB}{219, 133, 29}
\definecolor{codegrey}{RGB}{79, 95, 95}

\definecolor{codeblue}{RGB}{27, 118, 204}

\lstdefinestyle{Erlang}{
    language=Erlang,
    rangeprefix=\%\%\ ,
    rangesuffix=\ \%\%,
    includerangemarker=false,
    deletekeywords={error,?,ok,!,->,self,spawn,{,},[,],(,),exit},
    morekeywords={fun,?,!},
    keywordstyle=\color{codered},
    classoffset=1,
    morekeywords={ok,ko,start,data,stop,a,b,label,normal,timeout,x5,nb_p,x_geq_5,timer_x5,timer_x0,done},
    keywordstyle=\color{codeorange},
    classoffset=2,
    morekeywords={spawn,self,MODULE,NextState,OtherState,loop,main,exit,loop_s1,get_data,SomeData,MONITORED,MONITOR_SPEC,s0,s1,s2,s3,s4,s5,s6,s7,s8,s9,s10,S,get_payload_b,get_payload_a,get_timer,set_timer,stopping,save_msg,get_msg,get_msgs,D,E,send_before,is_integer,is_atom,is_pid,nonblocking_payload,nonblocking_selection,read_timer,exit,start_timer,error,perform_action,P1,P2,P0,P3,Q0,Q1,Q2,Q3,sleep,uniform_real,foo,bar,handle_timeout},
    keywordstyle=\color{codegreen},
    classoffset=3,
    morekeywords={{,},(,),[,]},
    keywordstyle=\color{codegrey},
    classoffset=4,
    morekeywords={stub,erlang,interleave,timer,rand},
    keywordstyle=\color{codeblue},
    classoffset=5,
    morekeywords={->},
    keywordstyle=\bfseries\color{codegrey},
    classoffset=0
}

\NewDocumentCommand{\erl}{m}{\lstinline[style=Erlang,breaklines=false,escapechar=£]{#1}}

%% file: macros/symbols.tex
\NewDocumentCommand{\Mon}{G{\PrcCalcP} e_ t' G{1}}{\WrapMathMode{%
        \text{\eyesymbol}
        (\IfValueTF{#2}{%
            {#1}\IfBooleanT{#3}{\RepeatChar{#4}{'}}_{#2}%
        }{%
            {#1}\IfBooleanT{#3}{\RepeatChar{#4}{'}}%
        })%
    }
}

\NewDocumentCommand{\MonOLD}{G{\PrcCalcP} e_ t' G{1}}{\WrapMathMode{%
        \IfValueTF{#2}{%
            \overset{\scalebox{0.5}{\text{\eyesymbol}}}{#1}
            \IfBooleanT{#3}{\RepeatChar{#4}{'}}_{#2}%
        }{%
            \overset{\scalebox{0.5}{\text{\eyesymbol}}}{#1}
            \IfBooleanT{#3}{\RepeatChar{#4}{'}}%
        }%
    }
}

\DeclareMathOperator{\Entails}{\vdash}
\DeclareMathOperator{\TypedBy}{\blacktriangleright}
\DeclareMathOperator{\Compat}{\bot}

\let\Diam\diamond

\DeclareMathOperator{\MathOpBimplies}{\Longleftrightarrow}
\NewDocumentCommand{\Bimplies}{}{\WrapMathMode{\MathOpBimplies}}

\DeclareMathOperator{\MathOpTrue}{\mathtt{true}}
\NewDocumentCommand{\True}{}{\WrapMathMode{\MathOpTrue}}

\DeclareMathOperator{\MathOpFalse}{\mathtt{false}}
\NewDocumentCommand{\False}{}{\WrapMathMode{\MathOpFalse}}

\DeclareMathOperator{\UTU}{\overset{\scriptscriptstyle\mathtt{unfold}}{\equiv}}
\NewDocumentCommand{\utu}{}{\WrapMathMode{\UTU}}

\DeclareMathOperator{\NUTU}{\overset{\scriptscriptstyle\mathtt{unfold}}{\phantom{\equiv}\mathllap{\smash{\not\equiv}}}}

\DeclareMathOperator{\MathOpSetNatural}{\mathbb{N}}
\NewDocumentCommand{\SetNatural}{}{\WrapMathMode{\MathOpSetNatural}}

\DeclareMathOperator{\MathOpSetRealZ}{\mathbb{R}_{\geq 0}}
\NewDocumentCommand{\SetRealZ}{}{\WrapMathMode{\MathOpSetRealZ}}

\DeclareMathOperator{\MathOpSetOfClocks}{\mathbb{X}}
\NewDocumentCommand{\SetOfClocks}{}{\WrapMathMode{\MathOpSetOfClocks}}

\DeclareMathOperator{\MathOpSetOfTimers}{\mathbb{T}}
\NewDocumentCommand{\SetOfTimers}{}{\WrapMathMode{\MathOpSetOfTimers}}

\NewDocumentCommand{\RecTypeEnv}{e_ t' G{1}}{\WrapMathMode{%
\IfValueTF{#1}{%
    {A}\IfBooleanT{#2}{\RepeatChar{#3}{'}}_{#1}%
}{%
    {A}\IfBooleanT{#2}{\RepeatChar{#3}{'}}%
}%
}}
\NewDocumentCommand{\FutureEnv}{e_ t' G{1}}{\WrapMathMode{%
\IfValueTF{#1}{%
    {\gamma}\IfBooleanT{#2}{\RepeatChar{#3}{'}}_{#1}%
}{%
    {\gamma}\IfBooleanT{#2}{\RepeatChar{#3}{'}}%
}%
}}

\NewDocumentCommand{\Comm}{}{\WrapMathMode{\mbox{\tiny{$\square$}}}}

\NewDocumentCommand{\TypeSend}{}{\WrapMathMode{!}}
\NewDocumentCommand{\TypeRecv}{}{\WrapMathMode{?}}

\DeclareMathOperator{\MathOpPrcSend}{\vartriangleleft}
\NewDocumentCommand{\PrcSend}{}{\WrapMathMode{\MathOpPrcSend}}

\DeclareMathOperator{\MathOpPrcRecv}{\vartriangleright}
\NewDocumentCommand{\PrcRecv}{}{\WrapMathMode{\MathOpPrcRecv}}

%% file: macros/plaintext.tex
\NewDocumentCommand{\RefTypeCheck}{}{\Cref{fig:type_checking,fig:type_checking_communication}}

\def\st{\WrapMathMode{\text{s.t.}}}

\NewDocumentCommand{\wt}{s}{\WrapMathMode*{\IfBooleanTF{#1}{\texttt{wt}}{\emph{well-typed}}}}

\NewDocumentCommand{\mc}{s}{\WrapMathMode*{\IfBooleanTF{#1}{\emph{mixed-choice}}{mixed-choice}}}

\NewDocumentCommand{\tc}{s}{\WrapMathMode*{\IfBooleanTF{#1}{\texttt{tc}}{\emph{type checking}}}}
\NewDocumentCommand{\wf}{s}{\WrapMathMode*{\IfBooleanTF{#1}{\texttt{wf}}{\emph{well-formed}}}}
\NewDocumentCommand{\wfness}{}{\WrapMathMode*{\emph{well-formedness}}}
\NewDocumentCommand{\wtness}{}{\WrapMathMode*{\emph{well-typedness}}}
\NewDocumentCommand{\live}{}{\WrapMathMode*{\emph{live}}}
\NewDocumentCommand{\compat}{}{\WrapMathMode*{\emph{compatible}}}

\NewDocumentCommand{\bal}{}{\WrapMathMode*{\emph{balanced}}}
\NewDocumentCommand{\fbal}{}{\WrapMathMode*{\emph{fully-balanced}}}
\NewDocumentCommand{\balness}{}{\WrapMathMode*{\emph{balancedness}}}
\NewDocumentCommand{\fbalness}{}{\WrapMathMode*{\emph{fully-balancedness}}}
\NewDocumentCommand{\del}{}{\WrapMathMode*{\emph{delayable}}}
\NewDocumentCommand{\dual}{}{\WrapMathMode*{\emph{dual}}}

\NewDocumentCommand{\bt}{}{\WrapMathMode*{\emph{base-type}}}

\NewDocumentCommand{\IFthenelse}{}{\WrapMathMode{\mathtt{if}\text{-}\mathtt{then}\text{-}\mathtt{else}}}

\NewDocumentCommand{\recvafter}{}{\WrapMathMode{\mathtt{receive}\text{-}\mathtt{after}}}

\NewDocumentCommand{\Bal}{}{\WrapMathMode{\text{Bal}}}
\NewDocumentCommand{\ST}{}{\WrapMathMode{\text{s.t.\ }}}

\NewDocumentCommand{\fe}{s t= t| t! t? t'}{%
    \IfBooleanTF{#2}{
        \IfBooleanTF{#4}{
            \WrapMathMode{\stackrel{!}{\IfBooleanT{#3}{\not}\Rightarrow}}%
        }{\IfBooleanTF{#5}{
            \WrapMathMode{\stackrel{?}{\IfBooleanT{#3}{\not}\Rightarrow}}%
        }{
            \WrapMathMode{\stackrel{\Comm\IfBooleanT{#6}{'}}{\IfBooleanT{#3}{\not}\Rightarrow}}%
        }}
    }{
        \IfBooleanTF{#1}{\WrapMathMode{\mathtt{fe}}%
        }{\WrapMathMode*{\emph{future-enabled}}}%
    }%
}

\NewDocumentCommand{\le}{s}{\IfBooleanTF{#1}{\WrapMathMode{\mathtt{le}}}{\WrapMathMode*{\emph{latest-enabled}}}}

\NewDocumentCommand{\ne}{s}{\IfBooleanTF{#1}{\WrapMathMode{\mathtt{ne}}}{\WrapMathMode*{\emph{next-enabled}}}}
\NewDocumentCommand{\ae}{s}{\IfBooleanTF{#1}{\WrapMathMode{\mathtt{ae}}}{\WrapMathMode*{\emph{after-enabled}}}}

\newcommand*\circled[1]{\tikz[baseline=(char.base)]{
    \node[shape=circle,draw,inner sep=0.7pt] (char) {\WrapMathMode{#1}};}}

\NewDocumentCommand{\FormatRecVarStyle}{m g}{\WrapMathMode{\IfValueTF{#2}{\mathbf{\uppercase{\scriptstyle #1}}_{{#2}}}{\mathbf{\uppercase{\scriptstyle #1}}}}}

\NewDocumentCommand{\FormatRecCheckStyle}{m g}{\WrapMathMode{\IfValueTF{#2}{\boldsymbol{#1}_{{#2}}}{\boldsymbol{#1}}}}

\NewDocumentCommand{\Rule}{s o m o t*}{\WrapMathMode{%
        \IfBooleanTF{#1}{\left\lfloor}{\left[}
        {%
            \IfValueT{#2}{{\IfBooleanTF{#5}{\mathtt{#2}}{#2}}\text{-}}
            \mathtt{{#3}}
            \IfValueT{#4}{\text{-}{\IfBooleanTF{#5}{\mathtt{#4}}{#4}}}
        }%
        \IfBooleanTF{#1}{\right\rfloor}{\right]}
    }%
}

\NewDocumentCommand{\Dual}{s m e_ t' G{1}}{\WrapMathMode{%
        \IfBooleanTF{#1}{%
        {%
            \hphantom{%
                \hspace{2.5pt}
                \IfValueTF{#3}{%
                    {#2}\IfBooleanT{#4}{\RepeatChar{#5}{'}}_{#3}%
                }{%
                    {#2}\IfBooleanT{#4}{\RepeatChar{#5}{'}}%
                }%
                \hspace{2.5pt}
            }%
            \mathllap{%
                \IfValueTF{#3}{%
                    \overline{{#2}\IfBooleanT{#4}{\RepeatChar{#5}{'}}_{#3}}%
                }{%
                    \overline{{#2}\IfBooleanT{#4}{\RepeatChar{#5}{'}}}%
                }%
                \hphantom{\hspace{2.5pt}}
            }%
        }%
        }{%
            \IfValueTF{#3}{%
                \overline{{#2}\IfBooleanT{#4}{\RepeatChar{#5}{'}}_{#3}{\mathstrut}}%
            }{%
                \overline{{#2}\IfBooleanT{#4}{\RepeatChar{#5}{'}}{\mathstrut}}%
            }%
        }%
    }%
}

\NewDocumentCommand{\Past}{s G{\GeneralTheoryConstraints} e_ t' G{1} tV e_}{\WrapMathMode{%
        \IfValueTF{#3}{%
            \downarrow{%
                \IfBooleanT{#6}{%
                    \IfBooleanTF{#1}{%
                        \bigvee_{\IfValueTF{#7}{#7}{i \in I}}%
                    }{%
                        \bigvee\limits_{\IfValueTF{#7}{#7}{i \in I}}%
                    }%
                }#2%
                \IfBooleanT{#4}{\RepeatChar{#5}{'}}_{#3}%
            }%
        }{%
            \downarrow{%
                \IfBooleanT{#6}{%
                    \bigvee\limits_{\IfValueTF{#7}{#7}{i \in I}}%
                }#2%
                \IfBooleanT{#4}{\RepeatChar{#5}{'}}%
            }%
        }%
    }%
}

%% file: macros/environments.tex
\setlist[itemize]{leftmargin=4ex,itemsep=1ex}
\setlist[enumerate]{itemsep=0.75ex}

\newlist{claims}{enumerate}{1}
\setlist[claims,1]{label={\itshape{Claim~\arabic*.}},ref={\arabic*},topsep=1ex,itemsep=1ex,itemindent=0pt,leftmargin=0pt,left=0pt,wide,labelindent=0pt,parsep=\parskip}

\newlist{caseanalysis}{enumerate}{4}

\setlist[caseanalysis,1]{label={\itshape{Case~\arabic*.}},ref={\arabic*},topsep=1ex,itemsep=1ex,parsep=0.5ex,leftmargin=5ex,itemindent=4ex,left=0ex,listparindent=4ex}

\setlist[caseanalysis,2]{label={\itshape{\roman*}.},ref={\arabic{caseanalysisi}.\roman*},topsep=0.5ex,itemsep=0.75ex,parsep=0.5ex,leftmargin=4ex,itemindent=0ex,left=-2ex,listparindent=4ex}

\setlist[caseanalysis,3]{label={\itshape{\alph*}.},ref={\arabic{caseanalysisi}.\roman{caseanalysisii}.\alph*},topsep=0.5ex,itemsep=0.75ex,parsep=0.5ex,leftmargin=4ex,itemindent=0ex,left=-2ex,listparindent=4ex}

\setlist[caseanalysis,4]{label={\Roman*.},ref={\arabic{caseanalysisi}.\roman{caseanalysisii}.\alph{caseanalysisiii}.\Roman*},topsep=0.5ex,itemsep=0.75ex,parsep=0.5ex,leftmargin=4ex,itemindent=0ex,left=-2ex,listparindent=4ex}

%% file: macros/environment_tweaks/proof.tex
\IfPackageNotLoaded{xpatch}{\usepackage{xpatch}}

\makeatletter
\xpatchcmd{\proof}{\topsep6\p@\@plus6\p@\relax}{}{}{}
\makeatother

%% file: macros/environment_tweaks/exa.tex
\let\exaCopy\exa%
\let\endexaCopy\endexa%

\NewCommandCopy{\proofqedsymbol}{\qedsymbol}
\newcommand{\exampleqedsymbol}{$\triangle$}

\NewDocumentCommand{\xqed}{s O{$\triangle$} O{1ex}}{%
    \IfBooleanTF{#1}{%
        \leavevmode\unskip\penalty9999 \hbox{}\nobreak\hfill\quad%
        \smash{%
            \raisebox{
                -#3%
            }{#2}}%
    }{%
        \leavevmode\unskip\penalty9999 \hbox{}\nobreak\hfill\quad\hbox{#2}%
    }%
}

\NewDocumentCommand{\ManualExaQED}{s}{%
    \IfBooleanTF{#1}{
        \unskip\ignorespacesafterend\xqed*[.]
    }{\unskip\ignorespacesafterend\xqed[\exampleqedsymbol]}
}

\NewDocumentEnvironment{exaDenoted}{o}%
{%
\IfValueTF{#1}{\exaCopy[#1]}{\exaCopy}%
}%
{%
\xqed[\exampleqedsymbol]%
\endexaCopy%
}%

\let\exa\exaDenoted
\let\endexa\endexaDenoted

%% file: macros/environment_tweaks/restatable_equation.tex
\NewDocumentCommand{\NewRestateEnvironment}{s m m}{
    \IfBooleanTF{#1}{%
        \NewDocumentEnvironment{#2}{s m o +b}{%
            \IfValueTF{#3}{%
                \IfBooleanTF{##1}{
                    \begin{restatable*}[##3]{#3}{##2}
                    ##4
                    \end{restatable*}%
                }{
                    \begin{restatable}[##3]{#3}{##2}
                    ##4
                    \end{restatable}%
                }
            }{
                \IfBooleanTF{##1}{
                    \begin{restatable*}{#3}{##2}
                    ##4
                    \end{restatable*}%
                }{
                    \begin{restatable}{#3}{##2}
                    ##4
                    \end{restatable}%
                }
            }
        }{}
    }{%
        \NewDocumentEnvironment{#2}{s m +b}{%
            \IfBooleanTF{##1}{
                \begin{restatable*}{donothing}{##2}
                \begin{#3}
                ##3
                \end{#3}
                \end{restatable*}%
            }{
                \begin{restatable}{donothing}{##2}
                \begin{#3}
                ##3
                \end{#3}
                \end{restatable}%
            }
        }{\unskip\ignorespacesafterend}
    }%
}

\NewRestateEnvironment{requation}{equation}
\NewRestateEnvironment{rtypecheck}{typecheck}
\NewRestateEnvironment{rstype}{stype}
\NewRestateEnvironment{rfigure}{figure}

\NewRestateEnvironment{rproof}{proof}

\NewRestateEnvironment*{rdefi}{defi}

%% file: macros/cref_config.tex
\crefname{claimsi}{claim}{claims}
\crefformat{claimsi}{claim~#2#1#3}
\crefmultiformat{claimsi}{claims~#2#1#3}{ and~#2#1#3}{, #2#1#3}{ and~#2#1#3}
\crefrangeformat{claimsi}{claims~#3#1#4 to~#5#2#6}

\Crefname{claimsi}{Claim}{Claims}
\Crefformat{claimsi}{Claim~#2#1#3}
\Crefmultiformat{claimsi}{Claims~#2#1#3}{ and~#2#1#3}{, #2#1#3}{ and~#2#1#3}
\Crefrangeformat{claimsi}{Claims~#3#1#4 to~#5#2#6}

\crefname{pty}{property}{properties}
\crefformat{pty}{property~#2#1#3}
\crefmultiformat{pty}{properties~#2#1#3}{ and~#2#1#3}{, #2#1#3}{ and~#2#1#3}
\crefrangeformat{pty}{properties~#3#1#4 to~#5#2#6}

\Crefname{pty}{Property}{Properties}
\Crefformat{pty}{Property~#2#1#3}
\Crefmultiformat{pty}{Properties~#2#1#3}{ and~#2#1#3}{, #2#1#3}{ and~#2#1#3}
\Crefrangeformat{pty}{Properties~#3#1#4 to~#5#2#6}

\crefname{caseanalysisi}{case}{cases}
\crefformat{caseanalysisi}{case~#2#1#3}
\crefmultiformat{caseanalysisi}{cases~#2#1#3}{ and~#2#1#3}{, #2#1#3}{ and~#2#1#3}
\crefrangeformat{caseanalysisi}{cases~#3#1#4 to~#5#2#6}

\Crefname{caseanalysisi}{Case}{Cases}
\Crefformat{caseanalysisi}{Case~#2#1#3}
\Crefmultiformat{caseanalysisi}{Cases~#2#1#3}{ and~#2#1#3}{, #2#1#3}{ and~#2#1#3}
\Crefrangeformat{caseanalysisi}{Cases~#3#1#4 to~#5#2#6}

\crefname{caseanalysisii}{case}{cases}
\crefformat{caseanalysisii}{case~#2#1#3}
\crefmultiformat{caseanalysisii}{cases~#2#1#3}{ and~#2#1#3}{, #2#1#3}{ and~#2#1#3}
\crefrangeformat{caseanalysisii}{cases~#3#1#4 to~#5#2#6}

\Crefname{caseanalysisii}{Case}{Cases}
\Crefformat{caseanalysisii}{Case~#2#1#3}
\Crefmultiformat{caseanalysisii}{Cases~#2#1#3}{ and~#2#1#3}{, #2#1#3}{ and~#2#1#3}
\Crefrangeformat{caseanalysisii}{Cases~#3#1#4 to~#5#2#6}

\crefname{caseanalysisiii}{case}{cases}
\crefformat{caseanalysisiii}{case~#2#1#3}
\crefmultiformat{caseanalysisiii}{cases~#2#1#3}{ and~#2#1#3}{, #2#1#3}{ and~#2#1#3}
\crefrangeformat{caseanalysisiii}{cases~#3#1#4 to~#5#2#6}

\Crefname{caseanalysisiii}{Case}{Cases}
\Crefformat{caseanalysisiii}{Case~#2#1#3}
\Crefmultiformat{caseanalysisiii}{Cases~#2#1#3}{ and~#2#1#3}{, #2#1#3}{ and~#2#1#3}
\Crefrangeformat{caseanalysisiii}{Cases~#3#1#4 to~#5#2#6}

\crefname{caseanalysisiv}{case}{cases}
\crefformat{caseanalysisiv}{case~#2#1#3}
\crefmultiformat{caseanalysisiv}{cases~#2#1#3}{ and~#2#1#3}{, #2#1#3}{ and~#2#1#3}
\crefrangeformat{caseanalysisiv}{cases~#3#1#4 to~#5#2#6}

\Crefname{caseanalysisiv}{Case}{Cases}
\Crefformat{caseanalysisiv}{Case~#2#1#3}
\Crefmultiformat{caseanalysisiv}{Cases~#2#1#3}{ and~#2#1#3}{, #2#1#3}{ and~#2#1#3}
\Crefrangeformat{caseanalysisiv}{Cases~#3#1#4 to~#5#2#6}

\Crefname{equation}{Eq}{Eqs}
\Crefformat{equation}{Eq~(#2#1#3)}
\Crefmultiformat{equation}{Eqs~(#2#1#3)}{ and~(#2#1#3)}{, #2#1#3}{ and~#2#1#3)}
\Crefrangeformat{equation}{Eqs~(#3#1#4 to~#5#2#6)}

\Crefname{lemma}{Lemma}{Lemmas}
\Crefformat{lemma}{Lemma~#2#1#3}
\Crefmultiformat{lemma}{Lemmas~#2#1#3}{ and~#2#1#3}{, #2#1#3}{ and~#2#1#3}
\Crefrangeformat{lemma}{Lemmas~#3#1#4 to~#5#2#6}

\Crefname{itm}{}{}
\Crefformat{itm}{(#2#1#3)}
\Crefmultiformat{itm}{(#2#1#3)}{ and~(#2#1#3)}{, #2#1#3}{ and~(#2#1#3)}
\Crefrangeformat{itm}{(#3#1#4 to~#5#2#6)}

\Crefname{item}{}{}
\Crefformat{item}{(#2#1#3)}
\Crefmultiformat{item}{(#2#1#3)}{ and~(#2#1#3)}{, #2#1#3}{ and~(#2#1#3)}
\Crefrangeformat{item}{(#3#1#4 to~#5#2#6)}

\Crefname{cond}{Condition}{Conditions}
\Crefformat{cond}{Condition~(#2#1#3)}
\Crefmultiformat{cond}{Conditions~(#2#1#3)}{ and~(#2#1#3)}{, #2#1#3}{ and~(#2#1#3)}
\Crefrangeformat{cond}{Conditions~(#3#1#4 to~#5#2#6)}

\Crefname{case}{Case}{Cases}
\Crefformat{case}{Case~(#2#1#3)}
\Crefmultiformat{case}{Cases~(#2#1#3)}{ and~(#2#1#3)}{, #2#1#3}{ and~(#2#1#3)}
\Crefrangeformat{case}{Cases~(#3#1#4 to~#5#2#6)}


%% file: macros/citations.tex
\NewDocumentCommand{\citeST}{}{%
    \cite{%
        Bettini2008,%
        Honda2008,%
        Vasconcelos2012,%
        Yoshida2007%
    }%
}

\NewDocumentCommand{\citeStructCongruence}{}{%
    \cite{%
        Honda2008,%
        Milner1999,%
        Pierce2002,%
        Vasconcelos2012,%
        Yoshida2007%
    }%
}



%% file: macros/theory/general.tex
\ProvideDocumentCommand{\l}{}{}
\ProvideDocumentCommand{\Constraints}{}{}

\ProvideDocumentCommand{\t}{}{}
\ProvideDocumentCommand{\Msg}{}{}

\ProvideDocumentCommand{\Nat}{}{}
\ProvideDocumentCommand{\Bool}{}{}
\ProvideDocumentCommand{\String}{}{}
\ProvideDocumentCommand{\None}{}{}

\ProvideDocumentCommand{\FN}{}{}
\ProvideDocumentCommand{\CN}{}{}
\ProvideDocumentCommand{\Dom}{}{}

\NewDocumentCommand{\GeneralTheoryConstraints}{e_ t' G{1} tV e_}{\WrapMathMode{%
\IfValueTF{#1}{%
            \IfBooleanT{#4}{\bigvee\limits_{\IfValueTF{#5}{#5}{i \in I}}}
            {\delta}\IfBooleanT{#2}{\RepeatChar{#3}{'}}_{#1}%
        }{%
            \IfBooleanT{#4}{\bigvee\limits_{\IfValueTF{#5}{#5}{i \in I}}}
            {\delta}\IfBooleanT{#2}{\RepeatChar{#3}{'}}%
        }%
    }%
}

\NewDocumentCommand{\GeneralTheoryl}{e_ t' G{1}}{\WrapMathMode{%
        \IfValueTF{#1}{%
            {l}\IfBooleanT{#2}{\RepeatChar{#3}{'}}_{#1}%
        }{%
            {l}\IfBooleanT{#2}{\RepeatChar{#3}{'}}%
        }%
    }%
}

\NewDocumentCommand{\GeneralTheoryt}{s e_ t' G{1}}{\WrapMathMode{%
        \IfValueTF{#2}{%
            {t}\IfBooleanT{#3}{\RepeatChar{#4}{'}}_{#2}%
        }{%
            {t}\IfBooleanT{#3}{\RepeatChar{#4}{'}}%
        }%
        \IfBooleanT{#1}{\in\SetRealZ}%
    }%
}

\NewDocumentCommand{\GeneralTheoryMsg}{e_ t' G{1}}{\WrapMathMode{%
        \IfValueTF{#1}{%
            {m}\IfBooleanT{#2}{\RepeatChar{#3}{'}}_{#1}%
        }{%
            {m}\IfBooleanT{#2}{\RepeatChar{#3}{'}}%
        }%
    }%
}

\NewDocumentCommand{\GeneralTheoryNat}{}{\WrapMathMode{\mathtt{Nat}}}
\NewDocumentCommand{\GeneralTheoryBool}{}{\WrapMathMode{\mathtt{Bool}}}
\NewDocumentCommand{\GeneralTheoryString}{}{\WrapMathMode{\mathtt{String}}}
\NewDocumentCommand{\GeneralTheoryNone}{}{\WrapMathMode{\mathtt{None}}}

\NewDocumentCommand{\GeneralTheoryFN}{G{\STypeS}}{\WrapMathMode{%
        \mathtt{fn}({#1})%
    }%
}
\NewDocumentCommand{\GeneralTheoryCN}{G{\STypeS}}{\WrapMathMode{%
        \mathtt{cn}({#1})%
    }%
}

\NewDocumentCommand{\GeneralTheoryDom}{G{\TypeCheckSession}}{\WrapMathMode{%
        \text{dom}({#1})%
    }%
}

\NewDocumentCommand{\LoadGeneralTheory}{}{%
    \let\Templ\l%
    \let\TempConstraints\Constraints%
    \let\Tempt\t%
    \let\TempMsg\Msg%
    \let\TempNat\Nat%
    \let\TempBool\Bool%
    \let\TempString\String%
    \let\TempNone\None%
    \let\TempFN\FN%
    \let\TempCN\CN%
    \let\TempDom\Dom%
    \let\l\GeneralTheoryl%
    \let\Constraints\GeneralTheoryConstraints%
    \let\t\GeneralTheoryt%
    \let\Msg\GeneralTheoryMsg%
    \let\Nat\GeneralTheoryNat%
    \let\Bool\GeneralTheoryBool%
    \let\String\GeneralTheoryString%
    \let\None\GeneralTheoryNone%
    \let\FN\GeneralTheoryFN%
    \let\CN\GeneralTheoryCN%
    \let\Dom\GeneralTheoryDom%
}

\NewDocumentCommand{\UnloadGeneralTheory}{}{%
    \let\l\Templ%
    \let\Constraints\TempConstraints%
    \let\t\Tempt%
    \let\Msg\TempMsg%
    \let\Nat\TempNat%
    \let\Bool\TempBool%
    \let\String\TempString%
    \let\None\TempNone%
    \let\FN\TempFN%
    \let\CN\TempCN%
    \let\Dom\TempDom%
}

%% file: macros/theory/sessiontypes.tex
\ProvideDocumentCommand{\Type}{}{}
\ProvideDocumentCommand{\S}{}{}
\ProvideDocumentCommand{\D}{}{}

\ProvideDocumentCommand{\Resets}{}{}

\ProvideDocumentCommand{\C}{}{}
\ProvideDocumentCommand{\DualChoice}{}{}
\ProvideDocumentCommand{\Choice}{}{}
\ProvideDocumentCommand{\Option}{}{}

\ProvideDocumentCommand{\T}{}{}
\ProvideDocumentCommand{\lT}{}{}
\ProvideDocumentCommand{\Del}{}{}
\ProvideDocumentCommand{\End}{}{}

\NewDocumentCommand{\STypeType}{m e_ t' G{1}}{\WrapMathMode{%
        \IfValueTF{#2}{%
            {#1}\IfBooleanT{#3}{\RepeatChar{#4}{'}}_{#2}%
        }{%
            {#1}\IfBooleanT{#3}{\RepeatChar{#4}{'}}%
        }%
    }%
}

\NewDocumentCommand{\STypeS}{e_ t' G{1}}{\WrapMathMode{%
\STypeType{S}_{#1}'{\IfBooleanTF{#2}{#3}{0}}%
}%
}

\NewDocumentCommand{\STypeD}{e_ t' G{1}}{\WrapMathMode{%
\overline{\STypeType{S}_{#1}'{\IfBooleanTF{#2}{#3}{0}}\mathstrut}%
}%
}

\NewDocumentCommand{\STypeResets}{e_ t' G{1}}{\WrapMathMode{%
        \IfValueTF{#1}{%
            {\lambda}\IfBooleanT{#2}{\RepeatChar{#3}{'}}_{#1}%
        }{%
            {\lambda}\IfBooleanT{#2}{\RepeatChar{#3}{'}}%
        }%
    }%
}%

\NewDocumentCommand{\STypeC}{s e_ o t' G{1} t|}{\WrapMathMode{%
\IfBooleanTF{#1}{\mathtt{C}\IfBooleanT{#4}{\RepeatChar{#5}{'}}\IfValueT{#2}{_{#2}}}{
\IfValueTF{#2}{%
\left\{{
        {\mathtt{C}}\IfBooleanT{#4}{\RepeatChar{#5}{'}}_{#2}
    }\right\}_{{#2}\in\IfValueTF{#3}{#3}{\uppercase{#2}}}%
}{
\left\{{%
        {\mathtt{C}}\IfBooleanT{#4}{\RepeatChar{#5}{'}}
    }\right\}%
}%
}%
}%
}

\NewDocumentCommand{\STypeDualChoice}{e_ o}{%
    \WrapMathMode{%
        \begin{array}[t]\{{@{} l @{}}\}%
            {\Dual{\Comm}\IfValueT{#1}{_{#1}}}\,%
            {\GeneralTheoryl\IfValueT{#1}{_{#1}}}%
            {\left\langle{\STypeT\IfValueT{#1}{_{#1}}}\right\rangle}%
            {\left({{\GeneralTheoryConstraints\IfValueT{#1}{_{#1}}},%
            {\STypeResets\IfValueT{#1}{_{#1}}}}\right)}.%
            {\Dual{\S}\IfValueT{#1}{_{#1}}}%
        \end{array}%
        \IfValueT{#1}{_{\lowercase{#1}\in\IfValueTF{#2}{\uppercase{#2}}{\uppercase{#1}}}}%
    }%
}

\NewDocumentCommand{\STypeChoice}{s t| m o g O{\STypeResets} e. e_ o}{\WrapMathMode{%
        \IfBooleanTF{#2}{
            \IfBooleanTF{#1}{
                \IfValueTF{#7}{%
                    \begin{array}[t]{@{}l@{}}%
                        \CSVListToArrayLines{#3}%
                    \end{array}%
                }{%
                    \begin{array}[c]{@{}l@{}}%
                        \CSVListToArrayLines{#3}%
                    \end{array}%
                }%
                \IfValueT{#5}{_{\!\!\textstyle {\vspace{-25pt} #5}}}
            }{
                \IfValueTF{#7}{%
                    \begin{array}[t]\{{@{}l@{}}\}%
                        \CSVListToArrayLines{#3}%
                    \end{array}%
                }{%
                    \begin{array}[c]\{{@{}l@{}}\}%
                        \CSVListToArrayLines{#3}%
                    \end{array}%
                }%
                \IfValueT{#5}{_{\!\!\textstyle {\vspace{-25pt} #5}}}
            }%
            \IfValueT{#8}{_{\smash{\mathrlap{\lowercase{#8}\in\uppercase{#9}}}}}%
        }{
            \IfBooleanTF{#1}{
                \begin{array}[t]{@{}l@{}}%
                    {\Comm\IfValueT{#8}{_{#8}}}\,%
                    {#3\IfValueT{#8}{_{#8}}}%
                    \IfValueT{#4}{\left\langle{#4\IfValueT{#8}{_{#8}}}\right\rangle}%
                    {\left({{\IfValueTF{#5}{#5}{\GeneralTheoryConstraints}\IfValueT{#8}{_{#8}}},%
                    {#6\IfValueT{#8}{_{#8}}}}\right)}.%
                    {\IfValueTF{#7}{#7}{\STypeS}\IfValueT{#8}{_{#8}}}%
                \end{array}%
            }{
                \begin{array}[t]\{{@{} l @{}}\}%
                    {\Comm\IfValueT{#8}{_{#8}}}\,%
                    {#3\IfValueT{#8}{_{#8}}}%
                    \IfValueT{#4}{\left\langle{#4\IfValueT{#8}{_{#8}}}\right\rangle}%
                    {\left({{\IfValueTF{#5}{#5}{\GeneralTheoryConstraints}\IfValueT{#8}{_{#8}}},%
                    {#6\IfValueT{#8}{_{#8}}}}\right)}.%
                    {\IfValueTF{#7}{#7}{\STypeS}\IfValueT{#8}{_{#8}}}%
                \end{array}%
                \IfValueT{#8}{_{\lowercase{#8}\in\IfValueTF{#9}{#9}{\uppercase{#8}}}}%
            }%
        }%
    }%
}

\NewDocumentCommand{\STypeOption}{t! t? t- m o G{\True} o e. e_}{\WrapMathMode{%
\IfBooleanTF{#1}{
    \IfBooleanTF{#2}{
            \Comm\IfValueT{#9}{_{#9}}%
    }{
            \SendSym\IfValueT{#9}{_{#9}}%
    }%
}{
    \IfBooleanTF{#2}{
            \RecvSym\IfValueT{#9}{_{#9}}%
    }{
            \Comm\IfValueT{#9}{_{#9}}%
    }%
}%
{#4\IfValueT{#9}{_{#9}}}%
{\IfValueT{#5}{\left\langle{#5\IfValueT{#9}{_{#9}}}\right\rangle}}%
\left({#6\IfValueT{#9}{_{#9}}},{\IfValueTF{#7}{\IfBooleanTF{#3}{#7}{\left\{#7\right\}}}{\emptyset}\IfValueT{#9}{_{#9}}}\right)%
\IfValueT{#8}{
    .{#8\IfValueT{#9}{_{#9}}}%
}%
}%
}

\NewDocumentCommand{\STypeT}{t< t> e_ t' G{1}}{\WrapMathMode{%
        \IfBooleanT{#1}{\left\langle}%
        \IfValueTF{#3}{%
            {T}\IfBooleanT{#4}{\RepeatChar{#5}{'}}_{#3}%
        }{%
            {T}\IfBooleanT{#4}{\RepeatChar{#5}{'}}%
        }%
        \IfBooleanT{#2}{\right\rangle}%
    }%
}

\NewDocumentCommand{\STypelT}{t< t> e_ t' G{1}}{\WrapMathMode{%
        \IfValueTF{#3}{%
            {\GeneralTheoryl}\IfBooleanT{#4}{\RepeatChar{#5}{'}}_{#3}%
        }{%
            {\GeneralTheoryl}\IfBooleanT{#4}{\RepeatChar{#5}{'}}%
        }%
        \IfBooleanT{#1}{\left\langle}%
        \IfValueTF{#3}{%
            {\STypeT}\IfBooleanT{#4}{\RepeatChar{#5}{'}}_{#3}%
        }{%
            {\STypeT}\IfBooleanT{#4}{\RepeatChar{#5}{'}}%
        }%
        \IfBooleanT{#2}{\right\rangle}%
    }%
}

\NewDocumentCommand{\STypeDel}{g t' G{1}}{\WrapMathMode{%
    \begin{array}[c]({@{}l@{}})%
        \IfValueTF{#1}{%
            #1%
        }{%
            \GeneralTheoryConstraints\IfBooleanT{#2}{\RepeatChar{#3}{'}},~\STypeS\IfBooleanT{#2}{\RepeatChar{#3}{'}}%
        }%
    \end{array}%
    }%
}

\NewDocumentCommand{\STypeEnd}{}{\WrapMathMode{\mathtt{end}}}

\ProvideDocumentCommand{\Val}{}{}
\ProvideDocumentCommand{\Reset}{}{}
\ProvideDocumentCommand{\Que}{}{}
\ProvideDocumentCommand{\Cfg}{}{}

\NewDocumentCommand{\STypeVal}{e_ t' G{1}}{\WrapMathMode{%
        \IfValueTF{#1}{%
            {\nu}\IfBooleanT{#2}{\RepeatChar{#3}{'}}_{#1}%
        }{%
            {\nu}\IfBooleanT{#2}{\RepeatChar{#3}{'}}%
        }%
    }%
}

\NewDocumentCommand{\STypeReset}{G{\STypeResets} t- t> t0 e_}{\WrapMathMode{%
        \left[{{#1\IfValueT{#5}{_{#5}}}\mapsto0}\right]%
    }%
}

\NewDocumentCommand{\STypeQue}{e_ t' G{1}}{\WrapMathMode{%
        \IfValueTF{#1}{%
            \mathbf{\mathtt{M}}\IfBooleanT{#2}{\RepeatChar{#3}{'}}_{#1}%
        }{%
            \mathbf{\mathtt{M}}\IfBooleanT{#2}{\RepeatChar{#3}{'}}%
        }%
    }%
}

\NewDocumentCommand{\STypeCfg}{s G{\STypeVal} E,{\STypeS} t; g o e_ t' G{1}}{\WrapMathMode{%
\IfBooleanTF{#1}{
    \IfValueTF{#7}{
        {\IfBooleanTF{#4}{\mathbf{S}}{\mathbf{s}}}\IfBooleanT{#8}{\RepeatChar{#9}{'}}_{#7}%
    }{
        {\IfBooleanTF{#4}{\mathbf{S}}{\mathbf{s}}}\IfBooleanT{#8}{\RepeatChar{#9}{'}}%
    }%
}{
    (%
        {#2\IfBooleanT{#8}{\RepeatChar{#9}{'}}\IfValueT{#7}{_{#7}},}%
        {~#3\IfValueTF{#7}{\IfBooleanT{#8}{\RepeatChar{#9}{'}}_{#7}}{\IfBooleanT{#8}{\RepeatChar{#9}{'}}}\IfBooleanT{#4}{,}}%
        \IfBooleanT{#4}{
            \IfValueTF{#5}{
                {~#5\IfValueTF{#7}{\IfBooleanT{#8}{\RepeatChar{#9}{'}}_{#7}}{\IfBooleanT{#8}{\RepeatChar{#9}{'}}}}%
            }{
                {~\emptyset}%
            }%
            \IfValueT{#6}{
                ;~{#6}%
            }%
        }%
    )%
}%
}%
}

\NewDocumentCommand{\LoadSType}{s}{%
    \IfBooleanT{#1}{\LoadGeneralTheory}%
    \let\TempType\STypeType%
    \let\TempS\S%
    \let\TempD\D%
    \let\TempResets\Resets%
    \let\TempC\C%
    \let\TempDualChoice\DualChoice%
    \let\TempChoice\Choice%
    \let\TempOption\Option%
    \let\TempT\T%
    \let\TemplT\lT%
    \let\TempDel\Del%
    \let\TempEnd\End%
    \let\TempVal\Val%
    \let\TempReset\Reset%
    \let\TempQue\Que%
    \let\TempCfg\Cfg%
    \let\Type\STypeType%
    \let\S\STypeS%
    \let\D\STypeD%
    \let\Resets\STypeResets%
    \let\C\STypeC%
    \let\Choice\STypeChoice%
    \let\DualChoice\STypeDualChoice%
    \let\Option\STypeOption%
    \let\T\STypeT%
    \let\lT\STypelT%
    \let\Del\STypeDel%
    \let\End\STypeEnd%
    \let\Val\STypeVal%
    \let\Reset\STypeReset%
    \let\Que\STypeQue%
    \let\Cfg\STypeCfg%
}

\NewDocumentCommand{\UnloadSType}{s}{%
    \IfBooleanT{#1}{\UnloadGeneralTheory}%
    \let\Type\TempType%
    \let\S\TempS%
    \let\D\TempD%
    \let\Resets\TempResets%
    \let\C\TempC%
    \let\Choice\TempChoice%
    \let\DualChoice\TempDualChoice%
    \let\Option\TempOption%
    \let\T\TempT%
    \let\lT\TemplT%
    \let\Del\TempDel%
    \let\End\TempEnd%
    \let\Val\TempVal%
    \let\Reset\TempReset%
    \let\Que\TempQue%
    \let\Cfg\TempCfg%
}

\NewDocumentCommand{\SType}{s m}{%
    \LoadGeneralTheory%
    \LoadSType%
    \IfBooleanTF{#1}{\WrapMathMode*{#2}}{\WrapMathMode{#2}}%
    \UnloadSType%
    \UnloadGeneralTheory%
}

\NewDocumentEnvironment{stype}{s o}{
    \LoadGeneralTheory%
    \LoadSType%
    \IfBooleanT{#1}{\begin{equation*}}%
            \IfBooleanF{#1}{\begin{equation}}%
                    }{
                    \IfBooleanF{#1}{\end{equation}}%
            \IfBooleanT{#1}{\end{equation*}}%
    \UnloadSType%
    \UnloadGeneralTheory\!%
}

%% file: macros/theory/processcalculus.tex
\ProvideDocumentCommand{\Prc}{}{}
\ProvideDocumentCommand{\P}{}{}
\ProvideDocumentCommand{\Q}{}{}

\ProvideDocumentCommand{\Role}{}{}
\ProvideDocumentCommand{\p}{}{}
\ProvideDocumentCommand{\q}{}{}
\ProvideDocumentCommand{\b}{}{} 

\ProvideDocumentCommand{\On}{}{}
\ProvideDocumentCommand{\Send}{}{}
\ProvideDocumentCommand{\Recv}{}{}
\ProvideDocumentCommand{\B}{}{}
\ProvideDocumentCommand{\Branch}{}{}
\ProvideDocumentCommand{\After}{}{}
\ProvideDocumentCommand{\v}{}{}
\ProvideDocumentCommand{\lv}{}{}
\ProvideDocumentCommand{\e}{}{}
\ProvideDocumentCommand{\Cond}{}{}
\ProvideDocumentCommand{\Duration}{}{}
\ProvideDocumentCommand{\n}{}{}
\ProvideDocumentCommand{\h}{}{}

\ProvideDocumentCommand{\If}{}{}
\ProvideDocumentCommand{\Then}{}{}
\ProvideDocumentCommand{\Else}{}{}

\ProvideDocumentCommand{\Delay}{}{}
\ProvideDocumentCommand{\Set}{}{}

\ProvideDocumentCommand{\Def}{}{}
\ProvideDocumentCommand{\Var}{}{}
\ProvideDocumentCommand{\Rec}{}{}
\ProvideDocumentCommand{\Scope}{}{}

\ProvideDocumentCommand{\Term}{}{}
\ProvideDocumentCommand{\Err}{}{}
\ProvideDocumentCommand{\Stuck}{}{}

\ProvideDocumentCommand{\Time}{}{}
\ProvideDocumentCommand{\Wait}{}{}
\ProvideDocumentCommand{\NEQ}{}{}

\ProvideDocumentCommand{\WF}{}{}
\ProvideDocumentCommand{\FT}{}{}
\ProvideDocumentCommand{\FN}{}{}
\ProvideDocumentCommand{\FPV}{}{}
\ProvideDocumentCommand{\FQ}{}{}

\ProvideDocumentCommand{\Single}{}{}
\ProvideDocumentCommand{\NumSessions}{}{}

\ProvideDocumentCommand{\pq}{}{}
\ProvideDocumentCommand{\qp}{}{}

\NewDocumentCommand{\PrcCalcPrc}{m e_ t' G{1}}{\WrapMathMode{%
      \IfValueTF{#2}{%
         {#1}\IfBooleanT{#3}{\RepeatChar{#4}{'}}_{#2}%
      }{%
         {#1}\IfBooleanT{#3}{\RepeatChar{#4}{'}}%
      }%
   }%
}

\NewDocumentCommand{\PrcCalcP}{e_ t' G{1}}{%
\PrcCalcPrc{P}_{#1}'{\IfBooleanTF{#2}{#3}{0}}%
}

\NewDocumentCommand{\PrcCalcQ}{e_ t' G{1}}{%
\PrcCalcPrc{Q}_{#1}'{\IfBooleanTF{#2}{#3}{0}}%
}

\NewDocumentCommand{\PrcCalcRole}{m e_ t' G{1}}{\WrapMathMode{%
      \IfValueTF{#2}{%
         \lowercase{#1}\IfBooleanT{#3}{\RepeatChar{#4}{'}}_{#2}%
      }{%
         \lowercase{#1}\IfBooleanT{#3}{\RepeatChar{#4}{'}}%
      }%
   }%
}

\NewDocumentCommand{\PrcCalcp}{e_ t' G{1}}{%
\PrcCalcRole{p}_{#1}'{\IfBooleanTF{#2}{#3}{0}}%
}

\NewDocumentCommand{\PrcCalcq}{e_ t' G{1}}{%
\PrcCalcRole{q}_{#1}'{\IfBooleanTF{#2}{#3}{0}}%
}

\NewDocumentCommand{\PrcCalcpq}{e_ t' G{1}}{\WrapMathMode{%
        \IfValueTF{#1}{%
            {\PrcCalcp\PrcCalcq}\IfBooleanT{#2}{\RepeatChar{#3}{'}}_{#1}%
        }{%
            {\PrcCalcp\PrcCalcq}\IfBooleanT{#2}{\RepeatChar{#3}{'}}%
        }%
    }%
}

\NewDocumentCommand{\PrcCalcqp}{e_ t' G{1}}{\WrapMathMode{%
        \IfValueTF{#1}{%
            {\PrcCalcq\PrcCalcp}\IfBooleanT{#2}{\RepeatChar{#3}{'}}_{#1}%
        }{%
            {\PrcCalcq\PrcCalcp}\IfBooleanT{#2}{\RepeatChar{#3}{'}}%
        }%
    }%
}

\NewDocumentCommand{\PrcCalcb}{e_ t' G{1}}{%
\PrcCalcRole{b}_{#1}'{\IfBooleanTF{#2}{#3}{0}}%
}

\NewDocumentCommand{\PrcCalcOn}{m}{\WrapMathMode{#1}}

\NewDocumentCommand{\PrcCalcSend}{m o}{\WrapMathMode{%
      \,\SendSym*%
      \begin{array}[t]{@{}l@{}}{#1}\IfValueT{#2}{\left(#2\right)}\end{array}%
   }%
}

\NewDocumentCommand{\PrcCalcRecv}{s e^ t| m o t: g e_ o}{\WrapMathMode{%
      \IfValueTF{#2}{^{#2}\,}{~}%
      \RecvSym*%
      \IfBooleanTF{#3}{
         \IfBooleanTF{#1}{
            \begin{array}[t]{@{}l@{}}%
               \CSVListToArrayLines{#4}%
            \end{array}%
         }{
            \begin{array}[t]\{{@{}l@{}}\}%
               \CSVListToArrayLines{#4}%
            \end{array}%
         }%
        \IfValueT{#8}{_{\smash{\mathrlap{\lowercase{#8}\in\IfValueTF{#9}{#9}{\uppercase{#8}}}}}}%
      }{
         \IfBooleanTF{#1}{
            \begin{array}[t]{@{}l@{}}%
               {#4\IfValueT{#8}{_{#8}}}%
               {\IfValueT{#5}{\left({#5\IfValueT{#8}{_{#8}}}\right)}}:%
               \IfValueT{#7}{#7\IfValueT{#8}{_{#8}}}%
            \end{array}%
         }{
            \begin{array}[t]\{{@{} l @{}}\}%
               {#4\IfValueT{#8}{_{#8}}}%
               {\IfValueT{#5}{\left({#5\IfValueT{#8}{_{#8}}}\right)}}:%
               \IfValueT{#7}{#7\IfValueT{#8}{_{#8}}}%
            \end{array}%
            \IfValueT{#8}{_{\lowercase{#8}\in\IfValueTF{#9}{#9}{\uppercase{#8}}}}%
         }%
      }%
   }%
}

\NewDocumentCommand{\PrcCalcB}{s e_ o t' G{1} t|}{\WrapMathMode{%
      \IfBooleanTF{#1}{\mathtt{B}\IfBooleanT{#4}{\RepeatChar{#5}{'}}\IfValueT{#2}{_{#2}}}{
         \IfValueTF{#2}{%
            \left\{{
               {\mathtt{B}}\IfBooleanT{#4}{\RepeatChar{#5}{'}}_{#2}
            }\right\}_{{#2}\in\IfValueTF{#3}{#3}{\uppercase{#2}}}%
         }{
            \left\{{%
               {\mathtt{B}}\IfBooleanT{#4}{\RepeatChar{#5}{'}}
            }\right\}%
         }%
      }%
   }%
}

\NewDocumentCommand{\PrcCalcBranch}{m o e: e_}{\WrapMathMode{%
      {#1\IfValueT{#4}{_{#4}}}%
         {\IfValueT{#2}{\left({#2\IfValueT{#4}{_{#4}}}\right)}}:%
      {#3\IfValueT{#4}{_{#4}}}%
   }%
}

\NewDocumentCommand{\PrcCalcAfter}{t' t< t= G{\infty} e: t*}{\WrapMathMode{%
   \IfBooleanTF{#6}{
      \begin{array}[t]{@{}l@{}}
         \mathtt{after}%
         \\[0.5ex]\mbox{}\quad{#5}%
      \end{array}
   }{
      \mathtt{after}~{#5}%
      }%
   }%
}

\NewDocumentCommand{\PrcCalcv}{s e_ t' G{1}}{\WrapMathMode{%
      \IfValueTF{#2}{%
         {v}\IfBooleanT{#3}{\RepeatChar{#4}{'}}_{#2}%
      }{%
         {v}\IfBooleanT{#3}{\RepeatChar{#4}{'}}%
      }%
   }%
}

\NewDocumentCommand{\PrcCalclv}{s e_ t' G{1}}{\WrapMathMode{%
      \IfValueTF{#2}{%
        {\GeneralTheoryl}\IfBooleanT{#3}{\RepeatChar{#4}{'}}_{#2}%
        {\PrcCalcv}\IfBooleanT{#3}{\RepeatChar{#4}{'}}_{#2}%
      }{%
        {\GeneralTheoryl}\IfBooleanT{#3}{\RepeatChar{#4}{'}}%
        {\PrcCalcv}\IfBooleanT{#3}{\RepeatChar{#4}{'}}%
      }%
   }%
}

\NewDocumentCommand{\PrcCalce}{s e_ t' G{1}}{\WrapMathMode{%
      \IfValueTF{#2}{%
         {e}\IfBooleanT{#3}{\RepeatChar{#4}{'}}_{#2}%
      }{%
         {e}\IfBooleanT{#3}{\RepeatChar{#4}{'}}%
      }%
      \IfBooleanT{#1}{\in\SetRealZ}%
   }%
}

\NewDocumentCommand{\PrcCalcCond}{e_ t' G{1}}{\WrapMathMode{%
      \IfValueTF{#1}{%
         {c}\IfBooleanT{#2}{\RepeatChar{#3}{'}}_{#1}%
      }{%
         {c}\IfBooleanT{#2}{\RepeatChar{#3}{'}}%
      }%
   }%
}

\NewDocumentCommand{\PrcCalcDuration}{e_ t' G{1}}{\WrapMathMode{%
      \IfValueTF{#1}{%
         {d}\IfBooleanT{#2}{\RepeatChar{#3}{'}}_{#1}%
      }{%
         {d}\IfBooleanT{#2}{\RepeatChar{#3}{'}}%
      }%
   }%
}

\NewDocumentCommand{\PrcCalcn}{s e_ t' G{1}}{\WrapMathMode{%
      \IfValueTF{#2}{%
         {n}\IfBooleanT{#3}{\RepeatChar{#4}{'}}_{#2}%
      }{%
         {n}\IfBooleanT{#3}{\RepeatChar{#4}{'}}%
      }%
      \IfBooleanT{#1}{\in\SetRealZ}%
   }%
}

\NewDocumentCommand{\PrcCalch}{s e_ t' G{1}}{\WrapMathMode{%
      \IfValueTF{#2}{%
         {h}\IfBooleanT{#3}{\RepeatChar{#4}{'}}_{#2}%
      }{%
         {h}\IfBooleanT{#3}{\RepeatChar{#4}{'}}%
      }%
   }%
}

\NewDocumentCommand{\PrcCalcIf}{s m}{\WrapMathMode{%
      \mathtt{if}\IfBooleanTF{#1}{~#2}{\left({#2}\right)}%
   }%
}

\NewDocumentCommand{\PrcCalcThen}{t: t* m t* t~}{\WrapMathMode{%
      \IfBooleanTF{#4}{
         \begin{array}[t]{@{\hspace{-4ex}} l @{\hspace{4ex}} l@{}}
            & \mathrlap{\mathtt{then}\IfBooleanT{#1}{:}}
            \\[0.5ex]
            \mathrlap{#3}
         \end{array}
      }{
         \IfBooleanTF{#2}{
            \mathrlap{\mathtt{then}\IfBooleanTF{#1}{:}{~}{#3}}
         }{%
            \mathtt{then}\IfBooleanTF{#1}{:}{~}{#3}
            \IfBooleanT{#5}{~}
         }%
      }%
   }%
}
\NewDocumentCommand{\PrcCalcElse}{t: t* m t*}{\WrapMathMode{
      \IfBooleanTF{#2}{
         \IfBooleanTF{#4}{
            \begin{array}[t]{@{\hspace{-4ex}} l @{\hspace{4ex}} l@{}}
               \\[0.5ex]
               \\[0.5ex]
               \mathrlap{\mathtt{else}\IfBooleanT{#1}{:}}
               \\[0.5ex]
               & {#3}
            \end{array}
         }{
            \begin{array}[t]{@{}l@{}}
               \\[0.5ex]
               \mathtt{else}\IfBooleanTF{#1}{:}{~}{#3}
            \end{array}
         }%
      }{
         \mathtt{else}\IfBooleanTF{#1}{:}{~}{#3}
      }
   }
}

\NewDocumentCommand{\PrcCalcDelay}{g}{\WrapMathMode{%
      \mathtt{delay}\left({\IfValueTF{#1}{#1}{\delta}}\right)%
   }%
}
\NewDocumentCommand{\PrcCalcSet}{m}{\WrapMathMode{%
      \mathtt{set}\,\circled{#1}\,%
   }%
}

\NewDocumentCommand{\PrcCalcDef}{G{\PrcCalcRec*=\PrcCalcP} E:{\PrcCalcQ} t*}{\WrapMathMode{%
      \IfBooleanTF{#3}{
         \begin{array}[t]{@{}l@{}}
            \mathtt{def\ }{#1}~\mathtt{in}:%
            \\[0.5ex]%
            \mbox{}\quad{#2}%
         \end{array}
      }{
         \mathtt{def\ }{#1}~\mathtt{in\ }
         {#2}%
      }
   }%
}

\NewDocumentCommand{\PrcCalcVar}{}{\WrapMathMode{X}}

\NewDocumentCommand{\PrcCalcRec}{s G{\vec{v};\vec{r}} t| G{\vec{T};\FormatRecCheckStyle{\theta};\FormatRecCheckStyle{\Delta}} e_}{\WrapMathMode{%
      \IfBooleanTF{#3}{
         {\PrcCalcVar\!:({#4})}%
      }{
            \IfValueTF{#5}{%
            \PrcCalcVar_{#5}
                }{%
                \PrcCalcVar
                }%
            \IfBooleanTF{#1}{
                ({#2})%
                }{%
                    \langle{#2}\rangle%
                }%
      }%
   }%
}

\NewDocumentCommand{\PrcCalcScope}{m G{\PrcCalcP}}{\WrapMathMode{%
({\nu}{#1})~{#2}%
}%
}

\NewDocumentCommand{\PrcCalcTerm}{}{\WrapMathMode{\texttt{0}}}
\NewDocumentCommand{\PrcCalcErr}{}{\WrapMathMode{\mathtt{error}}}
\NewDocumentCommand{\PrcCalcStuck}{}{\WrapMathMode{\mathtt{stuck}}}

\NewDocumentCommand{\PrcCalcTime}{G{\PrcCalcP} E_{t} t' G{1}}{\WrapMathMode{\Phi_{#2\IfBooleanT{#3}{\RepeatChar{#4}{'}}}\left({#1}\right)}}

\NewDocumentCommand{\PrcCalcWait}{G{\PrcCalcP} t0}{\WrapMathMode{\mathtt{Wait}\left({#1}\right)\IfBooleanT{#2}{=\emptyset}}}

\NewDocumentCommand{\PrcCalcNEQ}{G{\PrcCalcP} t0}{\WrapMathMode{\mathtt{NEQ}\left({#1}\right)\IfBooleanT{#2}{=\emptyset}}}

\NewDocumentCommand{\PrcCalcSingle}{G{\PrcCalcP} t0}{\WrapMathMode{{\substack{\mathtt{Single}\\\mathtt{Session}}}\left({#1}\right)\IfBooleanT{#2}{=\emptyset}}}

\NewDocumentCommand{\PrcCalcNumSessions}{D(){\PrcCalcP} t1}{\WrapMathMode{{\mathtt{Num}}\left({#1}\right)\IfBooleanT{#2}{=1}}}

\NewDocumentCommand{\PrcCalcWF}{G{\PrcCalcP} t1}{\WrapMathMode{\mathtt{wf}\left({#1}\right)\IfBooleanT{#2}{=\MathOpTrue}}}

\NewDocumentCommand{\PrcCalcFT}{G{\PrcCalcP} t1}{\WrapMathMode{\mathtt{ft}\left({#1}\right)\IfBooleanT{#2}{=\emptyset}}}

\NewDocumentCommand{\PrcCalcFN}{G{\PrcCalcP} t1}{\WrapMathMode{\mathtt{fn}\left({#1}\right)\IfBooleanT{#2}{=\emptyset}}}

\NewDocumentCommand{\PrcCalcFPV}{G{\PrcCalcP} t1}{\WrapMathMode{\mathtt{fpv}\left({#1}\right)\IfBooleanT{#2}{=\emptyset}}}

\NewDocumentCommand{\PrcCalcFQ}{G{\PrcCalcP} t0}{\WrapMathMode{\mathtt{fq}\left({#1}\right)\IfBooleanT{#2}{=\emptyset}}}

\ProvideDocumentCommand{\Timers}{}{}
\ProvideDocumentCommand{\PCfg}{}{}

\NewDocumentCommand{\PrcCalcTimers}{e_ t' G{1} d() o}{\WrapMathMode{%
      \IfValueTF{#1}{%
         {\theta}\IfBooleanT{#2}{\RepeatChar{#3}{'}}_{#1}%
      }{%
         {\theta}\IfBooleanT{#2}{\RepeatChar{#3}{'}}%
      }%
      \IfValueT{#4}{
         (#4)%
      }%
      \IfValueT{#5}{
         \left[{#5\mapsto 0}\right]%
      }%
   }%
}

\NewDocumentCommand{\PrcCalcPCfg}{O{\PrcCalcTimers} e_ t' G{1} t, G{\PrcCalcP} e_ t' G{1}}{\WrapMathMode{%
      ({%
         \IfBooleanTF{#5}{
            \IfValueTF{#2}{
               {#1}\IfBooleanT{#3}{\RepeatChar{#4}{'}}_{#2}%
            }{
               {#1}\IfBooleanT{#3}{\RepeatChar{#4}{'}}%
            },\,
            \IfValueTF{#7}{
               {#6}\IfBooleanT{#8}{\RepeatChar{#9}{'}}_{#7}%
            }{
               {#6}\IfBooleanT{#8}{\RepeatChar{#9}{'}}%
            }%
         }{
            \IfValueTF{#7}{
               \IfValueTF{#2}{
                  {#1}\IfBooleanTF{#5}{\IfBooleanT{#3}{\RepeatChar{#4}{'}}}{\IfBooleanT{#8}{\RepeatChar{#9}{'}}}_{#2}%
               }{
                  {#1}\IfBooleanTF{#5}{\IfBooleanT{#3}{\RepeatChar{#4}{'}}}{\IfBooleanT{#8}{\RepeatChar{#9}{'}}}%
               },\,
               \IfValueTF{#7}{
                  {#6}\IfBooleanT{#8}{\RepeatChar{#9}{'}}_{#7}%
               }{
                  {#6}\IfBooleanT{#8}{\RepeatChar{#9}{'}}%
               }%
            }{%
               \IfBooleanTF{#8}{
                  \IfValueTF{#2}{
                     {#1}\IfBooleanTF{#5}{\IfBooleanT{#3}{\RepeatChar{#4}{'}}}{\IfBooleanT{#8}{\RepeatChar{#9}{'}}}_{#2}%
                  }{
                     {#1}\IfBooleanTF{#5}{\IfBooleanT{#3}{\RepeatChar{#4}{'}}}{\IfBooleanT{#8}{\RepeatChar{#9}{'}}}%
                  },\,
                  \IfValueTF{#7}{
                     {#6}\IfBooleanT{#8}{\RepeatChar{#9}{'}}_{#7}%
                  }{
                     {#6}\IfBooleanT{#8}{\RepeatChar{#9}{'}}%
                  }%
               }{
                  \IfValueTF{#2}{
                     {#1}\IfBooleanT{#3}{\RepeatChar{#4}{'}}_{#2},
                     {#6}\IfBooleanT{#3}{\RepeatChar{#4}{'}}_{#2}%
                  }{
                     {#1}\IfBooleanT{#3}{\RepeatChar{#4}{'}},
                     {#6}\IfBooleanT{#3}{\RepeatChar{#4}{'}}%
                  }%
               }%
            }%
         }%
      })%
   }%
}

\NewDocumentCommand{\LoadPrcCalc}{s}{%
   \IfBooleanT{#1}{\LoadGeneralTheory}%
   \let\TempPrc\Prc%
   \let\TempP\P%
   \let\TempQ\Q%
   \let\TempRole\Role%
   \let\Tempp\p%
   \let\Tempq\q%
   \let\Tempb\b%
   \let\TempOn\On%
   \let\TempSend\Send%
   \let\TempRecv\Recv%
   \let\TempB\B%
   \let\TempBranch\Branch%
   \let\TempAfter\After%
   \let\Tempv\v%
   \let\Tempe\e%
   \let\TempCond\Cond%
   \let\TempDuration\Duration%
   \let\Tempn\n%
   \let\Temph\h%
   \let\TempIf\If%
   \let\TempThen\Then%
   \let\TempElse\Else
   \let\TempDelay\Delay%
   \let\TempSet\Set%
   \let\TempDef\Def%
   \let\TempVar\Var%
   \let\TempRec\Rec%
   \let\TempScope\Scope%
   \let\TempTerm\Term%
   \let\TempErr\Err%
   \let\TempStuck\Stuck%
   \let\TempTime\Time%
   \let\TempWait\Wait%
   \let\TempNEQ\NEQ%
   \let\TempWF\WF%
   \let\TempFT\FT%
   \let\TempFN\FN%
   \let\TempFPV\FPV%
   \let\TempFQ\FQ%
   \let\TempSingle\Single%
   \let\TempNumSessions\NumSessions%
   \let\TempTimers\Timers%
   \let\TempPCfg\PCfg%
   \let\Temppq\pq%
   \let\Tempqp\qp%
   \let\Templv\lv%
   \let\Prc\PrcCalcPrc%
   \let\P\PrcCalcP%
   \let\Q\PrcCalcQ%
   \let\Role\PrcCalcRole%
   \let\p\PrcCalcp%
   \let\q\PrcCalcq%
   \let\b\PrcCalcb%
   \let\On\PrcCalcOn%
   \let\Send\PrcCalcSend%
   \let\Recv\PrcCalcRecv%
   \let\B\PrcCalcB%
   \let\Branch\PrcCalcBranch%
   \let\After\PrcCalcAfter%
   \let\v\PrcCalcv%
   \let\e\PrcCalce%
   \let\Cond\PrcCalcCond%
   \let\Duration\PrcCalcDuration%
   \let\n\PrcCalcn%
   \let\h\PrcCalch%
   \let\If\PrcCalcIf%
   \let\Then\PrcCalcThen%
   \let\Else\PrcCalcElse%
   \let\Delay\PrcCalcDelay%
   \let\Set\PrcCalcSet%
   \let\Def\PrcCalcDef%
   \let\Var\PrcCalcVar%
   \let\Rec\PrcCalcRec%
   \let\Scope\PrcCalcScope%
   \let\Term\PrcCalcTerm%
   \let\Err\PrcCalcErr%
   \let\Stuck\PrcCalcStuck%
   \let\Time\PrcCalcTime%
   \let\Wait\PrcCalcWait%
   \let\NEQ\PrcCalcNEQ%
   \let\WF\PrcCalcWF%
   \let\FT\PrcCalcFT%
   \let\FN\PrcCalcFN%
   \let\FPV\PrcCalcFPV%
   \let\FQ\PrcCalcFQ%
   \let\Single\PrcCalcSingle%
   \let\NumSessions\PrcCalcNumSessions%
   \let\Timers\PrcCalcTimers%
   \let\PCfg\PrcCalcPCfg%
   \let\pq\PrcCalcpq%
   \let\qp\PrcCalcqp%
   \let\lv\PrcCalclv%
}

\NewDocumentCommand{\UnloadPrcCalc}{s}{%
    \IfBooleanT{#1}{\UnloadGeneralTheory}%
    \let\Prc\TempPrc%
    \let\P\TempP%
    \let\Q\TempQ%
    \let\Role\TempRole%
    \let\p\Tempp%
    \let\q\Tempq%
    \let\b\Tempb%
    \let\On\TempOn%
    \let\Send\TempSend%
    \let\Recv\TempRecv%
    \let\B\TempB%
    \let\Branch\TempBranch%
    \let\After\TempAfter%
    \let\v\Tempv%
    \let\e\Tempe%
    \let\Cond\TempCond%
    \let\Duration\TempDuration%
    \let\n\Tempn%
    \let\h\Temph%
    \let\If\TempIf%
    \let\Then\TempThen%
    \let\Else\TempElse%
    \let\Delay\TempDelay%
    \let\Set\TempSet%
    \let\Def\TempDef%
    \let\Var\TempVar%
    \let\Rec\TempRec%
    \let\Scope\TempScope%
    \let\Term\TempTerm%
    \let\Err\TempErr%
    \let\Stuck\TempStuck%
    \let\Time\TempTime%
    \let\Wait\TempWait%
    \let\NEQ\TempNEQ%
    \let\WF\TempWF%
    \let\FT\TempFT%
    \let\FN\TempFN%
    \let\FPV\TempFPV%
    \let\FQ\TempFQ%
    \let\Single\TempSingle%
    \let\NumSessions\TempNumSessions%
    \let\Timers\TempTimers%
    \let\PCfg\TempPCfg%
    \let\pq\Temppq%
    \let\qp\Tempqp%
    \let\lv\Templv%
}

\NewDocumentCommand{\PCalc}{s m}{%
   \LoadGeneralTheory%
   \LoadPrcCalc%
   \IfBooleanTF{#1}{\WrapMathMode*{#2}}{\WrapMathMode{#2}}%
   \UnloadPrcCalc%
   \UnloadGeneralTheory%
}

\NewDocumentEnvironment{pcalc}{s}{
   \LoadGeneralTheory%
   \LoadPrcCalc%
   \IfBooleanT{#1}{\begin{equation*}}%
         \IfBooleanF{#1}{\begin{equation}}%
               }{
               \IfBooleanF{#1}{\end{equation}}%
         \IfBooleanT{#1}{\end{equation*}}%
   \UnloadPrcCalc%
   \UnloadGeneralTheory\!%
}

%% file: macros/theory/actions.tex
\NewDocumentCommand{\SendSym}{s}{\WrapMathMode{%
        {\IfBooleanTF{#1}{\PrcSend}{\TypeSend}\mathstrut}%
    }%
}

\NewDocumentCommand{\RecvSym}{s}{\WrapMathMode{%
        {\IfBooleanTF{#1}{\PrcRecv}{\TypeRecv}\mathstrut}%
    }%
}

\NewDocumentCommand{\Trans}{s t| e:}{\WrapMathMode{%
        \IfValueTF{#3}{
            \IfBooleanTF{#2}{
                \,{{\xnrightarrow{\CSVListToSpacedLabels{#3}}}}\,%
            }{
                \,{\xrightarrow{\CSVListToSpacedLabels{#3}}}\,%
            }%
        }{
            \IfBooleanTF{#2}{
                {\not\rightarrow}%
            }{
                {\rightarrow}%
            }%
        }%
        \IfBooleanT{#1}{^{*}}
    }%
}

\NewDocumentCommand{\Reduce}{t- t~ t| t* e_ e^}{\WrapMathMode{%
    \,%
        \IfBooleanTF{#1}{
            \IfBooleanT{#3}{\not}\rightharpoonup%
        }{
            \IfBooleanTF{#2}{
                \IfBooleanT{#3}{\not}\rightsquigarrow%
            }{
                \IfBooleanT{#3}{\not}\longrightarrow%
            }%
        }%
        \IfBooleanTF{#4}{
            ^{*}\IfValueT{#5}{_{\hspace{-10pt}\mathstrut{#5}}}%
        }{
            \IfValueTF{#6}{
                ^{#6}\IfValueT{#5}{_{\hspace{-10pt}\mathstrut{#5}}}%
            }{%
                ^{\mbox{}}\IfValueT{#5}{_{\hspace{-10pt}\mathstrut{#5}}}%
            }%
        }%
    \,%
    }%
}

%% file: macros/theory/type_checking.tex
\ProvideDocumentCommand{\Session}{}{}
\ProvideDocumentCommand{\Gam}{}{}

\ProvideDocumentCommand{\Ses}{}{}

\NewDocumentCommand{\TypeCheckSession}{e_ t' G{1}}{\WrapMathMode{%
        \IfValueTF{#1}{%
            {\Delta}\IfBooleanT{#2}{\RepeatChar{#3}{'}}_{#1}%
        }{%
            {\Delta}\IfBooleanT{#2}{\RepeatChar{#3}{'}}%
        }%
    }%
}

\NewDocumentCommand{\TypeCheckGam}{e_ t' G{1} d() o}{\WrapMathMode{%
        \IfValueTF{#1}{%
        {\Gamma}\IfBooleanT{#2}{\RepeatChar{#3}{'}}_{#1}%
        }{%
        {\Gamma}\IfBooleanT{#2}{\RepeatChar{#3}{'}}%
        }%
        \IfValueT{#4}{
        (#4)%
        }%
        \IfValueT{#5}{
        \left[{#5\mapsto 0}\right]%
        }%
    }%
}

\NewDocumentCommand{\TypeCheckSes}{e_ t' G{1}}{\WrapMathMode{%
        \IfValueTF{#1}{%
            {\mathtt{S}}\IfBooleanT{#2}{\RepeatChar{#3}{'}}_{#1}%
        }{%
            {\mathtt{S}}\IfBooleanT{#2}{\RepeatChar{#3}{'}}%
        }%
    }%
}

\NewDocumentCommand{\LoadTypeChecking}{}{%
    \LoadGeneralTheory%
    \LoadSType%
    \LoadPrcCalc%
    \let\TempSession\Session%
    \let\TempGam\Gam%
    \let\TempSes\Ses%
    \let\Session\TypeCheckSession%
    \let\Gam\TypeCheckGam%
    \let\Ses\TypeCheckSes%
}

\NewDocumentCommand{\UnloadTypeChecking}{}{%
    \UnloadGeneralTheory%
    \UnloadSType%
    \UnloadPrcCalc%
    \let\Session\TempSession%
    \let\Gam\TempGam%
    \let\Ses\TempSes%
}

\NewDocumentCommand{\TypeCheck}{s m}{%
    \LoadTypeChecking%
    \IfBooleanTF{#1}{\WrapMathMode*{#2}}{\WrapMathMode{#2}}%
    \UnloadTypeChecking%
}

\NewDocumentEnvironment{typecheck}{s}{
    \LoadTypeChecking%
    \IfBooleanT{#1}{\begin{equation*}}%
            \IfBooleanF{#1}{\begin{equation}}%
                    }{
                    \IfBooleanF{#1}{\end{equation}}%
            \IfBooleanT{#1}{\end{equation*}}%
    \UnloadTypeChecking%
}

%% file: figs/type_semantics.tex
\NewDocumentCommand{\TypeSemanticsAct}{}{\WrapMathMode{%
\infer[\Rule{act}]{
\Cfg,{\Choice{\l}[\T]_i}%
\Trans:{\Comm_j\Msg}%
\Cfg{\Val\Reset_j},{\S_j}%
}{
\Val\models\delta_j%
\quad%
\Msg={\lT<>_j}%
\quad%
j\in I%
}%
}%
}

\NewDocumentCommand{\TypeSemanticsUnfold}{}{\WrapMathMode{%
    \infer[\Rule{unfold}]{
      \Cfg,{\mu\alpha.\S}%
      \Trans:{\ell}%
      \Cfg'%
    }{
      \Cfg,{\S\Subst{{\mu\alpha.\S}}{\alpha}}%
      \Trans:{\ell}%
      \Cfg'%
    }%
  }%
}

\NewDocumentCommand{\TypeSemanticsTick}{}{\WrapMathMode{%
    \Cfg\Trans:{\t}\Cfg{\Val+\t}%
    \quad%
    \Rule{tick}%
  }%
}

\NewDocumentCommand{\TypeSemanticsSend}{}{\WrapMathMode{%
\infer[\Rule{send}]{
\Cfg;{\Que}%
\Trans:{!\Msg}%
\Cfg{\Val'},{\S'};{\Que}%
}{
\Cfg%
\Trans:{!\Msg}%
\Cfg'%
}%
}%
}

\NewDocumentCommand{\TypeSemanticsRecv}{}{\WrapMathMode{%
\infer[\Rule{recv}]{
\Cfg;{\Msg;\,\Que}%
\Trans:{\tau}%
\Cfg{\Val'},{\S'};{\Que}%
}{
\Cfg%
\Trans:{?\Msg}%
\Cfg'%
}%
}%
}

\NewDocumentCommand{\TypeSemanticsTime}{}{\WrapMathMode{%
\infer[\Rule{time}]{
\Cfg;{\Que}%
\Trans:{\t}%
\Cfg{\Val'},{\S'};{\Que}%
}{
\begin{array}[t]{@{}c @{\quad} l@{}}
  %
  \Cfg\Trans:{\t}\Cfg{\Val'},{\S'}%
   & \text{(configuration)} %
  %
  \\[\ArrayLTSPremiseLineSpacing]%
  (\Cfg~\text{is \fe*})%
  \implies%
  (\Cfg{\Val'},{\S'}~\text{is \fe*})%
   & \text{(persistency)}   %
  %
  \\[\ArrayLTSPremiseTallLineSpacing]%
  \hphantom{\forall \t'<\t:\Cfg{\Val+\t'};{\Que}\Trans|:{\tau}}%
  \mathllap{\smash{\forall \t'<\t:\Cfg{\Val+\t'};{\Que}\Trans|:{\tau}}}%
   & \text{(urgency)}       %
\end{array}
}%
}%
}

\NewDocumentCommand{\TypeSemanticsQue}{}{\WrapMathMode{%
\Cfg;{\Que}\Trans:{?\Msg}\Cfg;{\Que;\,\Msg}%
\quad%
\Rule{que}%
}%
}

\NewDocumentCommand{\TypeSemanticsComL}{}{\WrapMathMode{%
\infer[\Rule{com}[l]*]{
{\Cfg*;{\Que}_1\mid\Cfg*;{\Que}_2}%
\Trans:{\tau}%
{\Cfg*;{\Que}_1'\mid\Cfg*;{\Que}_2'}%
}{
\Cfg*;{\Que}_1\Trans:{!\Msg}\Cfg*;{\Que}_1'%
\quad%
\Cfg*;{\Que}_2\Trans:{?\Msg}\Cfg*;{\Que}_2'%
}%
}%
}

\NewDocumentCommand{\TypeSemanticsComR}{}{\WrapMathMode{%
\infer[\Rule{com}[r]*]{
{\Cfg*;{\Que}_1\mid\Cfg*;{\Que}_2}%
\Trans:{\tau}%
{\Cfg*;{\Que}_1'\mid\Cfg*;{\Que}_2'}%
}{
\Cfg*;{\Que}_1\Trans:{?\Msg}\Cfg*;{\Que}_1'%
\quad%
\Cfg*;{\Que}_2\Trans:{!\Msg}\Cfg*;{\Que}_2'%
}%
}%
}

\NewDocumentCommand{\TypeSemanticsParL}{}{\WrapMathMode{%
\infer[\Rule{par}[l]*]{
{\Cfg*;{\Que}_1\mid\Cfg*;{\Que}_2}%
\Trans:{\tau}%
{\Cfg*;{\Que}_1'\mid\Cfg*;{\Que}_2}%
}{
\Cfg*;{\Que}_1\Trans:{\tau}\Cfg*;{\Que}_1'%
}%
}%
}

\NewDocumentCommand{\TypeSemanticsParR}{}{\WrapMathMode{%
\infer[\Rule{par}[r]*]{
{\Cfg*;{\Que}_1\mid\Cfg*;{\Que}_2}%
\Trans:{\tau}%
{\Cfg*;{\Que}_1\mid\Cfg*;{\Que}_2'}%
}{
\Cfg*;{\Que}_2\Trans:{\tau}\Cfg*;{\Que}_2'%
}%
}%
}

\NewDocumentCommand{\TypeSemanticsWait}{}{\WrapMathMode{%
\infer[\Rule{wait}]{
{\Cfg*;{\Que}_1\mid\Cfg*;{\Que}_2}%
\Trans:{\t}%
{\Cfg*;{\Que}_1'\mid\Cfg*;{\Que}_2'}%
}{
\Cfg*;{\Que}_1\Trans:{\t}\Cfg*;{\Que}_1'%
\quad%
\Cfg*;{\Que}_2\Trans:{\t}\Cfg*;{\Que}_2'%
}%
}%
}

%% file: figs/type_wellformedness_rules.tex
\NewDocumentCommand{\TypeWFEnd}{}{\WrapMathMode{%
    \infer[\Rule{end}]{
      \RecTypeEnv;~\True\Entails\End
    }{
    }
  }
}

\NewDocumentCommand{\TypeWFRec}{}{\WrapMathMode{%
    \infer[\Rule{rec}]{
      \RecTypeEnv;~\Constraints\Entails\mu\alpha.\S
    }{
      \RecTypeEnv,{\alpha:\Constraints};~\Constraints\Entails\S
    }
  }
}

\NewDocumentCommand{\TypeWFVar}{}{\WrapMathMode{%
    \infer[\Rule{var}]{
      \RecTypeEnv,{\alpha:\Constraints};~\Constraints\Entails\alpha
    }{
    }
  }
}

\NewDocumentCommand{\TypeWFChoice}{}{\WrapMathMode{%
\infer[\Rule{choice}]{
\RecTypeEnv;~\Past_iV_{i\in I}
\Entails
\Choice{\l}[\T]{\Constraints}[\Resets].{\S}_i
}{%
\begin{array}{@{}c @{\hspace{2ex}} l@{}}
  \forall i\in I:
  \RecTypeEnv;~\FutureEnv_i
  \Entails
  \S_i
  ~\land~
  \Constraints_i\Reset_i\models\FutureEnv_i
   & \text{(feasibility)}
  \\[\ArrayLTSPremiseLineSpacing]
  \forall i,j\in I:
  i \neq j
  \implies
  \Constraints_i\land\Constraints_j\models\False
  ~\vee~
  \Comm_i=\Comm_j
   & \text{(mixed-choice)}
  \\[\ArrayLTSPremiseLineSpacing]
  \forall i\in I:
  \T_i=\Del'
  \implies
  \emptyset;~\FutureEnv'\Entails\S'
  ~\land~
  \Constraints'\models\FutureEnv'
   & \text{(delegation)}
\end{array}
}
}
}

%% file: figs/process_reduction.tex

\NewDocumentCommand{\ProcessReductionReset}{s}{\WrapMathMode{%
        \PCfg[\Timers,\circled{x}:n],{\Set{x}.\P}%
        \!\Reduce-\!%
        \PCfg[\Timers,\circled{x}:0],{\P}%
        \IfBooleanTF{#1}{\hspace{0.75ex}\Rule{Reset}}{\smash{\mathrlap{\hspace{0.75ex}\Rule{Reset}}}\hspace{1.5ex}}
    }%
}

\NewDocumentCommand{\ProcessReductionIfT}{s}{\WrapMathMode{%
        \infer[\IfBooleanTF{#1}{\Rule{If}[T]*}{\smash{\mathrlap{\Rule{If}[T]*}}\hspace{4ex}}]{
            \PCfg,{\If*{\Cond}~\Then{\P}~\Else{\Q}}%
            \Reduce-%
            \PCfg,{\P}%
        }{
            \Timers\models\Cond%
        }%
    }%
}

\NewDocumentCommand{\ProcessReductionIfF}{s}{\WrapMathMode{%
        \infer[\IfBooleanTF{#1}{\Rule{If}[F]*}{\smash{\mathrlap{\Rule{If}[F]*}}\hspace{4ex}}]{
            \PCfg,{\If*{\Cond}~\Then{\P}~\Else{\Q}}%
            \Reduce-%
            \PCfg,{\P}%
        }{
            \Timers\not\models\Cond%
        }%
    }%
}

\NewDocumentCommand{\ProcessReductionRecv}{}{\WrapMathMode{%
        \infer[\Rule{Recv}]{
            \PCfg,{{\On{\p}\Recv^{\e}{\l}[\v]:{\P}_i}\mid{\qp:\lv\cdot\h}}%
            \Reduce-%
            \PCfg,{{\P_j\Subst{\v}{\v_j}}\mid{\qp:\h}}%
        }{
            j\in I%
            \quad%
            \l=\l_j%
        }%
    }%
}

\NewDocumentCommand{\ProcessReductionRecvT}{}{\WrapMathMode{%
        \infer[\Rule{Recv}[T]*]{
            \PCfg,{{\On{\p}\Recv^{\Diam\n}{\l}[\v]:{\P}_i~\After:{\Q}}\mid{\qp:\lv\cdot\h}}%
            \Reduce-%
            \PCfg,{{\P_j\Subst{\v}{\v_j}}\mid{\qp:\h}}%
        }{
            j\in I%
            \quad%
            \l=\l_j%
        }%
    }%
}

\NewDocumentCommand{\ProcessReductionSend}{}{\WrapMathMode{%
\PCfg,{{\On{\p}\Send{\l}[\v].{\P}}\mid{\pq:\h}}%
\Reduce-%
\PCfg,{{\P}\mid{\pq:\h\cdot\lv}}%
\quad \Rule{Send}%
}%
}

\NewDocumentCommand{\ProcessReductionDet}{}{\WrapMathMode{%
        \infer[\Rule{Det}]{
            \PCfg,{\Delay{\Duration}.\P}%
            \Reduce-%
            \PCfg,{\Delay{\t'}.\P}%
        }{
            \t'\models\Duration
            \Subst{\t'}{\t}%
            %
            %
        }%
    }%
}

\NewDocumentCommand{\ProcessReductionDelay}{s}{\WrapMathMode{%
        \PCfg%
        \Reduce~%
        \PCfg[\Timers+\t],{\Time{\P}}%
        \IfBooleanTF{#1}{\hspace{0.75ex} \Rule{Delay}}{\smash{\mathrlap{\hspace{0.75ex} \Rule{Delay}}}\hspace{4ex}}%
    }%
}

\NewDocumentCommand{\ProcessReductionScope}{s}{\WrapMathMode{%
        \infer[\IfBooleanTF{#1}{\Rule{Scope}}{\smash{\mathrlap{\Rule{Scope}}}
        }]{
            \PCfg,{\Scope{\pq}{\P\IfBooleanT{#1}{''}}}%
            \Reduce-%
            \PCfg',{\Scope{\pq}{\P\IfBooleanTF{#1}{'''}{'}}}%
        }{
            \PCfg,{\P\IfBooleanT{#1}{''}}%
            \Reduce-%
            \PCfg',{\P\IfBooleanTF{#1}{'''}{'}}%
        }%
    }%
}

\NewDocumentCommand{\ProcessReductionParL}{s}{\WrapMathMode{%
        \infer[\IfBooleanTF{#1}{\Rule{Par}[L]*}{\smash{\mathrlap{\Rule{Par}[L]*}}\hspace{5ex}}]{
            \PCfg[\Timers_1,\Timers_2],{\P\IfBooleanT{#1}{''}\mid\Q}%
            \Reduce-%
            \PCfg[\Timers_1',\Timers_2],{\P\IfBooleanTF{#1}{'''}{'}\mid\Q}%
        }{
            \PCfg[\Timers_1],{\P\IfBooleanT{#1}{''}}%
            \Reduce-%
            \PCfg[\Timers_1'],{\P\IfBooleanTF{#1}{'''}{'}}%
        }%
    }%
}

\NewDocumentCommand{\ProcessReductionParR}{s}{\WrapMathMode{%
        \infer[\IfBooleanTF{#1}{\Rule{Par}[R]*}{\smash{\mathrlap{\Rule{Par}[R]*}}
        }]{
            \PCfg[\Timers_1,\Timers_2],{\P\mid\Q\IfBooleanT{#1}{''}}%
            \Reduce-%
            \PCfg[\Timers_1,\Timers_2'],{\P\mid\Q\IfBooleanTF{#1}{'''}{'}}%
        }{
            \PCfg[\Timers_2],{\Q\IfBooleanT{#1}{''}}%
            \Reduce-%
            \PCfg[\Timers_2'],{\Q\IfBooleanTF{#1}{'''}{'}}%
        }%
    }%
}

\NewDocumentCommand{\ProcessReductionDef}{s}{\WrapMathMode{%
        \infer[\Rule{Def}]{
            \PCfg,{\Def{\Rec*{\vec{v};\vec{r}}=\P}:{\Q}}%
            \Reduce-%
            \PCfg',{\Def{\Rec*{\vec{v};\vec{r}}=\P}:{\Q'}}%
        }{
            \PCfg,{\Q}%
            \Reduce-%
            \PCfg',{\Q}'%
        }%
    }%
}

\NewDocumentCommand{\ProcessReductionCall}{s}{\WrapMathMode{%
        \hspace{-2ex}%
            \PCfg,{\Def{\Rec*{\vec{v'};\vec{r'}}=\P\IfBooleanT{#1}{''}}:{\Rec\mid\Q}}%
            \Reduce-%
                \PCfg,{\Def{\Rec*{\vec{v'};\vec{r'}}=\P\IfBooleanT{#1}{''}}:{{\P\IfBooleanT{#1}{''}\Subst{\vec{v};\vec{r}}{\vec{v'};\vec{r'}}}\mid\Q}}%
        \IfBooleanTF{#1}{\quad \Rule{Call}}{\smash{\mathrlap{\ \ \Rule{Call}}}}%
    }%
}

\NewDocumentCommand{\ProcessReductionStr}{s}{\WrapMathMode{%
        \infer[\IfBooleanTF{#1}{\Rule{Str}}{\smash{\mathrlap{\Rule{Str}}}\hspace{4ex}}]{
            \PCfg\Reduce\PCfg[\Timers'],{\Q}%
        }{
            \P\equiv\P'%
            \quad%
            \PCfg,'\Reduce\PCfg',{\Q}'%
            \quad%
            \Q\equiv\Q'%
        }%
    }%
}

\begin{rfigure}{figProcessReduction}[t]
    \begin{requation}{eqReduceStandard}\label{eq:prc_reduction_standard}
            %
            \begin{array}[c]{@{}l@{}}
                \begin{array}[c]{@{}c @{\hspace{3.25ex}} c@{}}
                    \begin{array}[c]{@{}c@{}}{\ProcessReductionStr}\end{array}
                    & %
                    \begin{array}[c]{@{}c@{}}{\ProcessReductionScope}\end{array}
                \end{array}
                \\[-0.5ex]\\%
                \begin{array}[c]{@{}c @{\hspace{3.75ex}} c@{}}
                    \begin{array}[c]{@{}c@{}}{\ProcessReductionParL}\end{array}
                    & %
                    \begin{array}[c]{@{}c@{}}{\ProcessReductionParR}\end{array}
                \end{array}
            \end{array}
            %
    \end{requation}
    \FigHRule%
    \begin{requation}{eqReduceComm}\label{eq:prc_reduction_communication}
            %
            \begin{array}[c]{@{}c@{}}
                \begin{array}[c]{@{}c@{}}{\ProcessReductionSend}\end{array}
                \\[-0.5ex]\\%
                \begin{array}[c]{@{}c@{}}{\ProcessReductionRecv}\end{array}
                \\[-0.5ex]\\%
                \begin{array}[c]{@{}c@{}}{\ProcessReductionRecvT}\end{array}
            \end{array}
            %
    \end{requation}
    \FigHRule%
    \begin{requation}{eqReduceRecursion}\label{eq:prc_reduction_recursion}
            %
            \begin{array}[c]{@{}c@{}}
                \begin{array}[c]{@{}c@{}}{\ProcessReductionDef}\end{array}
                \\[-0.5ex]\\%
                \begin{array}[c]{@{}c@{}}{\ProcessReductionCall}\end{array}
            \end{array}
            %
    \end{requation}
    \FigHRule%
    \begin{requation}{eqReduceTime}\label{eq:prc_reduction_time}
            %
        \begin{array}[c]{@{}c@{}}
            \begin{array}[c]{@{}l@{}}
                \begin{array}[c]{@{}c @{\hspace{1.75ex}} c@{}}
                    \begin{array}[c]{@{}c@{}}{\ProcessReductionIfT*}\end{array}
                    & %
                \begin{array}[c]{@{}c@{}}{\ProcessReductionReset}\end{array}
                \end{array}
                \\[-0.5ex]\\%
                \begin{array}[c]{@{}c @{\hspace{3.25ex}} c@{}}
                \begin{array}[c]{@{}c@{}}{\ProcessReductionIfF*}\end{array}
                & %
                \begin{array}[c]{@{}c@{}}{\ProcessReductionDelay}\end{array}
                \end{array}
            \end{array}
                \\[-0.5ex]\\%
                    \begin{array}[c]{@{}c@{}}{\ProcessReductionDet}\end{array}
        \end{array}
            %
    \end{requation}
    \caption{Rules of process reduction}\label{fig:prc_reduction}
\end{rfigure}

%% file: figs/process_func_wait_neq.tex
\begin{rfigure}{figFuncWaitNEQ}[tb]
  \footnotesize\begin{pcalc}*
    \begin{array}[c]{@{}l@{}}
      \begin{array}[c]{@{}l@{\ }c@{\ }l@{}}
        \Wait & = & \begin{cases}
                      \{\p\}
                       & \text{if\ } \P \in \begin{array}[t]\{{@{}l@{}}\}
                                  \On{\p}\Recv^{\Diam\n}{\l}[\v]:{\P}_i~\After:{\Q},%
                                  ~
                                  \On{\p}\Recv^{\e}{\l}[\v]:{\P}_i
                                \end{array}
                      \\[0.5ex]
                      \Wait{\Q}\setminus\{\p,\q\}
                       & \text{if\ } \P=\Scope{\pq}{\Q}
                      \\[0.5ex]
                      \Wait{\Q}
                       & \text{if\ } \P=\Def{\Rec*{\vec{v};\vec{r}}=\P'}:{\Q}
                      \\[0.5ex]
                      \Wait{\P'}\cup\Wait{\Q}
                       & \text{if\ } \P=\P'\mid\Q
                      \\[0.5ex]
                      \emptyset
                       & \text{otherwise}
                    \end{cases}
        \\\\
        \NEQ  & = & \begin{cases}
                      \{\p\}
                       & \text{if\ } \P=\qp:\h\land\h\neq\emptyset
                      \\[0.5ex]
                      \NEQ{\Q}\setminus\{\p,\q\}
                       & \text{if\ } \P=\Scope{\pq}{\Q}
                      \\[0.5ex]
                      \NEQ{\Q}
                       & \text{if\ } \P=\Def{\Rec*{\vec{v};\vec{r}}=\P'}:{\Q}
                      \\[0.5ex]
                      \NEQ{\P'}\cup\NEQ{\Q}
                       & \text{if\ } \P=\P'\mid\Q
                      \\[0.5ex]
                      \emptyset
                       & \text{otherwise}
                    \end{cases}
      \end{array}
    \end{array}
  \end{pcalc}
  \caption{Definition of \PCalc{\Wait}\ and \PCalc{\NEQ}.}
  \label{fig:func_wait_neq}
\end{rfigure}

%% file: figs/message_throttling.tex
\begin{figure}[b]
  \centering%
  \begin{tikzpicture}[scale=0.215]
\tikzstyle{every node}+=[inner sep=0pt]
\draw [black,fill={rgb:black,1;white,10}] (9,-19.4) circle (3);
\draw (9,-19.4) node {$S_0$};
\draw [black,fill={rgb:black,1;white,10}] (30.1,-19.4) circle (3);
\draw (30.1,-19.4) node {$S_1$};
\draw [black,fill={rgb:black,1;white,10}] (51,-19.4) circle (3);
\draw (51,-19.4) node {$S_2$};
\draw [black,fill={rgb:black,1;white,10}] (72.3,-19.4) circle (3);
\draw [black,fill=white] (72.3,-19.4) circle (2.4);
\draw (72.3,-19.4) node {$\mathtt{end}$};
\draw [black] (11.049,-17.217) arc (130.27609:49.72391:13.15);
\fill [black] (28.05,-17.22) -- (27.76,-16.32) -- (27.12,-17.08);
\draw (19.55,-14.19) node [above] {$!\mathtt{msg}\mbox{ }(x\geq3,\mbox{ }\{x\})$};
\draw [black] (32.343,-17.416) arc (125.41888:54.58112:14.161);
\fill [black] (48.76,-17.42) -- (48.39,-16.55) -- (47.82,-17.36);
\draw (40.55,-14.3) node [above] {$!\mathtt{msg}\mbox{ }(x\geq3,\mbox{ }\{x\})$};
\draw [black] (53.283,-17.462) arc (124.51268:55.48732:14.767);
\fill [black] (70.02,-17.46) -- (69.64,-16.6) -- (69.07,-17.42);
\draw (61.65,-14.36) node [above] {$!\mathtt{tout}\mbox{ }(x\geq3,\mbox{ }\emptyset)$};
\draw [black] (27.75,-21.257) arc (-57.33551:-122.66449:15.193);
\fill [black] (11.35,-21.26) -- (11.75,-22.11) -- (12.29,-21.27);
\draw (19.55,-24.16) node [below] {$?\mathtt{ack}\mbox{ }(x<3,\mbox{ }\{x\})$};
\draw [black] (48.701,-21.319) arc (-56.04512:-123.95488:14.593);
\fill [black] (32.4,-21.32) -- (32.78,-22.18) -- (33.34,-21.35);
\draw (40.55,-24.31) node [below] {$?\mathtt{ack}\mbox{ }(x<3,\mbox{ }\{x\})$};
\end{tikzpicture}
  \caption{Message throttling protocol for $m=2$.}
  \label{fig:message_throttling}
\end{figure}

%% file: figs/process_type_checking.tex
\NewDocumentCommand{\ProcessTypeCheckingEnd}{}{\WrapMathMode{%
    \Gam;\Timers\Entails\Term\TypedBy\emptyset%
    \quad \Rule{End}%
  }%
}

\NewDocumentCommand{\ProcessTypeCheckingWeak}{}{\WrapMathMode{%
\infer[\Rule{Weak}]{
\Gam;\Timers\Entails{\P}\TypedBy\Session,~{\p:\Cfg,{\End}}%
}{
\Gam;\Timers\Entails{\P}\TypedBy\Session%
}%
}%
}

\NewDocumentCommand{\ProcessTypeCheckingEmpty}{}{\WrapMathMode{%
    \Gam;\Timers\Entails{\pq:\emptyset}\TypedBy{\pq:\emptyset}%
    \quad \Rule{Empty}%
  }%
}

\NewDocumentCommand{\ProcessTypeCheckingPar}{}{\WrapMathMode{%
    \infer[\Rule{Par}]{
      \Gam;\Timers_1,\Timers_2\Entails{\P\mid\Q}\TypedBy\Session_1,\Session_2%
    }{
      \Gam;\Timers_1\Entails{\P}\TypedBy\Session_1%
      \quad%
      \Gam;\Timers_2\Entails{\Q}\TypedBy\Session_2%
    }%
  }%
}

\NewDocumentCommand{\ProcessTypeCheckingRes}{}{\WrapMathMode{%
    \infer[\Rule{Res}]{
      \Gam;\Timers\Entails{\Scope{\pq}{\P}}\TypedBy\Session%
    }{
      \begin{array}[c]{c}
        \Gam;\Timers\Entails{\P}\TypedBy\Session,%
        ~{\p:\Cfg_1},%
        ~{\qp:\Que_1},%
        ~{\q:\Cfg_2},%
        ~{\pq:\Que_2}%
        \\[\ArrayTallLineSpacing]
        \Cfg;{\Que}_1\Compat\Cfg;{\Que}_2%
        \qquad%
        \forall i \in \{1,2\}.\, \S_i \text{\ well-formed\ against\ }\nu_i
      \end{array}
    }%
  }%
}

\NewDocumentCommand{\ProcessTypeCheckingVQue}{}{\WrapMathMode{%
\infer[\Rule{VQue}]{
\Gam;\Timers\Entails{\qp:{\lv}\cdot\h}\TypedBy\Session,%
~{\qp:{\lT<>};\Que}%
}{
\T~\text{\bt}%
\quad%
\Gam\Entails{\v:\T}%
\quad%
\Gam;\Timers\Entails{\qp:{\h}}\TypedBy\Session,%
~{\qp:{\Que}}%
}%
}%
}

\NewDocumentCommand{\ProcessTypeCheckingDQue}{}{\WrapMathMode{%
\infer[\Rule{DQue}]{
\Gam;\Timers\Entails{\qp:{\l\q}\cdot\h}\TypedBy\Session,%
~{\qp:{\lT<>};\Que},%
~{\q:\Cfg}%
}{
\T=\Del%
\quad%
\Val\models\Constraints%
\quad%
\Gam;\Timers\Entails{\qp:{\h}}\TypedBy\Session,%
~{\qp:{\Que}}%
}%
}%
}

\NewDocumentCommand{\ProcessTypeCheckingTimer}{}{\WrapMathMode{%
    \infer[\Rule{Timer}]{
      \Gam;\Timers,\circled{x}:n\Entails{\Set{x}.\P}\TypedBy\Session%
    }{
      \Gam;\Timers,\circled{x}:0\Entails{\P}\TypedBy\Session%
    }%
  }%
}

\NewDocumentCommand{\ProcessTypeCheckingDelDelta}{s}{\WrapMathMode{%
    \infer[\IfBooleanTF{#1}{\Rule{Del}[\Duration]}{\smash{\mathrlap{\Rule{Del}[\Duration]}}\hspace{4ex}}]{
      \Gam;\Timers\Entails{\Delay{\Duration}.\P}\TypedBy\Session%
    }{
      \forall \t\in\Duration:%
      \Gam;\Timers\Entails{\Delay{\t}.\P}\TypedBy\Session%
    }%
  }%
}

\NewDocumentCommand{\ProcessTypeCheckingDelT}{}{\WrapMathMode{%
    \infer[\Rule{Del}[t]]{
      \Gam;\Timers\Entails{\Delay{\t}.\P}\TypedBy\Session%
    }{
      \Gam;\Timers+\t\Entails{\P}\TypedBy\Session+\t%
      \quad%
      \Session~\text{not \t-reading}
    }%
  }%
}

\NewDocumentCommand{\ProcessTypeCheckingBranch}{s}{\WrapMathMode{%
\infer[\IfBooleanTF{#1}{\Rule{Branch}}{\smash{\mathrlap{\Rule{Branch}}}\hspace{4ex}}]{
\Gam;\Timers\Entails%
\On{\p}\Recv^{\e}{\l}[\v]:{\P}_{i}%
\TypedBy\Session,~{\p:\Cfg,{\Choice{\l}[\T]_{j}}}%
}{
\begin{array}[t]{@{}c@{}}%
  \neg (\left\lvert J \right\rvert = \left\lvert I \right\rvert = 1)
  \quad
  \begin{array}[t]{@{}l@{}}%
    \forall j\in J: %
    (\Val\models\Constraints_j)%
    \implies%
    (\Comm_j = ?)%
    ~\land~%
    \exists i\in I:
    (\l_i=\l_j)%
    ~\land~\mbox{}%
    \\[\ArrayTallLineSpacing]%
    \mbox{}\qquad
    \Gam;\Timers\Entails{\On{\p}\Recv*^{\e}{\l}[\v]:{\P}_{i}}%
    \TypedBy\Session,~{\p:\Cfg,{\Choice*{\l}[\T]_{j}}}%
  \end{array}%
\end{array}%
}%
}%
}

\NewDocumentCommand{\ProcessTypeCheckingTimeout}{s}{\WrapMathMode{%
\infer[\IfBooleanTF{#1}{\Rule{Timeout}}{\smash{\mathrlap{\Rule{Timeout}}}\hspace{4ex}}]{
\Gam;\Timers\Entails%
\On{\p}\Recv^{\Diam\n}{\l}[\v]:{\P}_{i}%
~\After{\n}:{\Q}%
\TypedBy\Session,~{\p:\Cfg,{\C_j}}%
}{
\begin{array}[t]{@{}c@{}}%
  \begin{array}[t]{@{}l@{}}%
    \Gam;\Timers\Entails%
    \On{\p}\Recv^{\Diam\n}{\l}[\v]:{\P}_{i}%
    \TypedBy\Session,~{\p:\Cfg,{\C_{j}}}%
  \end{array}%
  \\[\ArrayTallLineSpacing]
  \begin{array}[t]{@{}l@{}}%
    \Gam;\Timers+\n\Entails{\Q}%
    \TypedBy\Session+\n,~{\p:\Cfg{\Val+\n},{\C_{j}}}%
  \end{array}%
\end{array}%
}%
}%
}

\NewDocumentCommand{\ProcessTypeCheckingVRecv}{}{\WrapMathMode{%
\infer[\Rule{VRecv}]{
\Gam;\Timers\Entails{\On{\p}\Recv*^{\e}{\l}[\v]:{\P}}
\TypedBy\Session,~{\p:\Cfg,{\Option?-{\l}[\T]{\delta}[\lambda].{\S}}}%
}{
\begin{array}[t]{@{}c@{}}%
  \T~\text{\bt}%
  \quad%
  \Session~\text{not\ $\e$-reading}
  \quad
  {%
    \forall\t:
    (\Val+\t\models\Constraints)
    \Bimplies
    (\t\in\e)
  }
  \\[\ArrayTallLineSpacing]
  \forall\t\in\e: %
  \Gam,{\v:\T};\Timers+\t\Entails{\P}\TypedBy\Session+\t,%
  ~{\p:\Cfg{\Val+\t\Reset},{\S}}%
\end{array}%
}%
}%
}

\NewDocumentCommand{\ProcessTypeCheckingDRecv}{}{\WrapMathMode{%
\infer[\Rule{DRecv}]{
\Gam;\Timers\Entails{\On{\p}\Recv*^{\e}{\l}[\q]:{\P}}
\TypedBy\Session,~{\p:\Cfg,{\Option?-{\l}[\T]{\delta}[\lambda].{\S}}}%
}{
\begin{array}[t]{@{}c@{}}%
  \T=\Del'%
  \quad%
  \Val'\models\Constraints'%
  \quad%
  \Session~\text{not\ $\e$-reading}
  \quad
  {\forall\t:
    (\Val+\t\models\Constraints)
    \Bimplies
    (\t\in\e)
  }
  \\[\ArrayTallLineSpacing]
  \forall\t\in\e: %
  \Gam;\Timers+\t\Entails{\P}\TypedBy\Session+\t,%
  ~{\p:\Cfg{\Val+\t\Reset},{\S}},%
  ~{\q:\Cfg'}%
\end{array}%
}%
}%
}

\NewDocumentCommand{\ProcessTypeCheckingIfTrue}{}{\WrapMathMode{%
    \infer[\Rule{IfTrue}]{
      \Gam;\Timers\Entails{\If*{\Cond}~\Then{\P}~\Else{\Q}}\TypedBy\Session%
    }{
      \Timers\models\Cond%
      \quad%
      \Gam;\Timers\Entails{\P}\TypedBy\Session%
    }%
  }%
}

\NewDocumentCommand{\ProcessTypeCheckingIfFalse}{}{\WrapMathMode{%
    \infer[\Rule{IfFalse}]{
      \Gam;\Timers\Entails{\If*{\Cond}~\Then{\P}~\Else{\Q}}\TypedBy\Session%
    }{
      \Timers\not\models\Cond%
      \quad%
      \Gam;\Timers\Entails{\Q}\TypedBy\Session%
    }%
  }%
}

\NewDocumentCommand{\ProcessTypeCheckingVSend}{s}{\WrapMathMode{%
\infer[\IfBooleanTF{#1}{\Rule{VSend}}{\smash{\mathrlap{\Rule{VSend}}}\hspace{0.95ex}}]{
\Gam;\Timers\Entails{\On{\p}\Send{\l}[\v].{\P}}\TypedBy\Session,
~{\p:\Cfg,{\Choice{\l}[\T]_i}}%
}{
\begin{array}[t]{@{}c@{}}
  \exists i\in I: %
  (\l=\l_i)
  \land
  (\Val\models\Constraints_i)%
  \land%
  (\T_i~\text{\bt})%
  \land%
  (\Gam\Entails\v:\T_i)
  \land
  \ifInlineTypeCheck\else
  \mbox{} \\[\ArrayTallLineSpacing]%
  \fi
  (\Comm_i=!)
  \land
  \Gam;\Timers\Entails{\P}\TypedBy\Session,{\p:\Cfg{\Val\Reset_i},{\S_i}}%
\end{array}%
}%
}%
}

\NewDocumentCommand{\ProcessTypeCheckingDSend}{s}{\WrapMathMode{%
\infer[\IfBooleanTF{#1}{\Rule{DSend}}{\smash{\mathrlap{\Rule{DSend}}}\hspace{0.95ex}}]{
\Gam;\Timers\Entails{\On{\p}\Send{\l}[\b].{\P}}\TypedBy\Session,%
~{\p:\Cfg,{\Choice{\l}[\T]_i}},%
~{\b:\Cfg'}%
}{
\begin{array}[t]{@{}c@{}}%
  \exists i\in I: %
  (\l=\l_i)
  \land
  (\Val\models\Constraints_i)%
  \land%
  (\T_i=\Del')%
  \land%
  (\Val'\models\Constraints')%
  \land
  \ifInlineTypeCheck\else
  \mbox{} \\[\ArrayTallLineSpacing]%
  \fi
  (\Comm_i=!)
  \land
  \Gam;\Timers\Entails{\P}\TypedBy\Session,{\p:\Cfg{\Val\Reset_i},{\S_i}}%
\end{array}%
}%
}%
}

\NewDocumentCommand{\ProcessTypeCheckingVar}{}{\WrapMathMode{%
    \infer[\Rule{Var}]{
      \Gam,{\Rec|};\Timers\Entails{\Rec}\TypedBy
      \vec{r}:(\vec {\Val},\vec {\S})%
    }{
      (\vec {\Val},\vec {\S})\in\FormatRecCheckStyle{\Delta}%
      \quad%
      \Timers\in\FormatRecCheckStyle{\theta}%
      \quad%
      \forall i: \Gam\Entails \vec{v}_i : \vec{T}_i%
    }%
  }%
}

\NewDocumentCommand{\ProcessTypeCheckingRec}{}{\WrapMathMode{%
\infer[\Rule{Rec}]{
\Gam;\Timers\Entails{\Def{\Rec*{\vec{v};\vec{r}}=\P}:{\Q}}\TypedBy\Session%
}{
\begin{array}[t]{@{}c@{}}%
  \forall (\vec{\Val},~\vec{\S})
  \in\FormatRecCheckStyle{\Delta},
  {\theta'}\in\FormatRecCheckStyle{\theta}: %
  \Gam,
  {\vec{v}:\vec{T}},%
  {\Rec|}%
  ;\Timers'
  \Entails{\P}
  \TypedBy
  \vec{r}: (\vec{\Val},~\vec{\S})%
  \\[\ArrayLineSpacing]
  \Gam,{\Rec|};\Timers\Entails{\Q}\TypedBy\Session%
\end{array}%
}%
}%
}

%% file: figs/process_type_checking_fig_standard.tex
\begin{rfigure}{figTypeCheckStandard}[t]
  \begin{typecheck}\label{eq:type_checking_standard}
    \begin{array}[c]{@{}c@{}}
      \begin{array}{@{}c @{\qquad} c@{}}
        \begin{array}[c]{@{}c@{}}{\ProcessTypeCheckingEnd}\end{array}
         &
        \begin{array}[c]{@{}c@{}}{\ProcessTypeCheckingEmpty}\end{array}
      \end{array}
      \\[-0.5ex]\\%
      \begin{array}{@{}c @{\qquad} c@{}}
        \begin{array}[c]{@{}c@{}}{\ProcessTypeCheckingWeak}\end{array}
         &
        \begin{array}[c]{@{}c@{}}{\ProcessTypeCheckingPar}\end{array}
      \end{array}
      \\[-0.5ex]\\%
      \begin{array}[c]{@{}c@{}}{\ProcessTypeCheckingRes}\end{array}
      \\[-0.5ex]\\%
      \begin{array}[c]{@{}c@{}}{\ProcessTypeCheckingRec}\end{array}
      \\[-0.5ex]\\%
      \begin{array}[c]{@{}c@{}}{\ProcessTypeCheckingVar}\end{array}
    \end{array}
  \end{typecheck}
  \FigHRule%
  \begin{typecheck}\label{eq:type_checking_time_sensitive}
    \begin{array}[c]{@{}c@{}}
      \begin{array}[c]{@{}l@{}}
        \begin{array}{@{}c @{\hspace{1.25ex}} c@{}}
          \begin{array}[c]{@{}c@{}}{\ProcessTypeCheckingIfTrue}\end{array}
           &
          \begin{array}[c]{@{}c@{}}{\ProcessTypeCheckingTimer}\end{array}
        \end{array}
        \\[-0.5ex]\\%
        \begin{array}{@{}c @{\hspace{1.25ex}} c@{}}
          \begin{array}[c]{@{}c@{}}{\ProcessTypeCheckingIfFalse}\end{array}
           &
          \begin{array}[c]{@{}c@{}}{\ProcessTypeCheckingDelDelta}\end{array}
        \end{array}
      \end{array}
      \\[-0.5ex]\\%
      \begin{array}[c]{@{}c@{}}{\ProcessTypeCheckingDelT}\end{array}
    \end{array}
  \end{typecheck}
  \caption{Standard \& time-sensitive type checking rules}\label{fig:type_checking}
\end{rfigure}

%% file: figs/process_type_checking_fig_communication.tex
\begin{rfigure}{figTypeCheckComm}[t]
  \begin{typecheck}\label{eq:type_checking_send}
    \begin{array}[c]{@{}c@{}}
      \begin{array}[c]{@{}c@{}}{\ProcessTypeCheckingVSend}\end{array}
      \\[-0.5ex]\\%
      \begin{array}[c]{@{}c@{}}{\ProcessTypeCheckingDSend}\end{array}
    \end{array}
  \end{typecheck}
  \FigHRule%
  \begin{typecheck}\label{eq:type_checking_recv}
    \begin{array}[c]{@{}c@{}}
      \begin{array}[c]{@{}c@{}}{\ProcessTypeCheckingBranch}\end{array}
      \\[-0.5ex]\\%
      \begin{array}[c]{@{}c@{}}{\ProcessTypeCheckingTimeout}\end{array}
    \end{array}
  \end{typecheck}
  \FigHRule%
  \begin{typecheck}\label{eq:type_checking_recv_single}
    \begin{array}[c]{@{}c@{}}
      \begin{array}[c]{@{}c@{}}{\ProcessTypeCheckingVRecv}\end{array}
      \\[-0.5ex]\\%
      \begin{array}[c]{@{}c@{}}{\ProcessTypeCheckingDRecv}\end{array}
    \end{array}
  \end{typecheck}
  \FigHRule%
  \begin{typecheck}\label{eq:type_checking_queues}
    \begin{array}[c]{@{}c@{}}
      \begin{array}[c]{@{}c@{}}{\ProcessTypeCheckingVQue}\end{array}
      \\[-0.5ex]\\%
      \begin{array}[c]{@{}c@{}}{\ProcessTypeCheckingDQue}\end{array}
    \end{array}
  \end{typecheck}
  \caption{Rules for type-checking communicating processes against TOAST}\label{fig:type_checking_communication}
\end{rfigure}

%% file: figs/process_func_fq.tex
\begin{rdefi}{defFuncFQ}[Free-queues of a Process]\label{def:func_fq}
  The function \FQ{\P}\ returns the set of \emph{free-queues} in a given process \P, defined inductively below:
  \small\[
    \begin{array}[t]{@{}l@{}}
      \FQ{\P}
      = \begin{cases}
          \{\pq\}                                                 &
          \text{if\ }\P=\pq:\h
          \\[\ArrayLineSpacing]
          \FQ{\P'}\setminus\{\pq,\qp\}                            &
          \text{if\ }\P=\Scope{\pq}{\P'}
          \\[\ArrayExtraTallLineSpacing]
          \FQ{\P'}                                                &
          \text{if\ }\P\in\begin{array}[c]\{{@{}l@{}}\}
                          \On{\p}\Send{\l}[\v]{\P'},
                          ~\Def{\Rec*{\vec{v};\vec{r}}=\P}:{\P'},
                          \\
                          \Set{x}.\P',
                          ~\Delay{\Duration}.\P',
                          ~\Delay{\t}.\P'
                        \end{array}
          \\[\ArrayExtraTallLineSpacing]
          \bigcup_{i\in I} \FQ{\P_i}                              &
          \text{if\ }\P =
          \On{\p}\Recv^{\e}{\l}[\v]:{\P}_i
          \\[\ArrayExtraTallLineSpacing]
          \FQ{\P'} \cup \FQ{\Q}                                   &
          \begin{array}[t]{@{}l@{}}
          \text{if\ }\P\in\begin{array}[c]\{{@{}l@{}}\}
                            {\P'\mid\Q},
                            ~\If*{\Cond}~\Then{\P'}~\Else{\Q}
                          \end{array}
        \end{array}
          \\[\ArrayExtraTallLineSpacing]
          \FQ{\On{\p}\Recv^{\Diam\n}{\l}[\v]:{\P}_i} \cup \FQ{\Q} &
          \text{if\ }\P =
          \On{\p}\Recv^{\Diam\n}{\l}[\v]:{\P}_i~\After:{\Q}
          \\[\ArrayExtraTallLineSpacing]
          \emptyset                                               &
          \text{if\ }\P =
          \Term,
          ~\Rec
        \end{cases}
    \end{array}
  \]
\end{rdefi}

%% file: figs/process_func_ft.tex
\begin{rdefi}{defFuncFT}[Free-timers of a Process]\label{def:func_ft}
  The function \FT{\P}\ returns the set of \emph{free timers} in a given process \P, defined inductively below:
  \small\[
    \begin{array}[t]{@{}l@{}}
      \FT{\P}
      = \begin{cases}
          \{\circled{x}\} \cup \FT{\P'}                           &
          \text{if\ }\P=\Set{x}.{\P'}
          \\[\ArrayExtraTallLineSpacing]
          \{\circled{x}\} \cup \FT{\P'} \cup \FT{\Q}              &
          \text{if\ }\P=\If*{\Cond}~\Then{\P'}~\Else{\Q}
          \text{\ and\ }
          {\circled{x}\in\Cond}
          %
          \\[\ArrayExtraTallLineSpacing]
          \FT{\P'}                                                &
          \text{if\ }\P\in\begin{array}[c]\{{@{}l@{}}\}
                          \On{\p}\Send{\l}[\v]{\P'},
                          ~\Def{\Rec*{\vec{v};\vec{r}}=\P}:{\P'},
                          \\
                          \Delay{\Duration}.\P',
                          ~\Delay{\t}.\P',
                          ~\Scope{\pq}{\P'}
                        \end{array}
          \\[\ArrayExtraTallLineSpacing]
          \bigcup_{i\in I} \FT{\P_i}                              &
          \text{if\ }\P =
          \On{\p}\Recv^{\e}{\l}[\v]:{\P}_i
          \\[\ArrayExtraTallLineSpacing]
          \FT{\P'} \cup \FT{\Q}                                   &
          \text{if\ }\P={\P'\mid\Q}
          \\[\ArrayExtraTallLineSpacing]
          \FT{\On{\p}\Recv^{\Diam\n}{\l}[\v]:{\P}_i} \cup \FT{\Q} &
          \text{if\ }\P =
          \On{\p}\Recv^{\Diam\n}{\l}[\v]:{\P}_i~\After:{\Q}
          \\[\ArrayExtraTallLineSpacing]
          \emptyset                                               &
          \text{if\ }\P =
          \Term,
          ~\Rec,
          ~{\pq:\h}
        \end{cases}
    \end{array}
  \]
\end{rdefi}

%% file: figs/process_func_fn.tex
\begin{rdefi}{defFuncFN}[Free-names of a Process]\label{def:func_fn}
  The function \FN{\P}\ returns the set of \emph{free-names} in a given process \P, defined inductively below:
  \small\[
    \begin{array}[t]{@{}l@{}}
      \FN{\P}
      = \begin{cases}
          \{\pq\} \cup \FN{\P'}                                   &
          \text{if\ }\P=\Scope{\pq}{\P'}
          \\[\ArrayExtraTallLineSpacing]
          \FN{\P'}                                                &
          \text{if\ }\P\in\begin{array}[c]\{{@{}l@{}}\}
                          \On{\p}\Send{\l}[\v]{\P'},
                          ~\Def{\Rec*{\vec{v};\vec{r}}=\P}:{\P'},
                          \\
                          \Set{x}.\P',
                          ~\Delay{\Duration}.\P',
                          ~\Delay{\t}.\P'
                        \end{array}
          \\[\ArrayExtraTallLineSpacing]
          \bigcup_{i\in I} \FN{\P_i}                              &
          \text{if\ }\P =
          \On{\p}\Recv^{\e}{\l}[\v]:{\P}_i
          \\[\ArrayExtraTallLineSpacing]
          \FN{\P'} \cup \FN{\Q}                                   &
          \begin{array}[t]{@{}l@{}}
          \text{if\ }\P\in\begin{array}[c]\{{@{}l@{}}\}
                            {\P'\mid\Q},
                            ~\If*{\Cond}~\Then{\P'}~\Else{\Q}
                          \end{array}
        \end{array}
          \\[\ArrayExtraTallLineSpacing]
          \FN{\On{\p}\Recv^{\Diam\n}{\l}[\v]:{\P}_i} \cup \FN{\Q} &
          \text{if\ }\P =
          \On{\p}\Recv^{\Diam\n}{\l}[\v]:{\P}_i~\After:{\Q}
          \\[\ArrayExtraTallLineSpacing]
          \emptyset                                               &
          \text{if\ }\P =
          \Term,
          ~\Rec,
          ~\pq:\h
        \end{cases}
    \end{array}
  \]
\end{rdefi}

%% file: figs/process_func_fpv.tex
\begin{rdefi}{defFuncFPV}[Free process variables of a Process]\label{def:func_fpv}
  The function \FPV{\P}\ returns the set of \emph{free process variables} in a given process \P, defined inductively below:
  \small\[
    \begin{array}[t]{@{}l@{}}
      \FPV{\P}
      = \begin{cases}
          \{\Var\}                                                  &
          \text{if\ }\P=\Rec
          \\[\ArrayLineSpacing]
          \{\Var\} \cup \FPV{\P'}                                   &
          \text{if\ }\P=\Def{\Rec*{\vec{v};\vec{r}}=\P}:{\P'}
          \\[\ArrayExtraTallLineSpacing]
          \FPV{\P'}                                                 &
          \text{if\ }\P\in\begin{array}[c]\{{@{}l@{}}\}
                          \On{\p}\Send{\l}[\v]{\P'},
                          ~\Scope{\pq}{\P'},
                          ~\Set{x}.\P',
                          \\
                          \Delay{\Duration}.\P',
                          ~\Delay{\t}.\P'
                        \end{array}
          \\[\ArrayExtraTallLineSpacing]
          \bigcup_{i\in I} \FPV{\P_i}                               &
          \text{if\ }\P =
          \On{\p}\Recv^{\e}{\l}[\v]:{\P}_i
          \\[\ArrayExtraTallLineSpacing]
          \FPV{\P'} \cup \FPV{\Q}                                   &
          \begin{array}[t]{@{}l@{}}
          \text{if\ }\P\in\begin{array}[c]\{{@{}l@{}}\}
                            {\P'\mid\Q},
                            ~\If*{\Cond}~\Then{\P'}~\Else{\Q}
                          \end{array}
        \end{array}
          \\[\ArrayExtraTallLineSpacing]
          \FPV{\On{\p}\Recv^{\Diam\n}{\l}[\v]:{\P}_i} \cup \FPV{\Q} &
          \text{if\ }\P =
          \On{\p}\Recv^{\Diam\n}{\l}[\v]:{\P}_i~\After:{\Q}
          \\[\ArrayExtraTallLineSpacing]
          \emptyset                                                 &
          \text{if\ }\P =
          \Term,
          ~\pq:\h
        \end{cases}
    \end{array}
  \]
\end{rdefi}

%% file: figs/session_reduction.tex
\begin{figure}[h!]
  \begin{typecheck}*
    \begin{array}[t]{@{}c@{}}
      \begin{array}[t]{c @{\quad} c}
        \infer[\Rule[\Session]{L}]{
          \Session_1,~\Session_2
          \Reduce
          \Session_1',~\Session_2
        }{
          \Session_1
          \Reduce
          \Session_1'
        }
         & %
        \infer[\Rule[\Session]{Recv}]{
        {\p:\Cfg},
        ~{\qp:{\lT<>;\Que}}
        \Reduce
        {\p:\Cfg'},
        ~{\qp:\Que}
        }{
        \Cfg;{\lT<>;\Que}
        \Trans:{\tau}
        \Cfg{\Val'},{\S'};{\Que}
        }
      \end{array}
      \\[-0.5ex]\\%
      \begin{array}[t]{c @{\quad} c}
        \infer[\Rule[\Session]{R}]{
          \Session_1,~\Session_2
          \Reduce
          \Session_1,~\Session_2'
        }{
          \Session_2
          \Reduce
          \Session_2'
        }
         & %
        \infer[\Rule[\Session]{Send}]{
        {\p:\Cfg},
        ~{\pq:\Que},
        \Reduce
        {\p:\Cfg'},
        ~{\pq:\Que;\Msg}
        }{
        \Cfg
        \Trans:{!\Msg}
        \Cfg'
        }
      \end{array}
    \end{array}
  \end{typecheck}
  \caption{Reduction for Session Environments}\label{fig:session_reduction}
\end{figure}